\documentclass[12pt, twoside]{report}
\usepackage{amssymb}
\usepackage{amsmath}
\usepackage{stmaryrd}
\usepackage{cancel}
\usepackage{graphicx}
\usepackage{psfrag}
\usepackage{tocloft}
\usepackage{amsthm}

\newcommand{\txPlane}{\ensuremath{tx\text{-}\mathsf{Plane}}}
\newcommand{\enc}{\mathsf{Enc}}
\newcommand{\der}{\mathsf{Der}}
 
\newcommand{\conc}{\hat{\;}}
\newcommand{\wl}{\mathsf{wl}} \newcommand{\lc}{\mathsf{lc}}
\newcommand{\m}{\mathsf{m}} \newcommand{\cen}{\mathsf{cen}}
\newcommand{\coll}{\mathsf{coll}}
\newcommand{\inecoll}{\mathsf{inecoll}}
 
\newcommand{\convex}{\mathsf{Conv}}
\newcommand{\concave}{\mathsf{Conc}} \newcommand{\Bw}{\mathsf{Bw}}
\newcommand{\mtwp}{\mathsf{CP}}
\newcommand{\llangle}{\mbox{$\langle\hspace*{-3pt}\langle$}}
\newcommand{\rrangle}{\mbox{$\rangle\hspace*{-3pt}\rangle$}}
\renewcommand{\time}{\mathsf{time}} \newcommand{\spc}{\mathsf{space}}
\newcommand{\loc}{\mathsf{loc}} 
\newcommand{\Loc}{\mathsf{Loc}} 
\newcommand{\com}{\succ}

\newcommand{\parrow}{\xrightarrow{\resizebox{!}{3.5pt}{$\circ$}}}
\newcommand{\Ob}{\mathrm{Ob}} \newcommand{\Ph}{\mathrm{Ph}}
\newcommand{\vs}{\vec{s}}

\newcommand{\vp}{\vec{{p}}}
\newcommand{\vpp}{\vec{p}\,}
\newcommand{\vq}{\vec{q}}
\newcommand{\vqq}{\vec{q}\,}
\newcommand{\vr}{\vec{r}}
\newcommand{\vrr}{\vec{r}\,}
\newcommand{\vx}{\vec{x}}
\newcommand{\vxx}{\vec{x}\,}
\newcommand{\vy}{\vec{y}}
\newcommand{\vyy}{\vec{y}\,}
\newcommand{\va}{\vec{a}}
\newcommand{\vc}{\vec{c}}
\newcommand{\vb}{\vec{b}}
\newcommand{\vo}{\vec{o}}
\newcommand{\voo}{\vec{o}\,}
\newcommand{\vz}{\vec{z}}
\newcommand{\vzz}{\vec{z}\,}
\newcommand{\vehm}{1^{h}_m}
\newcommand{\vekm}{1^{k}_m}
\newcommand{\vet}{\vec{\mathsf{1}}_t}
\newcommand{\vex}{\vec{\mathsf{1}}_x}
\newcommand{\vey}{\vec{\mathsf{1}}_y}
\newcommand{\vez}{\vec{\mathsf{1}}_z}
\newcommand{\vei}{\vec{\mathsf{1}}_i}
\newcommand{\ve}{\vec{\mathsf{1}}_1}
\renewcommand{\vee}{\vec{\mathsf{1}}_2}
\newcommand{\veee}{\vec{\mathsf{1}}_3}
\newcommand{\veeee}{\vec{\mathsf{1}}_4}
\newcommand{\bv}{\mathbf{v}}
\newcommand{\vv}{\vec{v}}
\newcommand{\fvv}{\vec{\mathbf{v}}}

\newcommand{\fvp}{\vec{\mathbf{p}}}
\newcommand{\fvvkm}{\vec{\mathbf{v}}^{\,k}_m}

\newcommand{\fvakm}{\vec{\mathbf{a}}^{\,k}_m}
\newcommand{\fvvh}{\vec{\mathbf{v}}^{\,h}}
\newcommand{\fvah}{\vec{\mathbf{a}}^{\,h}}

\newcommand{\Q}{\mathrm{Q}}
\newcommand{\B}{\mathrm{B}}
\newcommand{\W}{\mathrm{W}}
\newcommand{\M}{\mathrm{M}}
\newcommand{\IB}{\mathrm{IB}}
\newcommand{\DB}{\mathrm{DB}}
\newcommand{\IOb}{\mathrm{IOb}}

\renewcommand{\time}{\mathsf{time}}
\newcommand{\dist}{\mathsf{dist}}
\newcommand{\leteq}{\mbox{$:=$}}
\newcommand{\mort}{\bot_\mu}
\newcommand{\up}{\uparrow}
\newcommand{\upp}{\uparrow\hspace*{-1pt}\uparrow\!}
\newcommand{\scom}{\succ\hspace{-9pt}\prec} 
\newcommand{\pheq}{\,\lambda\,}
\newcommand{\seq}{\,\sigma\,}
\newcommand{\teq}{\,\tau\,}
\newcommand{\simrad}{\thicksim^{rad}}
\newcommand{\simph}{\thicksim^{ph}}
\newcommand{\simmu}{\thicksim^\mu}
\newcommand{\then}{\enskip \rightarrow\ }
\newcommand{\Then}{\enskip \Longrightarrow\ }

\renewcommand{\iff}{\enskip \leftrightarrow\ }
\newcommand{\Iff}{\enskip \Longleftrightarrow\ }
\newcommand{\rship}{\mbox{$>\hspace{-6pt}\big|$}b,k,c\mbox{$\big>_{\!rad}$}}
\newcommand{\mship}{\mbox{$>\hspace{-6pt}\big|$}b,k,c \mbox{$\big>_{\!\!\mu}$}}
\newcommand{\ship}{\mbox{$>\hspace{-6pt}\big|$}b,k,c \big>}
\newcommand{\lland}{\;\land\;}
\newcommand{\llor}{\;\lor\;}
\newcommand{\setclose}{\}}
\newcommand{\Setclose}{\,\right\}}
\newcommand{\setmid}{\::\:}
\newcommand{\setopen}{\{}
\newcommand{\Setopen}{\left\{\,}

\newcommand{\phsum}{\rightthreetimes}

\newcommand{\dom}{Dom\,}
\renewcommand{\d}{\mathit{d}}
\newcommand{\ran}{Ran\,}

\newcommand{\R}{\mathbb{R}}
\newcommand{\N}{\mathbb{N}}
\newcommand{\Z}{\mathbb{Z}}
\renewcommand{\sup}{sup\,}
\renewcommand{\L}{\ensuremath{\mathcal{L}}}
\newcommand{\LL}{\ensuremath{\textsc{l}}}

\usepackage{color} 
\definecolor{thmcolor}{rgb}{0,0,.4} 
\definecolor{remarkcolor}{rgb}{0,.2,0} 
\definecolor{proofcolor}{rgb}{.4,0,0} 
\definecolor{quecolor}{rgb}{.2,.2,0} 
\definecolor{axcolor}{rgb}{.23,0,.23}
\definecolor{axbgcolor}{rgb}{1,.9,1} 
\definecolor{defbgcolor}{rgb}{1,1,0.8} 
\definecolor{thmbgcolor}{rgb}{0.9,0.9,1} 
\definecolor{rmbgcolor}{rgb}{0.9,1,0.9} 
\definecolor{proofbgcolor}{rgb}{1,0.9,0.9}

\renewcommand{\qedsymbol}{\colorbox{proofbgcolor}{\textcolor{proofcolor}{$\blacksquare$}}} 
 
\newcommand{\ax}[1]{\textcolor{axcolor}{\ensuremath{\mathsf{#1}}}} 
\newcommand{\Ax}[1]{\textcolor{axcolor}{\colorbox{axbgcolor}{\ensuremath{\mathsf{#1}}}}} 
\newcommand{\df}[1]{{\bf #1}}

\newcommand{\Df}[1]{\setlength{\fboxsep}{2pt}\colorbox{defbgcolor}{\ensuremath{#1}}\setlength{\fboxsep}{3pt}} 
\newcommand{\Dff}[1]{\setlength{\fboxsep}{0pt}\colorbox{defbgcolor}{\ensuremath{#1}}\setlength{\fboxsep}{3pt}}

\theoremstyle{definition} \newtheorem{thm}{\colorbox{thmbgcolor}{\textcolor{thmcolor}{Theorem}}}[section] 
\theoremstyle{definition} \newtheorem{cor}[thm]{\colorbox{thmbgcolor}{\textcolor{thmcolor}{Corollary}}} 
\theoremstyle{definition} \newtheorem{lem}[thm]{\colorbox{thmbgcolor}{\textcolor{thmcolor}{Lemma}}}
\theoremstyle{definition} \newtheorem{prop}[thm]{\colorbox{thmbgcolor}{\textcolor{thmcolor}{Proposition}}}
 
\theoremstyle{remark} \newtheorem{conv}[thm]{\colorbox{rmbgcolor}{\sc\textcolor{remarkcolor}{Convention}}} 
\theoremstyle{remark} 
\theoremstyle{definition} \newtheorem{que}[thm]{\colorbox{rmbgcolor}{\textcolor{quecolor}{Question}}} 
\theoremstyle{definition} \newtheorem{rem}[thm]{\colorbox{rmbgcolor}{\textcolor{remarkcolor}{Remark}}}
\theoremstyle{definition} \newtheorem{example}[thm]{\colorbox{rmbgcolor}{\textcolor{remarkcolor}{Example}}}

\makeindex

\begin{document}

\pagestyle{plain}
\begin{titlepage}


\centerline{\Large {PhD Thesis}}

\vskip42pt

\centerline{\Large \textbf{First-Order Logic Investigation of Relativity Theory}} 
\vskip4pt
\centerline{\Large \textbf{with an Emphasis on Accelerated Observers}}

\vskip42pt
\centerline{\Large \bf Gergely Sz\'ekely}

\vskip52pt

\centerline{
\begin{tabular}{ll}
\large Advisers:&  \Large Hajnal  Andr{\'e}ka,\\ &\large  head of department, D.Sc.\\ & \\
       &    \Large Judit X. Madar{\'a}sz,\\& \large research fellow, Ph.D.
\end{tabular}}

\vskip42pt

\centerline{\Large Mathematics Doctoral School}
\centerline{\Large Pure Mathematics Program}
\vskip12pt
\centerline{\large School Director: \Large Prof.\ \Large Mikl{\'o}s Laczkovich}
\centerline{\large Program Director: \Large Prof.\ \Large Andr{\'a}s Sz{\H u}cs}


\vfill 
\begin{figure}[h!btp]
\small
\begin{center}
\includegraphics[scale=0.50]{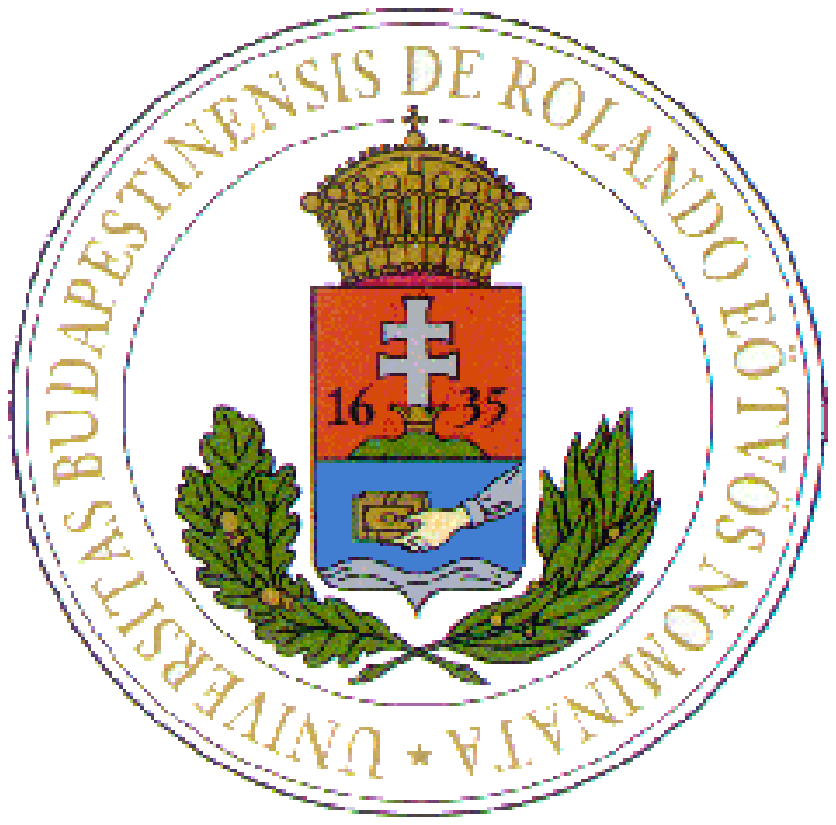}
\end{center}
\end{figure}

\centerline{E{\"o}tv{\"o}s Lor{\'a}nd University}
\centerline{Faculty of Sciences, Institute of Mathematics}
\centerline{2009}

\end{titlepage}

\setcounter{page}{1}
\tableofcontents

\chapter{Introduction}

This work is a continuation of the works by Andr{\'e}ka, Madar{\'a}sz,
N{\'e}meti and their coauthors, e.g., \cite{AAMN}, \cite{pezsgo},
\cite{AMNsamples}, \cite{logst}, \cite{Mphd}.  Our research is
directly related to Hilbert's sixth problem of axiomatization of
physics. Moreover, it goes beyond this problem since its general aim
is not only to axiomatize physical theories but to investigate the
relationship between basic assumptions (axioms) and predictions
(theorems). Another general aim of ours is to provide a foundation for
physics similar to that of mathematics.

For good reasons, the foundation of mathematics was performed strictly
within first-order logic (FOL).\index{FOL} A reason for this fact is that
staying within FOL helps to avoid tacit assumptions. Another reason is
that FOL has a complete inference system while second-order logic (and
thus any higher-order logic) cannot have one, see, e.g., \cite[\S
  IX. 1.6]{ETF}. For further reasons for staying within FOL, see,
e.g., Chap.~\ref{chp-whyfol} and \cite{ax}, \cite[\S Appendix: Why
  FOL?]{pezsgo}, \cite{vaananen}, \cite{wolenski}.

Why is it useful to apply the axiomatic method to relativity theory?
For one thing, this method makes it possible for us to understand the
role of any particular axiom. We can check what happens to our theory
if we drop, weaken or replace an axiom. For instance, it has been
shown by this method that the impossibility of faster than light
motion is not independent from other assumptions of special
relativity, see \cite[\S 3.4]{pezsgo}, \cite{AMNsamples}.  More
boldly: it is superfluous as an axiom because it is provable as a
theorem from much simpler and more convincing basic assumptions.  The
linearity of the transformations between observers (reference frames)
can also be proven from some plausible assumptions, therefore it need
not be assumed as an axiom, see Thm.~\ref{thm-poi} and \cite{pezsgo},
\cite{AMNsamples}. Getting rid of unnecessary axioms of a physical
theory is important because we do not know whether an axiom is true or
not, we just assume so. We can only be sure of experimental facts but
they typically correspond not to axioms but to (preferably
existentially quantified) intended theorems.

Not only can we get rid of superfluous assumptions by applying the
axiomatic method, but we can discover new, interesting and physically
relevant theories.  That happened in the case of the axiom of
parallels in Euclid's geometry; and this kind of investigation led to the
discovery of hyperbolic geometry. Our FOL theory of
accelerated observers ($\ax{AccRel}$), which nicely fills the gap
between special and general relativity theories, is also a good
example of such a theory.

\label{why-questions}
Moreover, if we have an axiom system, we can ask which axioms are
responsible for a certain prediction of our theory. This kind of
reverse thinking helps to answer the why-type questions of
relativity. For example, we can take the twin paradox theorem and
check which axiom of special relativity was and which one was not
needed to derive it. The weaker an axiom system is, the better answer
it offers to the question: ``Why is the twin paradox true?''. The twin
paradox is investigated in this manner in Chap.~\ref{chp-twp}, while
its inertial approximation (called the clock paradox) in
Chap.~\ref{chp-cp}.  For details on answering why-type questions of
relativity by the methodology of the present work, see \cite{wqp}. We
hope that we have given good reasons why we use the axiomatic method
in our research into spacetime theories. For more details and further
reasons, see, e.g., Guts~\cite{guts}, Schutz~\cite{schutz},
Suppes~\cite{suppes}.

This work is structured in the following way: in Chap.~\ref{chp-frame}
we introduce our FOL frame and our basic notation; then, in
Chap.~\ref{chp-sr}, we recall a FOL axiomatization of special
relativity by our research group. Based on this axiomatization first
we investigate the logical connection between the clock paradox
theorem and the axiom system in Chap.~\ref{chp-cp}. First we give a
geometrical characterization theorem for the clock paradox, see
Thm.~\ref{thm-clp}; then we prove some surprising consequences for
both Newtonian and relativistic kinematics. Thm.~\ref{thm-simdist}
answers Question 4.2.17 of Andr\'eka--Madar\'asz--N\'emeti
\cite{pezsgo}. 

In Chap.~\ref{chp-dyn} we extend our geometrical
approach to relativistic dynamics and investigate the relations
between our purely geometrical key axioms of dynamics and the
conservation postulates of the standard approaches. For example, we
show that the conservation postulates are not needed to prove the
relativistic mass increase theorem $m_0=\sqrt{1-v^2/c^2}\cdot m$,
which is the first step to capture Einstein's insight $E=mc^2$. 

In
Chap.~\ref{chp-accrel} we extend the theory introduced in
Chap.~\ref{chp-sr} to accelerated observers by introducing our
aforementioned theory $\ax{AccRel}$, which is the main subject of this
thesis. In Chap.~\ref{chp-twp} we investigate the twin paradox within
\ax{AccRel}; we show that a nontrivial assumption is required if we
want the twin paradox to be a consequence of our theory
$\ax{AccRel}$. 
In Chap.~\ref{chp-grav} we prove two formulations of
the gravitational time dilation from a streamlined and small set of axioms
(\ax{AccRel}), by using Einstein's
equivalence principle.  

In Chap.~\ref{chp-gr} we ``derive'' a
FOL axiom system of general relativity from our theory
\ax{AccRel} in one natural step. The technical parts of the proofs and the development of
the necessary tools are presented in Chap.~\ref{chp-a}. And 
in
Chap.~\ref{chp-whyfol} we go into the details of the reasons for 
choosing FOL in our investigation.

\begin{conv}
Throughout this work, there appear ``highlighted'' statements, such as
\ax{AxCenter^+} in Chap.~\ref{chp-dyn}, which associate the name
\ax{AxCenter^+} with a formula of our FOL language. It
is important to note that these formulas are not automatically
elevated to the rank of axiom. Instead, they serve as potential axioms
or even as potential statements to appear in theorems, hence they are
nothing more than formulas distinguished in our language.
\end{conv}

We try to be as self-contained as possible. First occurrences of
concepts used in this work are set in boldface to make them easier to
find. We also use 
colored text
and boxes to help the reader
to find the axioms, notations, etc. Throughout this work,
if-and-only-if is abbreviated as \df{iff}.\index{iff} We hope that the
Index at the end of this thesis also helps find the individual
definitions and notations introduced.

\section*{ACKNOWLEDGMENTS}
I wish to express my heartfelt thanks to my advisers Hajnal Andr\'eka,
Judit X.\ Madar\'asz and Istv\'an N\'emeti for the invaluable
inspiration and guidance I received from them for my work. I am also
grateful to Mike Stannett for his many helpful comments and
suggestions. My thanks also go to Zal\'an Gyenis, Ram\'on Horv\'ath,
Victor Pambuccian and Adrian Sfarti for our interesting discussions on
the subject. This research is supported by Hungarian National Foundation for Scientific Research
grants No T73601.

\begin{figure}[htp]
\small
\begin{center}
\psfrag{1}[t][t]{Chap.~1} 
\psfrag{2}[t][t]{Chap.~2} 
\psfrag{3}[t][t]{Chap.~3} 
\psfrag{4}[t][t]{Chap.~4} 
\psfrag{5}[t][t]{Chap.~5}
\psfrag{6}[t][t]{Chap.~6} 
\psfrag{7}[t][t]{Chap.~7} 
\psfrag{8}[t][t]{Chap.~9} 
\psfrag{9}[t][t]{Chap.~8} 
\psfrag{10}[t][t]{Chap.~10} 
\psfrag{11}[t][t]{Chap.~11}
\includegraphics[keepaspectratio, width=0.9\textwidth]{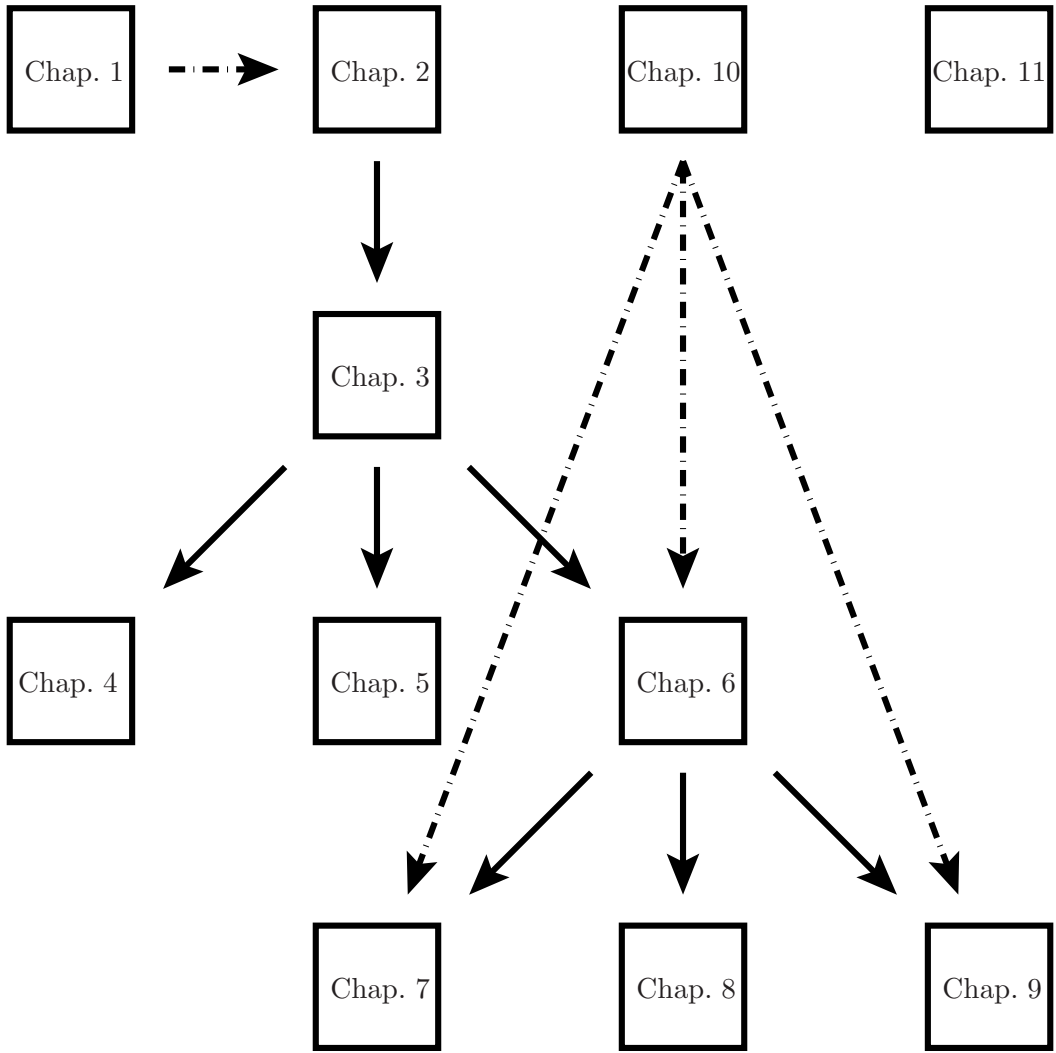}
\caption{\label{fig-secmap} Illustration of the connection of the Chapters}
\end{center}
\end{figure}

\chapter{FOL frame}
\label{chp-frame}

In this chapter we specify the FOL frame
within which we will work.

\section{Frame language}

Here we explain our basic concepts. This thesis mainly deals
with the kinematics of relativity, i.e., with the motion of {\it
  bodies} (things which can move, e.g., test-particles, reference
frames, center-lines). However, we briefly discuss dynamics in
Chap.~\ref{chp-dyn}, and our co-authored papers \cite{dyn},
\cite{dyn-studia} and \cite{msz-wku} are fully devoted to dynamics. We
represent motion as the changing of spatial location in time. Thus we
use reference frames for coordinatizing events (meetings of bodies). {\it
  Quantities} are used for marking time and space.  The structure of
quantities is assumed to be an ordered field in place of the field of
real numbers. For simplicity, we associate reference frames with
certain bodies called {\it observers}. The coordinatization of events
by observers is formulated by means of the \textit{worldview
  relation}. We visualize an observer as ``sitting'' at the origin of
the space part of its reference frame, or equivalently, ``living'' on
the time-axis of the reference frame. We distinguish {\it inertial}
and {\it noninertial} {observers}. For the time being, inertiality is
only a label on observers which will be defined later by our
axioms. We also use another special kind of bodies called {\it
  photons}. We use photons only for labeling light paths, so here we
do not consider any of their quantum dynamical properties.

In an axiomatic approach to relativity, it is more natural to take
bodies instead of events as basic concepts. This choice is not
uncommon in the literature, see, e.g., Ax~\cite{ax},
Benda~\cite{benda}. However, a large variety of choosing basic concepts occur
in the different axiomatizations of special relativity, see, e.g.,
Goldblatt~\cite{goldblatt}, Mundy~\cite{mundy-oaomstg, mundy-tpcomg},
Pambuccian~\cite{pambuccian}, Robb~\cite{Robb}, \cite{Robb2}
Suppes~\cite{suppes-sopitposat}, Schutz~\cite{schutz}, \cite{Schu},
\cite{schutz-aasfmst}.

\newpage
Allowing ordered fields in place of the field of real numbers
increases the flexibility of our theory and minimizes the amount of
our mathematical presuppositions. For further motivation in this
direction, see, e.g., Ax~\cite{ax}. Similar remarks apply to
our flexibility-oriented decisions below, e.g., to
treat the dimension of spacetime as a variable.

There are many reasons for using observers (or coordinate systems, or
reference frames) instead of a single observer-independent spacetime
structure.  One is that it helps to weed unnecessary axioms from our
theories.  Nevertheless, we state and emphasize the logical
equivalence\footnote{By logical equivalence, we mean definitional
  equivalence.} of observer-oriented and observer-independent
approaches to relativity theory, see \cite[\S 3.6]{logst}, \cite[\S
  4.5]{Mphd}.

Keeping the foregoing in mind, let us now set up the FOL
language of our axiom systems. First we fix a natural number
$\Df{d}\ge 2$\index{$d$} for the dimension of spacetime. Our language
contains the following nonlogical symbols:
\begin{itemize}
\item unary relation symbols \Df{\B}\index{$\B$}
  (\df{bodies}\index{bodies}), \Df{\Ob}\index{$\Ob$}
  (\df{observers}\index{observers}), \Df{\IOb}\index{$\IOb$}
  (\df{inertial observers}\index{inertial observers}), \Df{\Ph}
  (\df{photons}\index{photons}) and $\Df{\Q}$\index{$\Q$}
  (\df{quantities}\index{quantities});
\item binary function symbols \Df{+}\index{$+$},
 \Df{\cdot}\index{$\cdot$} and a binary relation symbol
 \Df{<}\index{$<$} (field operations and ordering on
 $\Q$); and
\item a $2+d$-ary relation symbol \Df{\W}\index{$\W$} (\df{worldview
  relation}\index{worldview relation}).
\end{itemize}
$\B(x)$, $\Ob(x)$, $\IOb(x)$, $\Ph(x)$ and $\Q(x)$ are translated as
``$x$ is a body,'' ``$x$ is an observer,'' ``$x$ is an inertial
observer,'' ``$x$ is a photon'' and ``$x$ is a quantity,'' respectively.
We use the worldview relation $\W$ to speak about coordinatization by
translating $\W(x,y,z_1,\ldots, z_d)$ as ``observer $x$ coordinatizes
body $y$ at spacetime location $\langle z_1,\ldots,z_d\rangle$,''
(i.e., at space location $\langle z_2,\ldots,z_d\rangle$ and at instant
$z_1$).

$\B(x)$, $\Ob(x)$, $\IOb(x)$, $\Ph(x)$, $\Q(x)$, $\W(x,y,z_1,\ldots,
z_d)$, $x=y$ and $x<y$ are the so-called {atomic formulas} of our
FOL language, where $x$, $y $, $z_1,\dots,z_d$ can be
arbitrary variables or terms built up from variables by using the
field operations. The \df{formulas}\index{formulas} of our
FOL language are built up from these atomic formulas by using
the logical connectives {\it not} (\Df{\lnot}), {\it and}
(\Df{\land}), {\it or} (\Df{\lor}), {\it implies}
(\Df{\rightarrow}), {\it if-and-only-if}
(\Df{\leftrightarrow}), and the quantifiers {\it exists} $x$
(\Df{\exists x}) and {\it for all $x$} (\Df{\forall x}) for every
variable $x$. To abbreviate formulas of FOL we often
omit parentheses according to the following convention: quantifiers
bind as long as they can, and $\land$ binds stronger than
$\rightarrow$. For example, we write $\forall x\enskip
\varphi\land\psi\rightarrow\exists y\enskip \delta\land\eta$ instead of
$\forall x\big((\varphi\land\psi)\rightarrow\exists
y(\delta\land\eta)\big)$.

 We use the notation $\Df{\Q^n}\leteq \Q\times\ldots\times
\Q$\index{$Q^n$} ($n$-times) for the set of all $n$-tuples of elements
of $\Q$. If $\vpp\in \Q^n$, we assume that $\Df{\vpp}=\langle
p_1,\ldots,p_n\rangle$\index{$\vpp$}, i.e., $p_i\in\Q$ denotes the
$i$-th component of the $n$-tuple $\vpp$. Specially, we write $\W(m,b,\vpp)$ in
place of $\W(m,b,p_1,\dots,p_d)$, and we write $\forall \vpp$ in place
of $\forall p_1\dots\forall p_d$, etc.

To abbreviate formulas, we also use bounded quantifiers in the
following way: $\exists x\; \varphi(x)\land \psi$ and $\forall x\;
\varphi(x)\rightarrow \psi$ are abbreviated to $\exists
x\in\varphi\enskip \psi$ and $\forall x\in\varphi\enskip \psi$,
respectively. For example, to formulate that every observer observes a
body somewhere, we write
\begin{equation*}\forall m\in \Ob\;\exists
b\in\B\;\exists\vp\in\Q^d\quad W(m,b,\vpp)
\end{equation*}
instead of 
\begin{equation*}\forall
m\; \Ob(m)\rightarrow \exists b\;\B(b)\land \exists \vp \enskip 
 \Q(p_1)\land\ldots\land\Q(p_d)\land W(m,b,\vpp).
\end{equation*}

We use FOL set theory as a metatheory to speak about model
theoretical concepts, such as models, validity, etc.

The \df{models}\index{models} of this language are of the form
\begin{equation*}
\Df{\mathfrak{M}} = \langle U; \B_\mathfrak{M}, \Ob_\mathfrak{M},
\IOb_\mathfrak{M}, \Ph_\mathfrak{M}, \Q_\mathfrak{M}, +_\mathfrak{M},
\cdot_\mathfrak{M}, <_\mathfrak{M},\W_\mathfrak{M}\rangle,
\end{equation*}
where $U$ is a nonempty set, and $\B_\mathfrak{M}$,
$\Ob_\mathfrak{M}$, $\IOb_\mathfrak{M}$, $\Ph_\mathfrak{M}$ and
$\Q_\mathfrak{M}$ are unary relations on $U$, etc. Formulas are
interpreted in $\mathfrak{M}$ in the usual way.

Let $\Sigma$ and $\Gamma$ be sets of formulas, and let $\varphi$ and
$\psi$ be formulas of our language. Then $\Sigma$ \df{logically
  implies}\index{logically implies} $\varphi$, in symbols
$\Sigma\Df{\models}\varphi$\index{$\models$}, iff $\varphi$ is true in
every model of $\Sigma$, (i.e., $\varphi$ is a logical consequence of
$\Sigma$). $\Sigma\not\models\varphi$ denotes that there is a model of
$\Sigma$ in which $\varphi$ is not true. To simplify our notations, we
use the plus sign between formulas and sets of formulas in the
following way: $\Df{\Sigma+\Gamma}\leteq \Sigma\cup\Gamma$,
$\Df{\varphi+\psi}\leteq \setopen\varphi,\psi\setclose$ and
$\Df{\Sigma+\varphi}\leteq
\Sigma\cup\setopen\varphi\setclose$.\index{$\Sigma+\Gamma$}\index{$\varphi+\psi$}\index{$\Sigma+\varphi$}
\begin{rem}
Let us note that the fewer axioms $\Sigma$ contains, the stronger
the logical implication  $\Sigma\models\varphi$ is, and similarly the more axioms $\Sigma$
contains the stronger the counterexample $\Sigma\not\models\varphi$ is.
\end{rem}

\begin{rem}
By G\"odel's completeness theorem, all the theorems of this thesis remain
valid if we replace the relation of logical consequence ($\models$) by
the deducibility relation of FOL ($\vdash$).\index{$\vdash$}
\end{rem}

\section{Frame axioms}

Here we introduce two axioms that are going to be treated as part of
our logic frame. Our first axiom expresses very basic assumptions,
such as: both photons and observers are bodies, inertial observers are
also observers, etc.

\begin{description}\index{\ax{AxFrame}}
\item[\Ax{AxFrame}] $\Ob\cup \Ph\subseteq \B$, $\IOb\subseteq \Ob$,
 $\W\subseteq \Ob \times \B\times \Q^d$, $\B\cap \Q=\emptyset$; $+$
 and $\cdot$ are binary operations, and $<$ is a binary relation on
 $\Q$\footnote{These statements can easily be translated to our FOL language, e.g., formula $\forall xy\enskip x<y\rightarrow \Q(x)\land\Q(y)$ means that ``$<$ is a binary relation on $\Q$.''}.
\end{description}
Instead of using this axiom we could also use many-sorted FOL
language as in \cite{pezsgo} and \cite{logst}, and only assume that
$\IOb\subseteq\Ob$.

To be able to add, multiply and compare measurements by observers, we
provide an algebraic structure for the set of quantities with the help
of our next axiom.

\begin{description}
\item[\Ax{AxEOF}]\index{\ax{AxEOF}} The \df{quantity
 part}\index{quantity part} $\left< \Q;+,\cdot, < \right>$ is a
 Euclidean ordered field, i.e., a linearly ordered field in which
 positive elements have square roots.
\end{description}
 For the FOL definition of linearly ordered field, see, e.g.,
\cite{chang-keisler}. We use the usual field operations \Dff{0, 1, -,
   /, \sqrt{\phantom{i}}} and binary relation $\Df{\le}$, definable
 within FOL. We also use the vector-space structure of
  $\Q^n$, i.e., $\Df{\vpp+\vqq, -\vpp, \lambda\cdot\vpp}\in \Q^n$ if
$\vpp,\vqq\in \Q^n$ and $\lambda\in \Q$; and $\Df{\vo}\,\leteq \langle
0,\ldots,0\rangle$\index{$\vo$} denotes the \df{origin}\index{origin}.

\begin{conv}\label{conv-frame}
We treat \ax{AxFrame} and \ax{AxEOF} as part of our logic frame.
Hence without any further mentioning, they are always assumed and will
be part of each axiom system we propose herein, except in some of the
theorems of Chap.~\ref{chp-a}.
\end{conv}

\section{Basic definitions and notations}
Let us collect here the basic definitions and notations that are going
to be used in the following chapters.
\begin{rem}
In our formulas we seek to use only FOL definable concepts. So we will
always warn the reader whenever we introduce a concept which is not
FOL definable in our language.
\end{rem}

The ordered field of {real numbers}, which is not FOL definable in our
language, is denoted by $\Df{\R}$\index{$\R$}.  The
\df{composition}\index{composition} of binary relations $R$ and $S$ is
defined as:
\begin{equation*}\index{$\circ$}
\Df{R \circ S}\leteq \Setopen\langle a,c\rangle: \exists b\enskip \langle a,b\rangle \in R \;\land\; \langle b,c\rangle \in S \Setclose.
\end{equation*}
The \df{domain}\index{domain} and the \df{range}\index{range} of a binary relation $R$ are denoted by 
\begin{equation*}\index{$\dom$}\index{$\ran$}
\Df{\dom R}\leteq \Setopen a \setmid \exists b\enskip \langle
a,b\rangle \in R \Setclose\quad \text{ and }\quad \Df{\ran R}\leteq
\Setopen b \setmid \exists a \enskip \langle a,b\rangle\in R
\Setclose,
\end{equation*}
 respectively.
$R^{-1}$ denotes the \df{inverse} of $R$, i.e., 
\begin{equation*}\index{$R^{-1}$}
\Df{R^{-1}}\leteq \Setopen\langle b,a\rangle \setmid \langle a,b\rangle \in R\Setclose.
\end{equation*}
\begin{rem}
We think of a \df{function}\index{function} as a special binary
relation. Notation $\Df{f:A\rightarrow B}$ denotes that $f$ is a
function from $A$ to $B$, i.e., $\dom f=A$ and $\ran f\subseteq B$.
Note that if $f$ and $g$ are functions, then
\begin{equation*}
\boxed{(f \circ g) (x)=g\big(f(x)\big)}
\end{equation*}
for all $x\in \dom f\circ g$. Notation $\Dff{f:A\parrow
 B}$\index{$\parrow$} denotes that $f$ is a partial function on $A$, i.e., $\dom
f\subseteq A$ and $\ran f\subseteq B$.
\end{rem}
\noindent
The {\bf identity map}\index{identity map} on $H\subseteq \Q^d$ is defined as:
\begin{equation*}\index{$Id_H$}
\Df{Id_H}\leteq \Setopen \langle\vp,\vpp\rangle\in \Q^d\times\Q^d
\setmid \vp\in H\Setclose,
\end{equation*}
and the \df{restriction}\index{restriction} of a function $f$ to a set $H$ is defined as: 
\begin{equation*}\index{$f\big\vert_H$}
\Df{f\big|_H}\leteq\Setopen \langle x,y\rangle \setmid x\in\dom f \cap H \lland f(x)=y\Setclose.
\end{equation*}
The set of positive elements of $\Q$ is denoted by 
\begin{equation*}
\Df{\Q^+}\leteq \setopen x\in \Q:0<x\setclose,\index{$\Q^+$}
\end{equation*} 
and the different kinds of interval between $x,y\in \Q$ are defined as:
\begin{equation*}
 \begin{aligned}
\Df{(x,y)}\leteq& \Setopen t\in\Q\setmid x<t<y \;\text{ or }\; y<t<x\Setclose,\\ 
\Df{[x,y]}\leteq& \Setopen t\in \Q\setmid x\le t \le y \;\text{ or }\; y\le t \le x \Setclose, \\
\Df{[x,y)}\leteq& \Setopen t\in \Q\setmid x\le t<y \;\text{ or }\; y<t\le x     \Setclose, \text{ and} \\ 
\Df{(x,y]}\leteq& \Setopen t\in \Q\setmid x<t\le y \;\text{ or }\; y\le t<x   \Setclose.
 \end{aligned}\index{$(x,y)$}\index{$[x,y]$}\index{$(x,y]$}\index{$[x,y)$}
\end{equation*}
We use this nonstandard but convenient notion of intervals to avoid
inconveniences of empty intervals, such as $(1,0)$ in the standard
notion. By our definition $(1,0)$ is not the empty set but the
interval $(0,1)$.

\noindent
For any $n\ge 1$, the \df{Euclidean length}\index{Euclidean length} of $\vpp\in \Q^n$ is defined as: 
\begin{equation*}
\Df{|\vpp|}\leteq \sqrt{p_1^2+\ldots+p_n^2}.\index{$\vert\vpp\vert$}
\end{equation*}
Hence $|x|$ is the absolute value of $x$ if $x\in\Q$.
The \df{(open) ball}\index{open ball} with center $\vp\in\Q^d$ and radius $r\in\Q^+$ is defined as: 
\begin{equation*}\index{$B_r(\vpp)$}
\Df{B_r(\vpp)}\leteq \Setopen \vqq\in\Q^n\setmid |\vp-\vqq|<r\Setclose, 
\end{equation*}
A set $G\subseteq \Q^n$ is called \df{open} \index{opens set} iff for
all $\vp\in G$ there is an $\varepsilon\in \Q^+$ such that
$B_\varepsilon(\vpp)\subset G$. A set $H\subseteq\Q$ is called \df{connected}\index{connected set} iff
$(x,y)\subseteq H$ for all $x,y\in H$.  We say that a function
$\gamma:\Q\parrow \Q^d$ is a \df{curve}\index{curve} if $\dom\gamma$
is connected and has at least two distinct elements.
The {\bf standard basis vectors}\index{standard basis vectors} of $\Q^d$ are denoted by $\Df{\vei}$, i.e.,
\begin{equation*}
\vei\leteq \langle0,\ldots,\stackrel{i}{1},\ldots,0\rangle\index{$\vei$}
\end{equation*}
for all
$1\le i\le d$. We also use notations $\vet$\index{$\vet$},
$\vex$\index{$\vex$}, $\vey$\index{$\vey$} and $\vez$\index{$\vez$}
instead of $\ve$, $\vee$, $\veee$, and $\veeee$, respectively.
The line passing through $\vp$ and $\vq$ is defined as:
\begin{equation*}\index{$line(\vpp,\vqq)$}
\Df{line(\vpp,\vqq)}\leteq \Setopen \vpp+\lambda(\vpp-\vqq) \setmid\lambda\in\Q\Setclose.
\end{equation*}
Let us note that $line(\vp,\vpp)=\setopen \vpp\setclose$ by this definition.
It is practical to introduce a notation for the $tx$-plane: 
\begin{equation*}\index{$\txPlane$}
\Df{\txPlane}\leteq \Setopen \vpp\in \Q^d\setmid p_3=\ldots=p_d=0\Setclose.
\end{equation*}

\section{Some fundamental concepts related to \\ relativity}

\begin{figure}[htp]
\small
\begin{center}
\psfrag{mm}[bl][bl]{$\wl_m(m)$} \psfrag{mk}[br][br]{$\wl_m(k)$}
\psfrag{mb}[t][t]{$\wl_m(b)$} \psfrag{mph}[tl][tl]{$\wl_m(ph)$}
\psfrag{kk}[br][br]{$\wl_k(k)$} \psfrag{km}[b][bl]{$\wl_k(m)$}
\psfrag{kb}[t][t]{$\wl_k(b)$} \psfrag{kph}[tl][tl]{$\wl_k(ph)$}
\psfrag{p}[r][r]{$\vpp$} \psfrag{k}[bl][bl]{$k$}
\psfrag{b}[tl][tl]{$b$} \psfrag{m}[bl][bl]{$m$}
\psfrag{ph}[tl][tl]{$ph$} \psfrag{evm}[tl][tl]{$ev_m$}
\psfrag{evk}[bl][bl]{$ev_k$} \psfrag{Evm}[r][r]{$Ev_m$}
\psfrag{Evk}[r][r]{$Ev_k$} \psfrag{Ev}[r][r]{$Ev$}
\psfrag{T}[l][l]{$e=ev_m(\vpp)=ev_k(\vqq)$} \psfrag{q}[r][r]{$\vqq$}
\psfrag{Cdk}[l][l]{$Cd_k$} \psfrag{Cdm}[l][l]{$Cd_m$}
\psfrag{Crdk}[l][l]{$\loc_k$} \psfrag{Crdm}[l][l]{$\loc_m$}
\psfrag{t}[lb][lb]{$$} \psfrag{o}[t][t]{$\vo$}
\psfrag{fkm}[t][t]{$w^k_m$} \psfrag*{text1}[cb][cb]{worldview of $k$}
\psfrag*{text2}[cb][cb]{worldview of $m$}
\includegraphics[keepaspectratio, width=\textwidth]{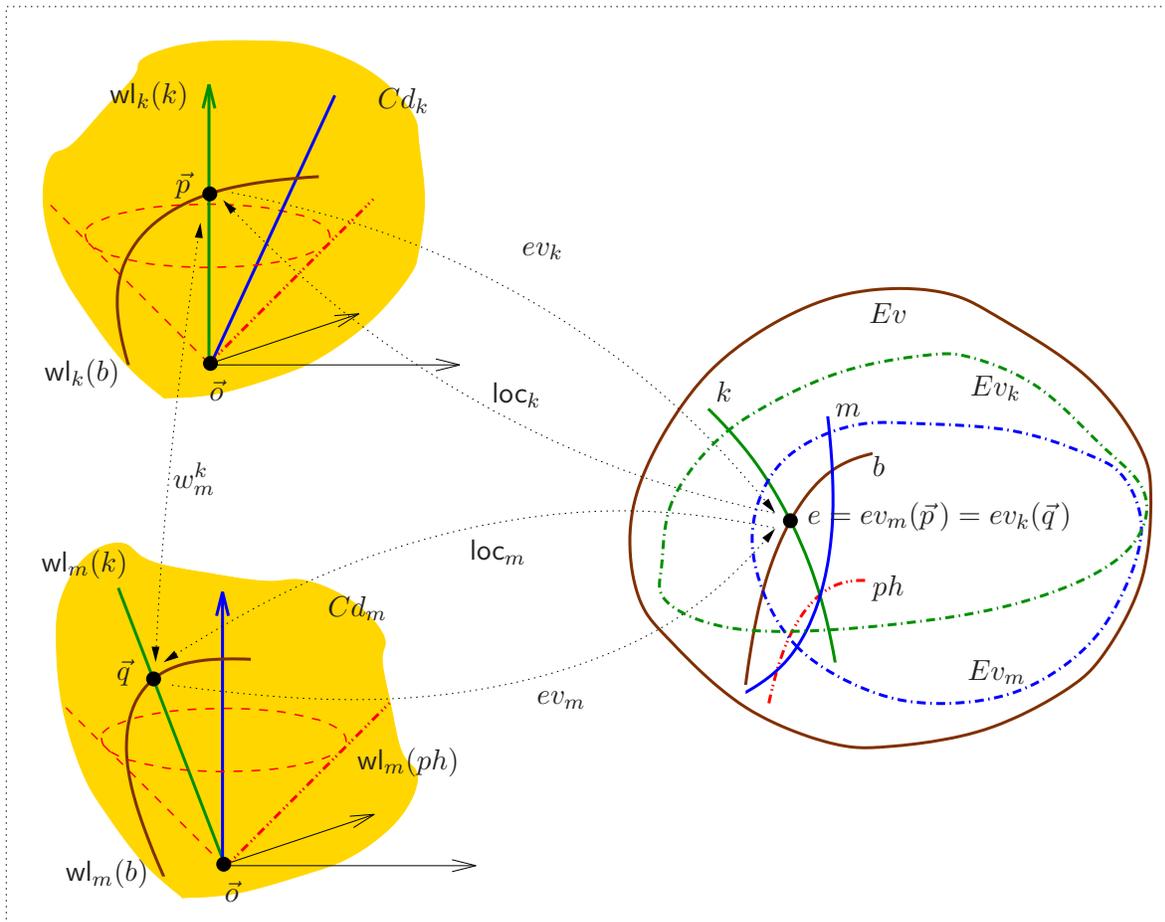}
\caption{\label{fig-fmk} Illustration of the basic definitions}
\end{center}
\end{figure}

Let us gather here some fundamental definitions and notations which
are used in the following chapters.  The set $\Q^d$ is
called the \df{coordinate system} and its elements are referred to as \df{
  coordinate points}\index{coordinate points}\index{coordinate
  system}\index{$\vpp_\sigma$}\index{$p_\tau$}.  We use the notations
\begin{equation*}\index{$\vpp_\sigma$}\index{$p_\tau$}
\Df{\vpp_\sigma}\leteq \langle p_2,\ldots, p_d\rangle \quad \text{ and }\quad \Df{p_\tau}\leteq p_1
\end{equation*} 
for the \df{space component}\index{space component} and the
\df{time component}\index{time component} of $\vpp\in\Q^d$,
respectively.  The \df{event}\index{event} $ev_m(\vpp)$ is defined as
the set of bodies observed by observer $m$ at coordinate point $\vpp$,
i.e.,
\begin{equation*}\index{$ev_m(\vpp)$}
\Df{ev_m(\vpp)}\leteq \Setopen b \setmid \W(m,b,\vpp)\Setclose.
\end{equation*}
The function that maps $\vp$ to $ev_m(\vpp)$ is also denoted by
$ev_m$.  Event $e$ is said to be \df{encountered}\index{encounter} by
observer $k$ if there is a coordinate point $\vq$ such that $k\in
ev_k(\vqq)=e$.  Let $Ev_m$ denote the set of nonempty events
coordinatized by observer $m$, i.e.,
\begin{equation*}\index{$Ev_m$}
\Df{Ev_m}\leteq \Setopen e \setmid \exists \vp\in\Q^d\enskip ev_m(\vpp)=e\neq\emptyset\Setclose,
\end{equation*}
and let $Ev$ denote the set of all observed events, i.e.,
\begin{equation*}\index{$Ev$}
\Df{Ev}\leteq \Setopen e \setmid \exists m\in \Ob\enskip e\in Ev_m
\Setclose.
\end{equation*}

We say that events $e_1$ and $e_2$ are \df{simultaneous}\label{sim}
 for observer $m$, in symbols
$e_1\Df{\!\rule{0pt}{8pt}\!\sim_m\!} e_2$, iff there are coordinate
points $\vp$ and $\vq$ such that $ev_m(\vpp)=e_1$, $ev_m(\vqq)=e_2$,
and $p_\tau=q_\tau$.\index{simultaneous}\index{$\sim$}
\begin{rem} 
It is easy to see that $\sim_m$ is a reflexive and symmetric relation
for every observer $m$; however, it is not an equivalence relation
unless we assume further axioms.
\end{rem}

The \df{coordinate-domain}\index{coordinate-domain} of observer $m$,
in symbols $Cd_m$, is the set of coordinate points where $m$ observes
something (a nonempty event), i.e.,
\begin{equation*}\index{$Cd_m$}
\Df{Cd_m}\leteq \Setopen \vpp \in \Q^d\setmid ev_m(\vpp)\neq \emptyset \Setclose.
\end{equation*}

The \df{worldview transformation}\index{worldview transformation}
between the coordinate-systems of observers $k$ and $m$ is defined as
the set of the pairs of coordinate points in which $k$ and $m$
coordinatize the same nonempty event, i.e.,
\begin{equation*}\index{$w^k_m$}
\Df{w^k_m}\leteq \Setopen\langle \vp,\vqq\rangle\in \Q^d\times
\Q^d\setmid ev_k(\vpp)=ev_m(\vqq)\neq\emptyset\Setclose.
\end{equation*}
Let us note that by this definition worldview transformations are
only binary relations but axiom \ax{AxPh_0}, defined below on p.\pageref{axph0}, turns them into functions, see
Prop.~\ref{prop-sr0}.

\begin{conv}\label{conv-wkm}
Whenever we write $w^k_m(\vpp)$, we mean there is a unique $\vq\in
\Q^d$ such that $\langle \vpp,\vqq\rangle \in w^k_m$, and
$w^k_m(\vpp)$ denotes this unique $\vq$. That is, if we talk about
the value $w^k_m(\vpp)$ of $w^k_m$ at $\vq$, we postulate that it
exists and is unique (by the present convention).
\end{conv}

Since in axiomatic approaches we only assume what is 
explicitly stated by the axioms, we have to prove every other
statement, even the plausible ones. So let us
prove a proposition here about the basic properties of worldview
transformations.

\begin{prop}\label{prop-wkm}
Let $m$ and $k$ be observers. Then 
\begin{enumerate}
\item \label{item-subwkk} $w^k_k\supseteq Id_{Cd_k}$, and
\item \label{item-eqwkk} $w^k_k=Id_{Cd_k}$ iff $k$ does not see any nonempty event twice, i.e., $w^k_k$ is a function. 
\item \label{item-subwkm} $w^k_m\circ w^m_k \supseteq Id_{\dom
 w^k_m}$, and
\item \label{item-eqwkm} $w^k_m\circ w^m_k=Id_{\dom w^k_m}$ iff
 $w^k_m$ is injective.
\item \label{item-subwcomp} $w^k_h\circ w^h_m\supseteq w^k_m$, and 
\item \label{item-eqwcomp} $w^k_h\circ w^h_m=w^k_m$ iff $Ev_k\cap Ev_m\subseteq Ev_h$.
\end{enumerate}
\end{prop}

\begin{proof}
Items \eqref{item-subwkk} and \eqref{item-eqwkk} can be easily proved
by checking the respective definitions.

\smallskip
\noindent
To prove Item \eqref{item-subwkm}, let $\vp\in \dom w^k_m$. Then, by
our definitions, there is a $\vq\in\Q^d$ such that
$ev_k(\vpp)=ev_m(\vqq)\neq\emptyset$, i.e., $\langle
\vp,\vqq\rangle\in w^k_m$. Then, by our definition of worldview
transformation, $\langle \vq,\vpp\rangle \in w^m_k$. Consequently,
$\langle \vp,\vpp\rangle\in w^k_m\circ w^m_k$, which was to be proved.

\smallskip
\noindent 
Let us now prove Item \eqref{item-eqwkm}. If
$w^k_m$ is not injective, there are $\vp_1,\vp_2,\vq\in\Q^d$ such that
$\vp_1\neq \vp_2$ and $\langle \vp_1,\vqq\rangle, \langle
\vp_2,\vqq\rangle\in w^k_m$. Then $\langle \vp_1,\vp_2\rangle\in
w^k_m\circ w^m_k$. So $w^k_m$ has to be injective if $w^k_m\circ
w^m_k=Id_{\dom {w^k}_m}$.

\noindent
To prove the converse implication, let $\langle \vp,\vrr\rangle\in w^k_m\circ
w^m_k$. Then there is a $\vq\in\Q^d$ such that $\langle
\vp,\vqq\rangle\in w^k_m$ and $\langle \vq,\vrr\rangle\in w^m_k$. So
$\langle \vr,\vqq\rangle \in w^k_m$, and thus we get that $\vp=\vr$ by
the injectivity of $w^k_m$. So $\langle \vp,\vrr\rangle\in Id_{\dom
 w^k_m}$ as it was required.

\smallskip
\noindent
Items \eqref{item-subwcomp} and \eqref{item-eqwcomp} can be easily
proved by checking the respective definitions.
\end{proof}

The \df{world-line}\index{world-line} of body $b$ according to
observer $m$ is defined as the set of coordinate points where $b$ was
observed by $m$, i.e.,
\begin{equation*}\index{$\wl_m(b)$}
\Df{\wl_m(b)}\leteq \Setopen \vp\in \Q^d \setmid \W(m,b,\vpp)\Setclose.
\end{equation*}
Let us note here that both $\vp\in\wl_k(b)$ and $b\in ev_k(\vpp)$ represent the atomic formula $\W(k,b,\vpp)$, but from slightly different aspects.

The \df{location} \index{location}$\Df{\loc_m(e)}$\index{$\loc$} of
event $e$ according to observer $m$ is defined as $\vp$ if
$ev_m(\vpp)=e$ and there is only one such $\vp\in\Q^d$; otherwise
$\loc_m(e)$ is undefined.  Event $e$ is called
\df{localized}\index{localized event} by observer $m$ if it has a
unique coordinate according to $m$, i.e., $\loc_m(e)$ is defined. To
express that in our formulas, we use $\Df{\Loc_m(e)}$\index{$\Loc_m(e)$} as
an abbreviation for the following formula: 
\begin{equation*}\exists \vp\in\Q^d\quad ev_m(\vpp)=e \lland
\forall \vq\in\Q^d\quad ev_m(\vqq)=e \then \vp=\vqq.
\end{equation*}

\begin{conv}\label{conv-eq}
We use the equation sign ``='' in the sense of existential equality,
i.e., $\alpha=\beta$ denotes that both $\alpha$ and $\beta$ are
defined and they are equal. We also use the same convention for
 other relations (e.g., for ``$<$''). See \cite[Conv.2.3.10,
 p.31]{Mphd} and \cite[Conv.2.3.10, p.61]{pezsgo}.
\end{conv}

\begin{rem}
 Let us note that $\loc_k(e)=\vp$ means that $ev_k(\vpp)=e$ and
 $\vp=\vq$ for all $\vq$ for which $ev_k(\vqq)=e$ by Conv.~\ref{conv-eq}.
\end{rem}

\begin{rem} Let us note that
$\loc_m\big(ev_k(\vpp)\big)$ is defined iff $w^k_m(\vpp)$ is so and in this
 case they are the same, i.e.,
\begin{equation*}
w^k_m(\vpp)=\loc_m\big(ev_k(\vpp)\big).
\end{equation*}
\end{rem}
\noindent
The \df{time of event}\index{time} $e$ according to observer $m$ is
defined as:
\begin{equation*}\index{$\time$}
\Df{\time_m(e)}\leteq \loc_m(e)_\tau
\end{equation*}
if $e$ is localized by $m$; otherwise $\time_m(e)$ is undefined.
The \df{elapsed time}\index{elapsed time} between events $e_1$ and
$e_2$ measured by observer $m$ is defined as:
\begin{equation*}
\Df{\time_m(e_1,e_2)}\leteq |\time_m(e_1)-\time_m(e_2)|
\end{equation*}
if $e_1$ and $e_2$ are localized by $m$;  otherwise
$\time_m(e_1,e_2)$ is undefined. $\time_m(e_1,e_2)$ is
called the \df{proper time}\index{proper time}\label{propertime} measured by $m$
between $e_1$ and $e_2$ if $m\in e_1\cap e_2$. The \df{spatial
 location}\index{spatial location} of event $e$ according to observer
$m$ is defined as:
\begin{equation*}\index{$\spc$}
\Df{\spc_m(e)}\leteq \loc_m(e)_\sigma
\end{equation*}
if $e$ is localized by $m$;  otherwise $\spc_m(e)$ is undefined.
The \df{spatial distance}\index{spatial distance} between events $e_1$
and $e_2$ according to observer $m$ is defined as:
\begin{equation*}\index{$\dist$}
\Df{\dist_m(e_1,e_2)}\leteq |\spc_m(e_1)-\spc_m(e_2)|
\end{equation*}
if $e_1$ and $e_2$ are localized by $m$;  otherwise $\dist_m(e_1,e_2)$ is
undefined.

Spacetime vector $\vr\in\Q^d$ is called \df{spacelike}\index{spacelike vector} iff $|\vr_\sigma|>|r_\tau|$, 
\df{lightlike}\index{lightlike vector} iff $|\vr_\sigma|=|r_\tau|$, and
\df{timelike}\index{timelike vector} iff $|\vr_\sigma|<|r_\tau|$.
Spacetime vectors $\vpp$ and $\vqq$ are called
\df{spacelike-separated}\index{spacelike-separated}, in symbols
$\vpp\Df{\seq}\vqq$\index{$\seq$}, iff $\vp-\vq$ is a spacelike
vector; \df{lightlike-separated}\index{lightlike-separated}, in
symbols $\vpp\Df{\pheq}\vqq$\index{$\pheq$}, iff $\vp-\vq$ is a
lightlike vector; and \df{timelike-separated}\index{timelike-separated},
in symbols $\vpp\Df{\teq}\vqq$\index{$\teq$}, iff $\vp-\vq$ is a
timelike vector. 
Events $e_1$ and $e_2$ which are localized by every {\it inertial} observer are called
spacelike-separated (lightlike-separated; timelike-separated), in
symbols $e_1\seq e_2$ ($e_1\pheq e_2$; $e_1\teq e_2$), iff
$\loc_m(e_1)$ and $\loc_m(e_2)$ are such for every {\it inertial}
observer $m$.
 A curve $\gamma$
is called \df{timelike}\index{timelike curve} iff it is
differentiable (see p.\pageref{p-diff}), and $\gamma'(t)$ is timelike for all $t\in \dom
\gamma$.

Coordinate points $\vpp$ and $\vqq$ are called \df{Minkowski
  orthogonal}\index{Minkowski orthogonal}, in symbols
$\vpp\Df{\mort}\vqq$\index{$\mort$}, iff the following holds:
$p_\tau\cdot q_\tau=p_2\cdot
q_2+\ldots+p_d\cdot q_d.$
The (signed) \df{Minkowski length}\index{Minkowski length} of $\vpp\in \Q^d$ is
\begin{equation*}\index{$\mu(\vpp)$}
\Df{\mu(\vpp)}\leteq \left\{
\begin{array}{rll}
\sqrt{\rule{0pt}{11pt} p_\tau^2-|\vpp_\sigma|^2} & \text{ if}\quad p_\tau^2\ge|\vpp_\sigma|^2, \\
-\sqrt{\rule{0pt}{11pt}|\vpp_\sigma|^2-p_\tau^2} & \text{ in other cases, } 
\end{array}
\right.
\end{equation*}
and the \df{Minkowski distance}\index{Minkowski distance} between
$\vpp$ and $\vqq$ is
$\Df{\mu(\vpp,\vqq)}\leteq\mu(\vpp-\vqq).$\index{$\mu(\vpp,\vqq)$} We
use the signed version of the Minkowski length because it contains two
kinds of information: (i) the length of $\vpp$, and (ii) whether it is
spacelike, lightlike or timelike. A map $f:\Q^d\rightarrow \Q^d$ is
called a \df{Poincar\'e transformation}\index{Poincar\'e
  transformation} iff it is a transformation preserving the Minkowski
distance, i.e., $\mu\big(f(\vpp),f(\vqq)\big)=\mu(\vp,\vqq)$ for all
$\vp,\vq\in\Q^d$.  Like transformations preserving Euclidean distance,
Poincar\'e transformations are also affine ones.  Linear
transformations preserving the Minkowski distance are called
\df{Lorentz transformations}\index{Lorentz transformation}.

\chapter{Special relativity}
\label{chp-sr}

In this chapter we axiomatize special relativity within our FOL
frame. The axiom system \ax{SpecRel}, which we introduce in this
chapter, is a kind of basic axiom system that we will extend and
transform in the forthcoming chapters.  Here we also discuss some
important properties of \ax{SpecRel}, such as its completeness with
respect to Minkowskian geometries over Euclidean ordered fields or the
possible worldview transformations between {\it inertial} observers.

\section{Special relativity in four simple axioms}

Here we formulate four simple and plausible axioms that
capture special relativity.  We seek to formulate easily
understandable axioms in FOL.  We present each axiom at
two levels. First we give an intuitive formulation, then a precise
formalization using our logical notations. We give the pure
FOL version of the first three axioms only. However, all
the axioms in this thesis can also be translated easily into
FOL formulas by inserting the respective FOL
definitions into the formalizations of the axioms.

Let us now formulate our first axiom on observers.  Historically, this
natural axiom goes back to Galileo Galilei or even to d'Oresme of
around 1350, see, e.g., \cite[p.23, \S 5]{AMNsamples}, but it is very
probably a prehistoric assumption, see Rem.~\ref{rem-axself}.  It simply
states that each observer assumes that it rests at the origin of the
space part of its coordinate system.

\begin{description}\index{\ax{AxSelf_0}}\label{axself0}
\item[\Ax{AxSelf_0}] An observer coordinatizes itself at a coordinate
 point iff it is in the observer's coordinate domain and its space component is
 the origin:
\begin{equation*}
\forall m \in \Ob \;\forall \vpp\in \Q^d\quad m\in ev_m(\vpp) \iff
\vp\in Cd_m \lland \vpp_\sigma=\vo.
\end{equation*}
\end{description}
A purely FOL formula expressing \ax{AxSelf_0} is the following:
\begin{multline*}
\forall m\enskip \forall \vp\quad\Ob(m)\land \Q(p_1)\land\ldots\land \Q(p_d)\;
\rightarrow\\
\Big(\W(m,m,\vpp)\iff \exists b\ \B(b)\land\W(m,b,\vpp)\land p_2=0\land\ldots\lland
p_d=0\Big).
\end{multline*}

\noindent
Let us also introduce a strengthened version of axiom \ax{AxSelf_0}:

\begin{description}
\item[\Ax{AxSelf}]\index{\ax{AxSelf}} An {\it inertial} observer coordinatizes itself at
 a coordinate point iff its space component is the origin:
\begin{equation*}
\forall m \in \IOb \enskip \forall \vp\in \Q^d\quad W(m,m,\vpp) \iff \vp_\sigma=\vo.
\end{equation*}
\end{description}

\begin{rem}\label{rem-axself}
At first glance it is not clear why \ax{AxSelf_0} is so natural. As
an explanation, let us consider the following simple example. Let us
imagine that we are watching sunset. What do we see? We do not see
and feel that we are rotating with the Earth but that the Sun is
moving towards the horizon; and according to our (the Earth's)
reference system, we are absolutely right. But we learned at primary
school that ``the Earth rotates and goes around the Sun.'' So why
does not this (i.e., the adoption of the heliocentric system) mean
that \ax{AxSelf_0} and our impression above about the sunset  are
simply wrong? It is because the debate between geocentric and
heliocentric systems was not about \ax{AxSelf_0}, but about how to
choose the best reference frame if we want to study the
motions of planets in our solar system. See \cite{Russell}.%
\footnote{Here we consider only the basic idea of the two systems
  (i.e., whether the Earth or the Sun is stationary) and not their
  details (e.g., epicycles). Of course, Ptolemy's geocentric model was
  wrong in its details since even if we fix the Earth as a reference
  frame, the other planets will go around not the Earth but the Sun.
  It is interesting to note that Tycho Brahe worked out a correct
  geocentric system in which the Sun and the Moon move around the
  Earth and the other planets move around the Sun.} As reference
frames, those of the Earth, the Sun, and even the Moon are equally
good. However, if we would like to calculate the motions of the
planets, the Sun's reference frame is the most convenient.
\end{rem}

Now we formulate our next axiom on the constancy of the speed of
photons. For convenience, we choose $1$ for this speed. This choice
physically means using units of distance compatible with units of
time, such as light-year, light-second, etc.
\begin{description}
\item[\Ax{AxPh}]\index{\ax{AxPh}} For every \textit{inertial}
 observer, there is a photon through two coordinate points $\vpp$ and
 $\vqq$ iff the slope of $\vpp-\vqq$ is $1$:
\begin{equation*}
\forall m\in \IOb\enskip \forall \vpp,\vqq\in
\Q^d\quad |\vpp_\sigma-\vq_\sigma|=|p_\tau-q_\tau| \iff \Ph\cap
ev_m(\vpp)\cap ev_m(\vqq)\ne\emptyset.
\end{equation*}
\end{description}
A purely FOL formula expressing \ax{AxPh} is the following:
\begin{multline*}
\forall m\; \forall \vp\enskip \forall \vq \quad
\IOb(m)\land\Q(p_1)\land \Q(q_1)\land\ldots\land \Q(p_d)\land \Q(q_d)
\; \rightarrow\\
\Big((p_1-q_1)^2=(p_2-q_2)^2+\ldots+(p_d-q_d)^2
\\\iff
 \exists ph\ \Ph(ph)\land \W(m,ph,\vpp)\land\W(m,ph,\vqq)\Big).
\end{multline*}
Axiom \ax{AxPh} is a well-known assumption of Special Relativity, see,
e.g., \cite{logst}, \cite[\S 2.6]{dinverno}. We may weaken
\ax{AxPh} by allowing {\it inertial} observers to measure different but uniform speeds of light.
\begin{description}
\item[\Ax{AxPh_0}]\label{axph0}\index{\ax{AxPh_0}} For every
  \textit{inertial} observer, the speed of light is uniform and positive, and there
  can be a photon at any point and in any direction with this speed:
\begin{multline*}
\forall m\in \IOb\;\exists c_m\in\Q^+\enskip \forall \vpp,\vqq\in
\Q^d\quad |\vpp_\sigma-\vq_\sigma|=c_m \cdot|p_\tau-q_\tau| \\ \iff
\Ph\cap ev_m(\vpp)\cap ev_m(\vqq)\ne\emptyset.
\end{multline*}
\end{description}
The models of our theory \ax{SpecRel} (see p.\pageref{page-sr}) would
change to some extent if we replaced \ax{AxPh} by \ax{AxPh_0};
however, they would not be essentially different. We use \ax{AxPh} for
convenience only. Sfarti~\cite{sfarti} proves that the principle of
relativity and \ax{AxPh_0} imply \ax{AxPh}.

\begin{rem}
For convenience, we quantify over events, too. That does not mean
abandoning our FOL language. It is just simplifying the
formalization of our axioms. Instead of events we could speak about
observers and spacetime locations. For example, instead of $\forall
e\in Ev_m\; \phi$ we could write $\forall \vpp\in Cd_m\;
\phi[e\!\leadsto\! ev_m(\vpp)]$, where none of $p_1\ldots p_d$ occur
free in $\phi$, and $\phi[e\!\leadsto\!  ev_m(\vpp)]$ is the formula
obtained from $\phi$ by substituting $ev_m(\vpp)$ for $e$ in all free
occurrences. Similarly, we can replace $\forall e\in Ev\; \phi$ by $\forall
m\in\Ob\enskip \forall e\in Ev_m\; \phi$.
\end{rem}
\noindent
By our next axiom we assume that events observed by {\it inertial} observers are the same.
\begin{description}
\item[\Ax{AxEv}]\index{\ax{AxEv}} Every \textit{inertial} observer
 coordinatizes the very same set of events:
\begin{equation*}
\forall m,k\in \IOb \quad Ev_m=Ev_k.
\end{equation*}
\end{description}
A purely FOL formula expressing \ax{AxEv} is the following:
\begin{multline*}
\forall m\; \forall k\; \forall \vp\quad
\IOb(m)\land\IOb(k)\land\Q(p_1)\land\ldots\land\Q(p_d)\;
\then
\exists \vq\\
\Q(q_1)\land\ldots\land \Q(q_d)\land 
\Big(\forall b\quad \B(b)\; \rightarrow\; \big(\W(m,b,\vpp)\iff\W(k,b,\vqq)\big)\Big).
\end{multline*}
Let us now prove some consequences of the axioms introduced so far.
\begin{prop}\label{prop-sr0}
Let $h$ be an observer and let $m$ and $k$ be {\it inertial} observers.
Then
\begin{enumerate}
\item \label{item-crd} 
$Cd_m=\Q^d$ and $ev_m$ is injective if \ax{AxPh_0} is assumed.
\item \label{item-fun} $ev_m$ is a bijection from $Cd_m$ to $Ev_m$; $\loc_m$ is a bijection from $Ev_m$ to $Cd_m$; and $w^h_m$ is a function from $\Q^d$ to $\Q^d$ if $ev_m$ is injective on nonempty events.
\item \label{item-fmk} 
$w^k_m$ is a bijection from $\Q^d$ to $\Q^d$ if \ax{AxPh_0} and \ax{AxEv} are assumed.
\end{enumerate}
\end{prop}

\begin{proof}   
To prove Item \eqref{item-crd}, let $\vpp\in\Q^d$. Then by
\ax{AxPh_0}, there is a photon $ph$ such that $ph\in ev_m(\vpp)\cap
ev_m(\vp+\langle 1,0,\ldots,0,c_m,0,\ldots,0\rangle)$.  Hence
$ev_m(\vpp)\neq\emptyset$ for all $\vpp\in\Q^d$. So
$Cd_m=\Q^d$. Moreover, if $\vqq \in \Q^d$ and $\vqq\neq\vpp$, then it
is possible to choose this $ph$ such that $ph\not\in ev_m(\vqq)$ also
holds. Thus $ev_m$ is injective.

\smallskip\noindent Item \eqref{item-fun} is clear since, if $ev_m$ is
injective on nonempty events, both $\loc_m\leteq ev_m^{-1}$ and
$w^h_m\leteq ev_h\circ \loc_m$ are functions.

\smallskip\noindent Let us now prove Item \eqref{item-fmk}. By Item
\eqref{item-fun}, we already have that $ev_k$ is a bijection from
$Cd_k$ to $Ev_k$, and $\loc_m$ is a bijection from $Ev_m$ to $Cd_m$.
By $\ax{AxEv}$, $Ev_k=Ev_m$. Thus $w^k_m=ev_k\circ \loc_m$ is a
bijection from $Cd_k$ to $Cd_m$. But by Item \eqref{item-crd}, we
also have that $Cd_k=Cd_m=\Q^d$. Hence $w^k_m$ is a bijection from
$\Q^d$ to $\Q^d$.
\end{proof}

Let us now introduce a symmetry axiom called the symmetric distance
axiom, by which we assume that \textit{inertial} observers use the
same units of measurement.

\begin{description}
\item[\Ax{AxSymDist}]\index{\ax{AxSymDist}} \textit{Inertial}
  observers $m$ and $k$ agree as to the spatial distance between
  events $e_1$ and $e_2$ if they are simultaneous for both of them:
\begin{multline*}
\forall m,k\in \IOb\enskip \forall e_1,e_2 \in Ev_m\cap Ev_k \quad\\
e_1\sim_m e_2\lland e_1\sim_k e_2
\then\dist_m(e_1,e_2)=\dist_k(e_1,e_2).
\end{multline*}
\end{description}

\noindent
Let us introduce the following axiom system:
\begin{equation*}\label{page-sr}\index{\ax{SpecRel}}
\boxed{\ax{SpecRel}\leteq \Setopen \ax{AxSelf_0}, \ax{AxPh}, \ax{AxEv},\ax{AxSymDist} \Setclose}
\end{equation*}
Now we have a FOL theory of Special Relativity for each
natural number $d\ge2$.

Our symmetry axiom \ax{AxSymDist} has many equivalent versions, see
\cite[\S 2.8, \S 3.9, \S 4.2]{pezsgo}. Let us introduce one of them
here.
\begin{description}
\item[\Ax{AxSymTime}]\index{\ax{AxSymTime}} Any two \textit{inertial}
  observers see each others' clocks behaving in the same way:
\begin{equation*}
\forall k,m\in \IOb\enskip\forall \lambda\in \Q \quad
\left|w^k_m(\lambda\cdot \vet)_\tau-w^k_m(\voo)_\tau
\right|=\left|w^m_k(\lambda\cdot \vet)_\tau-w^m_k(\voo)_\tau \right|.
\end{equation*}
\end{description}

To prove that \ax{AxSymTime} is equivalent to \ax{AxSymDist}, let us
introduce a version of \ax{SpecRel} without this axiom:
\begin{equation*}\index{\ax{SpecRel_0}}
\boxed{\ax{SpecRel_0}\leteq \Setopen \ax{AxSelf_0}, \ax{AxPh}, \ax{AxEv} \Setclose}
\end{equation*}

\begin{thm}\label{thm-tdeq}
Let $d\ge3$ and assume \ax{AxSpecRel_0}. Then the following three statements are equivalent:
\begin{enumerate}
\item \label{item-symd} \ax{AxSymDist},
\item \label{item-symt} \ax{AxSymTime} and
\item \label{item-poi} $\forall k,m\in \IOb\enskip w^k_m$ is a Poincar{\'e} transformation.
\end{enumerate}
\end{thm}

\begin{proof}[\colorbox{proofbgcolor}{\textcolor{proofcolor}{On the proof}}]
By using the fact that every Poincar{\'e} transformation is the
composition of a translation, a space-isomorphism and a Lorentz boost,
it is not difficult to prove that Item \eqref{item-poi} implies Items
\eqref{item-symd} and \eqref{item-symt}.
\medskip

\noindent
Item (2) in Thm.~\ref{thm-poi} states that Item \eqref{item-symd}
implies Item \eqref{item-poi}.
\medskip

\noindent
Finally, the implication of Item \eqref{item-poi} by Item
\eqref{item-symt} can be proved analogously to Thm.~\ref{thm-poi},
i.e., by proving that both the field-automorphism-induced maps and the
dilations in the decomposition of $w^k_m$ and $w^m_k$ given by Item
(1) in Thm.~\ref{thm-poi} are the identity map.
\end{proof}

\section{worldview transformations in special relativity}
\label{sec-srwv}

To prove a theorem that characterizes the worldview
transformations between {\it inertial} observers if only \ax{AxPh} and
\ax{AxEv} are assumed, we need one more definition. A map
$\tilde\varphi:\Q^d\rightarrow \Q^d$ is called a 
\df{field-automorphism-induced map}\index{field-automorphism-induced
 map} iff there is an automorphism $\varphi$ of the field
$\langle\Q,\cdot,+\rangle$ such that $\tilde\varphi(\vpp)=\langle
\varphi(p_1),\ldots,\varphi(p_d)\rangle$ for every $\vp\in\Q^d$.
Now we can state the Alexandrov-Zeeman theorem generalized for fields.

\begin{thm}[Alexandrov-Zeeman]
 Let be $F$ a field and $d\ge3$.  Every bijection from $F^d$ to $F^d$
 that transforms lines of slope 1 to lines of slope 1 is a
 Poincar\'e transformation composed with a dilation and a
 field-automorphism-induced map.
\end{thm}
For the proof of this theorem, see, e.g., \cite{vro}, \cite{VKK}. From
this theorem we derive the following characterization of worldview
transformations.

\begin{thm}
\label{thm-poi}
Let $d\ge 3$.  Let $m$ and $k$ be {\it inertial} observers.  Then
\begin{enumerate}
\item if \ax{AxPh} and \ax{AxEv} are assumed, $w^k_m$ is a Poincar\'e
  transformation composed with a dilation $D$ and a
  field-automorphism-induced map $\tilde\varphi$;
\item if \ax{AxPh}, \ax{AxEv} and \ax{AxSymDist} are assumed, $w^k_m$
  is a Poincar\'e transformation.
\end{enumerate}
\end{thm}

\begin{onproof}
It is not hard to see that \ax{AxPh} and \ax{AxEv} imply that $w^k_m$
is a bijection from $\Q^d$ to $\Q^d$ that preserves lines of slope 1,
see Prop.~\ref{prop-sr0}. Hence Item (1) is a consequence of
the Alexandrov-Zeeman theorem generalized for fields.

Now let us see why Item (2) is true.  By using Item (1), it is easy to
see that there is a line $l$ such that both $l$ and its $w^k_m$ image
are orthogonal to the time-axis.  Thus by \ax{AxSymDist}, $w^k_m$
restricted to $l$ is distance-preserving.  Consequently, both the
dilation $D$ and the field-automorphism-induced map $\tilde\varphi$ in
Item (1) have to be the identity map.  Hence $w^k_m$ is a Poincar\'e
transformation.\hfill\qed
\end{onproof}

Thm.~\ref{thm-poi} shows that \ax{SpecRel} is a good axiom system for
Special Relativity if we restrict our interest to {\it inertial}
motion.  It also implies that the most frequently quoted predictions
of Special Relativity are provable from \ax{SpecRel}:
\begin{itemize}
\item[(i)] ``moving clocks slow down,'' 
\item[(ii)] ``moving meter-rods shrink'' and 
\item[(iii)] ``moving pairs of clocks get out of synchronism.'' 
\end{itemize}
Even if we only assume \ax{AxPh} and \ax{AxEv}, we can prove 
qualitative versions of the predictions above; \ax{AxSymDist} is
needed if we want to prove the quantitative versions, too. And
\ax{AxSelf} is only a simplifying axiom; it makes formulating the above predictions easier.
For more detail. See, e.g., \cite[\S 2.5]{pezsgo}, \cite[\S 1]{AMNsamples}, \cite[\S 2]{logst}.

The following consequence of Thm.~\ref{thm-poi} is the starting point
for building Minkowski geometry, which is the ``geometrization'' of
Special Relativity. It shows how time and space are intertwined in
Special Relativity.

\begin{thm} \label{mink-thm}
Let $d\ge 3$. Assume $\ax{SpecRel}$.
Then
\begin{equation*}
\time_m(e_1,e_2)^2-\dist_m(e_1,e_2)^2 =
\time_k(e_1,e_2)^2-\dist_k(e_1,e_2)^2 
\end{equation*}
\noindent for any {\it inertial} observers
$m$ and $k$ and events $e_1$ and $e_2$ coordinatized by both of them.
\end{thm}
\noindent
Let us finally state a corollary here about the slowing down of moving clocks.
\begin{cor} \label{srslow-cor}
Assume $\ax{SpecRel}$, $d\ge 3$.
Let $m,k\in\IOb$,
$e_1,e_2\in Ev_k$, and assume $k\in e_1\cap e_2$,
$\dist_m(e_1,e_2)\ne 0$.
Then
\begin{equation*}\time_m(e_1,e_2)>\time_k(e_1,e_2).\end{equation*}
\end{cor}

In the above corollary, a ``moving clock'' is represented by observer
$k$; the fact that it is moving relative to observer $m$ is expressed
by $\dist_m(e_1,e_2)\ne 0$ and $k\in e_1\cap e_2$; and that $k$'s time
is slowing down relative to $m$'s is expressed by
$\time_m(e_1,e_2)>\time_k(e_1,e_2)$. This ``clock slowing down'' is
only a relative effect, i.e., ``clocks moving relative to $m$ slow
down relative to $m$.'' But this relative effect leads to a new kind
of gravitation-oriented ``absolute slowing down of time'' effect, as
Chap.~\ref{chp-grav} will show.

We can summarize the results of this chapter (that standard special
relativity is provable from \ax{SpecRel}) as a kind of completeness
theorem of \ax{SpecRel} with respect to its ``intended models'':
 
\begin{cor}\label{cor-srcompl}
Assume $d\ge3$. Then \ax{SpecRel} is complete with respect to
Minkowskian geometries over Euclidean ordered fields.
\end{cor}

The formal meaning of Cor.~\ref{cor-srcompl} is completely
analogous to that of Thm.~\ref{thm-grn} (about general relativity) and is
explained under Thm.~\ref{thm-grn}. For further details, see \cite[\S 4]{Mphd}, too. 

\chapter{Clock paradox}
\label{chp-cp}

As one of our main aims is to trace back the surprising predictions of
relativity to some convincing axioms, first we investigate an
axiomatic basis of the clock paradox\footnote{Unfortunately, it is
  still not uncommon for people who misinterpret the word `paradox' to
  try to find contradictions in relativity theory, that is why we think
  it important to note here that its original meaning is ``a statement
  that is seemingly contradictory and yet is actually true,'' i.e., it
  has nothing to do with logical contradiction.  With the nearly
  century long fruitless debate in view, perhaps it would be better to
  call the paradoxes of relativity theory simply effects, thus saying
  ``clock effect'' instead of ``clock paradox,'' but for the time
  being it appears to be a hopeless effort to have this idea
  generally accepted. Anyway, we would like to emphasize that it is
  absolutely pointless to try to find a logical contradiction in
  relativity theory, as its consistency has been proved, see
  \cite[p.77]{pezsgo}, \cite[Cor. 11.12 p.644]{logst}.}  (CP), which
is an inertial approximation of the famous twin paradox.  A similar
logical investigation of the twin paradox needs a more complex
mathematical apparatus, see \cite{twp}, \cite{mythes} and
Chap.~\ref{chp-twp}.  The results of this chapter are based on
 \cite{clp}, \cite{mytdk} and \cite{mythes}.

CP is one of the most famous predictions of special relativity. It
concerns three {\it inertial} observers: one of them is the
stay-at-home twin and the other two simulate the accelerated twin in
the twin paradox.  This simulation is done by
replacing the accelerated twin by a leaving {\it
 inertial} observer and a returning one that synchronizes its clock with the leaving one's
when they meet.

In this chapter we mainly concentrate on the relation of CP to the axioms
and other consequences of special relativity, but we also
formulate and characterize variants of CP: one where the stay-at-home twin will be
the younger one (Anti-CP) and another where no differential aging will take place
(No-CP).

In Section~\ref{ax-sec} we introduce a very basic axiom system
\ax{Kinem_0} of kinematics in which no relativistic effect is assumed.
\ax{Kinem_0} is a subtheory of Newtonian kinematics and special
relativity. In Section~\ref{thm-sec} we formulate and prove a
geometrical characterization of CP, Anti-CP and No-CP each within the
models of \ax{Kinem_0}, see Cor.~\ref{cor-twp} and
Thm.~\ref{thm-clp}. In Secs.\ \ref{conk-sec} and \ref{consr-sec} we
prove some surprising logical consequences of our characterization.
In Thm.~\ref{thm-univtime} we show that the absoluteness of time (in
the Newtonian sense) is not equivalent to the lack of the clock
paradox (No-CP) without assuming a strong theoretical
axiom. Similarly, in Thm.~\ref{thm-slowtime} we show that the slowing
down of moving clocks is not equivalent to CP. In
Thm.~\ref{thm-simdist} we show that a symmetry axiom of special
relativity is strictly stronger than CP.

\section{A FOL axiom system of kinematics}
\label{ax-sec}

We characterize the CP under some very mild assumptions about
kinematics. To introduce this weak axiom system (\ax{Kinem_0}) we
formulate some further axioms.  Let us recall that $\vet=\langle
1,0,\ldots,0\rangle$; and let us define the \df{time-unit
  vector}\index{time-unit vector} of $k$ according to $m$ to be
\begin{equation*}\index{$\vekm$}
\Df{\vekm}\leteq w^k_m(\vet)-w^k_m(\voo).
\end{equation*}

\begin{description}
\item[\Ax{AxLinTime}]\index{\ax{AxLinTime}} The world-lines of {\it inertial} observers are lines and time is elapsing uniformly on them:
\begin{multline*}
\forall m,k \in \IOb\enskip \wl_m(k)=\Setopen
w^k_m(\voo)+\lambda\cdot\vekm \setmid \lambda\in\Q\Setclose\, \land
\\ \forall \vp,\vq \in \wl_m(k) \quad
\time_k\big(ev_m(\vpp),ev_m(\vqq)\big)\cdot\left|\vekm\right|=|\vp-\vqq|.
\end{multline*}
\end{description}
\noindent
Let us now introduce the aforementioned axiom system of kinematics:
\begin{equation*}\index{\ax{Kinem_0}}
\boxed{\ax{Kinem_0}\leteq \Setopen \ax{AxSelf}, \ax{AxLinTime}, \ax{AxEv} \Setclose}
\end{equation*}
Let us note that \ax{Kinem_0} is a very weak axiom system of
kinematics.  By using Item \eqref{item-crd} of Prop.~\ref{prop-sr0} and
Thm.~\ref{thm-poi}, it not difficult to show that \ax{AxSelf} and
\ax{AxLinTime} are consequences of \ax{SpecRel}. So \ax{Kinem_0} is
weaker than \ax{SpecRel}.

\section{Formulating the clock paradox}

To formulate CP, first we formulate the situations in which it
can occur.  We say that {\it inertial} observer $m$ observes {\it inertial} observers $a$, $b$ and
$c$ in a \df{clock paradox situation}\index{clock paradox situation}
at events $e$, $e_a$ and $e_c$ iff $a\in e_a\cap e$, $b\in e_a\cap
e_c$, $c\in e\cap e_c$, $b\not\in e$ and
$\time_m(e_a)<\time_m(e)<\time_m(e_c)$ or
$\time_m(e_a)>\time_m(e)>\time_m(e_c)$, see Fig.~\ref{twpp}.  This
situation is denoted by
$\Df{\mtwp_m(\widehat{ac},b)}(e_a,e,e_c)$.\index{$\mtwp$}

\begin{figure}[h!t]
\begin{center}
\psfrag{m}{$m$}
\psfrag{a}{$a$}
\psfrag{b}[tl][tl]{$b$}
\psfrag{c}{$c$}
\psfrag{p}[tl][tl]{$\vp$}
\psfrag{s}{$s$}
\psfrag{q}[cl][cl]{$\vq$}
\psfrag{0}[tl][tl]{$\vo$}
\psfrag{r}[bl][bl]{$\vr$}
\psfrag{e}{$e$}
\psfrag{ea}{$e_a$}
\psfrag{ec}{$e_c$}
\psfrag{Ev}[tr][tr]{$Ev$}
\psfrag{Crdm}{$\loc_m$}
\psfrag{1b}[rb][rb]{${b}^\ddag$}
\psfrag{1a}[rb][rb]{${a}^\ddag$}
\psfrag{1c}[rb][rb]{${c}^\ddag$}
\includegraphics[keepaspectratio, width=0.8\textwidth]{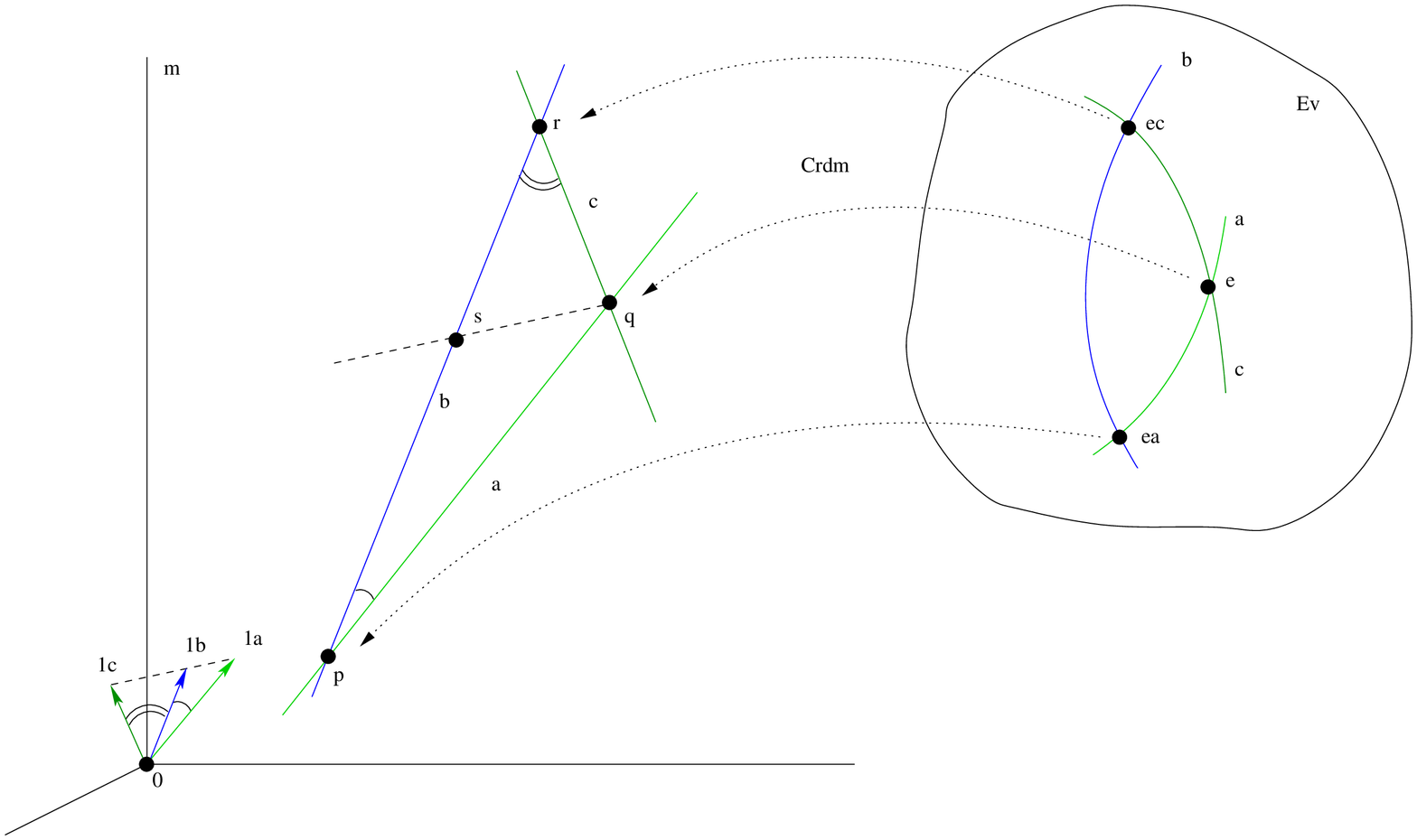}
\caption{\label{twpp} Illustration of relation
  $\mtwp_m(\widehat{ac},b)(e_a,e,e_c)$ and the proof of
  Prop.~\ref{prop-tp}}
\end{center}
\end{figure}

Let $a,b,c\in\IOb$ and $e_a,e,e_b\in Ev$.  Let
$\Df{\time(\widehat{ac}<b)}(e_a,e,e_b)$\index{$\time(\widehat{ac}<b)(e_a,e,e_b)$}
be an abbreviation for
$\time_a(e_a,e)+\time_c(e,e_c)<\time_b(e_a,e_c)$.  The definitions of
$\Df{\time(\widehat{ac}=b)}(e_a,e,e_b)$\index{$\time(\widehat{ac}=b)(e_a,e,e_b)$}
and
$\Df{\time(\widehat{ac}>b)}(e_a,e,e_b)$\index{$\time(\widehat{ac}>b)(e_a,e,e_b)$}
are analogous.
\noindent
Using this notation, we can formulate the clock paradox as follows:
\begin{description}
\item[\Ax{CP}]\index{\ax{CP}} Every {\it inertial} observer $m$ observes the clock paradox in every clock paradox situation:
\begin{equation*}
\forall m,c,a,b\in \IOb\enskip \forall e,e_a,e_c\in Ev_m\quad
\mtwp_m(\widehat{ac},b)(e_a,e,e_c) \then \time(\widehat{ac}<b)(e_a,e,e_c).
\end{equation*}
\end{description}
We define formulas \Ax{NoCP}\index{\ax{NoCP}} and \Ax{AntiCP}\index{\ax{AntiCP}} by replacing '$<$' by '$=$' and '$>$' in the formula \ax{CP}, respectively.

\section{Geometrical characterization of CP}
\label{thm-sec}

We say that $\vq\in\Q^d$ is (strictly) \df{between}\index{between}
$\vp\in\Q^d$ and $\vr\in\Q^d$ iff there is a $\lambda\in \Q$ such that
$\vq=\lambda\vp+(1-\lambda)\vr$ and $0<\lambda<1$.  This situation is
denoted by $\Df{\Bw}(\vp,\vq,\vrr)$.\index{$\Bw$}

Let $\vp,\vq,\vr \in \Q^d$ and $\mu\in\Q$ such that $\Bw(\vp,\mu\vq,
\vrr)$.  In this case we use notations
$\Df{\convex}(\vp,\vq,\vrr)$\index{$\convex$} and
$\Df{\concave}(\vp,\vq,\vrr)$\index{$\concave$} if $1<\mu$ and
$0<\mu<1$, respectively.

For convenience, we introduce the following notation:
\begin{equation*}\index{${}^\ddag p$}
\Df{{}^\ddag \vp}\leteq \left\{
\begin{array}{lll}
\phantom{-}\vp &\text{ if }& p_t \ge 0,\\
-\vp & \mbox{ if }& p_t < 0.
\end{array}
\right.
\end{equation*}

\begin{figure}[h!btp]
\begin{center}
\psfrag{p}{$\vp$}
\psfrag{q1}{$\vq_3$}
\psfrag{q2}{$\vq_2$}
\psfrag{q3}{$\vq_1$}
\psfrag{r}[bc][bc]{${}^\ddag \vr,\vr$}
\psfrag{p+}{${}^\ddag \vp$}
\psfrag{o}[tr][tr]{$\vo$}
\psfrag{convex}[bl][bl]{convex}
\psfrag{flat}[l][l]{flat}
\psfrag{concave}[bl][bl]{concave}
\includegraphics[keepaspectratio, width=0.6\textwidth]{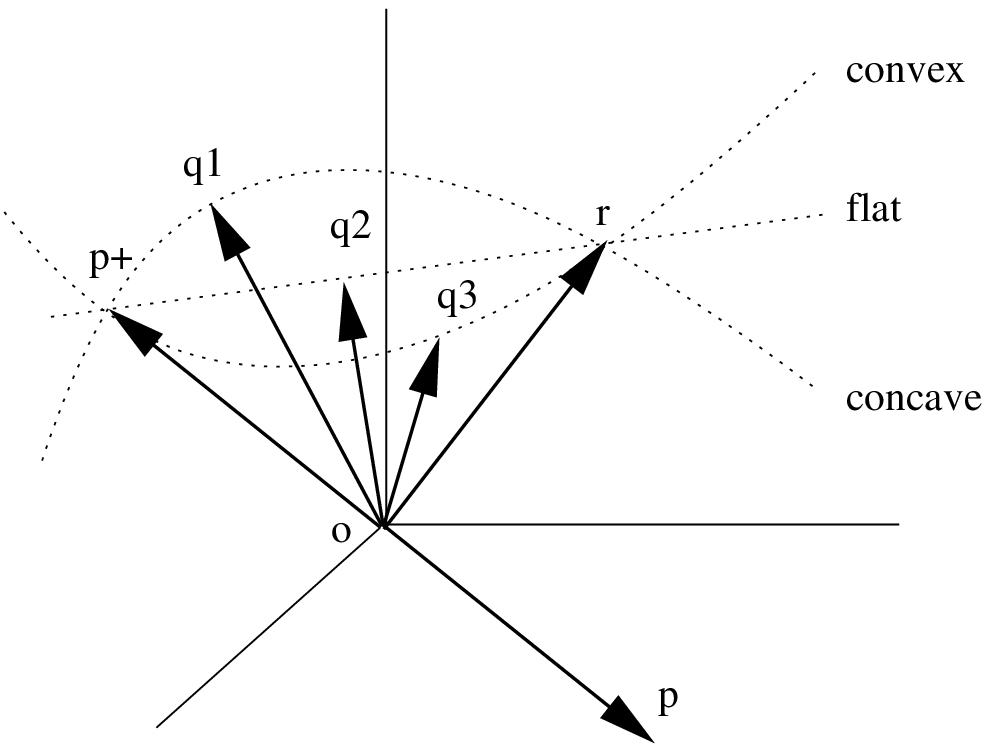}
\caption{\label{fig-conv} Illustration of relations $\convex({}^\ddag\vp,\vq_1,\vr)$, $\Bw({}^\ddag\vp,\vq_2,\vr)$ and $\concave({}^\ddag\vp,\vq_3,\vr)$}
\end{center}
\end{figure}

\begin{prop}\label{prop-tp}
Assume \ax{Kinem_0}.
Let $m$, $a$, $b$, and $c$ be {\it inertial} observers and $e$, $e_a$ and $e_c$ events such that $\mtwp_m(\widehat{ac},b)(e_a,e,e_c)$.
Then 
\begin{alignat*}{3}
&\time(\widehat{ac}<b)(e_a,e,e_c) &\quad &\Iff &\quad &\convex({}^\ddag1_m^a,{}^\ddag1_m^b,{}^\ddag1_m^c),\\
&\time(\widehat{ac}=b)(e_a,e,e_c) & &\Iff & &\Bw({}^\ddag1_m^a,{}^\ddag1_m^b,{}^\ddag1_m^c),\\
&\time(\widehat{ac}>b)(e_a,e,e_c) & &\Iff & &\concave({}^\ddag1_m^a,{}^\ddag1_m^b,{}^\ddag1_m^c).
\end{alignat*}
\end{prop}

\begin{proof}
Let $m$, $a$, $b$, and $c$ be {\it inertial} observers and $e$, $e_a$ and $e_c$
events such that $\mtwp_m(\widehat{ac},b)(e_a,e,e_c)$. Let us
abbreviate time-unit vectors ${}^\ddag \vekm$ as $k^\ddag$
throughout this proof. Let $\vp=\loc_m(e_a)$, $\vq=\loc_m(e)$ and
$\vr=\loc_m(e_c)$. We have that $\vp\neq \vr$ since $p_\tau<r_\tau$
or $r_\tau<p_\tau$. Therefore, by \ax{AxLinTime}, the triangle
$\vpp\vqq\vr$ is nondegenerate since $\vp,\vr\in \wl_m(b)$ but
$\vq\not\in \wl_m(b)$. Let us first show that $b$ measures the same
length of time between $e_a$ and $e_c$ as $a$ and $c$ together if
$\Bw(a^\ddag,b^\ddag,c^\ddag)$ holds. Let $\vs$ be the intersection
of $line(\vp,\vrr)$ and the line parallel to $line(a^\ddag,
c^\ddag)$ through $\vq$, see Fig.~\ref{twpp}. Then the triangles
$\voo a^\ddag b^\ddag$ and $\vpp\vqq\vs$ are similar; and the
triangles $\voo b^\ddag c^\ddag$ and $\vrr\vs\,\vqq$ are similar.
Thus
\begin{equation*}
\frac{|\vp-\vqq|}{|a^\ddag|}=\frac{|\vp-\vs\,|}{|b^\ddag|} \;\text{ and }\; \frac{|\vq-\vrr|}{|c^\ddag|}=\frac{|\vs-\vrr|}{|b^\ddag|}
\end{equation*}
hold.
From which, by \ax{AxLinTime}, it follows that
\begin{multline*}
\Big|\time_a(e_a,e)\Big|+\Big|\time_c(e,e_c)\Big| = \frac{|\vp-\vqq|}{|a^\ddag|}+\frac{|\vq-\vrr|}{|c^\ddag|}\\
=\frac{|\vp-\vs\,|+|\vs-\vrr|}{|b^\ddag|}=\frac{|\vr-\vpp|}{|b^\ddag|}= \Big|\time_c(e_a,e_c)\Big|.
\end{multline*}
Hence $\time(\widehat{ac}=b)(e_a,e,e_c)$ holds if
$\Bw(a^\ddag,b^\ddag,c^\ddag)$.  By \ax{AxLinTime}, $b$ measures more
(less) time between $e_a$ and $e_c$ iff its time-unit vector is
shorter (longer).  Thus we get that $\time(\widehat{ac}<b)(e_a,e,e_c)$
holds if $\convex(a^\ddag,b^\ddag,c^\ddag)$, and
$\time(\widehat{ac}>b)(e_a,e,e_c)$ holds if
$\concave(a^\ddag,b^\ddag,c^\ddag)$.  The converse implications also
hold since one of the relations $\convex$, $\Bw$ and $\concave$ holds
for $a^\ddag$, $b^\ddag$ and $c^\ddag$, and only one of the relations
$\time(\widehat{ac}<b)$, $\time(\widehat{ac}=b)$ and
$\time(\widehat{ac}>b)$ can hold for events $e_a$, $e$ and $e_c$.
This completes the proof.
\end{proof}

A set $H\subseteq \Q^d$ is called \df{convex}\index{convex} iff
$\convex(\vp,\vq,\vrr)$ for all $\vp,\vq,\vr\in H$ for which there is a $\mu\in
\Q^+$ such that $\Bw(\vp,\mu\vq,\vrr)$ holds.  We call $H$
\df{flat}\index{flat} or \df{concave}\index{concave} if
$\convex(\vp,\vq,\vrr)$ is replaced by $\Bw(\vq,\vr,\vpp)$ or
$\concave(\vr,\vp,\vqq)$, respectively.  
\begin{rem}\label{rem-conv}
If there are no $\vp,\vq,\vr\in H$ for which there is a $\mu\in \Q^+$
such that $\Bw(\vp,\mu\vq,\vrr)$ holds, then $H$ is convex, flat and
concave at the same time. To avoid
these undesired situations, let us call $H$
\df{nontrivial}\index{nontrivial convex set} if there are
$\vp,\vq,\vr\in H$ such that $\Bw(\vp,\mu\vq,\vrr)$ holds for a
$\mu\in \Q^+$. By the respective definitions, it is easy to see that
any nontrivial convex (flat, concave) set intersects a halfline at
most once. 
\end{rem}
\noindent 
Let us define the
\df{Minkowski sphere}\index{Minkowski sphere} as
\index{$MS^\ddag_m$}
$\Df{MS^\ddag_m}\leteq \Setopen {}^\ddag \vekm \setmid k\in\IOb\Setclose.$
\begin{rem}
Convexity as used here is not far from convexity as understood in
geometry or in the case of functions. For example, in the models of
$\ax{Kinem_0}+\ax{AxThExp^+}$ or $\ax{SpecRel_0}+\ax{AxThExp}$ the
Minkowski Sphere $MS^\ddag_m$ is convex in our sense iff the set of
points above it (i.e., $\setopen \vp\in\Q^d \setmid \exists \vq\in
MS^\ddag_m\enskip p_\tau\ge q_\tau\setclose$) is convex in the
geometrical sense. Axioms \ax{AxThExp^+} and \ax{AxThExp} are
introduced on pp.\ \pageref{axthex+} and \pageref{axthex},
respectively.
\end{rem}
\begin{rem}\label{rem-convMS}
By Rem.~\ref{rem-conv}, if $MS^\ddag_m$ is a nontrivial convex (flat,
concave) set, it intersects a line at most once.
\end{rem}

\noindent
Now we can state the following corollary of Prop.~\ref{prop-tp}.
\begin{cor} 
\label{cor-twp}
Assume \ax{Kinem_0}.
Then
\begin{alignat*}{3}
&\forall m\in\IOb\enskip MS^\ddag_m \text{ is convex} &\enskip &\Then &\enskip & \ax{CP}, \\
&\forall m\in\IOb\enskip MS^\ddag_m \text{ is flat} & &\Then & &\ax{NoCP}, \\
&\forall m\in\IOb\enskip MS^\ddag_m \text{ is concave}& &\Then & &\ax{AntiCP}.
\end{alignat*}
\end{cor}

The implications in Cor.~\ref{cor-twp} cannot be reversed because
there may be {\it inertial} observers that are not part of any clock paradox
situation.  We can solve this problem by using the following axiom to
shift {\it inertial} observers in order to create clock paradox situations.
\begin{description}
\item[\Ax{AxShift}]\index{\ax{AxShift}} Any {\it inertial} observer observing another
  {\it inertial} observer with a certain time-unit vector also observes still another
  {\it inertial} observer, with the same time-unit vector, at each coordinate point of
  its coordinate domain:
\begin{equation*}
\forall m,k\in \IOb\enskip \forall \vp\in Cd_m\; \exists h \in\IOb \quad h\in ev_m(\vpp) \lland \vekm=\vehm.
\end{equation*} 
\end{description}
Now we can reverse the above implications.
\begin{thm}\label{thm-clp}
Assume \ax{Kinem_0} and \ax{AxShift}.
Then
\begin{alignat*}{3}
&\ax{CP} &\enskip &\Iff &\enskip & \forall m\in\IOb\enskip MS^\ddag_m \text{ is convex,}\\
&\ax{NoCP} & &\Iff & &\forall m\in\IOb\enskip MS^\ddag_m \text{ is flat,}\\
&\ax{AntiCP} & &\Iff & &\forall m\in\IOb\enskip MS^\ddag_m \text{ is concave}.
\end{alignat*}
\end{thm}

\begin{proof}
By Cor.~\ref{cor-twp}, we have to prove the ``$\Longrightarrow$''
part only.  For that, let us take three points $a'$, $b'$ and $c'$
from $MS^\ddag_m$ for which there is a $\mu\in \Q$ satisfying
$\Bw({}^\ddag a',\mu b', {}^\ddag c')$. If there are no such points,
$MS^\ddag_m$ is convex, flat and concave at the same time, see Rem.~\ref{rem-conv}. Otherwise,  by \ax{AxShift}, there are {\it inertial} observers $a$, $b$ and $c$
in a clock paradox situation such that $1^a_m=a'$, $1^b_m=b'$ and
$1^c_m=c'$.  Thus from Prop.~\ref{prop-tp} we get that
$MS^\ddag_m$ has the desired property.
\end{proof}

\section{Consequences for Newtonian kinematics}
\label{conk-sec}

Let us investigate the connection between No-CP and the Newtonian
assumption of the absoluteness of time.

\begin{description}
\item[\Ax{AbsTime}]\index{\ax{AbsTime}} All {\it inertial} observers measure the same
   elapsed time between any two events:
\begin{equation*}
\forall m,k\in \IOb\enskip \forall e_1,e_2\in Ev\quad \time_m(e_1,e_2)=\time_k(e_1,e_2).
\end{equation*}
\end{description}

To strengthen our axiom system, we introduce two axioms that ensure
the existence of several {\it inertial} observers.

\begin{description}
\item[\Ax{AxThExp^+}]\label{axthex+}\index{\ax{AxThExp^+}} {\it Inertial} observers
  can move in any direction at any finite speed:
\begin{equation*}
\forall m\in \IOb\enskip \forall \vp,\vq\in \Q^d\quad p_\tau\neq q_\tau 
\then \exists k\in\IOb\quad k\in ev_m(\vpp)\cap ev_m(\vqq).
\end{equation*}
\end{description}
Let us also introduce a less theoretical version of this axiom.
\begin{description}
\item[\Ax{AxThExp^*}]\label{axthex*}\index{\ax{AxThExp^*}} {\it Inertial} observers
  can move in any direction at a speed which is arbitrarily close to any finite speed:
\begin{multline*}
\forall m\in \IOb\enskip \forall \vp,\vq\in \Q^d\enskip\forall
\varepsilon\in \Q^+ \quad p_\tau\neq q_\tau \\\then
\exists k\in\IOb\enskip \exists {\vqq}'\in \Q^d\quad
|\vq-{\vqq}'|<\varepsilon \lland k\in ev_m(\vpp)\cap ev_m({\vqq}').
\end{multline*}
\end{description}

By the following theorem, \ax{NoCP} logically implies \ax{AbsTime} if
\ax{AxThExp^+} (and some auxiliary axioms) are assumed; however, if we
assume the more experimental axiom \ax{AxThExp^*} instead of \ax{AxThExp^+}, \ax{AbsTime}
does not follow from \ax{NoCP}, which is an astonishing fact since it
means that without the theoretical assumption \ax{AxThExp^+} we would
not be able to conclude that time is absolute in the Newtonian sense
even if there were no clock paradox in our world.

\begin{thm} 
\label{thm-univtime}
\begin{align}
\label{eq-notwp1}
\ax{AbsTime}&\models \ax{NoCP}, \text{ and}\\
\label{eq-notwp1.5}
\ax{Kinem_0}+\ax{AxShift}+\ax{AxThExp^+}+\ax{NoCP}&\models \ax{AbsTime}, \text{ but}\\
\label{eq-notwp2}
\ax{Kinem_0}+\ax{AxShift}+\ax{AxThExp^*}+\ax{NoCP}&\not\models \ax{AbsTime}.
\end{align}
\end{thm}

\begin{proof}
Item \eqref{eq-notwp1} is obvious.

To prove \eqref{eq-notwp1.5}, let us note that $MS^\ddag_m$ is flat by
Thm.~\ref{thm-clp} since \ax{Kinem_0}, \ax{AxShift} and \ax{NoCP} are
assumed. By axiom \ax{AxThExp^+}, $MS^\ddag_m$ intersects any
nonhorizontal line. So $MS^\ddag_m$ has to be a horizontal hyperplane
containing $\langle 1,0,\ldots,0 \rangle$. Hence the time components
of time-unit vectors are the same for every {\it inertial}
observer. So \ax{AbsTime} follows from the assumptions.

To prove \eqref{eq-notwp2}, we construct a model in which
\ax{Kinem_0}, \ax{AxShift}, \ax{AxThExp^*} and \ax{NoCP} hold, but
\ax{AbsTime} does not.  Let $\langle \Q;+,\cdot,<\rangle$ be any
Euclidean ordered field.  Let $B\leteq \Q^d\times\Q^d$.  Let
$\IOb\leteq \setopen \langle \vp,\vqq\rangle\in\B\setmid p_\tau\neq
q_\tau \lland p_\tau-q_\tau\neq p_2-q_2\setclose$.  Let
\begin{equation*}
MS^\ddag_{\langle 1,0\rangle}\leteq \Setopen x\in\Q^d\setmid x_\tau-x_2=1 \lland x_\tau>0\Setclose.
\end{equation*}
Let $W(\langle1,0\rangle,\langle \vp,\vqq\rangle,\vrr)$ hold iff $\vr$
is in $line(\vp,\vqq)$.  Now the worldview relation is given for {\it
  inertial} observer $\langle1,0\rangle$.  For any other {\it
  inertial} observer $\langle \vp,\vqq\rangle$, let $w^{\langle
  \vp,\vqq\rangle}_{\langle1,0\rangle}$ be an affine transformation
that takes $\vo$ to $\vp$ while its linear part takes $\vet$ to
$MS^\ddag_{\langle 1,0\rangle}\cap\setopen
\lambda(\vp-\vqq):\lambda\in \Q\setclose$, and leaves the other basis
vectors fixed.  From these worldview transformations, it is easy to define
the worldview relations of other {\it inertial} observers, hence our
model is given.  It is not difficult to see that \ax{Kinem_0},
\ax{AxShift} and \ax{AxThExp^*} are true in this model.  Since
$MS^\ddag_{\langle 1,0\rangle}$ is flat and the worldview
transformations are affine ones, it is clear that $MS^\ddag_m$ is flat
for all $m\in \IOb$.  Hence \ax{NoCP} is also true in this model by
Cor.~\ref{cor-twp}.  It is easy to see that \ax{AbsTime} implies that
$(\vekm)_\tau=\pm1$ for all $m,k\in\IOb$.  Hence \ax{AbsTime} is not
true in this model, as we claimed.
\end{proof}

\section{Consequences for special relativity theory}
\label{consr-sec}

Now we investigate the consequences of our characterization for
special relativity. To do so, let us first note that if $d\ge3$, our
theory \ax{SpecRel_0} is strong enough to prove the most important
predictions of special relativity, such as that moving clocks get out
of synchronism, see Section~\ref{sec-srwv}. At the same time, \ax{SpecRel_0} is
weak enough not to prove every prediction of special
relativity. For example, it does not entail CP or the slowing down of
relatively moving clocks. Thus it is possible to compare these
predictions within models of \ax{SpecRel_0}. To investigate the
logical connection between them, let us formulate the slowing down
effect on moving clocks within our FOL framework.

\begin{description}
\item[\Ax{SlowTime}]\index{\ax{SlowTime}} Relatively moving {\it inertial} observers' clocks slow down:
\begin{equation*}
\forall m,k\in \IOb\quad \wl_m(k)\neq \wl_m(m) \then \left|(\vekm)_\tau\right|>1.
\end{equation*}
\end{description}

To prove a theorem about the logical connection between \ax{SlowTime}
and \ax{CP}, we need the following lemma, which states that the fact
that three {\it inertial} observers are in a clock paradox situation does not depend on
the {\it inertial} observer that watches them.

\begin{lem}
\label{lem-mtwp}
Let $d\ge3$.
Assume \ax{AxPh}, \ax{AxEv} and \ax{AxLinTime}.
Let $m,a,b,c\in\IOb$ and let $e_a,e,e_b\in Ev$.
Then
\begin{equation*}
\mtwp_m(\widehat{ac},b)(e_a,e,e_c)\iff \mtwp_b(\widehat{ac},b)(e_a,e,e_c).
\end{equation*}
\end{lem}

\begin{proof}
By (1) of Thm.~\ref{thm-poi}, \ax{AxPh} and \ax{AxEv} imply that
$w^b_m$ is a composition of a Poincar\'e transformation, a dilation
and a field-automorphism-induced map.  By \ax{AxLinTime}, the
field-automorphism is trivial.  Hence $\time_m(e)$ is between
$\time_m(e_a)$ and $\time_m(e_c)$ iff $\time_b(e)$ is between
$\time_b(e_a)$ and $\time_b(e_c)$.  This completes the proof since the
other parts of our definition of $\mtwp$ do not depend on {\it
  inertial} observers $m$ and $b$.
\end{proof}

We cannot consistently extend \ax{SpecRel_0} by axiom \ax{AxThExp^+} since
\ax{SpecRel_0} implies the impossibility of faster than light motion
of {\it inertial} observers if $d\ge3$, see, e.g., \cite{AMNsamples}. That is why we have
to weaken this axiom.

\begin{description}
\item[\Ax{AxThExp}]\index{\ax{AxThExp}}\label{axthex} {\it Inertial} observers can
  move in any direction at any speed slower than 1, i.e., the speed of
  light:
\begin{equation*}
\forall m\in \IOb\enskip \forall \vp,\vq\in \Q^d\quad
|\vp_\sigma-\vq_\sigma|<|p_\tau-q_\tau| \then \exists
k\in\IOb\quad k\in ev_m(\vpp)\cap ev_m(\vqq).
\end{equation*}
\end{description}

The following theorem shows that the slowing down of moving clocks
(\ax{SlowTime}) is logically stronger than \ax{CP}.

\begin{thm}
\label{thm-slowtime} 
Let $d\ge3$. Then
\begin{align}
\label{eq-slowtime1}
\ax{SpecRel_0}+\ax{AxLinTime}+\ax{SlowTime}&\models \ax{CP},\text{ but}\\
\label{eq-slowtime2}
\ax{SpecRel_0}+\ax{AxLinTime}+\ax{AxShift}+\ax{AxThExp}+\ax{CP}&\not\models \ax{SlowTime}.
\end{align}
\end{thm}

\begin{proof}
Item \eqref{eq-slowtime1} is clear by Lem.~\ref{lem-mtwp}.

To prove Item \eqref{eq-slowtime2}, let us construct a model in which
\ax{SpecRel_0}, \ax{AxLinTime}, \ax{AxShift}, \ax{AxThExp} and \ax{CP}
hold, but \ax{SlowTime} does not. Let $\langle \Q;+,\cdot,<\rangle$ be
any Euclidean ordered field. Let $B\leteq \Q^d\times\Q^d$. Let
$\IOb\leteq \setopen \langle \vp,\vqq\rangle\in\B\setmid
|\vp_\sigma-\vq_\sigma|<|p_\tau-q_\tau|\setclose$. It is easy to see
that there is a nontrivial convex subset $M$ of $\Q^d$ such that
$\vet\in M$ and $|p_\tau|<1$ for some $\vp\in M$. Let
$MS^\ddag_{\langle 1,0\rangle}$ be such a convex subset of $\Q^d$. Let
$W(\langle1,0\rangle,\langle \vp,\vqq\rangle,\vrr)$ hold iff $\vr$ is
in $line(\vp,\vqq)$. Now the worldview relation is given for {\it
  inertial} observer $\langle1,0\rangle$. By Rem.~\ref{rem-convMS},
$MS^\ddag_{\langle 1,0\rangle}$ intersects a line at most once. For
any other {\it inertial} observer $\langle \vp,\vqq\rangle$, let
$w^{\langle \vp,\vqq\rangle}_{\langle1,0\rangle}$ be such a
composition of a Lorentz transformation, a dilation and a translation
which takes $\vo$ to $\vp$ while its linear part takes $\vet$ to the
unique element of $MS^\ddag_{\langle 1,0\rangle}\cap\setopen
\lambda(\vp-\vqq):\lambda\in \Q\setclose$, and leaves the other basis
vectors fixed. It is easy to see that there is such a transformation. From
these worldview transformations, it is easy to define the worldview
relations of the other {\it inertial} observers. So the model is
given. It is not difficult to see that \ax{SpecRel_0}, \ax{AxShift},
\ax{AxLinTime} and \ax{AxThExp} are true in this model. Since
$MS^\ddag_{\langle 1,0\rangle}$ is convex and the worldview
transformations are affine ones, it is clear that $MS^\ddag_m$ is
convex for all $m\in \IOb$. Hence \ax{CP} is also true in this model
by Cor.~\ref{cor-twp}. It is clear that \ax{SlowTime} is not true
in this model since there is a $\vp\in MS^\ddag_{\langle 1,0\rangle}$
such that $|p_\tau|<1$ (i.e., there is $k\in\IOb$ such that
$|(1^k_{\langle 1,0\rangle})_\tau|<1$); and that completes the proof.
\end{proof}

Like the similar results of \cite{mytdk} and \cite{mythes},
Thm.~\ref{thm-simdist} answers Question 4.2.17 of
Andr\'eka--Madar\'asz--N\'emeti \cite{pezsgo}.  It shows that \ax{CP}
is logically weaker than the symmetric distance axiom of \ax{SpecRel}.

\begin{thm}
\label{thm-simdist}
Let $d\ge3$.
Then
\begin{align}
\label{eq-simdist1}
 \ax{SpecRel_0}+\ax{AxSymDist}&\models \ax{CP}, \text{ but}\\
\label{eq-simdist2}
 \ax{SpecRel_0}+\ax{AxLinTime}+\ax{AxShift}+\ax{AxThExp}+\ax{CP}&\not\models\ax{AxSymDist}.
\end{align}
\end{thm}

\begin{proof}
By (2) of Thm.~\ref{thm-poi}, \ax{SpecRel_0} and \ax{AxSymDist} imply
that $w^k_m$ is a Poincar\'e transformation for all $m,k\in\IOb$.
Hence
\begin{equation*}
 MS^\ddag_m\subseteq\Setopen \vp\in \Q^d\setmid p_\tau^2-
 {|\vp_\sigma|}^2 =1 \lland p_\tau>0\Setclose.
\end{equation*} 
Consequently, $MS^\ddag_m$ is convex.  So by Cor.~\ref{cor-twp},
\ax{CP} follows from \ax{SpecRel_0} and \ax{AxSymDist}.

Since \ax{SpecRel_0} and \ax{AxSymDist} imply \ax{SlowTime} if
$d\ge3$, Item \eqref{eq-simdist2} follows from Thm.~\ref{thm-slowtime}.
\end{proof}

It is interesting that, if the quantity part is the field of real
numbers, \ax{AxSymDist} and \ax{SlowTime} are equivalent in the models
of \ax{SpecRel_0} and some auxiliary axioms.  However, that the
quantity part is the field of real numbers
cannot be formulated in any FOL language of spacetime theories.  Consequently, nor can  Thm.~\ref{thm-eqv}, so it cannot be formulated and proved within our FOL frame
either.
\begin{thm}
\label{thm-eqv}
Assume \ax{SpecRel_0}, \ax{AxLinTime}, \ax{AxShift}, \ax{AxThExp}, and
that $\Q$ is the field of real numbers.  Then
\begin{align*}
 \ax{SlowTime}\Iff\ax{AxSymDist}.
\end{align*}
\end{thm}
For proof of Thm.~\ref{thm-eqv}, see \cite[\S 3]{mythes}.  This theorem
is interesting because it shows that assuming only that all moving
clocks slow down to some degree implies the exact ratio of the slowing
down of moving clocks (since $\ax{SpecRel_0}+\ax{AxSymDist}$ implies
that the worldview transformations are Poincar\'e ones, see
Thm.~\ref{thm-poi}).
\begin{que}
Does Thm.~\ref{thm-eqv} retain its validity if the assumption that $\Q$
is the field of real numbers is removed? If not, is it still possible
to replace it by a FOL assumption, e.g., by axiom schema \ax{CONT}
used in \cite{twp}, \cite{folfrt}, \cite{mythes} and Chaps.\
\ref{chp-twp}, \ref{chp-grav} and \ref{chp-a}?
\end{que}

We have seen that (the inertial approximation of) CP can be
characterized geometrically within a weak axiom system of kinematics.
We have seen some consequences of this characterization; in
particular, CP is logically weaker than the assumption of the slowing
down of moving clocks or the \ax{AxSymDist} axiom of special
relativity. A future task is to explore the logical connections
between other assumptions and predictions of relativity theories. For
example, in \cite{twp}, \cite{mythes} and Chap.~\ref{chp-twp},
\ax{SpecRel_0^\d}+\ax{AxSymDist} is extended to an axiom system
\ax{AccRel} logically implying the twin paradox (the accelerated
version of CP), but the natural question below, raised by
Thm.~\ref{thm-simdist}, has not been answered yet.
\begin{que} 
Is it possible to weaken \ax{AxSymDist} to \ax{CP} in \ax{AccRel}
without losing the twin paradox as a consequence? See \cite[Que.3.8]{twp}.
\end{que}

\chapter{Extending the axioms of special relativity for dynamics}
\label{chp-dyn}

Another surprising prediction of relativity theory is the equivalence
of mass and energy. To find an axiomatic basis to this prediction, we
have to extend our approach to dynamics.  The results of this chapter
are based on \cite{dyn} and \cite{dyn-studia}.

The idea is that we use collisions for measuring relativistic mass. We
could say that the relativistic mass of a body is a quantity that
shows the magnitude of its influence on the state of motion of the
other bodies it collides with. The bigger the relativistic mass of a
body is, the more it changes the motion of the bodies colliding with
it. To be able to formulate this idea, let us extend our FOL
language by a new $(\d+3)$-ary relation $\Df{\M}$\index{$\M$} for
relativistic mass.  We use this relation to speak about the
relativistic masses of bodies according to observers by translating
$\M(b,\vp,x,k)$ as ``the relativistic mass of body $b$ at coordinate
point $\vp$ is $x$ according to observer $k$.'' Since there can be
more than one $x$ which is $\M$-related to $b$, $\vp$ and $k$, we
introduce the following definition: \df{the relativistic mass of body
  $b$ at $\vp\in\Q^d$ according to observer $k$}\index{relativistic
  mass}, in symbols $\Df{\m_k(b,\vpp)}$\index{$\m_k(b,\vpp)$}, is
defined as $x$ if $\M(b,\vp,x,k)$ holds and there is only one such
$x\in\Q$; otherwise $\m_k(b,\vpp)$ is undefined.

\section{Axioms of dynamics}

In this section we introduce a FOL axiomatic
theory of special relativistic dynamics. In our first axiom on
relativistic mass, we assume that it is positive in meaningful and
zero in meaningless situations.

\newpage
\begin{description}
\item[\Ax{AxMass}]\index{\ax{AxMass}} According to any observer, the
  relativistic mass of a body $b$ at any coordinate point $\vp$ is
  defined and nonnegative, and it is zero iff $b$ is not present at $\vpp$:
\begin{equation*}
\forall k \in \Ob \enskip \forall b\in \B\enskip \forall \vp\in\Q^d \quad \m_k(b,\vpp)\ge0 \lland \big(\, \m_k(b,\vpp)=0 \iff b\not\in ev_k(\vpp)\,\big).
\end{equation*}
\end{description}
In our co-authored papers \cite{dyn} and \cite{dyn-studia}, this axiom
was built into the logic frame.

To formulate our other axioms on relativistic mass, first we have to define
collisions. To do so, we introduce the following concepts: the set
 of incoming bodies $in_k(\vqq)$ and
that of outgoing bodies $out_k(\vqq)$ of a collision
at coordinate point $\vq$ according to observer $k$ are defined as
bodies whose world-lines ``end'' and ``start'' at $\vq$, respectively (see
Fig.~\ref{inecoll}):
\begin{eqnarray*}\index{$in_k(\vqq)$}\index{$out_k(\vqq)$} 
\Dff{in_k(\vqq)} & \leteq & \Setopen b\in\B\: :\: \vq\in \wl_k(b)\lland
\forall \vp\in
\wl_k(b)\quad p_\tau<q_\tau \llor \vp=\vqq\Setclose,\\
\Dff{out_k(\vqq)} & \leteq & \Setopen b\in\B\: :\: \vq\in \wl_k(b)\lland
\forall \vp\in \wl_k(b)\quad p_\tau>q_\tau\llor \vp=\vqq\Setclose.
\end{eqnarray*}

Bodies $b_1,\ldots, b_n$ collide originating bodies $d_1,\ldots, d_m$
according to observer $k$, in symbols {$\Df{\coll_k(b_1,\ldots,
  b_n:d_1,\ldots, d_m)}$\index{$\coll$}}, iff $b_i\neq b_j$ and
$d_i\neq d_j$ whenever $i\neq j$ and there is a coordinate point $\vq$ such
that $in_k(\vqq)=\{b_1,\ldots,b_n\}$ and
$out_k(\vqq)=\{d_1,\ldots,d_m\}$. Inelastic collisions are just
collisions in which only one body is originated.
So in this case, we write $\Df{\inecoll_k(b_1,\ldots,
 b_n:d)}$\index{$\inecoll$} in place of $\coll_k(b_1,\ldots, b_n:d)$
and say that bodies $b_1,\ldots, b_n$ \df{collide inelastically}
originating body $d$ according to observer $k$. For the illustration
of these concepts, see Fig.~\ref{inecoll}.

\begin{figure}[ht]
\small
\begin{center}
\psfrag*{q}[l][l]{$\vq$} 
\psfrag*{b}[t][t]{$b$} 
\psfrag*{c}[t][t]{$c$}
\psfrag*{d}[b][b]{$d$} 
\psfrag*{qt}[r][r]{$q_\tau$}
\psfrag*{out}[b][b]{$out_k(\vqq)$} 
\psfrag*{in}[t][t]{$in_k(\vqq)$}
\psfrag*{soft}[l][l]{$\inecoll_k(b,c:d)$}
\psfrag*{k}[l][l]{$k$}
\includegraphics[keepaspectratio, width=0.8\textwidth]{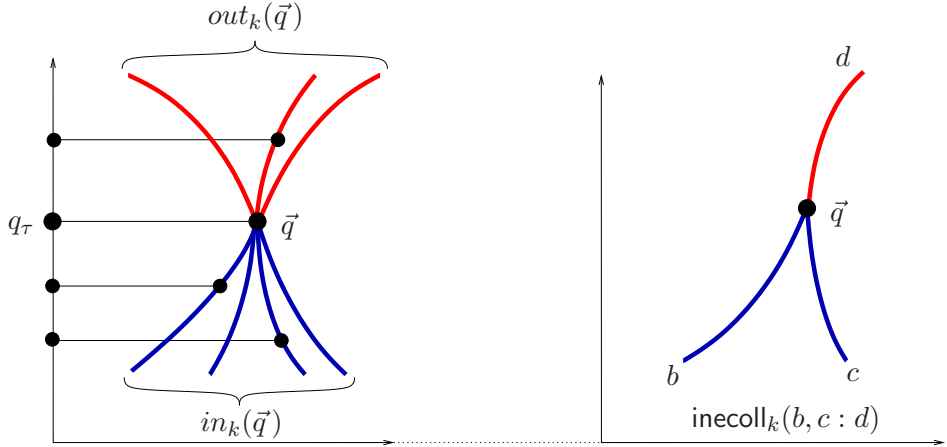}
\caption{\label{inecoll} Illustration of $in_k(\vqq)$, $out_k(\vqq)$ and $\inecoll_k(b,c:d)$ }
\end{center}
\end{figure}

The \df{spacetime location} {$\Dff{\loc_k^b(t)}$} 
of body $b$ at time instance $t\in \Q$ according to 
observer $k$ is defined as 
the coordinate point $\vp$ for which $\vp\in \wl_k(b)$ and $p_\tau=t$ hold if there is 
such a unique $\vp$;  otherwise $\loc_k^b(t)$ is undefined, see Fig.~\ref{med4}.

The \df{center of masses} $\Dff{\cen_k^{b_1,\ldots, b_n}(t)}$ of bodies $b_1,\ldots, b_n$ 
according to $k\in\Ob$ at time instance $t$ is defined by:
\begin{equation*}
\sum_{i=1}^{n} \m_k(b_i,\loc_k^{b_i}(t))\cdot\big(\cen_k^{b_1,\ldots, b_n}(t)-\loc_k^{b_i}(t)\big)=0
\end{equation*}
if $\loc_k^{b_i}(t)$ and $\m_k(b_i,\loc_k^{b_i}(t))$ are defined for all $1\le i\le n$;
  otherwise $\cen_k^{b_1,\ldots, b_n}(t)$ is undefined.
Let us note that the following is an explicit definition for
${\cen_k^{b_1,\ldots, b_n}(t)}$:
\begin{equation*}
 {\cen_k^{b_1,\ldots, b_n}(t)}=\sum^n_{i=1}\frac{\m_k(b_i,\loc^{b_i}_k(t))}{\m_k(b_1,\loc^{b_1}_k(t))+\ldots+\m_k(b_n,\loc^{b_n}_k(t))}\cdot\loc_k^{b_i}(t)
\end{equation*}
if $\loc_k^{b_i}(t)$ and $\m_k(b_i,\loc_k^{b_i}(t))$ are defined for all $1\le i\le n$.
The
\df{center-line of the masses} of bodies $b_1,\ldots, b_n$ according to
observer $k$ is defined as:
\begin{equation*}\index{$\cen_k({b_1,\ldots, b_n})$}
\Df{\cen_k({b_1,\ldots, b_n})}\leteq \Setopen \cen_k^{b_1,\ldots, b_n}(t)\: : \: t\in \Q\text{ and }
\cen_k^{b_1,\ldots, b_n}(t) \text{ is defined}\,\Setclose,
\end{equation*}
i.e., the center-line of mass is the world-line of the center
of mass.
\begin{rem}
Let us note that $\cen_k^b(t)=\loc_k^b(t)$ for all $k\in \Ob$, $b\in
\B$ and $t\in \Q$, and thus $\cen_k(b)=\wl_k(b)$ for every $k\in\Ob$
and $b\in\B$ if $\m_k(b,\loc_k^b(t))$ is defined and nonzero for all
$t\in \dom \loc_k^b$ (e.g., if \ax{AxMass} is assumed).
\end{rem}
\noindent
The segment determined by $\vp,\vq\in\Q^d$ is defined as:
\begin{equation*}\index{$[\vp,\vq\,]$}
\Dff{[\vp,\vq\,]}\leteq \Setopen
\lambda\vp+(1-\lambda)\vq\: : \: \lambda\in\Q,\
0\leq\lambda\leq 1\Setclose.
\end{equation*}
Let us call $H\subseteq \Q^d$ a \df{line segment} if
\begin{itemize} 
\item it is connected, (i.e., $[\vp,\vqq]\subseteq H$ for all $\vp,\vq \in H$),
\item it is a subset of a line, and
\item it has at least two elements.
\end{itemize}
Bodies whose world-lines are line segments according to every {\it inertial} observer are called \df{inertial bodies}, and their set is defined as:
\begin{equation*}\index{$\IB$}
\Dff{\IB}\leteq \setopen b\in \B\: :\: \forall k\in\IOb\quad \wl_k(b)
\text{ is a line segment}\setclose.
\end{equation*}

\begin{figure}[ht]
\small
\begin{center}
\psfrag*{t}[r][r]{$t$}
\psfrag*{b}[r][r]{$b$}
\psfrag*{ab}[r][r]{$\forall b$}
\psfrag*{ad}[l][l]{$\forall d$}
\psfrag*{c}[l][l]{$c$}
\psfrag*{ac}[lb][lb]{$\forall c$}
\psfrag*{k}[r][r]{$k$}
\psfrag*{ak}[r][r]{$\forall k$}
\psfrag*{c1}[t][t]{$\cen_k({b,c})$}
\psfrag*{lb}[rb][rb]{$\loc_k(b,t)$}
\psfrag*{lc}[lb][lb]{$\loc_k(c,t)$}
\psfrag*{mb}[lt][lt]{$\m_k(b)$}
\psfrag*{mc}[rt][rt]{$\m_k(c)$}
\psfrag*{mk}[b][b]{$\cen_k^{b,c}(t)$}
\psfrag*{m}[t][t]{$\cen_k({b,c})$}
\includegraphics[keepaspectratio, width=0.9\textwidth]{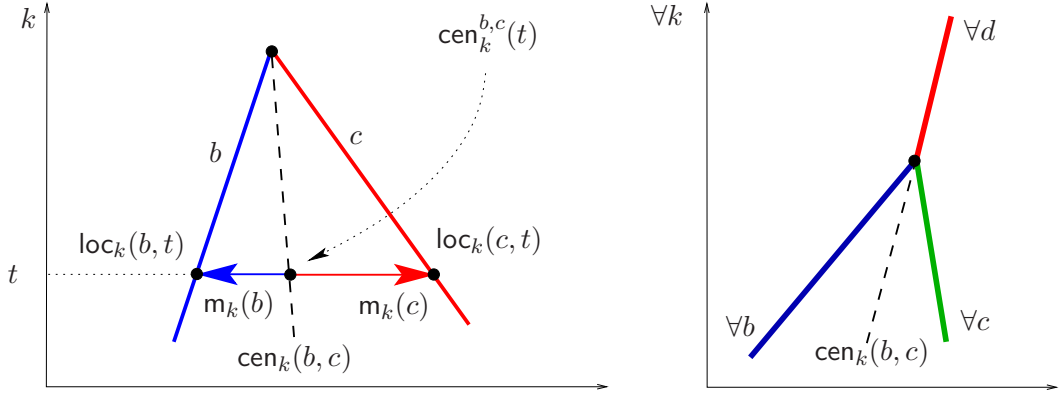}
\caption{\label{med4} Illustration of $\cen_k^{b,c}(t)$, $\cen_k({b,c})$ and of axiom \ax{AxCenter}}
\end{center}
\end{figure}

\begin{prop}
Let $k$ be an {\it inertial} observer and $b_1,\ldots,b_n$ {\it
  inertial} bodies such that, for all $1\le i\le n$, $\vp,\vq\in
\wl_k(b_i)$ imply $\m_k(b_i,\vpp)=\m_k(b_i,\vqq)>0$. Then the
following hold:
\label{prop-cen} 
\begin{enumerate}
\item $\cen_k({b_1,\ldots, b_n})$ is a line segment, a point or empty, 
\item $\cen_k({b_1,\ldots, b_n})$ is nonhorizontal, i.e., $\vr=\vs$ if
  $\vr,\vs\in\cen_k({b_1,\ldots, b_n})$ and $r_\tau=s_\tau$,
\item $\wl_k(b_1)\cap\ldots\cap \wl_k(b_n)\subseteq \cen_k({b_1,\ldots, b_n})$,
\item $\cen_k({b_1,\ldots, b_n})$ is a line segment if
  $\coll_k(b_1,\ldots,b_n:d_1,\ldots,d_m)$ or
  $\coll_k(d_1,\ldots,d_m:b_1,\ldots,b_n)$ for some (not necessarily
  {\it inertial}\,) bodies $d_1,\ldots,d_m$.
\end{enumerate}
\end{prop}
\noindent
Here we omit the easy proof.

Now we are ready to formalize that the relativistic mass of a body is
a quantity that shows the magnitude of its influence on the state of
motion of any other body it collides with.
\begin{description}
\item[\Ax{AxCenter}]\index{\ax{AxCenter}} The world-line of the {\it
 inertial} body originated by an inelastic collision of two {\it
 inertial} bodies is the continuation of the center-line of the masses of the
 colliding {\it inertial} bodies according to every {\it inertial} observer (see
 Fig.~\ref{med4}):
\begin{equation*}
\forall k\in\IOb\; \forall b,c,d\in \IB\quad
\inecoll_k(b,c:d)
\then \cen_k({b,c})\cup \wl_k(d)\subseteq\ell \text{ for some line }\ell.
\end{equation*}
\end{description}

The main axiom of \ax{SpecRelDyn} is \ax{AxCenter} which, in a certain
sense, can be taken as a definition of relativistic mass.
The other axioms of our axiom system will be simplifying or auxiliary ones
 to make life simpler.
We could only get rid of them at the expense
of sacrificing the simplicity of expressions.

\ax{AxCenter} is an axiom in Newtonian Dynamics, too, where the mass
$\m_k(b,\vpp)$ of a body $b$ does not depend on observer $k$ and coordinate point $\vp$.
However, in special relativity,
\ax{AxCenter} implies that the mass of a body necessarily depends on
the observer.
The reason for this fact is that the simultaneities of 
different observers in special relativity may differ from one another,
and this implies that the proportions involved in \ax{AxCenter}
change, too. See \cite[Prop.4.1]{dyn-studia}.

The \df{velocity}\index{velocity} $\bv_k^b(t)$ and \df{speed} $v_k^b(t)$ of body $b$ at instant $t\in\Q$ according to observer $k$ are defined as:
\begin{equation*}\index{$\bv_k^b(t)$}\index{$v_k^b(t)$}
\Df{\bv_k^b(t)}\leteq \big((\loc_k^b)_\sigma\big)'(t)\qquad\text{and}\qquad \Df{v_k^b(t)}\leteq|\bv_k^b(t)|
\end{equation*}
if $\loc_k^b(t)$ is defined and $\loc_k^b$ is differentiable at $t$;  otherwise they are undefined.
(For the FOL definition of $f'(t)$, see Section~\ref{sec-diff}.)
Let us note that 
\begin{equation*}
 \big((\loc_k^b)_\sigma\big)'(t)=\big((\loc_k^b)'(t)\big)_\sigma\qquad\text{and}\qquad\big(\loc_k^b(t)_\tau\big)'=\big((\loc_k^b)'\big)_\tau(t)=1
\end{equation*}
 if $\loc_k^b(t)$ is defined and differentiable.

The \df{rest mass}\index{rest mass} {$\Dff{\m_0(b)}$}\index{$\m_0(b)$}
of body $b$ is defined as $\lambda\in \Q$ if (1) there is an observer
according to which $b$ is at rest and the relativistic mass of $b$ is
$\lambda$, and (2) the relativistic mass of $b$ is $\lambda$ for every
observer according to which $b$ is at rest. That is, $\m_0(b)=\lambda$
if
\begin{equation*}
\begin{split}
&\exists k\in\Ob\enskip \forall t\in\dom v_k^b \quad v_k^b(t)=0 \enskip\, \lland\enskip\, \forall \vp\in\wl_k(b)\quad \m_k(b,\vpp)=\lambda \\
\lland &\forall k\in\Ob\enskip \forall t\in\dom v_k^b \quad v_k^b(t)=0\then \forall \vp\in\wl_k(b)\quad \m_k(b,\vpp)=\lambda
\end{split}
\end{equation*}
if there is such $\lambda$;  otherwise $\m_0(b)$ is undefined.

We have seen that \ax{AxCenter} implies that the relativistic mass
depends on both $b$ and $k$.  Our next axiom states that the
relativistic mass of a body depends on its rest mass and velocity at
the most.

\begin{description}
\item[\Ax{AxSpeed}]\index{\ax{AxSpeed}} According to any {\it
  inertial} observer, the relativistic masses of two {\it inertial} bodies are the
  same if both of their rest masses and speeds are equal:
\begin{multline*}
\forall k\in\IOb\enskip \forall b,c\in\B\enskip\forall \vp,\vq\in\Q^d \quad 
b\in ev_k(\vpp)\lland c\in ev_k(\vqq)\\
\lland\m_0(b)=\m_0(c)\lland v_k^b(p_\tau)=v_k^c(q_\tau)\then \m_k(b,\vpp)=\m_k(c,\vqq).
\end{multline*}
\end{description}
Let ${\B_0}$ be the set of bodies having rest mass,
i.e.,\index{$\B_0$}\index{$\IB_0$}
$\Dff{\B_0}\leteq \Setopen b\in \B\: :\: \m_0(b) \text{ is defined}\Setclose$,
and let  $\IB_0$ be the set of {\it inertial} bodies having rest mass, i.e., $\Df{\IB_0}\leteq \IB\cap\B_0$.

By the following proposition, \ax{AxSpeed} implies that the
relativistic mass of an {\it inertial} body having rest mass does not
change in time according to {\it inertial} observers.
\begin{prop}\label{prop-mib}
\begin{equation*}
\ax{AxSpeed}\models \forall k\in\IOb\; \forall b\in \IB_0\enskip \forall \vp,\vq\in \wl_k(b)\quad \m_k(b,\vpp)=\m_k(b,\vqq).
\end{equation*}
\end{prop}

\begin{proof}
By the respective definitions, it is easy to see that
$v_k^b(p_\tau)=v_k^b(q_\tau)$ for all $\vp,\vq\in\wl_k(b)$.  Hence by
\ax{AxSpeed}, $\m_k(b,\vpp)=\m_k(b,\vqq)$ if $\vp,\vq\in\wl_k(b)$, $k$
is an {\it inertial} observer, and $b$ is an {\it inertial} body
having rest mass.
\end{proof}

Prop.~\ref{prop-mib} leads us to introduce the following definition:
$\Df{\m_k(b)}$ is defined as $\m_k(b,\vpp)$ if
$\m_k(b,\vpp)=\m_k(b,\vqq)$ for all $\vp,\vq\in \wl_k(b)$; otherwise
$\m_k(b)$ is undefined.  So by Prop.~\ref{prop-mib} $\m_k(b)$ is
defined if $b\in\IB_0$, $k\in\IOb$ and \ax{AxSpeed} is assumed.
Similarly, we use notations $\Df{\bv_k(b)}$\index{$\bv_k(b)$} and
$\Df{v_k(b)}$\index{$v_k(b)$} instead of $\bv_k^b(t)$ and $v_k^b(t)$
when $b$ and $k$ are {\it inertial}, as in this case
$\bv_k^b(t_1)=\bv_k^b(t_2)$ for all $t_1,t_2\in\dom \bv_k^b$.

Our last axiom on dynamics states that every observer can make
experiments in which they make {\it inertial} bodies of arbitrary rest
masses and velocities collide inelastically: 

\begin{description}
\item[\Ax{Ax\forall\inecoll}]\index{\ax{Ax\forall\inecoll}} For any
  {\it inertial} observer, any possible kind of inelastic collision of
  {\it inertial} bodies can be realized:
\begin{multline*}
 \forall k\in\Ob\enskip \forall \bv_1,\bv_2\in\Q^{d-1}\; \forall m_1, m_2\in \Q
\quad  |\bv_1|<1\lland |\bv_2|<1 \\
\land \bv_1\neq\bv_2 \lland m_1>0 \lland m_2>0 \then
\exists b,c,d\in\IB \quad \inecoll_k(b,c:d)\;\\
 \lland \bv_k(b)=\bv_1 \lland \bv_k(c)=\bv_2\lland
\m_0(b)=m_1\lland \m_0(c)=m_2.
\end{multline*}
\end{description}

We often add axioms to \ax{SpecRel} which do not change the spacetime
structure, but are useful as auxiliary axioms.  For example,
\ax{AxThExp^\up} below states that every observer can make thought
experiments in which they assume the existence of ``slowly moving''
observers (see, e.g., \cite[p.622 and Thm.11.10]{logst}):

\begin{description}
\item[\Ax{AxThExp^\up}]\index{\ax{AxThExp^\up}} For any {\it inertial}
  observer, in any spacetime location, in any direction, at any
  speed slower than that of light it is possible to ``send out'' an
 {\it inertial} observer whose time flows ``forwards:''
\begin{multline*}
\forall k\in\IOb\;\forall \vp,\vq\in\Q^d\quad 
|(\vp-\vqq)_{\sigma}|<(\vp-\vqq)_{\tau} \\
\then \exists h\in\IOb\quad  h\in ev_k(\vpp)\cap
ev_k(\vqq) \lland w^k_h(\vqq)_\tau<w^k_h(\vpp)_\tau.
\end{multline*}

\end{description}
Let us extend \ax{SpecRel} by \ax{AxThExp^\up} and the axioms of dynamics above:
\begin{equation*}\index{\ax{SpecRelDyn}}
 \boxed{\ax{SpecRelDyn} \leteq 
\Setopen \ax{AxMass}, \ax{AxCenter}, \ax{AxSpeed},\ax{Ax\forall
\inecoll},\ax{AxThExp^\up}\Setclose \cup \ax{SpecRel}}
\end{equation*}
Let us note that \ax{SpecRelDyn} is provably consistent.
Moreover, it has
nontrivial models, see Prop.~\ref{rem}.

The following theorem provides the connection between the rest mass and the
relativistic mass of an {\it inertial} body.
Its conclusion is a well-known result of
special relativity.
We will see that our theorem is stronger than the
corresponding result in the literature since it contains fewer assumptions.
\begin{thm}
\label{thm1}
Let $d\geq 3$. Assume \ax{SpecRelDyn} and let $k$ be an {\it inertial} observer and $b$ be an {\it inertial} body having rest mass.
Then
\begin{equation*}
\m_0(b)={\sqrt{1-v_k(b)^2}}\cdot\m_k(b).
\end{equation*}
\end{thm}
\noindent A purely geometrical proof of Thm.~\ref{thm1} can be found in
\cite{dyn-studia}.

\begin{rem}
Assuming \ax{AxPh}, photons cannot have rest masses since their speed
is 1 according to any {\it inertial} observer. However, by
Thm.~\ref{thm1}, it is natural to extend our rest mass concept for
photons as $\m_0(ph)\leteq 0$ for all $ph\in\Ph$. After this extension
photons may be regarded as ``pure energy'' as they have zero rest
masses.

\end{rem}

\begin{rem}\label{discussion}
The conclusion of Thm.~\ref{thm1} fails if we omit any of the axioms
\ax{AxMass}, \ax{AxCenter}, \ax{AxSpeed}, \ax{Ax\forall \inecoll},
\ax{AxThExp^\up} from \ax{SpecRelDyn}.  However, it remains true if we
weaken \ax{Ax\forall \inecoll} and \ax{AxThExp^\up} to the following
two axioms, respectively:
\begin{description}
\item[\Ax{Ax\exists \inecoll}]\index{\ax{Ax\exists \inecoll}} According to every observer,
for every {\it inertial} body $a$ having rest mass, there are {\it inertial}
bodies $b$ and $c$ colliding inelastically such that $a$, $b$ and $c$ have the same rest masses, $a$ and
$b$ have the same speeds and the speed of $c$ is $0$ (see the left-hand side of Fig.~\ref{pirt}):
\begin{multline*}
 \forall k\in\IOb\; \forall a\in \IB_0\enskip \exists b,c,d\in \IB\quad
 \m_0(a)=\m_0(b)=\m_0(c)\\
\lland v_k(b)=v_k(a)\lland v_k(c)=0 \lland \inecoll_k(b,c:d);
\end{multline*}
\item[\Ax{AxMedian}]\index{\ax{AxMedian}}
For every two {\it inertial} bodies colliding
inelastically, there is an {\it inertial} observer
for which these two {\it inertial} bodies have opposite velocities and
collide inelastically (see the right-hand side of Fig.~\ref{pirt}):
\begin{multline*}
\forall k\in\IOb\;\forall b,c,d\in \IB \quad \inecoll_k(b,c:d)\\ \then
\exists h\in\IOb\quad\bv_h(b)=-\bv_h(c)\lland \inecoll_h(b,c:d).
\end{multline*}
\end{description}
\end{rem}

\begin{figure}[ht]
\small
\begin{center}
\psfrag*{a}[br][br]{$\forall a$} \psfrag*{b}[t][t]{$b$}
\psfrag*{c}[t][t]{$c$} \psfrag*{eb}[br][br]{$\exists b$}
\psfrag*{ec}[lb][lb]{$\exists c$} \psfrag*{ed}[lt][lt]{$\exists d$}
\psfrag*{ab}[br][br]{$\forall b$} \psfrag*{ac}[t][t]{$\forall c$}
\psfrag*{ad}[lt][lt]{$\forall d$} \psfrag*{k}[r][r]{$\forall k$}
\psfrag*{d}[b][b]{$d$} \psfrag*{h}[r][r]{$\exists h$}
\psfrag*{ax1}[l][l]{\ax{Ax\exists \inecoll}}
\psfrag*{ax2}[l][l]{\ax{AxMedian}}
\psfrag*{text}[r][r]{$\m_0(a)=\m_0(b)=\m_0(c)$}
\includegraphics[keepaspectratio, width=0.9\textwidth]{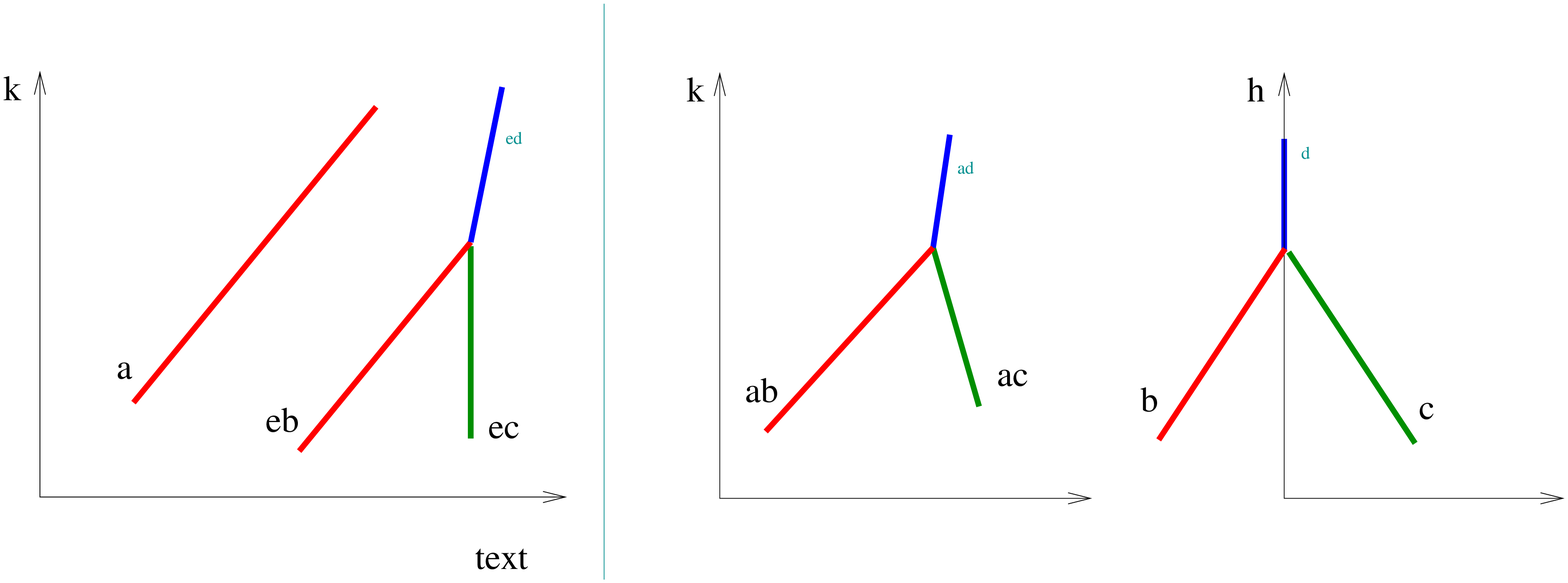}
\caption{\label{pirt} Illustration of axioms \ax{Ax\exists \inecoll}
and \ax{AxMedian}}
\end{center}
\end{figure}

\noindent {\it On Einstein's $E=mc^2$}: The conclusion
$\m_0(b)={\sqrt{1-v_k(b)^2}}\cdot\m_k(b)$ of our Thm.~\ref{thm1} above
is used in Rindler's relativity textbook \cite[pp.111-114]{Rin} to
explain the discovery and meaning of Einstein's famous insight
$\Dff{E=mc^2}$.  We could literally repeat this part of the text of
\cite{Rin} to arrive at $E=mc^2$ in the framework of our theory
\ax{SpecRelDyn} based on the axiom \ax{AxCenter}.  We postpone this to
a later point, because then we will have developed more
``ammunition,'' hence the didactics can be more inspiring.

\section{Conservation of relativistic mass and linear momentum}

In a certain sense \ax{AxCenter} states that the center of mass of an
isolated system consisting of two {\it inertial} bodies moves along a
line regardless whether the two bodies collide or not. It is natural
to generalize \ax{AxCenter} to more than two bodies (but permitting
only two-by-two inelastic collisions). Let $\ax{AxCenter_n}$ denote,
temporarily, a version of \ax{AxCenter} concerning any isolated system
consisting of $n$ bodies. Thus \ax{AxCenter} is just $\ax{AxCenter_2}$
in this series of increasingly stronger axioms. We will see that it
does not imply \ax{AxCenter_3}; thus \ax{AxCenter_3} is strictly
stronger than \ax{AxCenter} if certain auxiliary axioms are assumed,
see Cor.~\ref{cor-fmoment} and Prop.~\ref{prop-carcen}. However, it can
be proved that the rest of the axioms in this series are all
equivalent to $\ax{AxCenter_3}$ if \ax{AxCenter} is assumed, see
Cor.~\ref{cor-cenn}. That motivates us to introduce \ax{SpecRelDyn^+}
by replacing \ax{AxCenter} in \ax{SpecRelDyn} by the stronger
\ax{AxCenter_3}. Our theory \ax{SpecRelDyn^+} is still very geometric
and observation-oriented in spirit.  Let us now introduce
\ax{AxCenter_3}, and denote it as \ax{AxCenter^+}.

\newpage
\begin{description}
\item[\Ax{AxCenter^+}]\index{\ax{AxCenter^+}} If $a$ is an {\it inertial} body and {\it inertial}
bodies $b$ and $c$ collide inelastically originating {\it inertial} body
$d$, the center-line of the masses of $a$ and $d$ is the continuation of the
center-line of the masses of $a$, $b$ and $c$, see Fig.~\ref{centerplus}:
\begin{multline*}
\forall k\in\Ob\; \forall a,b,c,d\in \IB\quad\\
 \inecoll_k(b,c:d)\then
 \cen_k({a,b,c})\cup \cen_k({a,d})\subseteq\ell\text{ for some line
}\ell.
\end{multline*}
\end{description}
\begin{figure}[h!t]
\small
\begin{center}
\psfrag*{a}[br][br]{$\forall a$} \psfrag*{b}[br][br]{$\forall b$}
\psfrag*{c}[bl][bl]{$\forall c$} \psfrag*{d}[bl][bl]{$\forall d$}
\psfrag*{c1}[cb][cb]{$\cen_k({a,b,c})$}
\psfrag*{c2}[cb][cb]{$\cen_k({a,d})$} \psfrag*{k}[r][r]{$\forall k$}
\includegraphics[keepaspectratio, width=0.4\textwidth]{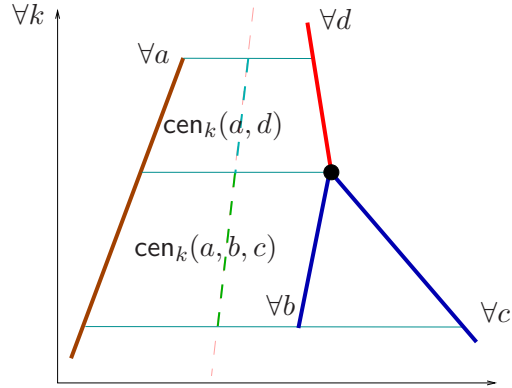}
\caption{\label{centerplus} Illustration of \ax{AxCenter^+}}
\end{center}
\end{figure}
Let us replace \ax{AxCenter} by \ax{AxCenter^+} in
\ax{SpecRelDyn}:
\begin{equation*}\index{\ax{SpecRelDyn^+}}
 \boxed{\ax{SpecRelDyn^+} \leteq 
\Setopen \ax{AxMass},\ax{AxCenter^+}, \ax{AxSpeed},\ax{Ax\forall
\inecoll},\ax{AxThExp^\up}\Setclose \cup \ax{SpecRel}}
\end{equation*}
Let us note that \ax{SpecRelDyn^+} is also consistent.
Moreover, it has
nontrivial models, see Prop.~\ref{rem}.

\ax{AxCenter} determines the velocity of the body emerging from an
inelastic collision, and we will see that \ax{AxCenter^+} also determines
the relativistic mass of the body emerging from the collision.

Let us now formulate the conservation of relativistic mass in our FOL language.
\begin{description}
\item[\Ax{ConsMass}]\index{\ax{ConsMass}} If {\it inertial} bodies $b$
 and $c$ collide inelastically originating {\it inertial} body $d$,
 the relativistic mass of $d$ is the sum of the relativistic masses
 of $b$ and $c$:
\begin{equation*}
\forall k\in\IOb\; \forall b,c,d\in \IB\quad \inecoll_k(b,c:d) 
\then \m_k(b)+\m_k(c)=\m_k(d).
\end{equation*}
\end{description}

The \df{linear momentum} of {\it inertial} body $b$ according to {\it
  inertial} observer $k$ is defined as $\m_k(b)\bv_k(b)$ if $\bv_k(b)$ and $\m_k(b)$ 
are defined, otherwise it is undefined. Now we can formulate the
conservation of linear momentum in our FOL language.
\begin{description}
\item[\Ax{ConsMomentum}]\index{\ax{ConsMomentum}} If {\it inertial} bodies
 $b$ and $c$ collide inelastically originating {\it inertial} body
 $d$, the linear momentum of $d$ is the sum of the linear momentum of
 $b$ and $c$:
\begin{equation*}
\forall k\in\IOb\; \forall b,c,d\in \IB\quad  
\inecoll_k(b,c:d)
\then \m_k(b)\bv_k(b) +\m_k(c)\bv_k(c) =\m_k(d)\bv_k(d).
\end{equation*}
\end{description}
To state a theorem on the connection of \ax{AxCenter}, \ax{ConsMass}
and \ax{ConsMomentum}, we need the following auxiliary axiom.
\begin{description}
\item[\Ax{AxInMass}] According to any {\it inertial} observer, the
  relativistic mass of every {\it inertial} body is constant:
\begin{equation*}
 \forall k\in\IOb\; \forall b\in \IB\enskip \forall \vp,\vq\in
 \wl_k(b)\quad \m_k(b,\vpp)=\m_k(b,\vqq).
\end{equation*}
\end{description}
That is a consequence of \ax{AxSpeed} for {\it inertial} bodies having
rest mass, see Prop.~\ref{prop-mib}.  In \cite{dyn} and
\cite{dyn-studia}, this axiom was also built in our logic
frame. 

The following theorem states that axiom \ax{AxCenter^+} is equivalent to the
conjunction of \ax{ConsMass} and either of the two formulas \ax{AxCenter}
and \ax{ConsMomentum} if certain auxiliary axioms are assumed.  That means
in a sense that \ax{ConsMass} represents the ``difference'' between
\ax{AxCenter} and \ax{AxCenter^+}, and the same holds if \ax{AxCenter} is replaced by
\ax{ConsMomentum}.

\begin{thm}
\label{thm2} Let us assume \ax{AxMass}, \ax{AxInMass}, \ax{AxSelf} and that $\IOb\subseteq\IB$. Then:
\begin{equation*}
\ax{AxCenter^+} \Iff
\ax{ConsMass}\lland\ax{ConsMomentum}\Iff
\ax{ConsMass}\lland\ax{AxCenter}.
\end{equation*}
\end{thm}
\noindent The proof of Thm.~\ref{thm2} is in \cite{dyn}.

\begin{cor}\label{cor}
Let us assume \ax{SpecRelDyn^+}.
Let $k$ be an {\it inertial} observer and $b$, $c$ and $d$ {\it inertial} bodies such that
$\inecoll_k(b,c:d)$ holds.
Then
\begin{eqnarray*}
\m_k(d)&=&\m_k(b)+\m_k(c),\quad \mbox{but}\\
\m_0(d)&>&\m_0(b)+\m_0(c),\quad \mbox{whenever } \bv_k(b)\ne
\bv_k(c).
\end{eqnarray*}
\end{cor}
\noindent
The proof itself is in \cite{dyn}, here we are only concerned with the idea of the proof.
\smallskip

\noindent {\it Returning to $E=mc^2$}: Cor.~\ref{cor} above can be used
to arrive at Einstein's insight $E=mc^2$ in the same way as it is done
in Rindler's~\cite{Rin} and d'Inverno's~\cite{dinverno} relativity
textbooks.  Namely, we have seen above that rest mass can be created
under appropriate conditions.  Created from what? Well, from kinetic
energy (energy of motion).  That points in the direction of Einstein's
connecting mass with energy.  In more detail, let us start with two
bodies $b_1$ and $b_2$ of rest mass $\m_0$.  Let us accelerate the two
bodies towards each other and let them collide inelastically, so that
they stick together forming the new body ``$b_1+b_2$'' (deliberately
sloppy notation).  Let us assume $b_1+b_2$ is at rest relative to the
observer conducting the experiment.  Then the rest mass
$\m_0(b_1+b_2)$ is the sum of relativistic masses $\m_k(b_1)$ and
$\m_k(b_2)$.  Assuming that at collision the speed of both $b_1$ and
$b_2$ were $v$, we have
$\m_0(b_1+b_2)=\m_0(b_1)/\sqrt{1-v^2}+\m_0(b_2)/\sqrt{1-v^2}$, which
is definitely greater than $\m_0(b_1)+\m_0(b_2)$ if $v\ne 0$.  So,
rest mass was created from the kinetic energy supplied to our test
bodies $b_1$ and $b_2$ when they were accelerated towards each other.
So far, we have a qualitative argument (based on our
\ax{SpecRelDyn^+}) in the direction that energy (in our example
kinetic) can be ``transformed'' to ``create'' mass.  A quantitatively
(and physically) more detailed analysis of $E=mc^2$ in terms of
Thm.~\ref{thm1} is given in \cite[pp.111-114]{Rin} where we refer the
reader for more detail and for the ``second part'' of the argument.
The ``first part'' was provided by Thm.~\ref{thm1} and Cor.~\ref{cor}.

\begin{prop}\label{Prop1}
\begin{eqnarray*}
\ax{SpecRelDyn} &\not\models&\ax{ConsMass},\quad \text{ and} \\
\ax{SpecRelDyn} &\not\models&\ax{ConsMomentum}.
\end{eqnarray*}
\end{prop}
\noindent
The proof of Prop.~\ref{Prop1} is in \cite{dyn}.

In the literature, the conservation of relativistic mass and that of
linear momentum are used to derive the conclusion of Thm.~\ref{thm1}.
By Prop.~\ref{Prop1} above, our axiom system \ax{SpecRelDyn} implies
neither \ax{ConsMass} nor \ax{ConsMomentum}.  By Thm.~\ref{thm2},
\ax{ConsMass} and \ax{ConsMomentum} together imply the key axiom
\ax{AxCenter} of \ax{SpecRelDyn}.  So Thm.~\ref{thm1} is stronger than
the corresponding result in the literature since it requires fewer
assumptions.

Thm.~\ref{thm2} also states that the conservation axioms can be
replaced by the natural, purely geometrical symmetry postulate
\ax{AxCenter^+} without loss of predictive power or expressive power.
Since the conservation axioms \ax{ConsMass} and \ax{ConsMomentum} are
not ``purely geometrical'' and they are less observation-oriented than
\ax{AxCenter^+}, we think that it may be more convincing to use
\ax{AxCenter} or \ax{AxCenter^+} in an axiom system when we introduce
the basics of relativistic dynamics. See \cite[p.22 footnote 22]{leszabo}.

\section{Four-momentum}

Neither relativistic mass nor linear momentum is
Lorentz-covariant.
However, they can be ``put together'' to obtain
a Lorentz-covariant quantity called four-momentum, as follows.
Let $k\in\Ob$ and $b\in\IB$.
The \df{four-momentum} $\Dff{\fvp_k(b)}$
of {\it inertial} body $b$ according to  {\it inertial} observer $k$ is defined as the
element of $\Q^d$ whose time component and space component are the
relativistic mass and linear momentum of $b$ according to
 $k$, respectively, see Fig.~\ref{fourmoment}.
i.e.,
\begin{equation*}
\fvp_k(b)_\tau=\m_k(b)\quad \text{ and }\quad
\fvp_k(b)_\sigma=\m_k(b)\bv_k(b).
\end{equation*}
It is not difficult to prove that $\fvp_k(b)$ is parallel to the
world-line of $b$ and its Minkowski length is $\m_0(b)$, see
Prop.~\ref{prop-fourmoment}. Hence, it is indeed a Lorentz-covariant
quantity. The \df{four-velocity} $\Dff{\fvv_k(b)}$ of {\it inertial}
body $b$ according to {\it inertial} observer $k$ is defined as
$\vq\in\Q^d$ if $q_\tau>0$, $\mu(\vqq)=1$ and $\vq$ is parallel to
$\wl_k(b)$. Let us note that $\fvv_k(b)$ is defined iff $\bv_k(b)$ is
defined; and
\begin{equation*}
\fvv_k(b)_\tau=\frac{1}{\sqrt{1-v_k(b)^2}}\quad \text{ and }\quad
\fvv_k(b)_\sigma=\frac{\bv_k(b)}{\sqrt{1-v_k(b)^2}},
\end{equation*}
see Fig.~\ref{fourmoment}.

\begin{figure}[ht]
\small
\begin{center}
\psfrag*{m}[r][r]{$\m_k(b)$} 
\psfrag*{b}[b][b]{$b$}
\psfrag*{mv}[t][t]{$\m_k(b)\bv_k(b)$}
\psfrag*{P}[lb][lb]{$\fvp_k(b)$} 
\psfrag*{b}[b][b]{$b$}
\psfrag*{V}[lt][lt]{$\fvv_k(b)$} 
\psfrag*{v}[rt][rt]{$\bv_k(b)$}
\psfrag*{1}[rt][rt]{$1$} 
\psfrag*{k}[rt][rt]{$k$}
\includegraphics[keepaspectratio, width=0.8\textwidth]{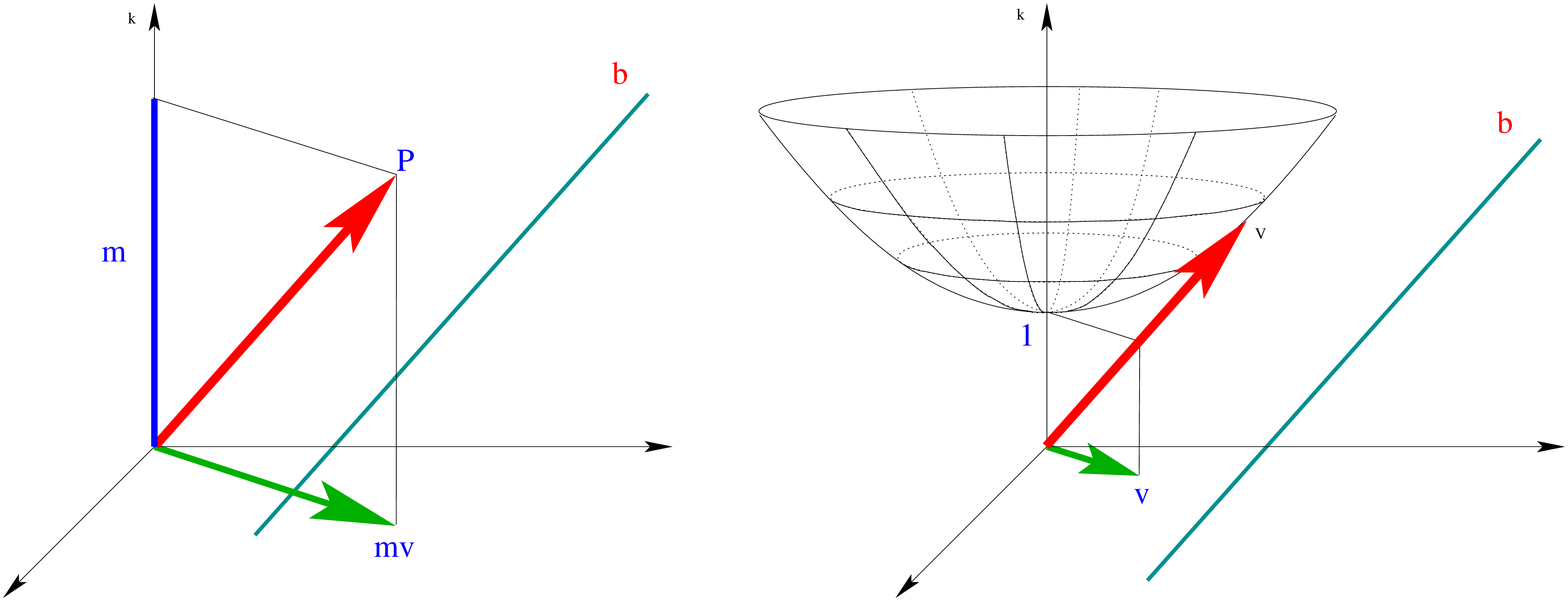}
\caption{\label{fourmoment} Illustration of four-momentum $\fvp_k(b)$
and four-velocity $\fvv_k(b)$} 
\end{center}
\end{figure}

\begin{prop} \label{prop-fourmoment} 
Let $d\geq 3$. Assume \ax{SpecRelDyn} and let $k$ be an {\it inertial}
observer and $b$ an {\it inertial} body having rest mass.  Then
\begin{equation*}
\fvp_k(b)=\m_0(b) \fvv_k(b).
\end{equation*}
\end{prop}
\begin{proof} 
By the definition of $\fvp_k(b)$, it is easy to see that
\begin{equation*}
\mu\big(\fvp_k(b)\big)=\sqrt{1-v_k(b)^2}\cdot \m_k(b).
\end{equation*}
Hence by Thm.~\ref{thm1}, $\mu\big(\fvp_k(b)\big)=\m_0(b)$.  Thus
$\fvp_k(b)=\m_0(b)\fvv_k(b)$.
\end{proof}

\begin{description}
\item[\Ax{ConsFourMoment}]\index{\ax{ConsFourMoment}} Conservation of four-momentum:
\begin{equation*}
\forall k\in\IOb\; \forall b,c,d\in \IB\quad \inecoll_k(b,c:d)\then \fvp_k(b) +\fvp_k(c) =\fvp_k(d).
\end{equation*}
\end{description}
The following can be easily proved from the definition of four-moment. 
\begin{prop}\label{prop-fourcor} 
$\ax{ConsFourMoment}\leftrightarrow\ax{ConsMass}\land\ax{ConsMoment}$.
\end{prop}
\noindent Hence the following is an immediate corollary of
Thm.~\ref{thm2}.
\begin{cor}\label{cor-fmoment}
\begin{equation*}
\ax{AxMass}+\ax{AxInMass}+\ax{AxSelf}+ \IOb\subseteq\IB\models\ax{AxCenter^+}\leftrightarrow \ax{ConsFourMoment}.
\end{equation*}
\end{cor}

Let us return to discussing the merits of using \ax{AxCenter^+} in
place of the more conventional preservation principles.
In the
context of Cor.~\ref{cor}, \ax{ConsFourMoment} has the advantage that
it is computationally direct and simple, while \ax{AxCenter^+} has
the advantage that it is more observational, more geometrical, and
more basic in some intuitive sense.

The following proposition shows the relation of \ax{ConsFourMoment} and \ax{AxCenter}.
By Cor.~\ref{cor-fmoment}, it also shows that \ax{AxCenter^+} is a strictly stronger axiom than \ax{AxCenter}. 
\begin{prop}
\label{prop-carcen} Assume \ax{AxMass} and \ax{AxInMass}. Then \ax{AxCenter} is equivalent to the following formula
\begin{equation*}
\forall k\in\Ob\enskip \forall b,c,d\in \IB\quad\inecoll_k(b,c:d)
\then \exists \lambda\in\Q\quad \fvp_k(b)+\fvp_k(c)=\lambda \fvp_k(d).
\end{equation*}
\end{prop}
\noindent Prop.~\ref{prop-extcarcen} on p.\pageref{prop-extcarcen} is an extension of this proposition.
The following is a consequence of Props.\ \ref{prop-fourcor} and \ref{prop-carcen}.
\begin{cor}
\begin{eqnarray*}
\ax{AxMass}+\ax{AxInMass}+\ax{ConsMass}&\models& \ax{AxCenter}\leftrightarrow\ax{ConsMoment}.\\
\ax{AxMass}+\ax{AxInMass}+\ax{ConsMoment}&\models& \ax{ConsMass}\rightarrow\ax{AxCenter}.\\
\ax{AxMass}+\ax{AxInMass}+ \ax{AxCenter}&\models& \ax{ConsMass}\rightarrow \ax{ConsMoment}.
\end{eqnarray*}
\end{cor}
\begin{rem}
  Let us, however, note that the two implications in the corollary
  above cannot be reversed since it is possible to construct a model in which there are an
  {\it inertial} observer $k$ and {\it inertial} bodies $b$, $c$
  and $d$ such that $\inecoll_k(b,c:d)$, $\bv_k(b)=-\bv_k(c)$,
  $v_k(d)=0$ and $ m_k(b)=m_k(c)=m_k(d) $; and in this model both
  \ax{AxCenter} and \ax{ConsMoment} hold while \ax{ConsMass} does not
  hold.
\end{rem}

Let us finally state a theorem about the existence of nontrivial
models of our axiom systems. The proof of Thm.~\ref{konzisztencia} can be found in \cite{dyn}.
\begin{thm}
\label{konzisztencia} \label{rem}
$\ax{SpecRelDyn^+}\cup\setopen\IOb\neq\emptyset\setclose$ is consistent.
\end{thm}

\section{Some possible generalizations}

Let us formulate $\ax{AxCenter}$ in a more general setting.
To do so, we introduce the set of bodies whose world-lines can be parametrized by differentiable curves according to any {\it inertial} observer:
\begin{equation*}\index{$\DB$}
\Dff{\DB}\leteq \Setopen b\in \B\: :\: \forall k\in\IOb\quad \loc^b_k
\text{ is a differentiable curve}\Setclose.
\end{equation*}
For the FOL definition of differentiability, see Section~\ref{sec-diff}.

\begin{description}
\item[\Ax{AxCenterDiff}]\index{\ax{AxCenterDiff}} If bodies
  $b,c\in\DB$ collide inelastically originating body $d\in\DB$, the
  world-line of $d$ is a differentiable continuation of the
  center-line of the masses of $b$ and $c$ according to any {\it
    inertial} observer (see Fig.~\ref{cendiff}):
\begin{multline*}
\forall k\in\IOb\enskip  \forall b,c,d\in \DB\enskip \forall t\in\Q \quad
\inecoll_k(b,c:d) \\
\lland \loc^b_k(t)=\loc^c_k(t)=\loc^d_k(t) \then (\cen_k^{b,c})'(t)=(\loc^d_k)'(t).
\end{multline*}
\end{description}

\begin{figure}[ht]
\small
\begin{center}
\psfrag*{cen}[ct][ct]{$\cen_k({b,c})$} 
\psfrag*{b}[rt][rt]{$\forall b$} 
\psfrag*{c}[lt][lt]{$\forall c$}
\psfrag*{d}[rb][rb]{$\forall d$} 
\psfrag*{t}[r][r]{$\forall t$}
\psfrag*{k}[r][r]{$\forall k$}
\includegraphics[keepaspectratio, width=0.5\textwidth]{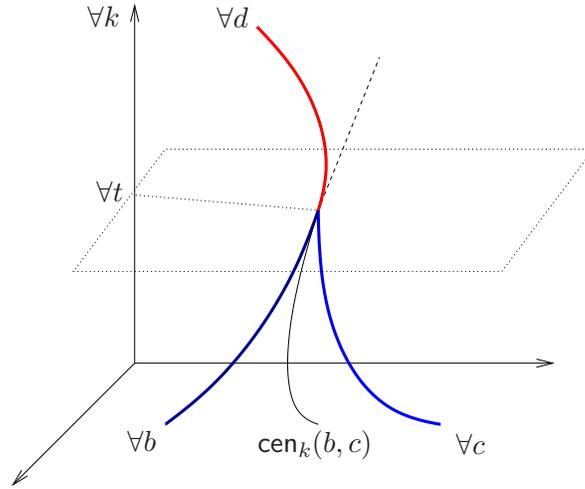}
\caption{\label{cendiff} Illustration of axiom \ax{AxCenterDiff}}
\end{center}
\end{figure}
\noindent
By the following proposition, \ax{AxCenterDiff} is an extension of \ax{AxCenter}.
\begin{prop}
$\ax{AxMass}+\ax{AxInMass}\models\ax{AxCenterDiff}\rightarrow\ax{AxCenter}$.
\end{prop}

\begin{proof}[\colorbox{proofbgcolor}{\textcolor{proofcolor}{On the proof}}]
The proof is based on the following two facts: (1)
$\IB\subseteq\DB$, and (2) \ax{AxMass} and \ax{AxInMass} imply that
$\cen_k({b,c})$ is a line segment if $k\in\IOb$, $b,c\in\IB$ such that
$\inecoll_k(b,c:d)$, see Prop.~\ref{prop-cen}.
\end{proof}

It is also natural to generalize \ax{AxCenter} to more than two bodies.
Now we formulate some of the possible generalizations.
\begin{description}
\item[\Ax{AxCenter_n}]\index{\ax{AxCenter_n}} After some two-by-two
  inelastic collisions of {\it inertial} bodies $b_1,\ldots, b_n$,
  their center-line of masses and the center-line of the masses of the
  last originated {\it inertial} body and the noncolliding bodies are
  in one line:
\begin{multline*}
 \forall k\in\Ob\enskip \forall b_1,\ldots,b_n\in
 \IB\quad\\ \bigwedge_{j=1}^{n-1} \Big[\;\forall d_1,\ldots,
   d_{j+1}\in\IB\enskip
   d_1=b_1\lland\bigwedge_{i=1}^{j}\inecoll_k(d_i,b_{i+1}:d_{i+1})\\ \then
   \cen_k({b_1,\ldots,b_n}) \cup
   \cen_k({d_{j+1},b_{j+2},\ldots,b_n})\subseteq\ell\text{ for some
     line }\ell\;\Big].
\end{multline*}
\end{description}

\begin{description}
\item[\Ax{AxCenter^*_n}]\index{\ax{AxCenter^*_n}} During all
  successive two-by-two inelastic collisions of {\it inertial} bodies
  $b_1,\ldots, b_n$, the center-line of mass after each collision is
  the continuation of the center-line of mass before the collision:
\begin{multline*}
\forall k\in\Ob\; \forall b_1,\ldots,b_n\in
\IB\quad\\ \bigwedge_{j=1}^{n-1} \Big[\,\forall d_1,\ldots,
 d_{j+1}\in\IB\enskip
 d_1=b_1\lland\bigwedge_{i=1}^{j}\inecoll_k(d_i,b_{i+1}:d_{i+1}) \then
 \\\cen_k({b_1,\ldots,b_n})\cup\bigcup_{s=1\le j}
 \cen_k({d_{s+1},b_{s+2},\ldots,b_n})\subseteq\ell\text{ for some
  line }\ell\,\Big].
\end{multline*}
\end{description}

\ax{AxCenter} is just \ax{AxCenter_2} or \ax{AxCenter^*_2} in these
two series of axioms.  Let us now decompose \ax{AxCenter_n} and
\ax{AxCenter^*_n} into the following fragments:

\begin{description}
\item[\Ax{AxCenter_{n,j}}]\index{\ax{AxCenter_{n,j}}} After $j$
  two-by-two inelastic collisions of {\it inertial} bodies
  $b_1,\ldots, b_n$, their center-line of masses and the center-line of
  the masses of the last originated {\it inertial} body and the
  noncolliding bodies are in one line:
\begin{multline*}
 \forall k\in\Ob\; \forall b_1,\ldots,b_n, d_1,\ldots, d_{j+1}\in\IB\quad d_1=b_1\lland\bigwedge_{i=1}^{j}\inecoll_k(d_i,b_{i+1}:d_{i+1})\\
 \then \cen_k({b_1,\ldots,b_n}) \cup \cen_k({d_{j+1},b_{j+2},\ldots,b_n})\subseteq\ell\text{ for some line
}\ell.
\end{multline*}
\end{description}
The definition of \ax{AxCenter^*_{n,j}} is analogous. Then \ax{AxCenter_n} and \ax{AxCenter^*_n} are equivalent to 
\begin{equation*}
\bigwedge_{j=1}^{n-1} \ax{AxCenter_{n,j}}\quad\text{ and }\quad\bigwedge_{j=1}^{n-1} \ax{AxCenter_{n,j}^*},
\end{equation*}
respectively. Let us now see some of the logical connections between the above two series of axioms.

\begin{prop} 
Let $x$, $y$, $n$ and $m$ be natural numbers such that $1\le x<y<n<m$.  Then
\begin{enumerate}
\item $\ax{AxCenter} \models \ax{AxCenter_{n,y}} \leftrightarrow \ax{AxCenter_{n,x}}$, and \\ 
$\ax{AxCenter} \models \ax{AxCenter^*_{n,y}} \leftrightarrow \ax{AxCenter^*_{n,x}}$.
\item a.) $\ax{AxCenter^*_{n,y}} \models \ax{AxCenter^*_{n,x}}$, but \\
 b.) $\ax{AxCenter_{n,y}} \not\models \ax{AxCenter_{n,x}}$.
\item a.) $\ax{AxCenter^*_{n,x}}\models \ax{AxCenter_{n,x}}$, but \\
b.) $\ax{AxCenter_{n,x}}\not\models \ax{AxCenter^*_{n,x}}$.
\item a.) $\ax{AxCenter_{n,x}}\leftrightarrow \ax{AxCenter_{m,x}}$, and \\
b.) $\ax{AxCenter^*_{n,x}}\leftrightarrow \ax{AxCenter^*_{m,x}}$.
\end{enumerate}
\end{prop}

\begin{proof} 
Item (1) can be proved by induction on $y-x$.

\noindent
Item (2a) is true since 
\begin{equation*}
\bigcup_{s=1\le x}\cen_k({d_{s+1},b_{s+2},\ldots,b_n})\subseteq\bigcup_{s=1\le y} \cen_k({d_{s+1},b_{s+2},\ldots,b_n}).
\end{equation*}
To prove (2b), let $\mathfrak{M}$ be a model such that $\IB\leteq
\setopen b_1,\ldots,b_n,d_1,\ldots d_{x+1}\setclose$, $\IOb\leteq
\setopen k\setclose$, $wl_k(k)=\emptyset$, $b_1=d_1$ and
\ax{AxCenter_{n,x}} is not valid. It is not difficult to see that there is
such a model $\mathfrak{M}$. Since there are no $n+y-1$ pieces of
distinct bodies in the required collision situation,
\ax{AxCenter_{n,y}} is valid (its condition is empty).
\medskip 

\noindent 
Item (3a) is true since 
\begin{equation*}
\cen_k({d_{x+1},b_{x+2},\ldots,b_n})\subseteq\bigcup_{s=1\le x} \cen_k({d_{s+1},b_{s+2},\ldots,b_n}).
\end{equation*}
Item (3b) is true since its opposite together with  (2a) and (3a) states that
$\ax{AxCenter_{n,x+1}}\models \ax{AxCenter^*_{n,x+1}}\models
\ax{AxCenter^*_{n,x}}\models \ax{AxCenter_{n,x}}$ which contradicts
  (2b).
\medskip 

\noindent
Item (4) is true since noncolliding bodies cannot affect the center-line of mass.
\end{proof}

Let us now introduce an axiom about general situations of two-by-two inelastic collisions.
Let $\Df{Sq_n}$ be the set of at most $n$-long sequences of natural numbers between $1$ and $n$.
Let us denote the concatenation of sequences $a,b\in Sq_n$ by $\Df{{a\conc b}}$.
We say that $\Df{GI_n}$ is a \df{generalized index set} of basis $n$ if
\begin{itemize}
\item $\setopen 1,\ldots,n\setclose\subseteq GI_n\subseteq Sq_n$;
\item $a_1<\ldots<a_k$ if $a_1\conc\ldots\conc a_k\in GI_n$ and
  $a_1,\ldots,a_k\in \setopen 1,\ldots,n\setclose$; and
\item For all $a\in GI_n\setminus\{1,\ldots,n\}$, there is a unique decomposition $a=b\conc c$ such that $b,c\in GI_n$.
\end{itemize}

\begin{description}
\item[\Ax{AxCenter^*GI_n}]\index{\ax{AxCenter^*GI_n}} During a collision situation of type $GI_n$ of {\it inertial} bodies $b_1,\ldots, b_n$, 
the center-line of mass after each collision is the continuation of the center-line of mass before the collision, see Fig.\ref{figgin}:
\begin{multline*}
 \forall k\in\Ob\enskip (\forall b_i\in\IB\enskip i\in GI_n)\quad
 \bigwedge_{i,j, i\conc j \in GI_n}\inecoll_k(b_i,b_j:b_{i\conc j}) \\
 \then \bigcup_{j_1,\ldots, j_m\in GI_n\ j_1\conc\ldots\conc j_m=1\conc\ldots\conc n} \cen_k({b_{j_1},b_{j_2},\ldots,b_{j_m}})\subseteq\ell\text{ for some line }\ell.
\end{multline*}
\end{description}

\begin{figure}[h!btp]
\small
\begin{center}
\psfrag{1}[rb][rb]{$b_1$}
\psfrag{2}{$b_2$}
\psfrag{3}[rb][rb]{$b_3$}
\psfrag{4}[rb][rb]{$b_4$}
\psfrag{5}{$b_5$}
\psfrag{12}[rb][rb]{$b_{1\conc2}$}
\psfrag{45}{$b_{4\conc5}$}
\psfrag{345}{$b_{3\conc4\conc5}$}
\psfrag{12345}{$b_{1\conc2\conc3\conc4\conc5}$}
\includegraphics[keepaspectratio, width=0.7\textwidth]{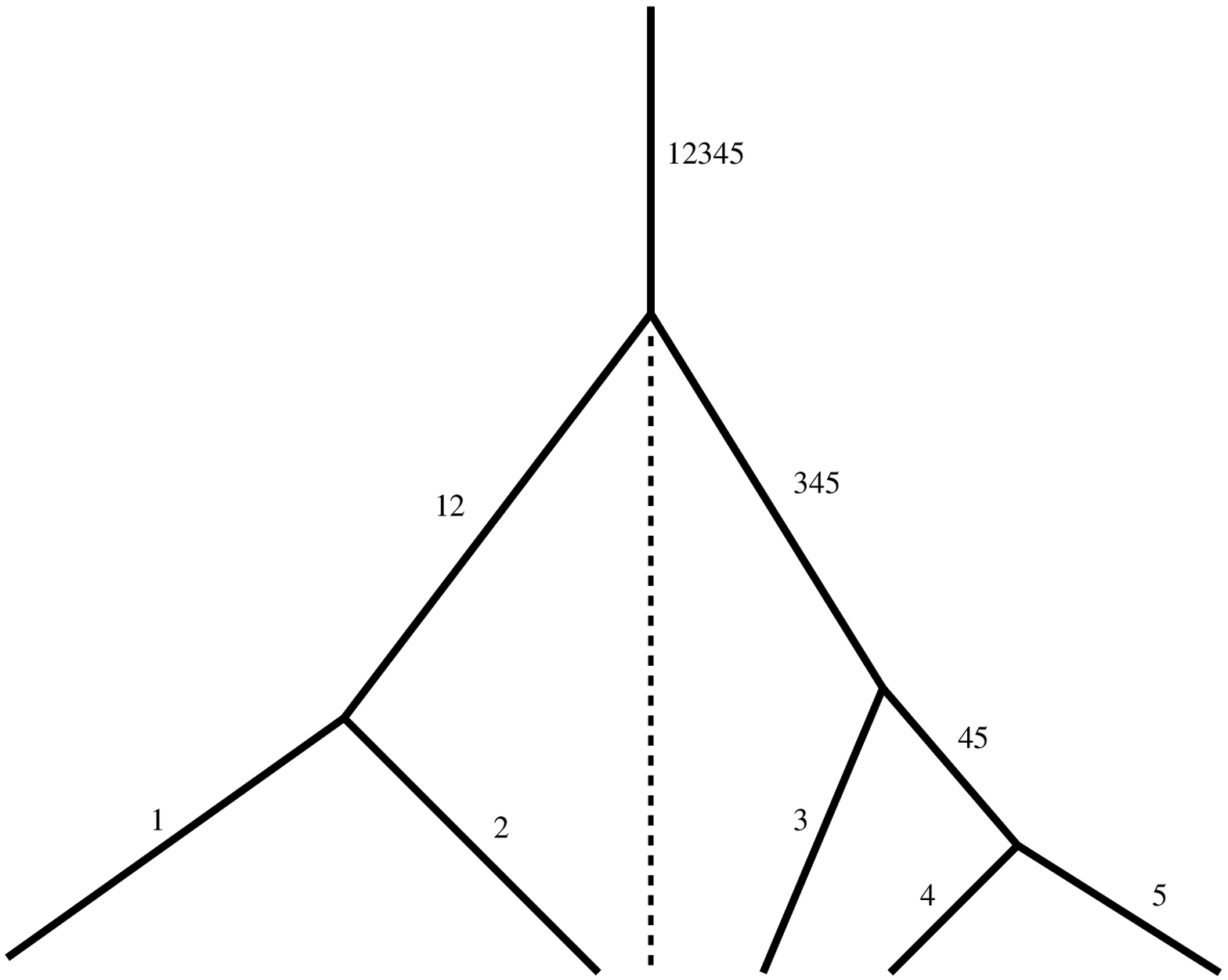}
\caption{\label{figgin} Illustration for \ax{AcCenter^*GI_n} if $GI_n=\{1,2,3,4,5,1\conc2,4\conc5,3\conc4\conc5,1\conc2\conc3\conc4\conc5\}$}
\end{center}
\end{figure}

Let us note that every two-by-two inelastic collision can be coded by
a generalized index set, but not every generalized index set can code a
collision situation, e.g., generalized index set
$\{1,2,3,1\conc2,1\conc3\}$ does not correspond to a collision
situation. In such a case \ax{AxCenter^*GI_n} states nothing since
its conditions cannot be satisfied.

\begin{thm}\label{thm-gin} 
Let $n$ be a natural number and let $GI_n$ be a generalized index
set. Then
\begin{equation*}
\ax{AxCenter}+\ax{AxCenter^+}\models \ax{AxCenter^*GI_n}.
\end{equation*}
\end{thm}
\begin{onproof}
The theorem can be proved by induction on the complexity of the generalized index set $GI_n$.\hfill\qed
\end{onproof}

\begin{cor}\label{cor-cenn} Let $n$ be a natural number. Then 
\begin{eqnarray*}
\ax{AxCenter}+\ax{AxCenter^+} &\models& \ax{AxCenter_n},\quad \text{ and} \\
\ax{AxCenter}+\ax{AxCenter^+} &\models& \ax{AxCenter^*_n}.
\end{eqnarray*}
\end{cor}
\noindent
Let us finally generalize \ax{AxCenter} to noninelastic collisions, too. 
\begin{description}
\item[\Ax{AxCenter_{n:m}}]\index{\ax{AxCenter_{n:m}}} If {\it
 inertial} bodies $b_1,\ldots, b_n$ collide originating {\it inertial}
 bodies $d_1,\ldots, d_m$, the center-line of the masses of $d_1,\ldots, d_m$
 is the continuation of the center-line of the masses of $b_1,\ldots, b_n$:
\begin{multline*}
\forall k\in\Ob\; \forall b_1,\ldots,b_n, d_1,\ldots, d_m\in \IB\quad
 \coll_k(b_1,\ldots,d_n:d_1,\ldots,d_m) \\\then 
 \cen_k({b_1,\ldots, b_n})\cup \cen_k({d_1,\ldots,d_m})\subseteq\ell\text{ for some line }\ell.
\end{multline*}
\end{description}
\ax{AxCenter} is just \ax{AxCenter_{2:1}} in this series of axioms since $\inecoll_k(b,c:d)$ is the same formula as $\coll_k(b,c:d)$. 

\begin{rem}
$\setopen\ax{AxCenter_{n:m}}: n,m\in \omega\setclose$ is an independent axiom system.
\end{rem}

Let us now extend the formula expressing the conservation of four-momentum of two inelastically colliding {\it inertial} bodies
(\ax{ConsFourMomentum}) to general collision situations.
\begin{description}
\item[\Ax{ConsFourMomentum_{n:m}}]\index{\ax{ConsFourMoment_{n:m}}} If
  {\it inertial} bodies $b_1,\ldots, b_n$ collide and originate {\it
    inertial} bodies $d_1,\ldots, d_m$, the sum of four-momentums of
  $d_1,\ldots, d_m$ and $b_1,\ldots, b_n$ is the same:
\begin{multline*}
\forall k\in\IOb\; \forall b_1,\ldots, b_n,d_1,\ldots, d_m\in \IB\quad \\ 
\coll_k(b_1,\ldots,b_n:d_1,\ldots,d_m)\then \sum_{i=1}^n \fvp_k(b_i)=\sum_{j=1}^m\fvp_k(d_j).
\end{multline*}
\end{description}

\begin{rem}
It is suggested by Rindler to assume all the formulas above as axioms
of relativistic dynamics, see \cite[p.109]{Rin}. However, the formulas
\ax{ConsFourMoment_{n:m}} are not natural and observational enough assumptions to
regard them as axioms. By Cor.~\ref{cor-cen+}, we can offer a list of more
natural formulas to be assumed, which is equivalent to the list
above.
\end{rem}

\begin{rem}The following theory is consistent and independent:
\begin{multline*}
\ax{SpecRel}\cup\Setopen \IOb\neq\emptyset, \ax{AxThExp} \Setclose\cup\Setopen \ax{AxInMass}, \ax{AxMass}\Setclose\\\cup\Setopen\ax{Ax\forall coll_{n:m}},\ax{ConsFourMoment_{n:m}}\,:\, n,m\in \omega\Setclose,
\end{multline*}
where \ax{Ax\forall coll_{n:m}} is a generalization of \ax{Ax\forall
  inecoll} which ensures the realization of every possible
collision of type ``$n:m$.''
\end{rem}

The following proposition extends Prop.~\ref{prop-carcen}. It
also shows the logical connection between \ax{AxCenter_{n:m}} and
\ax{ConsFourMoment_{n:m}}.
\begin{prop}\label{prop-extcarcen}
Assume \ax{AxMass} and \ax{AxInMass}. Let $k$ be an {\it inertial} observer and let $b_1,\ldots, b_n$ be {\it inertial} bodies.
Then the following is equivalent to \ax{AxCenter_{n:m}}:
\begin{multline*}
\forall k\in \IOb \enskip \forall b_1,\ldots,b_n,d_1,\ldots,d_m\in\IB \quad \\
\coll_k(b_1,\ldots,b_n:d_1,\ldots,d_m)
\then \exists \lambda\in\Q\quad \sum_{i=1}^n \fvp_k(b_i)=\lambda \sum_{j=1}^m\fvp_k(d_j).
\end{multline*}
\end{prop}

\begin{proof} 
By Lem.~\ref{lem-par}, $\cen_k({b_1,\ldots, b_n})$ and
$\cen_k({d_1,\ldots, d_n})$ are parallel iff $\sum_{i=1}^n
\fvp_k(b_i)$ and $\sum_{j=1}^m\fvp_k(d_j)$ are parallel.  We have that
$\cen_k({b_1,\ldots, b_n})\cap\cen_k({d_1,\ldots, d_n})\neq\emptyset$
if $\coll_k(b_1,\ldots,b_n:d_1,\ldots,d_m)$.  Thus \ax{AxCenter_{n:m}}
holds iff
\begin{equation*}
\sum_{i=1}^n \fvp_k(b_i)=\lambda \sum_{j=1}^m\fvp_k(d_j)
\end{equation*}
 for some $\lambda\in\Q$.
\end{proof}

Let us now introduce a list of axioms which are formulated in the
spirit of \ax{AxCenter} and whose elements are equivalent to the
corresponding \ax{ConsFourMomentum_{n:m}} formulas.
\begin{description}
\item[\Ax{AxCenter^+_{n:m}}]\index{\ax{AxCenter^+_{n:m}}} If $a$ is an
  {\it inertial} body and {\it inertial} bodies $b_1,\ldots, b_n$
  collide originating {\it inertial} bodies $d_1,\ldots, d_m$, the
  center-line of the masses of $a,b_1,\ldots,b_n$ is the continuation
  of the center-line of the masses of $a,d_1,\ldots,d_m$, i.e., there
  is a line that contains both of them:
\begin{multline*}
\forall k\in\IOb\; \forall b_1,\ldots, b_n,d_1,\ldots, d_m\in \IB\quad  
\coll_k(b_1,\ldots,b_n:d_1,\ldots,d_m)\\\then  \cen_k({a,b_1,\ldots,b_n})\cup \cen_k({a,d_1,\ldots,d_m})\subseteq\ell\text{ for some line }\ell.
\end{multline*}
\end{description}

The following corollary can be proved from Prop.~\ref{prop-extcarcen} in a strictly analogous way to the proof  of  Cor.~\ref{cor-fmoment}.
\begin{cor}\label{cor-cen+}
\begin{equation*}
\ax{AxInMass}+\ax{AxMass}+\ax{AxSelf}+\IOb\subseteq\IB \models
\ax{Center^+_{n:m}}\leftrightarrow\ax{ConsFourMomentum_{n:m}}
\end{equation*}
\end{cor}

\begin{lem}\label{lem-par}
Assume \ax{AxMass} and \ax{AxInMass}. Then for all $k\in\IOb$ and $b_1,\ldots, b_n\in \IB$, 
\begin{equation*}
\cen_k({b_1,\ldots, b_n})\quad \text{ is parallel to }\quad \sum_{i=1}^n\fvp_k(b_i)
\end{equation*}
if $\cen_k({b_1,\ldots, b_n})\neq\emptyset$.
\end{lem}

\begin{proof}
We prove the statement by induction on $n$.  It is clear
if $n=1$ since $\fvp_k(b)$ is parallel to $\cen_k(b)=\wl_k(b)$ if
\ax{AxMass} and \ax{AxInMass} are assumed, see
Prop.~\ref{prop-cen}.  Now we have to prove that if it is true
for $n$, then it is also true for $n+1$.  Let $b_1,\ldots,
b_n,b_{n+1}\in \IB$.  Since $\cen_k({b_1,\ldots,
  b_{n+1}})\neq\emptyset$ and translations preserve parallelism, we
can assume that $\vo\in\cen_k({b_1,\ldots, b_{n+1}})$.  By the
induction hypothesis, we can expand our model with {\it inertial} body
$c$ such that $\wl_k(c)=\cen_k({b_2,\ldots, b_{n+1}})$ and
$\fvp_k(c)=\sum_{i=2}^{n+1}\fvp_k(b_i)$.  Since
$\cen_k({b_1,\ldots,b_n})=\cen_k({b_1,c})$, we have to prove that
$\cen_k({b_1,c})$ is parallel to $\fvp_k(b_1)+\fvp_k(c)$.  Without
losing generality, we can assume that $\wl_k(b_1)$ and $\wl_k(c)$ are
lines.  Then
\begin{equation*}
\fvp_k(b_1)\cdot\frac{\m_k(b_1)+\m_k(c)}{\m_k(b_1)}\quad\text{ and }\quad \fvp_k(c)\cdot\frac{\m_k(b_1)+\m_k(c)}{\m_k(c)} 
\end{equation*}
have the same time component $\m_k(b_1)+\m_k(c)$; and they are in $\wl_k(b_1)$ and $\wl_k(c)$, respectively.
So $\cen_k^{b_1,c}\big(\m_k(b_1)+\m_k(c)\big)=\fvp_k(b_1)+\fvp_k(c)$ since $\vo\in\cen_k({b_1,c})$. 
Hence $\cen_k({b_1,c})$ and $\fvp_k(b_1)+\fvp_k(c)$ are parallel; and that is what was to be proved.
\end{proof}

\section{Collision duality}
To every formula $\varphi$ in the language of \ax{SpecRelDyn}, we can
define its \df{collision dual}\index{collision dual}
$\Df{\varphi^\otimes}$\index{$\varphi^\otimes$} in which the
subformulas of the form $\coll_k(b_1,\ldots, b_n:d_1,\ldots, d_m)$ are
replaced by $\coll_k(d_1,\ldots, d_m:b_1,\ldots, b_n)$, i.e., the
incoming and the outgoing bodies in the collisions are interchanged. By
this definition of collision dual, it is clear that $(\varphi^\otimes)^\otimes$ and
$\varphi$ are the same. Let $\Sigma$ be a set of formulas. The set of
collision duals of the formulas of $\Sigma$ is denoted by
$\Df{\Sigma^\otimes}$\index{$\Sigma^\otimes$}.

Formula $\varphi$ is called \df{self-dual} if $\varphi$ and
$\varphi^\otimes$ are the same. Every formula which is  not 
about collisions is self-dual. So most of the axioms of
\ax{SpecRelDyn} are self-dual. There are also self-dual axioms 
about collisions, e.g., \ax{AxCenter_{n:n}} is such for all natural number $n$.

\begin{prop}\label{prop-dual}
Let $\Sigma$ be a set of formulas and let $\varphi$ be formula in the language of \ax{SpecRelDyn}. Then
\begin{equation*}
\Sigma\models\varphi \quad \text{iff}\quad \Sigma^\otimes\models\varphi^\otimes.
\end{equation*}
\end{prop}

\begin{proof}
By G\"odel's completeness theorem, $\Sigma\models\varphi$ holds iff
$\Sigma\vdash\varphi$, i.e., there is a formal proof of $\varphi$ from
$\Sigma$. It is clear that $\Sigma\vdash\varphi$ iff
$\Sigma^\otimes\vdash\varphi^\otimes$ since we get a formal proof of
$\varphi^\otimes$ from $\Sigma$ by replacing every formula by its
collision in the formal proof of $\varphi$ from $\Sigma$.  From this,
we get $\Sigma^\otimes\models\varphi^\otimes$ by G\"odel's
completeness theorem.\end{proof}

By applying Prop.~\ref{prop-dual} to Thm.~\ref{thm1}, we get the
following as its conclusion is self-dual.

\begin{cor}\label{cor-fis}
Let $d\geq 3$ and assume \ax{SpecRelDyn^\otimes}.
Let $k$ be an {\it inertial} observer and $b$ be an {\it inertial} body having rest mass.
Then
\begin{equation*}
\m_0(b)={\sqrt{1-v_k(b)^2}}\cdot \m_k(b).
\end{equation*}
\end{cor}

\begin{rem}
It is natural to interpret $\inecoll$ as nuclear fusion. By this
interpretation $\inecoll^\otimes$ becomes nuclear fission.  In this
case axioms \ax{Ax\forall Center^\otimes} and \ax{Ax\forall Center}
(and even \ax{Ax\exists Center^\otimes} and \ax{Ax\exists Center}) are
too strong since they require the existence of several fusions and
fissions which might not exist in nature. However, it is not a problem
since all the theorems which use these axioms can be reformulated
without them by building them into the statements. So instead of
assuming that certain fusion/fission situations exist and proving a
statement from that, we can omit this assumption and prove the
statement only for the bodies which appear in the corresponding
fusion/fission situations.
\end{rem}

\section{Concluding remarks on dynamics}

We have introduced a purely geometrical axiom system of special
relativistic dynamics which is strong enough to prove the formula
connecting relativistic and rest masses of bodies. We have also
studied the connection of our key axioms \ax{AxCenter} and
\ax{AxCenter^+} and the usual axioms about the conservation of mass,
momentum and four-momentum.  We saw that the conservation postulates
are not needed to prove the relativistic mass increase theorem
$m_0=\sqrt{1-v^2/c^2}\cdot m$, see Prop.~\ref{Prop1} at
(p.\pageref{Prop1}) and Thm.~\ref{thm1} at (p.\pageref{thm1}). See also
\cite[p.22 footnote 22]{leszabo}. Connections with Einstein's insight $E=mc^2$
have also been discussed. The contents of the present chapter
represent only the first steps towards a logical conceptual analysis
of relativistic dynamics. A glimpse into Chap.\ 6 (pp.108-130)
``Relativistic particle mechanics'' of the textbook by
Rindler~\cite{Rin} suggests the topics to be covered by future work in
this line. In another direction, looking at the logical issues in
\cite{pezsgo} and \cite{logst} suggests questions and investigations
to be carried out in the future about the logical analysis of
relativistic dynamics.

\noindent
Let us mention here two tasks that should be done in the future:
\begin{que}
Analyzing the possibility/impossibility of faster than light motion of
colliding bodies within axiomatic special relativistic dynamics
similarly to what was done for observers within special relativistic
kinematic, see, e.g., \cite[Thm.11.7]{logst}, \cite[Thm.3, Thm.5]{MNT}.
\end{que}

\begin{que}
Extending the axiomatization of relativistic dynamics for accelerated observers.
\end{que}
\noindent
A work related to this chapter with somewhat different aims is \cite{Reldyn}.

\chapter{Extending the axioms of special relativity for accelerated observers}
\label{chp-accrel}

In this chapter we extend our axiomatization of special relativity to
non-{\it inertial} observers, too. Non-{\it inertial} observers are
also going to be called \df{accelerated observers}.\index{accelerated
  observers} We have two reasons for extending our aproach to
accelerated observers: to take a step towards a FOL axiomatization of
general relativity (see Chap.~\ref{chp-gr}) and to provide an
axiomatic basis of the twin paradox and other surprising predictions
of special relativity extended to non-{\it inertial} observers.  The
results of this chapter are based on \cite{AMNSZacekalm}, \cite{twp}
and \cite{mythes}. A further aim is to prove predictions of
general relativity from our theory of accelerated observers, by using
Einstein's equivalence principle, cf.\ our interpretation of answering
why-questions at p.\pageref{why-questions} or \cite{wqp}.

\section{The key axiom of accelerated observers}

It is clear that \ax{SpecRel} is too weak to answer any nontrivial
question about acceleration since \ax{AxSelf_0} is its only axiom that mentions
non-{\it inertial} observers. To extend \ax{SpecRel}, we now formulate
the key axiom about accelerated observers. It will state that the
worldviews of accelerated and {\it inertial} observers are locally
the same.

To connect the worldviews of the accelerated and the {\it inertial}
observers, we formulate the statement that, at each
moment of its world-line, each accelerated observer coordinatizes the
nearby world for a short while as an {\it inertial} observer does.  To
formalize that, first we introduce the relation of being a co-moving
observer.  Observer $m$ is a \df{co-moving observer}\index{co-moving
  observer} of observer $k$ at $\vq\in\Q^d$, in symbols $\Df{m
  \com_{\vq} k}$\index{$\com_{\vq}$}, iff $\vq\in \dom w^k_m$ and the
following holds:
\begin{equation*}
\forall \varepsilon \in \Q^+ \; \exists \delta \in \Q^+ \;
\forall \vp \in B_{\delta}(\vqq)\cap \dom w^k_m\quad
\big|w^k_m(\vpp)-\vpp\big| \leq\varepsilon\cdot|\vp-\vqq|.
\end{equation*}
Behind the definition of co-moving observers is the following intuitive
image: as we zoom in the neighborhood of the
coordinate point, the worldviews of the
two observers are getting more and more similar.

\begin{rem}\label{rem-comove} Let us note that 
$\vq\in Cd_m$ and $ev_m(\vqq)=ev_k(\vqq)$ if $m \com_{\vq} k$. It can
  be proved by choosing $\vp$ as $\vq\in B_{\delta}(\vqq)\cap \dom w^k_m$.
\end{rem}

\begin{rem}\label{rem-wkmfun}
By Conv.~\ref{conv-wkm}, there is a $\delta\in\Q^+$ such that
$w^k_m$ is a function on $B_\delta(\vqq)\cap \dom w^k_m$ if $m \com_{\vq}
k$. So $w^k_m$ is a function on a small enough neighborhood of $\vq$
if $m \com_{\vqq} k$ and $\dom w^k_m$ is open.
\end{rem}

The relation $\com_{\vqq}$ is transitive but it is neither reflexive
nor symmetric. It is not reflexive because if $k$ is an observer such
that $ev_k(\langle 0,\frac{1}{n},0\ldots,0\rangle)=ev_k(\voo)$ for all
$n\in\omega$, then $w^k_k$ is not a function on any neighborhood of
$\vo$. Thus $k\not\com_{\voo}k$ see Rem.~\ref{rem-wkmfun}.  Example
\ref{example-not} shows that $\com_{\vqq}$ is not symmetric. The
relation $\com_{\vqq}$ becomes an equivalence relation, e.g., if
$w^k_m$ is a function and defined in a small enough neighborhood of
$\vq$ for each $k,m\in \Ob$. This will be the case in our last axiom
system \ax{GenRel} in Chap.~\ref{chp-gr}.

Now we can formulate the key axiom of accelerated observers, called
the co-moving axiom. This axiom is about the connection between the
worldviews of {\it inertial} and accelerated observers:
\begin{description}
\item[\Ax{AxCmv}] \index{\ax{AxCmv}} For every observer and event
  encountered by it, there is a co-moving
 \textit{inertial} observer:
\begin{equation*}
\forall k \in \Ob \enskip \forall \vq \in \Q^d \quad k\in ev_k(\vqq) \then
\exists m\in \IOb \enskip m \com_{\vq}k.
\end{equation*}
\end{description}

\begin{rem}\label{rem-axacc} Let us note that 
 from \ax{AxCmv} and \ax{AxEv} follows that \textit{inertial}
 observers coordinatize every event encountered by an observer, i.e.,
 $e\in Ev_m$ for all event $e$ and {\it inertial} observer $m$
 whenever $k\in e\in Ev_k$ for some observer $k$.  That is true since
 \textit{inertial} observers coordinatize the same events by
 \ax{AxEv}; there is an $m\in\IOb$ such that $m\com_{\vq} k$ and
 $ev_k(\vqq)=e$ by \ax{AxCmv}; and $\vq\in Cd_m$ if $m\com_{\vq} k$,
 see Rem.~\ref{rem-comove}.
\end{rem}

Before we go on building our theory of accelerated observers, let us
prove a proposition reformulating the co-moving relation.  For the
notion of differentiability in our framework, see
Section~\ref{sec-diff}. Let us note that in our framework a function
differentiable at $\vq$ may have several derivatives at $\vq$.

\begin{prop}\label{prop-comdiff} 
Let $m$ and $k$ be observers and $\vq$ be a coordinate point. The
following two statements are equivalent:
\begin{itemize} 
\item[(a)] $m \com_{\vq} k$, i.e., $m$ is a co-moving observer of $k$
  at $\vq$.
\item[(b)] $w^k_m(\vqq)=\vq$, $w^k_m$ is differentiable at $\vq$, and
 one of its derivatives at $\vq$ is the identity map.
\end{itemize}
\end{prop}

\begin{proof}
To prove $(a)\Longrightarrow(b)$, let $m \com_{\vq} k$. Then $\vq\in
\dom w^k_m$ by definition. By Rem.~\ref{rem-wkmfun}, we have that $
w^k_m$ is a function on $B_{\delta_0}(\vqq)\cap \dom w^k_m$ for some
$\delta_0 \in \Q^+$. Thus $w^k_m(\vqq)$ is defined, and it is4s $\vq$ by
Rem.~\ref{rem-comove}. By the definition of $m\com_{\vqq} k$, for
all $\varepsilon \in \Q^+$, there is a $\delta \in \Q^+$ such that
the inequality $|w^k_m(\vpp)-\vpp|\leq\varepsilon|\vpp-\vqq|$ holds
for all $\vpp \in B_{\delta}(\vqq)\cap \dom w^k_m$. Since
$w^k_m(\vqq)=\vqq$, this inequality can be rewritten as
$|w^k_m(\vpp)-w^k_m(\vqq)-Id(\vpp-\vqq)|\le
\varepsilon|\vpp-\vqq|$. So $w^k_m$ is differentiable at $\vq$ and one
of its derivatives at $\vq$ is the identity map.

\smallskip \noindent To prove the converse implication, let $w^k_m$ be
differentiable at $\vq$ such that one of its derivatives at $\vq$ is
the identity map, and $w^k_m(\vqq)=\vqq$. Then for all $\varepsilon
\in \Q^+$, there is a $\delta \in \Q^+$ such that
$|w^k_m(\vpp)-w^k_m(\vqq)-Id(\vpp-\vqq)|\le \varepsilon|\vpp-\vqq|$
holds for all $\vpp \in B_{\delta}(\vqq)\cap\dom w^k_m$. And, since
$w^k_m(\vqq)=\vq$, this last inequality is the same as
$|w^k_m(\vpp)-\vpp|\leq\varepsilon|\vpp-\vqq|$. Thus $m\com_{\vq} k$.
\end{proof}

The world-line of an observer represents the set of coordinate points
where the observer is during its life but it does not tell how ``old''
the observer is at a certain event. So let us define the
\df{life-curve} $\lc^k_m$\index{life-curve} of observer $k$ according
to observer $m$ as the world-line of $k$ according to $m$ {\it
  parametrized by the time measured by $k$},\label{life-curve}
formally:
\begin{equation*}
\index{$\lc^k_m$}
\Df{\lc^k_m}\leteq \Setopen \langle t,\vpp \rangle\in \Q\times \Q^d \::\right.
\enskip\left.\exists \vqq\in \Q^d \quad k\in ev_k(\vqq)=ev_m(\vpp)\lland q_\tau=t\Setclose.
\end{equation*}
For the most important properties of life-curves, see Prop.~\ref{prop-lc}.

Let the natural embedding $\Df{\iota}: \Q \rightarrow \Q^d$ be\index{$\iota$}
defined as $\iota(x)\leteq \langle x,0,\dots,0\rangle$ for all $x\in \Q$.

\begin{lem}\label{lem-lc} Assume \ax{AxSelf_0}. Let $k$ and $m$ be observers. Then 
$\lc^k_m\leteq \iota\circ w^k_m$.
\end{lem}

\begin{proof}
By our definitions, 
\begin{equation*}
\iota \circ w^k_m=\Setopen \langle t,\vpp \rangle\in \Q\times \Q^d
\::\right.
\enskip\left. \exists \vq\in\Q^d \quad ev_k(\vqq)=ev_m(\vpp)\neq\emptyset \lland q_\tau=t \lland \vq_\sigma=\vo\Setclose.
\end{equation*}
By \ax{AxSelf_0} and the fact that $ev_k(\vqq)\neq\emptyset$, we have  $\vq_\sigma=\vo$ iff $k\in ev_k(\vqq)$.
So $\lc^k_m = \iota \circ w^k_m$. 
\end{proof}

Let us introduce here a very natural axiom about observers, which is
going to be used in the following proposition. 

\begin{description}
\item[\Ax{AxEvTr}]
\index{\ax{AxEvTr}}\label{axevtr}
 Every observer encounters the events in which it has been observed:
\begin{equation*}
\forall m\in \Ob\enskip \forall e\in Ev \quad m\in e \then
e\in Ev_m.
\end{equation*}
\end{description}

\begin{prop}
\label{prop-lc}
Let $m$, $k$ and $h$ be observers. 
Then
\begin{enumerate}
\item \label{item-trfiff} 
$\lc^k_m$ is a function iff 
\begin{itemize}
\item[(i)] event $e$ has a unique coordinate in $Cd_m$ whenever $k\in e\in Ev_m\cap Ev_k$, and 
\item[(ii)] $ev_k(\vqq)=ev_k(\vqq')$ if $\vqq,\vqq'\in Cd_k$ such that $ev_k(\vqq),ev_k(\vqq')\in Ev_m$, $k\in ev_k(\vqq)\cap ev_k(\vqq')$ and $q_\tau=q'_\tau$.
\end{itemize}
\item \label{item-trfunct}
$\lc^k_m$ is a function if $m$ is {\it inertial}, and \ax{AxPh_0} and \ax{AxSelf_0} are assumed.
\item \label{item-tr} $\lc^h_m\supseteq \lc^h_k\circ w^k_m$ always holds, and \\
$\lc^h_m= \lc^h_k\circ w^k_m$ holds if we assume $\ax{\mathit{Ev_m\subseteq Ev_k}}$.
\item \label{item-domtr} $\setopen q_\tau : k\in ev_k(\vqq)\setclose=\dom \lc^k_k\supseteq \dom \lc^k_m$ always holds, and \\
$\dom \lc^k_m=\dom \lc^k_k$ holds if we assume \ax{AxCmv} and $m\in \IOb$.
\item \label{item-rantr} $\ran \lc^k_m\subseteq \wl_m(k)$ always holds, and\\
$\ran \lc^k_m=\wl_m(k)$ if we assume \ax{AxEvTr}.
\end{enumerate}
\end{prop}

\begin{proof}
Item \eqref{item-trfiff} is a straightforward consequence of the
definition of $\lc^k_m$. To see that, let $R\leteq \setopen \langle
t,\vqq\rangle\in \Q\times Cd_k \setmid k\in ev_k(\vqq) \lland
q_\tau=t\setclose$. Then $\lc^k_m=R\circ w^k_m=R\circ ev_k\circ
\loc_m$. Since $ev_k$ is a function and $\loc_m$ is an inverse of a
function, it is easy to see that $\lc^k_m$ is a function iff $\loc_m$
is a function on $\ran (R\circ ev_k)$ and $R\circ ev_k$ is a function
to $\dom \loc_m=Ev_m$. It is clear that $\loc_m$ is a function on
$\ran (R\circ ev_k)$ iff (i) holds; and it is also clear that $R$ is a
function to $\dom \loc_m=Ev_m$ iff (ii) holds. Hence $\lc^k_m$ is a
function iff both (i) and (ii) hold.

To prove Item \eqref{item-trfunct}, we should check (i) and (ii) of Item \eqref{item-trfiff}.
By Item \eqref{item-crd} of Prop.~\ref{prop-sr0}, (i) is true.
By \ax{AxSelf_0} if $k\in ev_k(\vqq)\cap ev_k(\vqq')$, then $\vq_\sigma=\vo=\vqq'_\sigma$.
Thus if $q_\tau=q'_\tau$ also holds, then $\vqq=\vqq'$.
Hence (ii) is also true.

To prove Item \eqref{item-tr}, let $\langle t,\vpp\rangle\in
\lc^h_k\circ w^k_m$.  That means $\exists\vc\in Cd_k$ such that
$\langle t,\vc\,\rangle\in \lc^h_k$ and $\langle\vc,\vpp\rangle\in
w^k_m$, which is equivalent to $\exists\vqq\in Cd_h$ such that $h\in
ev_h(\vqq)=ev_k(\vc\,)$, $q_\tau=t$ and $ev_k(\vc\,)=ev_m(\vpp)$.
Thus $\langle t,\vpp\rangle\in \lc^h_m$.  To prove the converse
inclusion, let $\langle t,\vpp\rangle\in \lc^h_m$.  That means that
there is a coordinate point $\vqq\in Cd_h$ such that $h\in
ev_h(\vqq)=ev_m(\vpp)$ and $q_\tau=t$.  By the assumption
$Ev_m\subseteq Ev_k$, we have that $\exists \vc\in Cd_k$ such that
$ev_k(\vc\,)=ev_m(\vpp)$.  Thus $\langle t,\vpp\rangle\in \lc^h_k\circ
w^k_m$.  That proves Item \eqref{item-tr}.

To prove Item \eqref{item-domtr}, let us recall that $t\in \dom
\lc^k_m$ iff there are $\vpp\in Cd_m$ and $\vqq\in Cd_k$ such that
$k\in ev_m(\vpp)=ev_k(\vqq)$ and $q_\tau=t$.  From that, it easily
follows that $t\in \dom \lc^k_k$ iff there is a coordinate point
$\vqq\in Cd_k$ such that $q_\tau=t$ and $k\in ev_k(\vqq)$.  Thus
$\setopen q_\tau : k\in ev_k(\vqq)\setclose=\dom \lc^k_k\supseteq \dom
\lc^k_m$ is clear; and if we assume \ax{AxCmv} and $m\in\IOb$, then
$\dom \lc^k_k\subseteq \dom \lc^k_m$ is also clear since
\textit{inertial} observers coordinatize every event encountered by
observers, see Rem.~\ref{rem-axacc}.
 
To prove Item \eqref{item-rantr}, let us recall that $\vpp\in\ran
\lc^k_m$ iff $\vp\in Cd_m$ and there are $t\in \Q$ and $\vqq\in Cd_k$
such that $k\in ev_m(\vpp)=ev_k(\vqq)$ and $q_\tau=t$.  Thus $\ran
\lc^k_m\subseteq \wl_m(k)\leteq \setopen \vpp\in Cd_m \setmid k\in
ev_m(\vpp)\setclose$ is clear.  If $\vpp\in \wl_m(k)$, then $k\in
ev_m(\vpp)$.  Therefore, by \ax{AxEvTr}, we have that $ev_m(\vpp)\in
Ev_k$.  Thus there is a coordinate point $\vqq\in Cd_k$ such that
$ev_m(\vpp)=ev_k(\vqq)$.  Hence $\ran \lc^k_m=\wl_m(k)$.
\end{proof}

We call a timelike curve $\alpha$
\df{well-parametrized}\index{well-parametrized} if
$\mu\big(\alpha'(t)\big)=1$ for all $t\in \dom \alpha$. For the
FOL definition of $\alpha'$, see Section \ref{sec-diff}.

Assume $\mathfrak{Q}=\R$.  Then curve $f$ is well-parametrized iff $f$
is parametrized according to the Minkowski length, i.e., for all
$x,y\in \dom f$, the Minkowski length of $f$ restricted to $[x,y]$ is
$y-x$.  (By the Minkowski length of a curve we mean length according
to the Minkowski metric, e.g., in the sense used by Wald~\cite[p.43,
  (3.3.7)]{wald}).  If the {\it proper time} is defined as the
Minkowski length of a timelike curve, see, e.g., Wald~\cite[p.44,
  (3.3.8)]{wald}, Taylor-Wheeler~\cite[1-1-2]{TW00} or
d'Inverno~\cite[p.112, (8.14)]{dinverno}, a curve defined on a subset
of $\R$ is well-parametrized iff it is parametrized according to
proper time (see, e.g., \cite[p.112, (8.16)]{dinverno}). Hence for
well-parametrized curves, our definition of proper time (see
p.\pageref{propertime}) coincides with the definition of the
literature.

\begin{example} \label{example-wp}
Let us list some examples of well-parametrized curves here:
\begin{enumerate}
\item \label{item-ratfv} $\gamma(t)=1/2\cdot\langle t^3/3-1/t,
  t^3/3+1/t,0,\ldots, 0\rangle$ for all $t\in\Q^+$.
\item \label{item-sqrt} $\gamma(t)=\langle\sqrt{t^3}/3+\sqrt{t},\sqrt{t^3}/3-\sqrt{t},0,\ldots,0\rangle$ for all $t\in\Q^+$.
\item \label{item-unif} $\gamma(t)=\langle a\cdot sh(t/a),a\cdot ch(t/a),0,\ldots,0\rangle$ for all $a\in \R^+$ and $t\in \R$.
\item \label{item-spiral} $\gamma(t)=\langle \sqrt{a^2+1}\cdot t, \cos(a\cdot t),\sin(a\cdot t),0,\ldots,0\rangle$ for all $a\in \R^+$ and $t\in \R$.
\end{enumerate}
\end{example}
\noindent
Let us note that examples \eqref{item-ratfv} and \eqref{item-sqrt} do
not have well-parametrized extensions.
Definability in Prop.~\ref{prop-hyp} is meant in the same way as in Sec.~\ref{sec-cont}.

\begin{prop}\label{prop-hyp}
The vertical timelike unit-hyperbola
\begin{equation*}
\Df{Hyp}\leteq\Setopen\vp\in \Q^d \setmid p_2^2-p_t^2=1, p_3=\ldots=p_d=0\Setclose
\end{equation*}
 can be well-parametrized by a definable curve iff an {\bf exponential
   function} is definable over $\Q$, i.e.,
 there is a definable well-parametrized curve $\gamma:\Q\rightarrow
 \Q^d$ such that $\ran\gamma=Hyp$ iff there is a definable function
 $e:\Q\rightarrow \Q$ such that $e'(t)=e(t)$ and $e(-t)=1/e(t)$ for all $t\in\Q$.
\end{prop}
\begin{proof}
Using the fact that the Minkowski distance of $\vex$ and the points of
$Hyp$ tend to infinity in both directions, it can be proved that
$\dom\gamma=\Q$ for any well-parametrization $\gamma$ of $Hyp$.  Let
$\gamma=\langle \gamma_1,\gamma_2,0\ldots,0\rangle$ be a definable
well-parametrization of $Hyp$.  Then 
\begin{equation*}
\gamma_2(t)^2-\gamma_1(t)^2=1\quad
\text{for all}\quad t\in\Q.
\end{equation*}
By differentiating both sides of this equation, we
get that (see Sec.~\ref{sec-diff}) 
\begin{equation*}
\gamma_2(t)\gamma_2'(t)-\gamma_1(t)\gamma'_1(t) =0\quad\text{for
  all}\quad t\in\Q,
\end{equation*}
which means that $\gamma(t)$ and $\gamma'(t)$ are Minkowski orthogonal
since $\gamma'=\langle \gamma'_1,\gamma'_2,0\ldots,0\rangle$. Since
$\gamma$ is well-parametrized, $\gamma'(t)$ is of Minkowski length $1$
for all $t\in\Q$. Hence $\gamma'_1(t)^2-\gamma'_2(t)^2=1$ for all
$t\in\Q$.  So $\gamma(t)$ and $\gamma'(t)$ are two Minkowski orthogonal
vectors of the $\txPlane$, for which $\mu\big(\gamma(t)\big)=-1$ and
$\mu\big(\gamma'(t)\big)=1$.  Thus there are two possibilities: either
(1) $\gamma_1(t)=\gamma'_2(t)$ and $\gamma_2(t)=\gamma'_1(t)$, or (2)
$\gamma_1(t)=-\gamma'_2(t)$ and $\gamma_2(t)=-\gamma'_1(t)$. From the
differentiability of $\gamma$ it follows that only one of these two
cases can hold for all $t\in\Q$.  Let
$e(t)\leteq\gamma_1(t)+\gamma_2(t)$ in case (1) and let
$e(t)\leteq\gamma_1(-t)+\gamma_2(-t)$ in case (2). Then
$e:\Q\rightarrow \Q$ is a definable differentiable function for which
$e'(t)=e(t)$ and $e(-t)=1/e(t)$ for all $t\in\Q$.
\medskip

\noindent 
To prove the other direction, let $e:\Q\rightarrow\Q$ be a
definable differentiable function such that $e(t)'=e(t)$ and $e(-t)=1/e(t)$ for all
$t\in\Q$.  Then let us define functions $ch$ and $sh$ as follows:
\begin{equation*}
ch(t)\leteq\frac{e(t)+{e(-t)}}{2}\quad \text{and}\quad
sh(t)\leteq\frac{e(t)-{e(-t)}}{2}\quad \text{for all}\quad
t\in\Q.
\end{equation*} 
Then the following can be shown by a straightforward calculation:
\begin{equation*}
ch'(t)=sh(t),\quad sh'(t)=ch(t)\quad\text{and}\quad
ch(t)^2-sh(t)^2=1\quad \text{for all}\quad t\in\Q.
\end{equation*}
From these equations it is not difficult to prove
 that the following curve is a definable well-parametrization of $Hyp$:
\begin{equation*}
\gamma(t)\leteq\langle sh(t),ch(t),0\ldots,0\rangle\quad \text{for all}\quad t\in\Q.
\end{equation*}
That completes the proof of Prop.~\ref{prop-hyp}.
\end{proof}

\begin{rem}
It is well known that uniformly accelerated motion and hyperbolic
motion are the same, see \cite[\S 3.8]{dinverno}. Thus according to {\it inertial}
observers, the world-line of a uniformly accelerated observer is the
unit-hyperbola $Hyp$ distorted by a Poincar{\'e} transformation and a
dilation.  Thus Prop.~\ref{prop-hyp} implies that there can be
uniformly accelerated observers iff an exponential function of the
quantities is definable. See Question \ref{que-unif}.
\end{rem}

\noindent
Let us introduce an axiom here that we will use to strengthen \ax{AxSelf_0}:
\begin{description}\label{axself+}
\item[\Ax{AxSelf^+_0}] \index{\ax{AxSelf^+_0}} The set of
  time-instances in which an observer  encounters an event is
  connected and has at least two distinct elements, i.e.,
\begin{equation*}
\forall k\in\Ob\; \exists \vp,\vq\in\Q^d\enskip p_\tau\neq q_\tau
\lland k\in ev_k(\vpp)\cap ev_k(\vqq)\lland\{ r_\tau : k\in
ev_k(\vrr)\}\mbox{ is connected.}
\end{equation*}
\end{description}
Let us note here that axioms \ax{AxSelf_0} and \ax{AxSelf^+_0}
together are still weaker than \ax{AxSelf}.

Let now introduce an axiom system which is the extension of
\ax{SpecRel} by \ax{AxCmv} and some simplifying axioms:
\begin{equation*}\index{\ax{AccRel_0}}
\boxed{ \ax{AccRel_0}\leteq \Setopen \ax{AxSelf_0},\ax{AxSelf^+_0},
  \ax{AxPh},\ax{AxEv},\ax{AxEvTr},\ax{AxSymDist}, \ax{AxCmv}\Setclose}
\end{equation*}

\begin{rem}
\ax{AccRel_0} is an extension of \ax{SpecRel} since \ax{AxSelf} is implied by \ax{AxSelf_0}, \ax{AxPh} and \ax{AxEv}, see Prop.~\ref{prop-sr0}. Moreover, \ax{AccRel_0} is a conservative extension of \ax{SpecRel} with respect to accelerated observers.
\end{rem}

Our next theorem states that life-curves of accelerated observers in
the models of \ax{AccRel_0} are well-parametrized. That implies that in
the models of \ax{AccRel_0}, accelerated clocks behave as expected.
And Rem.~\ref{Trrem} states a kind of ``completeness theorem'' for
life-curves of accelerated observers.

\begin{thm}
\label{thm-wp} 
Let $d\ge 3$.
Assume \ax{AccRel_0}.
Let $k$ be an observer and $m$ be an {\it inertial} observer.
Then $\lc^k_m$ is a well-parametrized timelike curve.
\end{thm}

\begin{proof}
By \eqref{item-trfunct} in Prop.~\ref{prop-lc} we have that $\lc^k_m$
is a function. To prove that $\lc^k_m$ is also a curve, we need to
show that $\dom \lc^k_m$ is connected and has at least two distinct
elements. That is so because by Item \eqref{item-domtr} in
Prop.~\ref{prop-lc}, $\dom \lc^k_m=\setopen q_\tau : k\in
ev_k(\vqq)\setclose$ and the latter is connected and has at least two
distinct elements by \ax{AxSelf^+_0}. Hence $\lc^k_m$ is a curve.

To complete the proof, we have to show that $\lc^k_m$ is also timelike
and well-parametrized. Let $t\in \dom \lc^k_m$. We have to prove that
$\lc^k_m$ is differentiable at $t$ and its derivative at $t$ is of
Minkowski length $1$. By \eqref{item-domtr} of Prop.~\ref{prop-lc},
there is a $\vq\in Cd_k$ such that $k\in ev_k(\vqq)$ and
$q_\tau=t$. Thus, by \ax{AxCmv}, there is a co-moving {\it inertial}
observer of $k$ at $\vq$. By \ax{AxSelf_0}, $\iota(t)=\vq$. By
Prop.~\ref{propAff}, we can assume that $m$ is a co-moving {\it
  inertial} observer of $k$ at $\vq$, i.e., $m\com_{\vq}k$, because of
the following three statements. By \eqref{item-tr} of
Prop.~\ref{prop-lc} and \ax{AxEv}, for every $h\in \IOb$, both 
$\lc^k_m$ and $\lc^k_h$ can be obtained from the other by composing it
by a worldview transformation between {\it inertial} observers. By
Thm.~\ref{thm-poi}, worldview transformations between {\it inertial}
observers are Poincar\'e-transformations in the models of
\ax{SpecRel}.  Poincar\'e-transformations are affine and preserve the
Minkowski distance.

So let us assume that $m$ is a co-moving {\it inertial} observer of
$k$ at $\vq=\iota(t)$.  We prove that $\lc^k_m$ is
differentiable at $t$ and $\vet=\langle1,0,\ldots,0\rangle$ is its
derivative.  This will complete the proof since $\vet$ is a timelike
vector of Minkowski length $1$.  By Lem.~\ref{lem-lc},
$\lc^k_m=\iota\circ w^k_m$.  So by Chain Rule, the derivative of
$\lc^k_m$ at $t$ is the derivative of $w^k_m$ at $\iota(t)$ evaluated
on the derivative of $\iota$ at $t$, i.e.,
$(\lc^k_m)'(t)=d_{\vqq}w^k_m\big(\iota'(t)\big)$.  By
Prop.~\ref{prop-comdiff}, $d_{\vqq}w^k_m=Id$ since
$m\com_{\vq}k$.  It is clear that $\iota'(t)=\vet$.  Thus
$(\lc^k_m)'(t)=\vet$ as it was stated.
\end{proof}
\noindent
Let us note that we have not used \ax{AxEvTr} in the proof of
Thm.~\ref{thm-wp}.

\begin{rem}
\label{Trrem} Well-parametrized curves are exactly the life-curves of
accelerated observers in the models of \ax{AccRel_0}, by which we mean
the following.  Let $\mathfrak{Q}$ be an Euclidean ordered field and
let $f:\Q\parrow \Q^d$ be well-parametrized.  Then there are a model
$\mathfrak{M}$ of \ax{AccRel_0}, observer $k$ and {\it inertial}
observer $m$ such that $\lc^k_m=f$ and the quantity part of
$\mathfrak{M}$ is $\mathfrak{Q}$.  That is not difficult to prove by
using the methods of the present work, see Thm.~\ref{thm-wvcompf}.
\end{rem}

The co-moving relation $\com_{\vq}$ is not symmetric while the
intuitive image behind it is. Therefore, let us introduce a symmetric
version, too.  We say that observers $m$ and $k$ are \df{strong
  co-moving observers}\index{strong co-moving observers} at $\vq$, in
symbols $m \scom_{\vq} k$, iff both $m\com_{\vq} k$ and $k \com_{\vq}
m$ hold.  The following axiom gives a stronger
connection between the worldviews of {\it inertial} and accelerated
observers:

\begin{description}\label{axscmv}
\item[\Ax{AxSCmv}] \index{\ax{AxSCmv}} For every observer and event
  encountered by it, there is a strong co-moving
 \textit{inertial} observer:
\begin{equation*}
\forall k \in \Ob \enskip \forall \vq\in\Q^d\quad k\in ev_k(\vqq) \then
\exists m\in \IOb \enskip m \scom_{\vq} k.
\end{equation*}
\end{description}

\begin{thm}Let $d\ge3$.
Assume \ax{AxSCmv} and \ax{SpecRel}. Let $h$ and $k$ be observers and
let $\vq$ be a coordinate point such that
$\vq\in\wl_k(k)\cap\wl_k(h)$. Then the worldview transformation
$w^k_h$ is differentiable at $\vq$ and one of its derivatives is a
Lorentz transformation.

\end{thm}

\begin{proof}
By axiom \ax{AxSCmv}, there are {\it inertial} observers $h_0$ and
$k_0$ such that $h_0\com_{\vq}h$ and $k\com_{\vq}k_0$. By
Prop.~\ref{prop-comdiff}, $w^h_{h_0}(\vqq)=\vq=w^{k_0}_k(\vqq)$, and
$w^h_{h_0}$ and $w^{k_0}_k$ are differentiable at $\vq$ and one of
their derivatives at $\vq$ is the identity map. By Thm.~\ref{thm-poi},
$w^{h_0}_{k_0}$ is a Poincar\'e transformation. So $w^{h_0}_{k_0}$ is
differentiable and its derivative is a Lorentz transformation. Hence,
by Thm.~\ref{thm-chn}, the composition of $w^h_{h_0}$, $w^{h_0}_{k_0}$
and $w^{k_0}_k$ is differentiable at $\vq$ and one of its derivatives
is a Lorentz transformation. By Prop.~\ref{prop-wkm} this composition
extends $w^k_h$. So $w^k_h$ is also differentiable at $\vq$ and one of
its derivatives is a Lorentz transformation, see Rem.~\ref{rem-diff}.
\end{proof}

\section{Models of the extended theory }

First let us note that it is easy to construct nontrivial models of
\ax{AccRel_0}, for example, the construction in
Misner\,--\,Thorne\,--\,Wheeler~\cite[\S 6, especially pp.172-173 and
  \S 13.6 on pp.327-332]{MTW} can be used for constructing models for
\ax{AccRel_0}.

To characterize the worldview transformations between {\it inertial}
and accelerated observers in the models of \ax{AccRel_0}, let us
introduce the following definition. A function $f:\Q^d\parrow \Q^d$ is
called {\bf worldview compatible} iff $\Setopen p_\tau\setmid
\vp\in\dom f \land \vp_\sigma=\voo\Setclose$ is connected and has at
least two distinct elements, $f$ is differentiable at every
$\vp\in\Q^d$ for which $\vp\in\dom f$ and $\vp_\sigma=\vo$, and one of
its derivatives is a Lorentz transformation at $\vp$ in this case.

\begin{rem}
The worldview transformation $w^k_m$ between observer $k$ and {\it
  inertial} observer $m$ is worldview compatible if $d\ge3$ and
\ax{AccRel_0} is assumed. This can be proved by using
Thms.\ \ref{thm-wp} and \ref{thm-poi} and the fact that $\{p_{\tau} :
\vp\in\dom w^k_m\land \vp_\sigma=\voo\}=\{p_\tau: k\in ev_k(\vpp)\}$,
which follows by Rem.~\ref{rem-axacc}.
\end{rem}

\begin{thm}\label{thm-wvcompf} 
Let $f:\Q^d\parrow\Q^d$ be a worldview compatible function. Then
there is a model of \ax{AccRel_0}, and there are an observer $k$ and an {\it
  inertial} observer $m$ in this model such that $w^k_m=f$.
\end{thm}
\begin{proof}
We construct a model of \ax{AccRel_0} over the field
$\Q$. Let
\begin{equation*}
\Ph\leteq\Setopen line(\vp,\vqq)\setmid \vp,\vq\in\Q^d\lland
|\vp_\sigma-\vq_\sigma|=|p_\tau-q_\tau|\Setclose,
\end{equation*}
\begin{equation*}
\IOb\leteq\Setopen m_{\vr}\setmid\vr\in \dom f\lland \vr_\sigma=\vo\Setclose \cup
\setopen m\setclose,\quad
 \Ob\leteq \IOb\cup\setopen k\setclose 
\quad\B\leteq \Ob\cup \Ph.
\end{equation*}
To finish the construction of the model, we should give the worldview
relation $\W$, too. Instead, it is enough to give the event functions
of all observers or to give the event function of one particular
observer and the worldview transformations that define the event
functions of the other observers. So let us first give the event
function of observer $m$.  For all $\vr\in\dom f$ if $\vr_\sigma=\vo$,
let $d_{\vrr}f$ be the Lorentz transformation which is a derivative of
$f$ at $\vrr$ (since $f$ is worldview compatible, there is such a
Lorentz transformation).  Let
\begin{equation*}
ph\in ev_m(\vpp)\quad \text{iff}\quad \vp\in ph,\qquad m\in
ev_m(\vpp)\quad\text{iff}\quad \vp_\sigma =\vo,
\end{equation*}
\begin{equation*}
k\in ev_m(\vpp)\quad\text{iff}\quad \vp=f(\vrr)\text{ for some
}\vrr\in\dom f \text{ for which }\vr_\sigma=\vo,
\end{equation*}
\begin{equation*}
m_{\vr}\in ev_m(\vpp)\quad \text{iff}\quad \vp=f(\vrr) +\lambda\cdot
d_{\vrr} f(\vet)\text{ for some }\lambda\in\Q\text{, i.e.,}
\end{equation*}
  iff the $line(\vp,f(\vrr))$ is the
tangent line of $\wl_m(k)$. Now we have arranged every body in the
events observed by $m$, thus we have given the event function $ev_m$.
Since $d_{\vrr}f$ is a Lorentz transformation so is its inverse $\big[d_{\vrr}f\big]^{-1}$.
Let 
\begin{equation*}
w^m_{m_{\vr}}(\vpp)\leteq\big[d_{\vrr}f\big]^{-1}(\vp-f(\vrr))+\vr,
\end{equation*}
which is a Poincar\'e transformation since $\big[d_{\vrr}f\big]^{-1}$
is a Lorentz transformation. And let $w^k_m\leteq f$. Now we have
given the model since the worldview relation $\W$ can be defined from
the worldview transformations and $ev_m$. Let us check the axioms. It
is easy to see that \ax{AxSelf_0} is valid by the definition of
$ev_m$. \ax{AxSelf^+_0} is valid since $\setopen p_\tau: \vp\in\dom
f\land \vp_\sigma=\voo\setclose$ is connected and has at least two distinct
elements. \ax{AxEv}, \ax{AxPh} and \ax{AxSymDist} are valid by the
definition of $ev_m$ and the fact that $w^m_{m_{\vr}}$ are Poincar\'e
transformations. \ax{AxEvTr} is valid by the definition of $ev_m$ and
$w^k_m$. To prove that \ax{AxCmv} is valid, we show that
$m_{\vr}\com_{\vr} k$, i.e., $m_{\vr}$ is a co-moving inertial
observer of $k$ at $\vr$. By Prop.~\ref{prop-comdiff}, we have
to check two things, (1) $w^k_{m_{\vr}}(\vrr)=\vr$ and (2)
$w^k_{m_{\vr}}$ is differentiable at $\vr$ and the identity map is one
of its derivatives.
\begin{equation*}
w^k_{m_{\vr}}(\vrr)=\big[d_{\vrr}f\big]^{-1}(f(\vrr)-f(\vrr))+\vr=\vr
\end{equation*}
since $[d_{\vrr}f]^{-1}(\voo)=\vo$ by the linearity of
$[d_{\vrr}f]^{-1}$.  By Thm.~\ref{thm-chn},
\begin{equation*}
d_{\vrr}w^k_{m_{\vr}}(\vrr)=d_{\vrr}f\circ d_{f(\vrr)}w^m_{m_{\vr}}=d_{\vrr}f\circ[d_{\vrr}f]^{-1}=Id_{\Q^d}
\end{equation*}
since the derivative of $w^k_{m_{\vr}}$ at $f(\vrr)$ is its linear part $[d_{\vrr}f]^{-1}$.
That completes the proof of the theorem.\end{proof}

\begin{figure}[h!btp]
\small
\begin{center}
\psfrag{f}{$f$}
\psfrag{Id}{$Id$}
\psfrag{pi}{$\pi_\sigma$}
\includegraphics[keepaspectratio, width=0.5\textwidth]{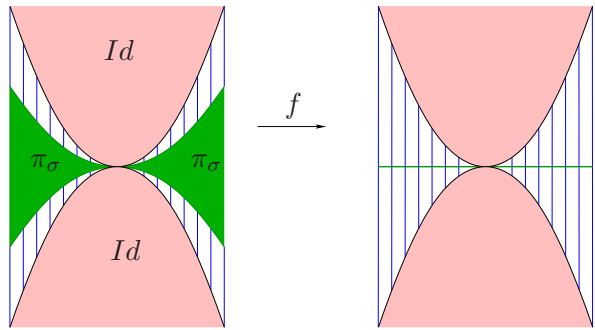}
\caption{\label{fignot} Illustration for Example \ref{example-not}.}
\end{center}
\end{figure}

\begin{rem}
Let $f$ be a worldview compatible transformation. It is not hard to
see, by the proof above, that we can extend any model of \ax{AccRel_0}
and $m\in \IOb$ such that $w^k_m=f$ for some $k\in\Ob$ in the extended
model.
\end{rem}

\begin{example}\label{example-not}
The function 
\begin{equation*}
f(\vpp)\leteq \left\{
\begin{array}{cll}
\langle 0,\vp_\sigma \rangle & \text{ iff } &\hspace{44.8pt}|p_\tau|\le
|\vp_\sigma|^2\\
\langle\frac{|p_\tau|-|\vp_\sigma|^2}{|\vp_\sigma|^2},\vp_\sigma\rangle
& \text{ iff } &\hspace{6pt}|\vp_\sigma|^2 < |p_\tau|<2|\vp_\sigma|^2\\
\vp & \text{ iff } &2|\vp_\sigma|^2\le |p_\tau|,
\end{array}
\right.
\end{equation*} 
see Fig.~\ref{fignot}, is worldview compatible, thus it can define a
worldview of an accelerated observer. Nevertheless, $f$ is not
injective in any neighborhood of the origin.
\end{example}

\begin{rem}
\ax{AccRel_0} is flexible enough to allow an accelerated observer's
coordinate domain to be a subset of the time-axis, i.e., there can be
an observer $h$ such that $\vp_\sigma=\vo$ for all $\vp\in Cd_h$.
Observers of this kind behave as accelerated clocks because they only
use the time coordinate of their coordinate systems, so we can use
them to define {\bf accelerated clocks}\index{accelerated clock}
within \ax{AccRel_0}.
\end{rem}
\chapter{The twin paradox}
\label{chp-twp}

The results of this chapter are based on \cite{twp} and
\cite{mythes}. Here we investigate the logical connection of our
accelerated relativity theory and the twin paradox, which is the
accelerated version of the clock paradox, see
Chap.~\ref{chp-cp}. According to the twin paradox (TwP)\index{TwP}, if
a twin makes a journey into space (accelerates), he will return to
find that he has aged less than his twin brother who stayed at home
(did not accelerate).  However surprising TwP is, it is not a
contradiction. It is only a fact that shows that the concept of time
is not as simple as it seems to be.

A more optimistic consequence of  TwP  is the following. Suppose
you would like to visit a distant galaxy 200 light years away. You are told
it is impossible because even light travels there for 200 years. But you do not
despair, you accelerate your spaceship nearly to the speed of light. Then
you travel there in 1 year of your time. You study there whatever you
wanted, and you come back in 1 year subjective time. When you arrive back,
you aged only  2 years. So you are happy, but of course you cannot tell the
story to your brother, who stayed on Earth. Alas you can  tell it to your
grand-\ldots-grand-children only.

\section{Formulating the twin paradox}
\label{sec-notwp}

To do logical investigation on
TwP, first we have to formulate it in our FOL language.
To formulate TwP, let us denote the \df{set of events encountered by observer $m$ between events $e_1$ and $e_2$} localized by $m$\index{set of encountered events} as
\begin{equation*}
\index{$\enc_m(e_1,e_2)$}
\Df{\enc_m(e_1,e_2)}\,\leteq\, \setopen \,
e\in Ev_m\setmid m\in e \lland \exists \vp\in\Q^d 
 ev_m(\vpp)=e\lland \time_m(e_1)\le p_\tau\le\time_m(e_2)\,\setclose.
\end{equation*}
Then TwP in our FOL setting can be formulated as follows:
\begin{description}
\item[\Ax{TwP}] \index{\ax{TwP}} Every \textit{inertial} observer $m$
 measures at least as much time as any other observer $k$ between any
 two events $e_1$ and $e_2$ in which they meet and which are
 localized by both of them; and they measure the same time iff they
 have encountered the very same events between $e_1$ and $e_2$:
\begin{multline*}
\forall m \in \IOb\enskip \forall k\in \Ob\enskip  \forall e_1,e_2\in Ev \quad 
\Loc_m(e_1)\lland\Loc_m(e_2)\lland\Loc_k(e_1)\\\lland\Loc_k(e_2)\lland k,m\in e_1\cap e_2 \then \time_k(e_1,e_2)\le\time_m(e_1,e_2)\\
\lland\big(\time_m(e_1,e_2)=\time_k(e_1,e_2)\iff \enc_m(e_1,e_2)=\enc_k(e_1,e_2) \big). 
\end{multline*}
\end{description}

Let us also formulate a property of clocks which we call the Duration
Determining Property of Events (DDPE)\index{DDPE}. This property
states that the clocks of any two observers with the same world-line
are synchronized, i.e., they measure the same amount of time
between any two events that they encounter. DDPE is such a basic
property of clocks that it is a possible candidate for assuming it as
an axiom  (if it is not provable from the other axioms).
\begin{description}
\item[\Ax{DDPE}]\index{\ax{DDPE}} If each of
two observers encounters the very same (nonempty) events between two given events,
they measure the same time between these two events:
\begin{multline*}
\forall k,m\in \Ob\enskip \forall e_1,e_2 \in Ev\quad  m,k\in e_1\cap e_2\\ 
\lland \enc_m(e_1,e_2)=\enc_k(e_1,e_2)  \then \time_m(e_1,e_2)=\time_k(e_1,e_2),
\end{multline*}
see the right hand side of Fig.~\ref{fig-twp}.
\end{description}
\begin{figure}
\small
\begin{center}
\psfrag*{TwP}[l][l]{\ax{TwP}}
\psfrag*{Eqtime}[l][l]{\ax{DDPE}}
\psfrag*{f}[b][b]{$w^k_m$}
\psfrag*{p}[rt][rt]{$\vp$}
\psfrag*{q}[rb][rb]{$\vq$}
\psfrag*{t1}[l][l]{$\wl_m(k)$}
\psfrag*{p'}[lt][lt]{$\vpp'$}
\psfrag*{q'}[lb][lb]{$\vqq'$}
\psfrag*{a}[c][c]{$\Longrightarrow\; |q_\tau-p_\tau|<|q'_\tau-p'_\tau|$}
\psfrag*{b}[c][c]{$\Longrightarrow\; |q_\tau-p_\tau|=|q'_\tau-p'_\tau|$}
\psfrag*{m}[b][b]{$\wl_m(m)$}
\psfrag*{k}[b][b]{$\wl_k(k)$}
\psfrag*{t3}[r][r]{$ \wl_m(k) \not\supseteq$}
\psfrag*{t2}[r][r]{$ \wl_k(k) \supseteq$}
\psfrag*{c}[c][c]{same events}
\includegraphics[keepaspectratio, width=0.8\textwidth]{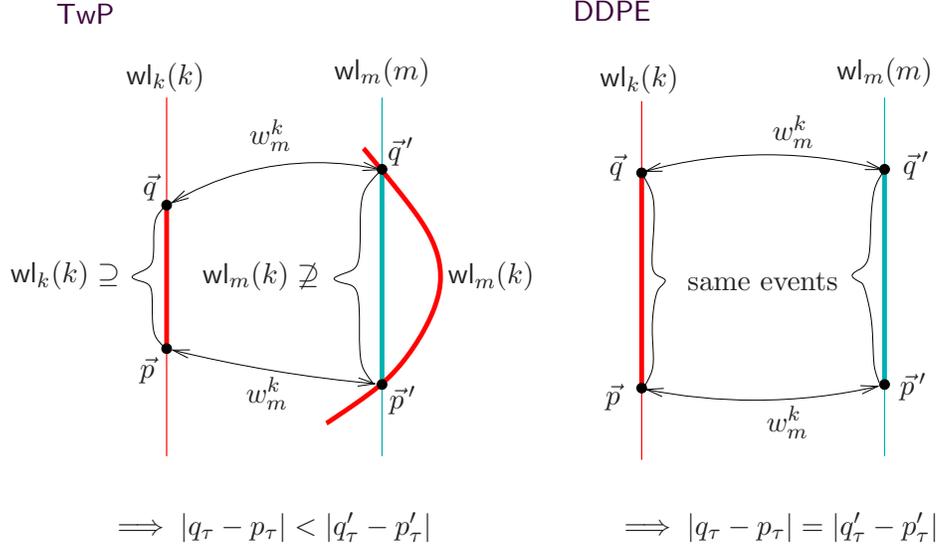}
\caption{\label{fig-twp} Illustration of \ax{TwP} and \ax{DDPE}}
\end{center}
\end{figure}

\begin{thm}\label{thmNoCONT}
For every Euclidean ordered field $\mathfrak{Q}$ not isomorphic to
$\R$, there is a model $\mathfrak{M}$ of $\ax{AccRel_0}$ such that the
quantity part of $\mathfrak{M}$ is $\mathfrak{Q}$ and
$\mathfrak{M}\not\models\ax{Twp}$; moreover,
$\mathfrak{M}\not\models\ax{DDPE}$.
\end{thm}

Thm.~\ref{thmNoCONT} is rather surprising since stationary {\it
  inertial} clocks are synchronized by \ax{SpecRel}, and \ax{AxCmv}
states that accelerated clocks locally behave like {\it inertial}
ones.  The proof of this theorem 
is at p.\pageref{thmNoCONT-proof}.

Thm.~\ref{thmNoCONT} also has strong consequences, it implies that to prove the
Twin Paradox or even DDPE, it does not suffice to add all the
FOL formulas valid in $\R$ to \ax{AccRel_0}. Let $Th(\R)$
denote the set of all FOL formulas valid in $\R$. The
following corollary formulates this strong consequence.

\begin{cor} 
\label{corNoCONT}
$Th(\R)+\ax{AccRel_0}\not\models \ax{TwP}$ 
and $Th(\R)+\ax{AccRel_0}\not\models \ax{DDPE}$.
\end{cor} 

\begin{proof}[\colorbox{proofbgcolor}{\textcolor{proofcolor}{Proof of Cor.~\ref{corNoCONT}}}]
 Let $\mathfrak{Q}$ be a field elementarily equivalent to $\R$, i.e.,
 all FOL formulas valid in $\R$ are valid in $\mathfrak{Q}$,
 too. Assume that $\mathfrak{Q}$ is not isomorphic to $\R$. For
 example, the field of the real algebraic numbers is such.  Let
 $\mathfrak{M}$ be a model of $\ax{AccRel_0}$ with quantity part
 $\mathfrak{Q}$ in which neither $\ax{TwP}$ nor $\ax{DDPE}$ is true.
 Such an $\mathfrak{M}$ exists by Thm.~\ref{thmNoCONT}. That shows that
 $Th(\R)+ \ax{AccRel_0}\not\models \ax{TwP}\lor \ax{DDPE}$ since
 $\mathfrak{M}\models Th(\R)$ by assumption.
\end{proof}

An ordered field is called {\bf non-Archimedean} if it has an element $a$ such
that, for every
positive integer $n$, 
\begin{equation*}-1<\underbrace{a+\ldots+a}_n<1.
\end{equation*}
We call these elements {\bf infinitesimally small}. These are not
FOL definable concepts in our language; however, that is not a
problem since we will not use them in formulas.

The following theorem says that, for countable or non-Archimedean
Euclidean ordered fields, there are quite
sophisticated models of \ax{AccRel_0} in which \ax{TwP} and \ax{DDPE}
are false.

\begin{thm}\label{thmMO}
For every Euclidean ordered field $\mathfrak{Q}$ which is
non-Archimedean or countable, there is a model $\mathfrak{M}$ of
$\ax{AccRel_0}$ such that $\mathfrak{M}\not\models\ax{TwP}$,
$\mathfrak{M}\not\models\ax{DDPE}$, the quantity part of
$\mathfrak{M}$ is $\mathfrak{Q}$ and (i)--(iv) below also hold in
$\mathfrak{M}$.
\begin{itemize}
\item[(i)] Every observer uses the whole coordinate system as
coordinate-domain:
\begin{equation*}
\forall m \in \Ob\quad Cd_m=\Q^d.
\end{equation*}
\item[(ii)] At any point in $\Q^d$, there is a co-moving {\it inertial} observer of any
observer:
\begin{equation*}
\forall k \in \Ob \enskip \forall q \in \Q^d\; \exists m \in \IOb\quad m\com_q k.
\end{equation*}
\item[(iii)] All observers coordinatize the same set of events:
\begin{equation*}
\forall m,k\in \Ob\enskip \forall \vp\in \Q^d\;\exists \vq\in \Q^d\quad ev_{m}(\vpp)=ev_{k}(\vqq).
\end{equation*}
\item[(iv)] Every observer coordinatizes every event only once:
\begin{equation*}
\forall m\in \Ob\enskip \forall \vp,\vq\in \Q^d\quad ev_m(\vpp)=ev_m(\vqq)\then p=q.
\end{equation*}
\end{itemize}
\end{thm}

\begin{proof}[\colorbox{proofbgcolor}{\textcolor{proofcolor}{Proofs of Thms.\ \ref{thmNoCONT} and \ref{thmMO}}}]
\label{thmNoCONT-proof}
\label{thmMO-proof}
We construct four models. Before the constructions let us
introduce a definition. For every $\vp\in \Q^d$, let
$m_{\vp}:\Q^d\rightarrow \Q^d$ denote the translation by vector $\vp$,
i.e., $m_{\vp}: \vq\mapsto \vq+\vp$. Function $f:\Q^d\rightarrow
\Q^d$ is called \df{translation-like}\index{translation-like} iff for
all $\vq\in\Q^d$, there is a $\delta \in \Q^+$ such that
$f(\vpp)=m_{f(\vqq)-\vqq}(\vpp)$ for all $\vp\in B_\delta(\vqq)$, and 
$f(\vpp)=f(\vqq)$ and $\vp_\sigma=\vo$ imply that $\vq_\sigma=\vo$ for
all $\vp,\vq\in \Q^d$.

Let $\mathfrak{Q}=\left<\Q;+,\cdot,< \right>$ be an Euclidean ordered
field and let $k:\Q^d\rightarrow \Q^d$ be a translation-like
map. First we construct a model $\mathfrak{M}_{(\mathfrak{Q},k)}$ of
\ax{AccRel_0} and (i) and (ii) of Thm.~\ref{thmMO},
 which will be a model of (iii) and (iv) of Thm.~\ref{thmMO} if $k$ is
 a bijection. Then we choose $\mathfrak{Q}$ and $k$ appropriately to
 get the desired models in which \ax{DDPE} and \ax{TwP} are false.

Let us now construct the model $\mathfrak{M}_{(\mathfrak{Q},k)}$. Let
$\IOb\leteq \{m_{\vp}:\vp\in \Q^d\}$, $\Ob\leteq \IOb\cup\{k\}$,
$\Ph\leteq \{l: \exists\vp,\vq\in\Q^d \enskip l=line(\vp,\vqq) \lland
|\vp_\sigma-\vq_\sigma|=|p_\tau-q_\tau|\}$, and $\B\leteq \Ob\cup
\Ph$. Recall that $\vo$ is the origin, i.e., $\langle
0,\ldots,0\rangle$. First we give the worldview of $m_{\vo}$, then we
give the worldview of an arbitrary observer $h$ by giving the
worldview transformation between $h$ and $m_{\vo}$. Let $
wl_{m_{\vo}}(ph)\leteq ph$ and $wl_{m_{\vo}}(h)\leteq \{h(\vxx):
\vx_\sigma=\vo\,\}$ for all $ph\in \Ph$ and $h\in \Ob$. And let
$ev_{m_{\vo}}(\vpp)\leteq \{b\in \B: \vp\in wl_{m_{\vo}}(b)\}$ for all
$\vp\in \Q^d$. Let $w^h_{m_{\vo}}\leteq h$ for all $h\in \Ob$. From
these worldview transformations, we can obtain the worldview of each
observer $h$ in the following way: $ev_h(\vpp)\leteq
ev_{m_{\vo}}\big(h(\vpp)\big)$ for all $\vp\in \Q^d$. And from the
worldviews, we can obtain the $\W$ relation as follows: for all $h\in
\Ob$, $b\in \B$ and $\vp\in \Q^d$, let $\W(h,b,\vpp)$ iff $b\in
ev_h(\vpp)$. Thus we have given the model
$\mathfrak{M}_{(\mathfrak{Q},k)}$. Let us note that $w^m_h=m\circ
h^{-1}$ and $m_{h(\vqq)-\vq}\com_{\vq} h$ for all $m,h\in \Ob$ and
$\vq\in \Q^d$. It is easy to check that the axioms of \ax{AccRel_0}
and (i) and (ii) of Thm.~\ref{thmMO} are true in
$\mathfrak{M}_{(\mathfrak{Q},k)}$ and that (iii) and (iv) of
Thm.~\ref{thmMO} are also true in $\mathfrak{M}_{(\mathfrak{Q},k)}$ if
$k$ is a bijection.

To construct the first model, we choose $\mathfrak{Q}$
and $k$ such that \ax{TwP} falls in $\mathfrak{M}_{(\mathfrak{Q},k)}$.
Let $\mathfrak{Q}$ be an Euclidean ordered field different from $\R$.
To define $k$ let $\{I_1, I_2, I_3, I_4, I_5\}$ be a partition%
\footnote{ i.e., $I_i$'s are disjoint and $\Q=I_1\cup I_2 \cup I_3
  \cup I_4 \cup I_5$.}  of $\Q$ such that every $I_i$ is open, $x\in
I_2 \iff x+1\in I_3 \iff x+2 \in I_4$, and for all $y\in I_i$ and
$z\in I_j$, $y\leq z \iff i\leq j$.  Such a partition can be easily
constructed.%
\footnote{ Let $H\subset \Q$ be a nonempty bounded set that does not
  have a supremum. Let $I_1\leteq \{x\in \Q: \exists h \in H \quad
  x<h\}$, $I_2\leteq \{x+1\in \Q: x\in I_1\}\setminus I_1$, $I_3\leteq
  \{x+1\in \Q: x\in I_2\}$, $I_4\leteq \{x+1\in \Q: x\in I_3\}$ and
  $I_5\leteq \Q\setminus(I_1\cup I_2 \cup I_3 \cup I_4)$.} Let
\begin{equation*}
k(\vpp)\leteq \left\{
\begin{array}{cll}
 \vp & \text{ if } & p_\tau\in I_1\cup I_5 , \\
\vp-\vet & \text{ if } & p_\tau\in I_4 , \\
\vp+\vet & \text{ if } & p_\tau\in I_3 , \\
\vp+\vex & \text{ if } & p_\tau \in I_2
 \end{array}
\right.
\end{equation*}
for every $\vp\in \Q^d$, see Fig.~\ref{fig-notwp}.
\begin{figure}
\small
\begin{center}
\psfrag*{p'}[r][r]{$k(\vpp)$} \psfrag*{q'}[r][r]{$k(\vqq)$}
\psfrag*{q}[r][r]{$\vq$} \psfrag*{p}[r][r]{$\vp$}
\psfrag*{text1}[b][b]{worldview of $m$}
\psfrag*{text2}[b][b]{worldview of $k$} \psfrag*{I0}[l][l]{$I_1$}
\psfrag*{I1}[l][l]{$I_2$} \psfrag*{I2}[l][l]{$I_3$}
\psfrag*{I3}[l][l]{$I_4$} \psfrag*{I4}[l][l]{$I_5$}
\psfrag*{tr}[r][r]{$wl_{k}(k)$} \psfrag*{tr1}[r][r]{$wl_m(k)$}
\psfrag*{k}[b][b]{$k$} \psfrag*{f}[t][t]{$w^{k}_m$}
\psfrag*{ff}[t][t]{$w^m_{k}$} \psfrag*{text3}[t][t]{\shortstack[c]{
  first model, \\$\mathfrak{Q}\neq\R$ is
  Euclidean,\\ \ax{TwP} is
  false:\\ $|p'_\tau-q'_\tau| < |p_\tau-q_\tau| $}}
\includegraphics[keepaspectratio, width=0.8\textwidth]{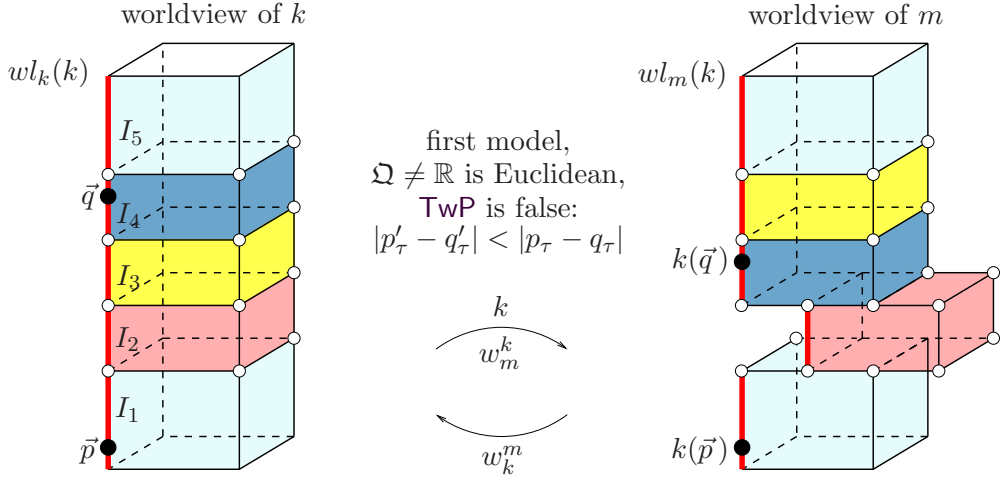}
\caption{\label{fig-notwp}Illustration for the proofs of Thms.\ \ref{thmNoCONT} and \ref{thmMO}}
\end{center}
\end{figure}
It is easy to see that $k$ is a translation-like bijection. Let
$\vp,\vq\in\Q^d$ be coordinate points such that
$\vp_\sigma=\vq_\sigma=\vo$ and $p_\tau\in I_1$, $q_\tau\in I_4$; and
let $m\leteq m_{\vo}$, $e_1\leteq ev_k(\vpp)$, $e_2\leteq ev_k(\vqq)$.
It is easy to see that \ax{TwP} is false in
$\mathfrak{M}_{(\mathfrak{Q},k)}$ for $k$, $m$, and $e_1$, $e_2$ since
\begin{equation*}
\time_m(e_1,e_2)=|k(\vpp)_\tau-k(\vqq)_\tau| < |p_\tau - q_\tau|=\time_k(e_1,e_2),
\end{equation*}
see Fig.~\ref{fig-notwp}.

\begin{figure}
\begin{center}
\small
\psfrag*{text2}[t][t]{worldview of $m$}
\psfrag*{text1}[t][t]{worldview of $k$}
\psfrag*{p'}[r][r]{$k(\vpp)$}
\psfrag*{q'}[r][r]{$k(\vqq)$}
\psfrag*{q}[r][r]{$\vq$}
\psfrag*{p}[r][r]{$\vp$}
\psfrag*{P}[rt][rt]{$\vp$}
\psfrag*{Q}[lb][lb]{$\vq$}
\psfrag*{P'}[rt][rt]{$k(\vpp)$}
\psfrag*{Q'}[rb][rb]{$k(\vqq)$}
\psfrag*{tr}[r][r]{$wl_k(k)$}
\psfrag*{tr1}[r][r]{$wl_m(k)$}
\psfrag*{k}[b][b]{$k$}
\psfrag*{f}[t][t]{$w^k_m$}
\psfrag*{ff}[t][t]{$w^m_k$}
\psfrag*{text6}[t][t]{\shortstack[c]{second model, \\$\mathfrak{Q}\neq\R$ is
Euclidean,\\
\ax{DDPE} is false:\\
$|p_\tau-q_\tau|\neq|k(\vpp)_\tau-k(\vqq)_\tau|$}}
\psfrag*{text4}[t][t]{\shortstack[c]{third model,\\ $\mathfrak{Q}$ is
non-Archimedean,\\ \ax{DDPE} is false:\\
$|p_\tau-q_\tau|\neq |k(\vpp)_\tau-k(\vqq')_\tau|$}}
\psfrag*{text5}[t][t]{\shortstack{fourth model,\\ $\mathfrak{Q}$ is countable
Archimedean,\\ \ax{DDPE} is false:\\
$|p_\tau-q_\tau|\neq |k(\vpp)_\tau-k(\vqq')_\tau|$}}
\psfrag*{tr2}[l][l]{$wl_m(m)$}
\psfrag*{tr0}[r][r]{$wl_m(m)$}
\psfrag*{a+1}[r][r]{$a+1$}
\psfrag*{a+2}[r][r]{$a+2$}
\psfrag*{a}[r][r]{$a$}
\psfrag*{1t}[l][l]{$\widehat{1}$}
\psfrag*{o}[l][l]{$o$}
\psfrag*{I1}[l][l]{$I_1$}
\psfrag*{I2}[l][l]{$I_2$}
\includegraphics[keepaspectratio, width=\textwidth]{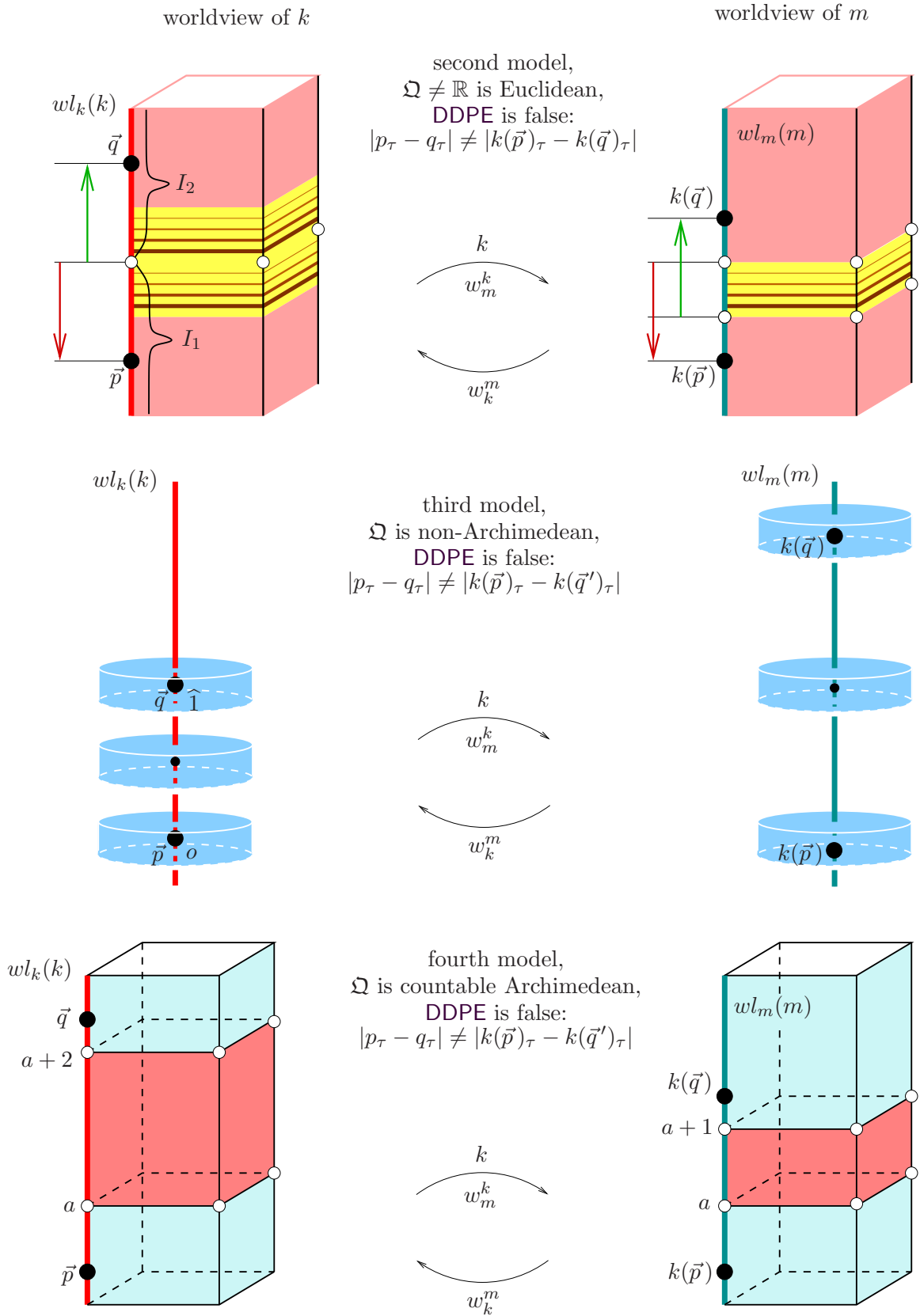}
\caption{\label{fig-noDDPE} Illustration for the proofs of Thms.\ \ref{thmNoCONT} and \ref{thmMO}}
\end{center}
\end{figure}

To construct the second model, let $\mathfrak{Q}$ be an arbitrary
Euclidean ordered field different from $\R$ and let $\{I_1, I_2\}$ be
a partition of $\Q$ such that $x<y$ for all $x\in I_1$ and $y\in I_2$.
Let
\begin{equation*}
k(\vpp)\leteq \left\{
\begin{array}{cll}
 \vp & \text{ if } & p_\tau\in I_1 , \\
\vp-\vet & \text{ if } & p_\tau\in I_2
\end{array}
\right.
\end{equation*}
for every $\vp\in \Q^d$, see Fig.~\ref{fig-noDDPE}. It is easy to see
that $k$ is translation-like. Let $\vp,\vq\in\Q^d$ such that
$\vp_\sigma=\vq_\sigma=\vo$; $p_\tau,p_\tau+1\in I_1$; and
$q_\tau,q_\tau-1\in I_2$. And let $m\leteq m_{\vo}$, $e_1\leteq
ev_k(\vpp)$, $e_2\leteq ev_k(\vqq)$. It is also easy to see that
\ax{DDPE} is false in $\mathfrak{M}_{(\mathfrak{Q},k)}$ for $k$, $m$
and $e_1$, $e_2$ since $m$ and $k$ encounter the very same events
between $e_1$ and $e_2$, however,
\begin{equation*}
\time_k(e_1,e_2)=|p_\tau - q_\tau|\neq |k(\vpp)_\tau-k(\vqq)_\tau|=\time_m(e_1,e_2),
\end{equation*}
see Fig.~\ref{fig-noDDPE}.  This completes the proof of
Thm.~\ref{thmNoCONT}.

To construct the third model, let $\mathfrak{Q}$ be an arbitrary
non-Archimedean, Euclidean ordered field. Let $a\sim b$ denote that
$a,b\in \Q$ and $a-b$ is infinitesimally small. It is not difficult to
see that $\sim$ is an equivalence relation. Let us choose an element
from every equivalence class of $\sim$; and let the chosen element
equivalent to $a\in \Q$ be denoted by $\tilde{a}$. Let $k(\vpp)\leteq
\langle p_\tau+\tilde{p}_\tau,\vp_\sigma\rangle$ for every $\vp\in
\Q^d$, see Fig.~\ref{fig-noDDPE}. It is easy to see that $k$ is a
translation-like bijection. Let $p\leteq \vo$, $q\leteq \vet$,
$k(\vpp)=\langle \tilde{0},0,\ldots,0\rangle$, $k(\vqq)=\langle
1+\tilde{1},0,\dots,0\rangle$. And let $m\leteq m_{\vo}$, $e_1\leteq
ev_k(\vpp)$, $e_2\leteq ev_k(\vqq)$. It is also easy to check that
\ax{DDPE} is false in $\mathfrak{M}_{(\mathfrak{Q},k)}$ for $k$, $m$
and $e_1$, $e_2$ since $m$ and $k$ encounter the very same events
between $e_1$ and $e_2$, however,
\begin{equation*}
\time_k(e_1,e_2)=|p_\tau - q_\tau|\neq |k(\vpp)_\tau-k(\vqq)_\tau|=\time_m(e_1,e_2),
\end{equation*}
 see Fig.~\ref{fig-noDDPE}.

To construct the fourth model, let $\mathfrak{Q}$ be an arbitrary
countable Archimedean Euclidean ordered field and let $k(\vpp)=\langle
f(p_\tau),\vp_\sigma\rangle$ for every $\vp\in \Q^d$ where
$f:\Q\rightarrow \Q$ is constructed as follows, see Figs.\ 
\ref{fig-noDDPE} and \ref{fig-arch}. We can assume that $\mathfrak{Q}$
is a subfield of $\R$ by \cite[Thm.1 in \S VIII]{Fuchs}. Let $a$ be a
real number that is not an element of $\Q$. Let us enumerate the
elements of $[a,a+2]\cap \Q$ and denote the $i$-th element by $r_i$.
First we cover $[a,a+2]\cap \Q$ with infinitely many disjoint
subintervals of $[a,a+2]$ such that the sum of their lengths is $1$,
the length of each interval is in $\Q$ and the distance of the left
endpoint of each interval from $a$ is also in $\Q$. We construct this
covering by recursion. In the $i$-th step, we will use only finitely
many new intervals such that the sum of their lengths is $1/{2^i}$. In
the first step, we cover $r_1$ with an interval of length $1/2$. Let
us suppose that we have covered $r_i$ for each $i<n$.  Since we have
used only finitely many intervals so far, we can cover $r_n$ with an
interval that is not longer than $1/{2^n}$. Since
$\sum_{i=1}^{n}1/2^i<1$, it is not difficult to see that we can choose finitely
many other subintervals of $[a,a+2]$ to be added to this interval such
that the sum of their lengths is $1/{2^n}$. We are given the covering
of $[a,a+2]$. Let us enumerate these intervals. Let $I_i$ be the
$i$-th interval, $d_i$ be the length of $I_i$, $d_0\leteq 0$ and
$a_i\geq 0$ the distance of $a$ and the left endpoint of $I_i$.
$\sum_{i=1}^{\infty}d_i=1$ since $\sum_{i=1}^{\infty}{1}/{2^i}=1$.
Let
\begin{equation*}
f(x)\leteq \left\{
\begin{array}{lll}
x & \text{ if } & x<a , \\
x-1 & \text{ if } & a+2\le x,\\
x-a_n+\sum\limits_{i=0}^{n-1}d_i & \text{ if } & x\in I_n
\end{array}
\right.
\end{equation*}
for all $x\in \Q$, see Fig.~\ref{fig-arch}.
It is easy to see that $k$ is a translation-like bijection.
Let $\vp,\vq\in \Q^d$ such that $p_\tau<a$ and $a+2<q_\tau$; and let $m\leteq m_{\vo}$, $e_1\leteq ev_k(\vpp)$, $e_2\leteq ev_k(\vqq)$. 
It is also easy to check that \ax{DDPE} is false in $\mathfrak{M}_{(\mathfrak{Q},k)}$ for $k$, $m$
 and $e_1$, $e_2$ since $m$ and $k$ encounters the very same same events between $e_1$ and $e_2$, however,
\begin{equation*}
\time_k(e_1,e_2)=|p_\tau - q_\tau|\neq |k(\vpp)_\tau-k(\vqq)_\tau|=\time_m(e_1,e_2),
\end{equation*}
see Fig.~\ref{fig-noDDPE}.

\begin{figure}[h!btp]
\begin{center}
\small
\psfrag*{a}[t][t]{$a$}
\psfrag*{r1}[t][t]{$r_1$}
\psfrag*{r2}[t][t]{$r_2$}
\psfrag*{r3}[t][t]{$r_3$}
\psfrag*{1}[r][r]{$1$}
\psfrag*{v}[l][l]{$\vdots$}
\psfrag*{12}[l][l]{$\frac{1}{2}$}
\psfrag*{14}[l][l]{$\frac{1}{4}$}
\psfrag*{18}[l][l]{$\frac{1}{8}$}
\psfrag*{etc}[rb][rb]{etc.}
\psfrag*{a+2}[t][t]{$a+2$}
\psfrag*{a1}[t][t]{$a_1$}
\psfrag*{I1}[b][b]{$I_1$}
\psfrag*{I2}[b][b]{$I_2$}
\includegraphics[keepaspectratio, width=0.7\textwidth]{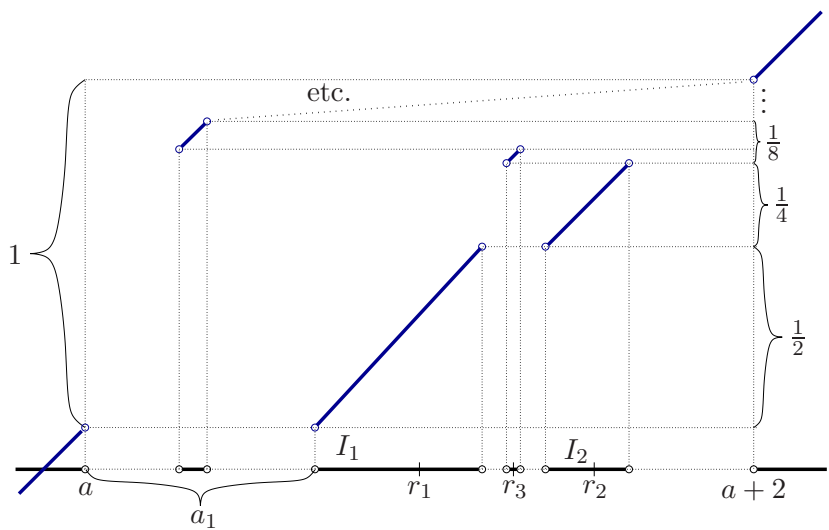}
\caption{\label{fig-arch}Illustration for the proofs of Thms.\ \ref{thmNoCONT} and \ref{thmMO}.}
\end{center}
\end{figure}

\end{proof}

\section{Axiom schema of continuity}\label{sec-cont}
As it was proved in Section \ref{sec-notwp}, \ax{AccRel_0} is not
strong enough to prove properties of accelerated clocks, such as the
twin paradox or even DDPE. The additional property we need is that
every bounded nonempty subset of the quantity part has a
supremum. That is a second-order logic property (because it concerns
all subsets) which we cannot use in a FOL axiom
system. Instead, we will use a kind of ``induction'' axiom schema. It
will state that every nonempty, bounded subset of the quantity part
which can be defined by a FOL formula (using possibly the
extra part of the model, e.g., using the worldview relation)
has a supremum. To formulate this FOL axiom schema, we need
some more definitions.

If $\varphi$ is a formula and $x$ is a variable, then we say that $x$
is a \df{free variable} \label{free variable} of $\varphi$ iff $x$
does not occur under the scope of either $\exists x$ or $\forall x$.
Sometimes we introduce a formula $\varphi$ as $\varphi(\vx\,)$, which
means that all the free variables of $\varphi$ lie in $\vx$.

If $\varphi(x,y)$ is a formula and $\mathfrak{M}=\langle
U;\ldots\rangle$ is a model, then whether $\varphi$ is true or false
in $\mathfrak{M}$ depends on how we associate elements of $U$ with the
free variables $x$ and $y$. When we associate $a\in U$ with $x$ and $b\in U$
with $y$, $\varphi(a,b)$ denotes this truth-value; so $\varphi(a,b)$
is either true or false in $\mathfrak{M}$. For example, if $\varphi$
is $x< y$, then $\varphi(0,1)$ is true while $\varphi(1,0)$ is
false in any ordered field. A formula $\varphi$ is said to be
\df{true}\index{true} in $\mathfrak{M}$ if $\varphi$ is true in
$\mathfrak{M}$ no matter how we associate elements with the free
variables. We say that a \df{subset $H$ of $\Q$ is} (parametrically)
\df{definable by} $\varphi(y,\vx\,)$ iff there is an $\va\in U^n$ such
that $H=\setopen b\in \Q\: :\: \varphi(b,\va\,)\text{ is true in
}\mathfrak{M}\setclose$. We say that a subset of $\Q$ is
\df{definable}\index{definable subset} iff it is definable by a
FOL formula.

Now we formulate the promised axiom schema. To do so, let $\phi(x,\vy\,)$ be a FOL formula of our language.
\begin{description}
\item[\Ax{AxSup_\phi}]\index{\ax{AxSup_\phi}} Every subset of $\Q$
 definable by $\phi(x,\vy\,)$ has a supremum if it is nonempty and
 \df{bounded}.
\end{description}
\noindent A FOL formula expressing \ax{AxSup_\phi} can be
found in Chap.~\ref{chp-a}. Our axiom schema \ax{CONT} below says
that every nonempty bounded subset of $\Q$ that is definable in our
language has a supremum:
\begin{equation*}\index{\ax{CONT}} 
\Ax{CONT} \leteq \Setopen \ax{AxSup_\varphi}\setmid
\varphi \text{ is a FOL formula of our language} \Setclose.
\end{equation*}
Let us note that \ax{CONT} is true in any model whose quantity part is
$\R$.
And let us call the collection of the axioms introduced so far \ax{AccRel}:
\begin{equation*}\index{\ax{AccRel}}
\boxed{ \ax{AccRel}\leteq \ax{AccRel_0}\cup\ax{CONT}}
\end{equation*}

Our next theorem states that \ax{DDPE} can be proved from our
FOL axiom system \ax{AccRel} if $d\ge 3$.

\begin{thm} \label{thmEq}
$\ax{AccRel} \models \ax{DDPE}$ if $d\ge 3$.
\end{thm}

\begin{proof}
\label{thmEq-proof}
Let $k$ and $m$ be observers, and let $e_1$ and $e_2$ be events
localizable by $m$ and $k$ such that $m,k\in e_1\cap e_2$ and
$\enc_m(e_1,e_2)=\enc_k(e_1,e_2)$. We have to prove that
$\time_m(e_1,e_2)=\time_k(e_1,e_2)$. Let $\vp\leteq\loc_k(e_1)$ and
$\vq\leteq\loc_k(e_2)$, and let $\vpp'\leteq\loc_m(e_1)$ and
$\vqq'\leteq\loc_m(e_2)$. Then $\time_k(e_1,e_2)=|p_\tau-q_\tau|$ and
$\time_m(e_1,e_2)=|p'_\tau-q'_\tau|$. See the right hand side of
Fig.~\ref{fig-twp}.

We can assume that $p_\tau\leq q_\tau$ and $p'_\tau\leq q'_\tau$. Let
$h\in \IOb$. We prove that
$\big|q_\tau-p_\tau\big|=\big|q'_\tau-p'_\tau\big|$, by applying
Thm.~\ref{thmJeq} as follows: let $[a,b]\leteq [p_\tau,q_\tau]$,
$[a',b']\leteq [p'_\tau,q'_\tau]$, $f\leteq \lc^k_h$ and $g\leteq
\lc^m_h$. By \ax{AxSelf^+_0} and \ax{AxCmv}, we conclude that
$[a,b]\subseteq \dom f$ and $[a',b']\subseteq \dom g$ since
$h\in\IOb$, see Prop.~\ref{prop-lc}. From \ax{AccRel_0} it follows that
$f$ and $g$ are definable and well-parametrized timelike curves, see
Thm.~\ref{thm-wp}. By \ax{AxSelf_0}, we have that
$\{\,f(r):r\in[a,b]\,\}=\{\,g(r'):r'\in [a',b']\,\}$ since
$\enc_k(e_1,e_2)=\enc_m(e_1,e_2)$. Thus, by Thm.~\ref{thmJeq}, we
conclude that $\big|q_\tau-p_\tau\big|=\big|q'_\tau-p'_\tau\big|$; and
that is what we wanted to prove.
\end{proof}

Now let us prove the following theorem stating that the twin paradox
is a logical consequence of \ax{AccRel} if $d\ge3$.

\begin{thm} \label{thmTwp}
$\ax{AccRel} \models \ax{TwP}$ if $d\ge 3$.
\end{thm}

\begin{proof}
\label{thmTwp-proof}
Let $m\in \IOb$ and $k\in \Ob$; and let $e_1$ and $e_2$ be events
localizable by $m$ and $k$ such that $m,k\in e_1\cap e_2$.  By
Thm.~\ref{thmEq}, \ax{DDPE} is provable from \ax{AccRel}. So we have to
prove the following only:
\begin{equation}\label{eq-twp0} \time_m(e_1,e_2)\ge\time_k(e_1,e_2),\enskip\text{ and} 
\end{equation}
\begin{equation}\label{eq-twp} 
\enc_m(e_1,e_2)=\enc_k(e_1,e_2)\enskip\text{ if }\enskip\time_m(e_1,e_2)=\time_k(e_1,e_2).
\end{equation}
To do so, let $\vp\leteq\loc_k(e_1)$, $\vq\leteq\loc_k(e_2)$; and let $\vpp'\leteq\loc_m(e_1)$, $\vqq'\leteq\loc_m(e_2)$.
Then 
\begin{equation*}
\time_k(e_1,e_2)=|p_\tau-q_\tau|\quad\text{ and }\quad\time_m(e_1,e_2)=|p'_\tau-q'_\tau|,
\end{equation*}
see Fig.~\ref{fig-twp}. Thus we have to prove that
$|q_\tau-p_\tau|\le|q'_\tau-p'_\tau|$, and that $\enc_m(e_1,e_2)=
\enc_k(e_1,e_2)$ if $|q_\tau-p_\tau|=|q'_\tau-p'_\tau|$. We are going
to prove them by applying Thm.~\ref{thmFtwp} to $\lc^k_m$ and $[p_\tau,
  q_\tau]$. From \ax{AccRel_0} we have that
\begin{equation}
\label{twp-e1}
\lc^k_m: \Q\parrow \Q^d \text{ is a definable and well-parametrized
 timelike curve,}
\end{equation}
see Thm.~\ref{thm-wp}. By \ax{AxSelf_0},
$\vp_\sigma=\vq_\sigma=\vpp'_\sigma=\vqq'_\sigma=\vo$ since $m,k\in
e_1\cap e_2$. By the definition of life-curve,
\begin{equation}
\label{twp-e2}
\lc^k_m(p_\tau)=\vpp'\quad\text{ and }\quad \lc^k_m(q_\tau)=\vqq'.
\end{equation}
By \ax{AxSelf^+_0}, we have
\begin{equation}
\label{twp-e3}
[p_\tau, q_\tau]\subseteq \dom \lc^k_m.
\end{equation}
Hence, by applying (i) of Thm.~\ref{thmFtwp} to $\lc^k_m$ and
$[p_\tau,q_\tau]$, we get that
\begin{equation*}
|q_\tau-p_\tau|\le|\lc^k_m(q_\tau)_\tau-\lc^k_m(p_\tau)_\tau|=|q'_\tau-p'_\tau|.
\end{equation*}
Consequently, $\time_k(e_1,e_2)\le\time_m(e_1,e_2)$. So
\eqref{eq-twp0} is proved.

We prove \eqref{eq-twp} by proving its contraposition.  Moreover, we
prove that $\time_k(e_1,e_2)<\time_m(e_1,e_2)$ if
$\enc_m(e_1,e_2)\neq\enc_k(e_1,e_2)$. That will be proved by applying
(ii) of Thm.~\ref{thmFtwp} to $\lc^k_m$ and $[p_\tau,q_\tau]$. To do
so, let us assume that $\enc_m(e_1,e_2)\neq\enc_k(e_1,e_2)$. Since
$\enc_m(e_1,e_2)\neq\enc_k(e_1,e_2)$, there are two possibilities:
either there is an event $e$ such that $e\in\enc_k(e_1,e_2)$ and
$e\not\in\enc_m(e_1,e_2)$, or there is an event $\bar e$ such that
$\bar e\in\enc_m(e_1,e_2)$ and $\bar e\not\in\enc_k(e_1,e_2)$. If
there is such $e$, there is an $x\in[p_\tau,q_\tau]$ such that
$\lc^k_m(x)=\loc_m(e)$. By \ax{CONT}-Bolzano Theorem,
$\lc^k_m(x)_\tau\in[q'_\tau,p'_\tau]$, since $\lc^k_m$ is a definable
timelike curve and $\lc^k_m(p_\tau)=\vpp'$,
$\lc^k_m(q_\tau)=\vqq'$. Thus, since $e\not\in\enc_m(e_1,e_2)$, we
have $\lc^k_m(x)_\sigma\neq\vo$. If $\bar e$ is such that $\bar
e\in\enc_m(e_1,e_2)$ and $\bar e\not\in\enc_k(e_1,e_2)$, then
$\lc^k_m(t)\neq\loc_m(\bar e)$ for all $t\in [p_\tau,q_\tau]$. By
\ax{CONT}-Bolzano Theorem, there is an $x\in[p_\tau,q_\tau]$ such that
$\lc^k_m(x)_\tau=\loc_m(\bar e)_\tau$. By \ax{AxSelf_0}, $\loc_m(\bar
e)_\sigma=\vo$. Therefore, $\lc^k_m(x)_\sigma\neq\vo$ since $\bar
e\not\in\enc_k(e_1,e_2)$. So in both cases there is an
$x\in[p_\tau,q_\tau]$ such that
$\lc^k_m(x)_\sigma\neq\vo=\lc^k_m(p_\tau)_\sigma$. Consequently, there
is an $x\in\dom\lc^k_m$ such that
\begin{equation*}
x\in [p_\tau,q_\tau]\quad\mbox{and}\quad \lc^k_m(x)_\sigma\neq \lc^k_m(p_\tau)_\sigma.
\end{equation*}
By (ii) of Thm.~\ref{thmFtwp}, we get that
\begin{equation*}
|q_\tau-p_\tau|<|\lc^k_m(q_\tau)_\tau-\lc^k_m(p_\tau)_\tau|=|q'_\tau-p'_\tau|.
\end{equation*}
Consequently, $\time_k(e_1,e_2)<\time_m(e_1,e_2)$ if $\enc_k(e_1,e_2)\neq \enc_m(e_1,e_2)$. That completes the proof of the theorem.\end{proof}

\begin{que}
Can the \ax{CONT} axiom schema be replaced by some natural assumptions on observers such that the theorem above remains valid?
\end{que}

\begin{rem} The assumption $d\ge 3$ cannot be omitted from
Thm.~\ref{thmTwp}.  However, Thms.\ \ref{thmTwp} and \ref{thmEq} remain
true if we omit the assumption $d\ge 3$ and assume the auxiliary
axioms \ax{AxThExp} of Chap.~\ref{chp-cp} and \ax{AxLine} defined
below, i.e.,
\begin{equation*}
\ax{AccRel}+\ax{AxThExp}+\ax{AxLine}\models \ax{TwP}\;\land\;\ax{DDPE}
\end{equation*}
holds for $d=2$, too.  A proof for the latter statement can be
obtained from the proofs of Thms.\ \ref{thmTwp} and \ref{thmEq} by
\cite[Items 4.3.1, 4.2.4, 4.2.5]{mythes} and
\cite[Thm.1.4(ii)]{AMNsamples}.

\begin{description}
\item[\Ax{AxLine}] World-lines of {\it inertial} observers are lines
 according to any {\it inertial} observer:
 \begin{equation*}
\qquad\qquad\quad\forall m,k \in \IOb\enskip \exists \vp,\vq\in\Q^d \quad wl_m(k)=line(\vp,\vqq).
\end{equation*}
\end{description}
\end{rem}

\begin{que}
Can the assumption $d\ge 3$ be omitted from Thm.~\ref{thmEq}, i.e.,
does $\ax{AccRel}\models\ax{DDPE}$ hold for $d=2$?
\end{que}

In the next chapter, we discuss how the present methods
and in particular \ax{AccRel_0} and \ax{CONT} can be used for
introducing gravity via Einstein's equivalence principle and for
proving that ``gravity causes time to run slow'' (also called gravitational
time dilation). In this connection we would like to point out that it is explained, in
Misner et al.\ \cite[pp.172-173, 327-332]{MTW}, that 
the theory of accelerated observers (in flat spacetime) is a rather
useful first step in building up general relativity by using the
methods of that book.

\chapter{Simulating gravitation by accelerated observers}
\label{chp-grav}

Before we derive a FOL axiom system of general
relativity from our theory \ax{AccRel}, let us investigate the
strength of \ax{AccRel} by proving some theorems on gravitation from
it.  The results of this chapter are based on \cite{logtw} and
\cite{folfrt}.  Here we investigate the effect of gravitation on
clocks in our FOL setting by proving theorems about
gravitational time dilation. This effect roughly  means that
``gravitation makes time flow slower,'' that is to say, clocks in the
bottom of a tower run slower than clocks in its top. We use Einstein's
equivalence principle to treat gravitation in \ax{AccRel}. This
principle says that a uniformly accelerated frame of reference is
indistinguishable from a rest frame in a uniform gravitational field,
see, e.g., d'Inverno~\cite[\S 9.4]{dinverno}. So instead of
gravitation we will talk about acceleration and instead of towers we
will talk about spaceships.  This way the gravitational time dilation
will become the following statement: ``Time flows more slowly in the
back of a uniformly accelerated spaceship than in its front.''

One of the reasons why gravitational time dilation is interesting and
important is that general relativistic hypercomputing is based on this
effect, see \cite{ann}, \cite{nemeti-dgy}. Another reason is that it
leads to other surprising effects, such as that ``time stops'' at the event
horizons of huge\footnote{This statement is true for any black hole but it is
  interesting in the case of huge ones.} (ca.\ $10^{10}$ solar mass)
black holes. That is true because at the event horizon ``gravitational
force'' (meant in the sense of Rindler~\cite[\S 11.2 p.230]{Rin}) tends to
infinity. The possibility of the existence of (traversable) wormholes is also related to
these ideas, see \cite[p.140]{kaku}, Novikov~\cite{novikov}, Thorne~\cite{thorne} and \cite{rkr}.

Here we concentrate on the general case when the spaceship is not
necessarily uniformly accelerated. This case corresponds to the
situation when the tower is in a possibly changing gravitational
field.  At first it is not clear whether the changing gravitational
field has any physical relevance. However, every ``physical''
gravitational field is changing slightly. For example, the source of
the gravitation may lose energy by radiation, which might
significantly change the gravitational field in the long run. {Black
  holes may radiate by Hawking's radiation hypothesis. Changing
  gravitational fields also play a key role in the theory of
  gravitational waves.}

\section{Formulating gravitational time dilation}
Let us formulate the sentence ``Time flows more slowly in the back
of an accelerated spaceship than in its front.'' in our FOL language.

\begin{figure}[h!t]
\small
\begin{center}
\psfrag{a}[tl][tl]{$k$}
\psfrag{m}[l][l]{$m$}
\psfrag{e}[b][b]{$e$}
\psfrag{e1}[tl][tl]{$e_1$}
\psfrag{e'}[bl][bl]{$e'$}
\psfrag{e2}[br][br]{$e_2$}
\psfrag{2l}[l][l]{$2\lambda$}
\psfrag{l}[b][b]{$\lambda$}
\psfrag{ph1}[tr][tr]{$ph_1$}
\psfrag{ph2}[br][br]{$ph_2$}
\psfrag{text1}[tl][tl]{$(a)$}
\psfrag{text2}[tl][tl]{$(b)$}
\includegraphics[keepaspectratio, width=0.7\textwidth]{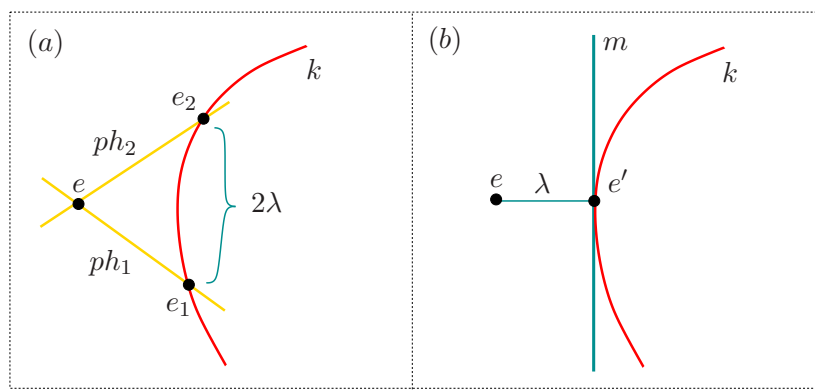}
\caption{\label{distfig} Illustrations of the radar distance and the Minkowski distance, respectively}
\end{center}
\end{figure}

To talk about spaceships, we need a concept of distance between
events and observers.
We have two natural candidates for that:
\begin{itemize}
\item Event $e$ is at \df{radar distance}\index{radar distance}
 $\lambda\in\Q^+$ from observer $k$ iff there are events $e_1$ and
 $e_2$ and photons $ph_1$ and $ph_2$ such that $k\in e_1\cap e_2$,
 $ph_1\in e\cap e_1$, $ph_2\in e\cap e_2$ and
 $\time_k(e_1,e_2)=2\lambda$. Event $e$ is at \df{radar distance}
 $0$ from observer $k$ iff $k\in e$. See $(a)$ of Fig.~\ref{distfig}.
\item Event $e$ is at \df{Minkowski distance}\index{Minkowski
 distance} $\lambda\in\Q$ from observer $k$ iff there is an event
 $e'$ such that $k\in e'$, $e\sim_m e'$ and $\dist_m(e,e')=\lambda$
 for every co-moving {\it inertial} observer $m$ of $k$ at $e'$. See
 $(b)$ of Fig.~\ref{distfig}.
\end{itemize}

We say body $b$ is at constant radar distance from observer $k$
according to $k$ iff the radar distance (from $k$) of every event in
which $b$ participates is the same.  The notion of constant Minkowski
distance is analogous.

To state that the {\it spaceship does not change its direction}, we
need to introduce another concept.  We say that observers $k$ and $b$
are \df{coplanar}\index{coplanar} iff $\wl_m(k)\cup \wl_m(b)$ is a
subset of a vertical plane in the coordinate system of an {\it
  inertial} observer $m$.  A plane is called a \df{vertical
  plane}\index{vertical plane} iff it is parallel to the time-axis.

Now we introduce two concepts of spaceship.  Observers $b$, $k$ and
$c$ form a \df{radar spaceship}\index{radar spaceship}, in symbols
$\Df{\rship}$\index{$\rship$}, iff $b$, $k$ and $c$ are coplanar and
$b$ and $c$ are at (not necessarily the same) constant radar distances
from $k$ according to $k$.  The definition of the \df{Minkowski
  spaceship}\index{Minkowski spaceship}, in symbols
$\Df{\mship}$\index{$\mship$}, is analogous.

We say that event $e_1$ \df{precedes}\index{precedes} event $e_2$
according to observer $k$ iff $\loc_m(e_1)_\tau\le \loc_m(e_2)_\tau$
for all co-moving \textit{inertial} observers $m$ of $k$. In this
case we also say that $e_2$ \df{succeeds}\index{succeeds} $e_1$
according to $k$. We need these concepts to distinguish the past and
the future light cones according to observers. Let us note that 
no time orientation is definable from \ax{AccRel}; so we can only
speak of orientation according to observers. However, there are
several possible axioms which make time orientation possible, e.g.,
\begin{equation*}
\forall m,k\in\IOb\quad w^k_m(\voo)_\tau<w^k_m(\vet)_\tau
\end{equation*} is such.

\begin{figure}[h!btp]
\small
\begin{center}
\psfrag{e}[r][r]{$e$}
\psfrag{he1}[l][l]{$\hat{e}_1$}
\psfrag{he2}[l][l]{$\hat{e}_2$}
\psfrag{hp1}[tr][tr]{$\hat{p}_1$}
\psfrag{hp2}[tl][tl]{$\hat{p}_2$}
\psfrag{te1}[br][br]{$\tilde{e}_1$}
\psfrag{te2}[l][l]{$\tilde{e}_2$}
\psfrag{tp1}[br][br]{$\tilde{p}_1$}
\psfrag{tp2}[bl][bl]{$\tilde{p}_2$}
\psfrag{e1}[r][r]{$e_1$}
\psfrag{e2}[lr][lr]{$e_2$}
\psfrag{ph1}[rt][rt]{$ph_1$}
\psfrag{ph2}[lt][lt]{$ph_2$}
\psfrag{m}[lb][lb]{$m$}
\psfrag{a}[tl][tl]{$k$}
\psfrag{l}[b][b]{$\lambda$}
\psfrag*{text1}[cb][cb]{(a)}
\psfrag*{text2}[cb][cb]{(b)}
\psfrag*{text3}[cb][cb]{(c)}
\includegraphics[keepaspectratio, width=0.8\textwidth]{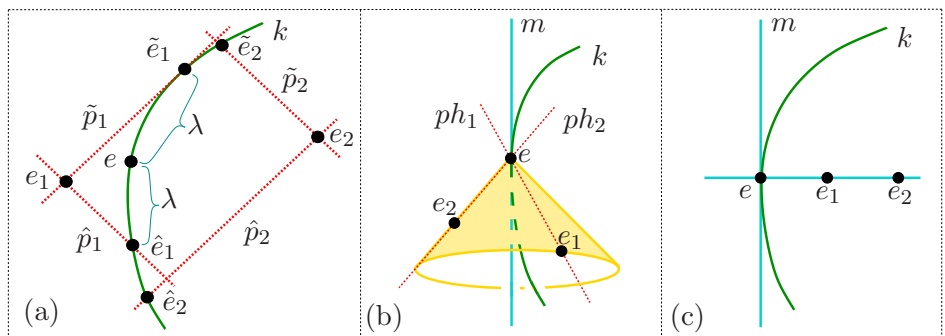}
\caption{\label{simfig} Illustrations of relations $e_1\simrad_k e_2$, $e_1\simph_k e_2$ and $e_1\simmu_k e_2$, respectively}
\end{center}
\end{figure}

We also need a concept to decide which events happen at the same
time according to an accelerated observer.
The following three
natural concepts offer themselves:
\begin{itemize}
\item Events $e_1$ and $e_2$ are \df{radar simultaneous}\index{radar
 simultaneous} for observer $k$, in symbols $e_1\Df{\simrad_k}
 e_2$\index{$\simrad$}, iff there are events $e$, $\hat e_1$, $\hat
 e_2$, $\tilde e_1$, $\tilde e_2$ and photons $\tilde{p}_1$,
 $\tilde{p}_2$, $\hat{p}_1$, $\hat{p}_2$ such that $k\in e\cap
 \tilde{e}_i\cap\hat{e}_i$, $\hat{p}_i\in e_i\cap\hat{e}_i$,
 $\tilde{p}_i\in e_i\cap\tilde{e}_i$, ($\tilde{e}_i\neq\hat{e}_i$ or
 $e_i=e$) and $\time_k(e,\hat{e}_i)=\time_k(e,\tilde{e}_i)$ if
 $i\in\setopen 1,2\setclose$, see Fig.~\ref{simfig}.
\item Events $e_1$ and $e_2$ are \df{photon simultaneous}\index{photon
 simultaneous} for observer $k$, in symbols $e_1\Df{\simph_k}
 e_2$\index{$\simph$}, iff there are an event $e$ and photons $ph_1$
 and $ph_2$ such that $k\in e$, $ph_1\in e\cap e_1$, $ph_2\in e\cap
 e_2$ and $e_1$ and $e_2$ precede $e$ according to $k$. See $(b)$ of
 Fig.~\ref{simfig}.
\item Events $e_1$ and $e_2$ are \df{Minkowski
  simultaneous}\index{Minkowski simultaneous} for observer $k$, in
  symbols $e_1\Df{\simmu_k} e_2$\index{$\simmu$}, iff there is an
  event $e$ such that $k\in e$ and $e_1$ and $e_2$ are simultaneous
  for any co-moving \textit{inertial} observer of $k$ at $e$. See
  $(c)$ of Fig.~\ref{simfig}.
\end{itemize}

\begin{rem}
Let us note that, for \textit{inertial} observers, the concepts of
radar simultaneity, Minkowski simultaneity and the concept of
simultaneity introduced on p.\pageref{sim} coincide, and any two of these
three simultaneity concepts coincide only for \textit{inertial}
observers.
\end{rem}

Radar simultaneity and Minkowski simultaneity are the two most natural
generalizations (for non-{\it inertial} observers) of the standard
simultaneity introduced by Einstein in \cite{Einstein}. In the case
of Minkowski simultaneity, the standard simultaneity of co-moving
\textit{inertial} observers is rigidly copied, while in the case of
radar simultaneity, the standard simultaneity is generalized in a more
flexible way. Dolby and Gull calculate and illustrate the radar
simultaneity of some coplanar accelerated observers in
\cite{Dolby-Gull}. 

Let us note that the Minkowski simultaneity of observer $k$ is an
equivalence relation if and only if $k$ does not accelerate. So one
can argue against regarding it as a simultaneity concept for non-{\it
  inertial} observers, too. We think, however, that it is so
straightforwardly generalized from the standard concept of
simultaneity that it deserves to be forgiven for its weakness and to
be called simultaneity. Let us also note that the Minkowski simultaneity
of $k$ is an equivalence relation on a small enough neighborhood of
the world-line of $k$ if this world-line is smooth enough.

The concept of photon simultaneity is the
least usual and the most naive. It is based on the simple idea that
an event is happening right now iff it is seen to be happening right
now. Some authors require from a simultaneity concept to be an
equivalence relation such that its equivalence classes are smooth
spacelike hypersurfaces, see, e.g., Matolcsi~\cite{Matolcsi}.
In spite of the fact that equivalence classes of $\simph_k$ are
neither smooth nor spacelike, we think that it might to be called
simultaneity, see, e.g., Hogarth
\cite{Hogarth} and Malament \cite{Malament}. This concept
occurs as a possible simultaneity concept in some of the papers
investigating the question of conventionality/definability of
simultaneity, see, e.g., Ben-Yami \cite{Ben-Yami}, Rynasiewicz
\cite{Rynasiewicz}, Sarkar and Stachel \cite{Sarkar-Stachel}. Let us
also note that all of the introduced simultaneity and distance
concepts are experimental ones, i.e., they can be determined by
observers by means of experiments with clocks and photons.

We distinguish the front and the back of the spaceship by the
direction of the acceleration, so we need a concept for direction. We
say that the \df{directions of $\vp\in \Q^d$ and $\vq\in \Q^d$ are the
  same}\index{direction}, in symbols
$\vpp\Df{\upp}\vqq$\index{$\upp$}, if $\vp$ and $\vq$ are spacelike
vectors, and there is a $\lambda \in \Q^+$ such that $\lambda\cdot
\vp_\sigma=\vq_\sigma$, see $(a)$ of Fig.~\ref{figupp}. When $\vp$ and
$\vq$ are timelike vectors, we also use this notation if $p_\tau
q_\tau>0$.

\begin{figure}[h!btp]
\small
\begin{center}
\psfrag{p}[lb][lb]{$\vpp$} 
\psfrag{ps}[lb][lb]{$\vpp_\sigma$}
\psfrag{qs}[lb][lb]{$\vq_\sigma$} 
\psfrag{ph}[rb][rb]{$ph$}
\psfrag{q}[r][r]{$\vqq$} 
\psfrag{q3}[l][l]{$\vq_3$}
\psfrag{q2}[l][l]{$\vq_2$} 
\psfrag{q1}[l][l]{$\vq_1$}
\psfrag{aa}[l][l]{$(a)$} 
\psfrag{bb}[l][l]{$(b)$}
\psfrag{o}[t][t]{$\vo$} 
\psfrag{b}[t][t]{$b$} 
\psfrag{k}[t][t]{$c$}
\psfrag{e}[t][t]{$e$} 
\psfrag{eb}[t][t]{$e_b$}
\psfrag{ek}[t][t]{$e_c$} 
\psfrag{k1}[t][t]{$c'$}
\psfrag{b1}[t][t]{$b'$}
\includegraphics[keepaspectratio, width=0.7\textwidth]{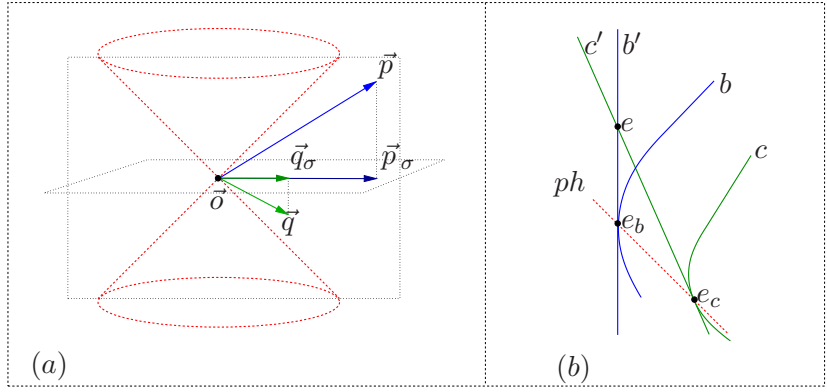}
\caption{\label{figupp} $(a)$ illustrates $\vpp\upp\vqq$, and
$(b)$ illustrates that observer $c$ is approaching observer $b$,
as seen by $b$ by photons.}
\end{center}
\end{figure}

Now let us focus on the definition of acceleration in our FOL
setting. The life-curves of observers and the derivative $f'$ of a
given function $f$ are both FOL definable concepts, see
pages \pageref{life-curve} and \pageref{derivative}. Thus if the life-curve of observer $k$
according to observer $m$ is a function, then the following
definitions are also FOL ones. The \df{relative
 velocity}\index{relative velocity} $\Dff{\fvvkm}$\index{$\fvvkm$} of
observer $k$ according to observer $m$ at instant $t\in\Q$ is the
derivative of the life-curve of $k$ according to $m$ at $t$ if it is
differentiable at $t$;  otherwise it is undefined. The \df{relative
 acceleration}\index{relative acceleration}
$\Dff{\fvakm}$\index{$\fvakm$} of observer $k$ according to observer
$m$ at instant $t\in\Q$ is the derivative of the relative velocity of
$k$ according to $m$ at $t$ if it is differentiable at $t$; otherwise it is
undefined.

We say that \df{the direction of the spaceship $\ship$ is the same as that of the} 
\df{acceleration of $k$} iff the following holds:
\begin{multline*}
\forall m \in \IOb \enskip \forall t \in \dom \fvakm\enskip \enskip \forall \vpp,\vqq \in Cd_m\quad \\c\in ev_m(\vpp) \lland 
b\in ev_m(\vqq)\lland \vpp\seq \vqq \then\fvakm(t)\upp (\vpp-\vqq).
\end{multline*}

The \df{acceleration}\index{acceleration} of observer $k$ at instant
$t\in\Q$ is defined as the unsigned Minkowski length of the relative
acceleration according to any {\it inertial} observer $m$ at $t$,
i.e.,
\begin{equation*}\index{$a_k(t)$}
\Df{a_k(t)}\leteq -\mu\big(\fvakm(t)\big).
\end{equation*}
The reason for the ``$-$'' sign in this definition is the fact that
$\mu\big(\fvakm(t)\big)$ is negative since $\fvakm(t)$ is a spacelike
vector, see Thm.~\ref{thm-wp} and Prop.~\ref{prop-vmorta}.  The
acceleration is a well-defined concept since it is independent of the
choice of the {\it inertial} observer $m$, see Thm.~\ref{thm-poi} and
Prop.~\ref{prop-inv}. We say that observer $k$ is \df{positively
  accelerated}\index{positively accelerated observer} iff $a_k(t)$ is
defined and greater than $0$ for all $t\in \dom \lc^k_k$. Observer $k$
is called \df{uniformly accelerated}\index{uniformly accelerated
  observer} iff there is an $a\in\Q^+$ such that $a_k(t)=a$ for all
$t\in \dom \lc^k_k$.

We say that \df{the clock of $b$ runs slower than the clock of $c$ as}
\df{seen by \,$k$\, by radar} iff
$\time_b(e_b,\bar{e}_b)<\time_c(e_c,\bar{e}_c)$ for all events
$e_b,\bar{e}_b, e_c, \bar{e}_c$ for which $b\in e_b\cap \bar{e}_b$,
$c\in e_c\cap \bar{e}_c$ and $e_b\simrad_k e_c$, $\bar{e}_b\simrad_k
\bar{e}_c$. If it is \df{seen by photons}, we use $\simph_k$ instead
of $\simrad_k$. Similarly, if it is \df{seen by Minkowski
 simultaneity}, we use $\simmu_k$ instead of $\simrad_k$.

\section{Proving gravitational time dilation}
Let us prove here two theorems about gravitational time
dilation. Both theorems state that gravitational time
dilation follows from \ax{AccRel}, they only differ in the
formulation of this statement.

Let us first prove a theorem about the clock-slowing effect of
gravitation in radar spaceships.

\begin{thm} \label{thm-rad}
Let $d\ge 3$.
Assume \ax{AccRel}.
Let $\rship$ be a radar spaceship such that:
\begin{itemize}
\item[(i)] observer $k$ is positively accelerated,
\item[(ii)] the direction of the spaceship is the same as that of the
 acceleration of observer $k$.
\end{itemize}
Then both (1) and (2) hold:
\begin{itemize}
\item[$(1)$] The clock of $b$ runs slower than the clock of $c$ as seen by $k$ by radar.
\item[$(2)$] The clock of $b$ runs slower than the clock of $c$ as seen by each of $k$, $b$ and $c$ by photons.
\end{itemize}
\end{thm}

\begin{figure}[h!btp]
\small
\begin{center} 
\psfrag{a}[l][l]{$\alpha$}
\psfrag{b}[l][l]{$\beta$}
\psfrag{c}[l][l]{$\gamma$}
\psfrag{at}[l][l]{$\alpha(t)$}
\psfrag{at+R}[r][r]{$\alpha(t+R)$}
\psfrag{at-R}[r][r]{$\alpha(t-R)$}
\psfrag{at+r}[r][r]{$\alpha(t+r)$}
\psfrag{at-r}[r][r]{$\alpha'(t-r)$}
\psfrag{a1t+R}[b][b]{$\alpha'(t+R)$}
\psfrag{a1t-R}[b][b]{$\alpha'(t-R)$}
\psfrag{a1t+r}[b][b]{$\alpha'(t+r)$}
\psfrag{a1t-r}[b][b]{$\alpha'(t-r)$}
\psfrag{a1t+rtl}[t][t]{$\alpha'(t+r)$}
\psfrag{a1t-rtr}[t][t]{$\alpha'(t-r)$}
\psfrag{bt}[l][l]{$\beta_*(t)$}
\psfrag{btr}[r][r]{$\beta_*(t)$}
\psfrag{atr}[r][r]{$\alpha(t)$}
\psfrag{ct}[l][l]{$\gamma_*(t)$}
\psfrag{a1t}[b][b]{$\alpha'(t)$}
\psfrag{b1t}[l][l]{$\beta'_*(t)$}
\psfrag{c1t}[l][l]{$\gamma'_*(t)$}
\psfrag{text}[l][r]{$\quad\quad\alpha_2(x)<\alpha_2(y) \text{ iff } x<y$}
\includegraphics[keepaspectratio, width=0.94\textwidth]{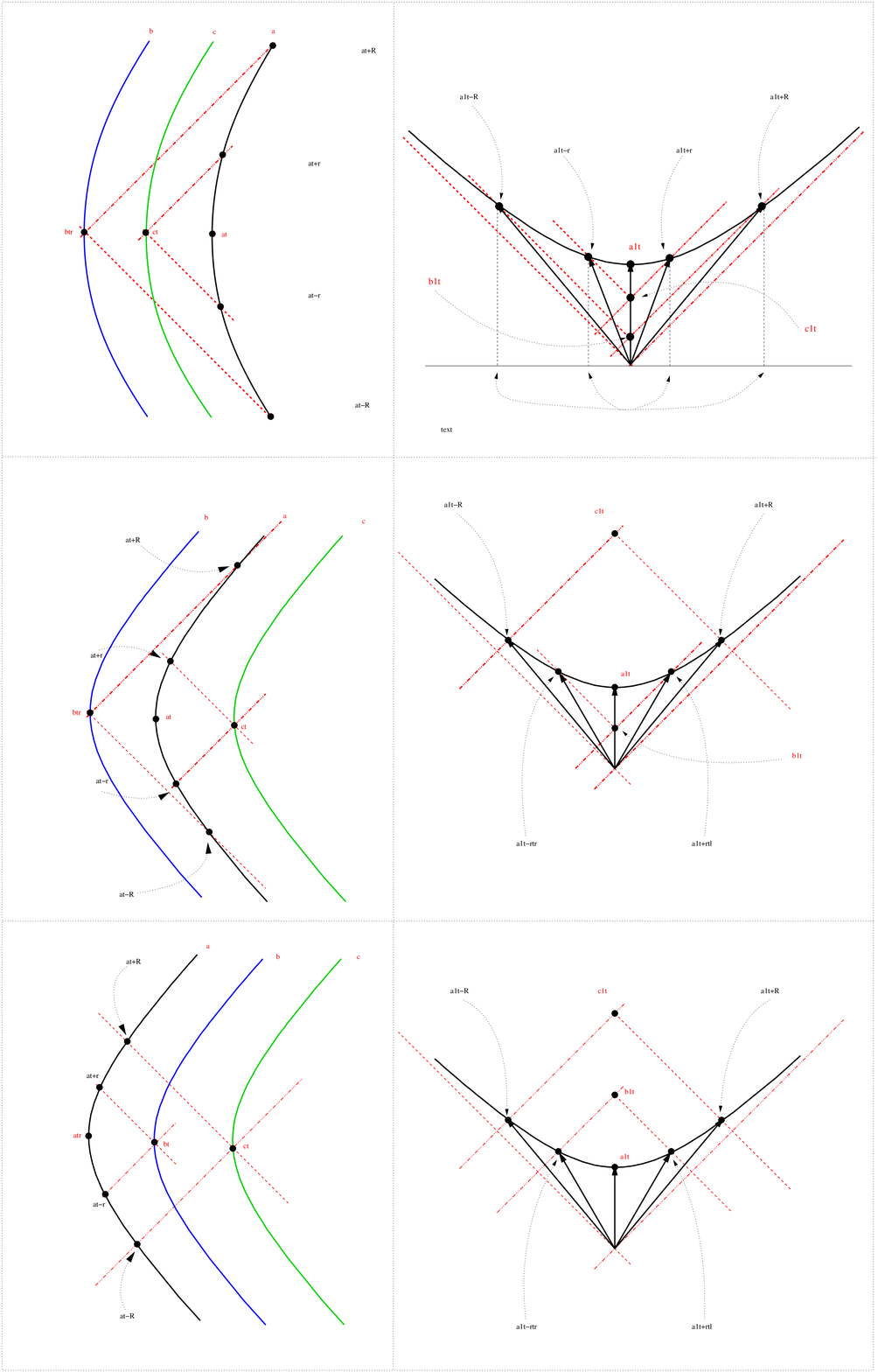}
\caption{\label{fig-radthm} Illustration for the proof of Item (1) in
  Thm.~\ref{thm-rad} verifying requirement (iii) in Lem.~\ref{lem-main}}
\end{center}
\end{figure}

\begin{proof}
To prove Item (1), let $\rship$ be a radar spaceship such that $k$ is
positively accelerated and the direction of the spaceship is the same
as that of the acceleration of $k$.  Let $e_b$, $\bar{e}_b$, $e_c$,
$\bar{e}_c$ be such events that $b\in e_b\cap\bar{e}_b$, $c\in e_c\cap
\bar{e}_c$ and $e_b\simrad_k e_c$, $\bar{e}_b\simrad_k\bar{e}_c$.  To
prove Item (1), we have to prove that
$\time_b(e_b,\bar{e}_b)<\time_c(e_c,\bar{e}_c)$.  Since $\rship$ is a
spaceship, there is an {\it inertial} observer $m\in\IOb$ such that
$\wl_m(b)\cup \wl_m(k)\cup \wl_m(c)$ is a subset of a vertical plane.
Let $m$ be such an {\it inertial} observer.  Without losing
generality, we can assume that this plane is the $\txPlane$.  We are
going to apply Lem.~\ref{lem-main}.  To do so, let $\beta=\lc^b_m$,
$\gamma=\lc^c_m$ and $\alpha=\lc^k_m$; and let $\beta_*$ and
$\gamma_*$ be the radar reparametrization of $\beta$ and $\gamma$
according to $\alpha$, respectively.  By Thm.~\ref{thm-wp}, $\beta$ and
$\gamma$ are definable and well-parametrized timelike curves.  By
Lems.\ \ref{lem-vmon} and \ref{lem-accdir}, we can assume that
$\alpha'_2$ is increasing and $\alpha'\upp\vet$.  By
Prop.~\ref{prop-ph}, $\beta_*$ and $\gamma_*$ are definable timelike
curves since the photon sum of any two timelike vectors of
$\ran\alpha'$ is also a timelike one.  Requirement (i) in
Lem.~\ref{lem-main} is clear by the definition of the radar
reparametrization.  It is also clear that there are
$x_\beta,y_\beta\in\dom\beta$, $x_\gamma,y_\gamma\in \dom\gamma$ and
$x,y\in\dom\beta_*\cap\dom\gamma_*$ such that
$\beta(x_\beta)=\loc_m(e_b)=\beta_*(x)$,
$\beta(y_\beta)=\loc_m(\bar{e}_b)=\beta_*(y)$ and
$\gamma(x_\gamma)=\loc_m(e_c)=\gamma_*(x)$,
$\gamma(y_\gamma)=\loc_m(\bar{e}_c)=\gamma_*(y)$.  Hence requirement
(ii) in Lem.~\ref{lem-main} also holds.  Since the direction of
$\rship$ is the same as that of the acceleration of $k$, there are
only three possible orders of the observers in the spaceship.  All
these three cases are illustrated by Fig.~\ref{fig-radthm}.  By
Prop.~\ref{prop-rad}, it is easy to see that
$\mu\big(\beta'_*(t)\big)<\mu\big(\gamma'_*(t)\big)$ for all $t\in
(x,y)$; and that is requirement (iii) in Lem.~\ref{lem-main}.  Hence by
Lem.~\ref{lem-main}, $|x_\beta-y_\beta|<|x_\gamma-y_\gamma|$.  Thus
$\time_b(e_b,\bar{e}_b)<\time_c(e_c,\bar{e}_c)$ since by
Lem.~\ref{lem-time}, $\time_i(e_i,\bar{e}_i)=|x_i-y_i|$ for all
$i\in\setopen b,c\setclose$; and that is what we wanted to prove.

To prove Item (2), there are many cases we should consider resulting
from which order is taken by the observers in the spaceship, and which
observer is watching the other two.  The proof in all the cases is
based on the very same ideas and lemmas as the proof of Item (1).  The
only difference is that we should use photon simultaneity and photon
reparametrization instead of radar ones, and we should use
Prop.~\ref{prop-ph} (and Lem.~\ref{lem-vmon}) when verifying requirement
(iii) in Lem.~\ref{lem-main}.  In Fig.~\ref{fig-radthmph}, we illustrate
the proof of requirement (iii) in Lem.~\ref{lem-main} in one of the
many cases.  In the other cases, this part of the proof can also be
attained by means of similar figures without any extra difficulty.
\end{proof}

\begin{figure}[h!btp]
\small
\begin{center} 
\psfrag{a}[bl][b]{$\alpha$}
\psfrag{b}[bl][bl]{$\beta$}
\psfrag{c}[bl][bl]{$\gamma$}
\psfrag{at}[tr][tr]{$\alpha(t)$}
\psfrag{at-2R}[tl][tl]{$\alpha(t-2R)$}
\psfrag{at-2r}[tl][tl]{$\alpha(t-2r)$}
\psfrag{a1t-2R}[tr][tl]{$\alpha'(t-2R)$}
\psfrag{a1t-2r}[tr][tl]{$\alpha'(t-2r)$}
\psfrag{b*t}[tr][tr]{$\beta_*(t)$}
\psfrag{c*t}[tl][tl]{$\gamma_*(t)$}
\psfrag{b1*t}[tr][tl]{$\beta'_*(t)$}
\psfrag{c1*t}[tl][tl]{$\gamma'_*(t)$}
\psfrag{a1t}[tl][tl]{$\alpha'(t)$}
\psfrag{b*1t}[tl][tl]{$\beta'_*(t)$}
\psfrag{c*1t}[tl][tl]{$\gamma'_*(t)$}
\includegraphics[keepaspectratio, width=\textwidth]{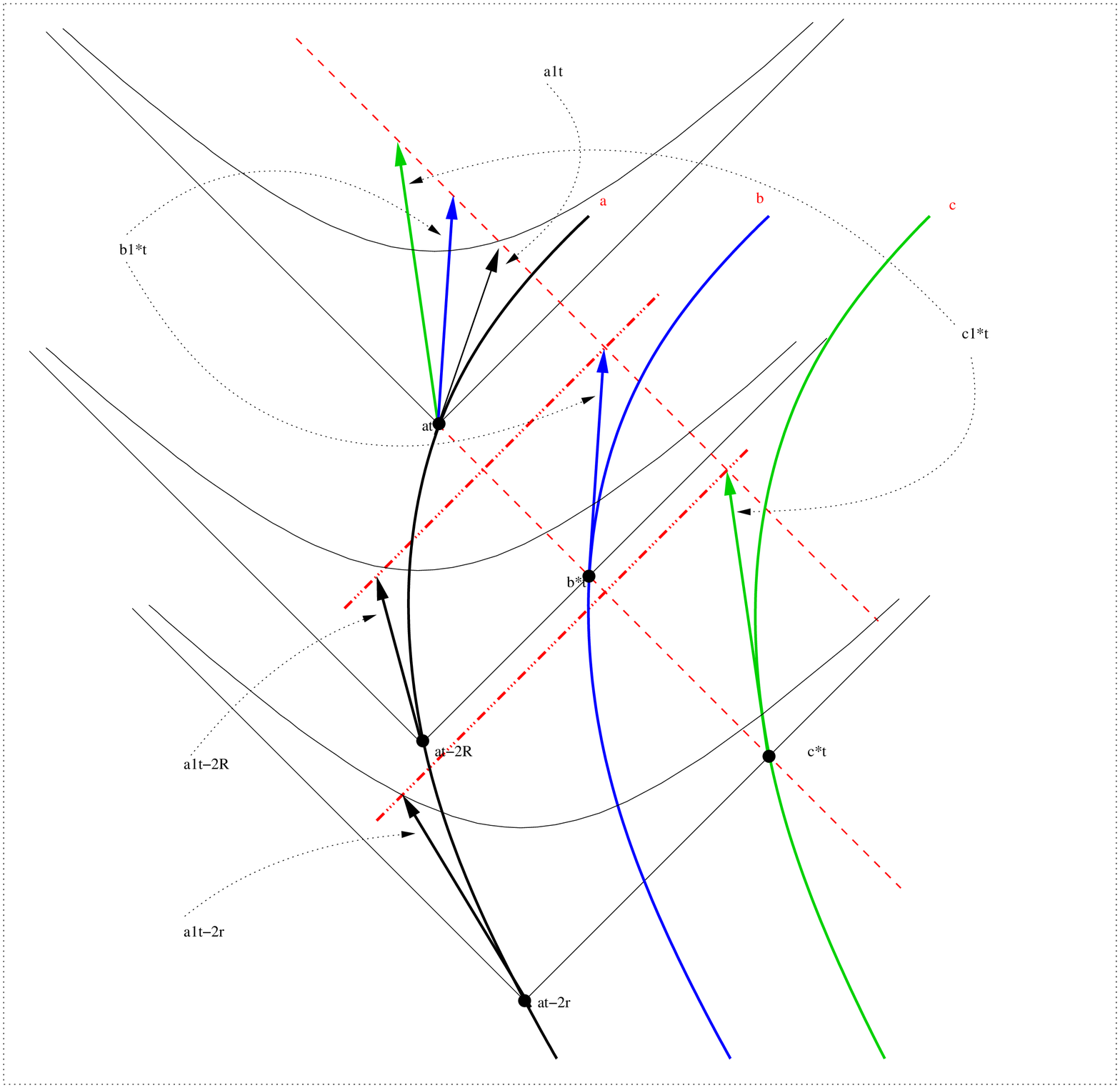}
\caption{\label{fig-radthmph} Illustration for the proof of Item (2)
  in Thm.~\ref{thm-rad} verifying requirement (iii) in
  Lem.~\ref{lem-main}}
\end{center}
\end{figure}

To prove a similar theorem for Minkowski spaceships, we need the
following concept.
We say that observer \df{$b$ is not too far behind}\index{not too far behind} 
the positively accelerated observer $k$ iff 
the following holds:
\begin{multline*}
\forall m \in \IOb \enskip \forall t \in \dom \fvakm\enskip \forall
\vpp,\vqq \in Cd_m \quad k\in ev_m(\vpp)\lland b\in ev_m(\vqq)\\\lland
ev_m(\vpp)\simmu_k ev_m(\vqq) \lland \fvakm(t)\upp (\vpp-\vqq) \then
\forall \tau \in \dom\fvakm \quad
\mu(\vpp,\vqq)<\frac{-1}{a_k(\tau)}.
\end{multline*}

Now we can state and prove our theorem about the clock-slowing effect of
gravitation in Minkowski spaceships.
\begin{thm}\label{thm-mu}
Let $d\ge 3$.
Assume \ax{AccRel}.
Let $\mship$ be a Minkowski spaceship such that:
\begin{itemize}
\item[(i)] observer $k$ is positively accelerated,
\item[(ii)] the direction of the spaceship is the same as that of the acceleration of observer $k$,
\item[(iii)] observer $b$ is not too far behind $k$.
\end{itemize}
Then both (1) and (2) hold:
\begin{enumerate}
\item The clock of $b$ runs slower than the clock of $c$ as
seen by $k$ by Minkowski simultaneity.
\item The clock of $b$ runs slower than the clock of $c$ as seen by each
of $k$, $b$ and $c$ by photons.
\end{enumerate}
\end{thm}

\begin{figure}[h!btp]
\small
\begin{center} 
\psfrag{a}[bl][bl]{$\alpha$}
\psfrag{b}[bl][bl]{$\beta$}
\psfrag{c}[bl][bl]{$\gamma$}
\psfrag{aa}[tl][tl]{(a)}
\psfrag{bb}[tl][tl]{(b)}
\psfrag{cc}[tl][tl]{(c)}
\psfrag{dd}[tl][tl]{(d)}
\psfrag{a1}[b][b]{$\alpha'(t)$}
\psfrag{a1t}[tl][tl]{$\alpha'(t)$}
\psfrag{at}[tl][tl]{$\alpha(t)$}
\psfrag{a11}[tl][tl]{$\alpha''(t)$}
\psfrag{b*t}[tr][tl]{$\beta_*(t)$}
\psfrag{c*t}[bl][bl]{$\gamma_*(t)$}
\psfrag{b1*t}[tr][tl]{$\beta'_*(t)$}
\psfrag{c1*t}[tl][tl]{$\gamma'_*(t)$}
\psfrag{b*1}[bl][bl]{$\beta'_*(t)$}
\psfrag{b*1r}[br][br]{$\beta'_*(t)$}
\psfrag{c*1}[bl][bl]{$\gamma'_*(t)$}
\psfrag{r}[bl][bl]{$r$}
\psfrag{R}[tl][tl]{$R$}
\psfrag{dc}[bl][bl]{$r\cdot\bar\alpha''(t)$}
\psfrag{db}[bl][bl]{$R\cdot\bar\alpha''(t)$}
\includegraphics[keepaspectratio, width=\textwidth]{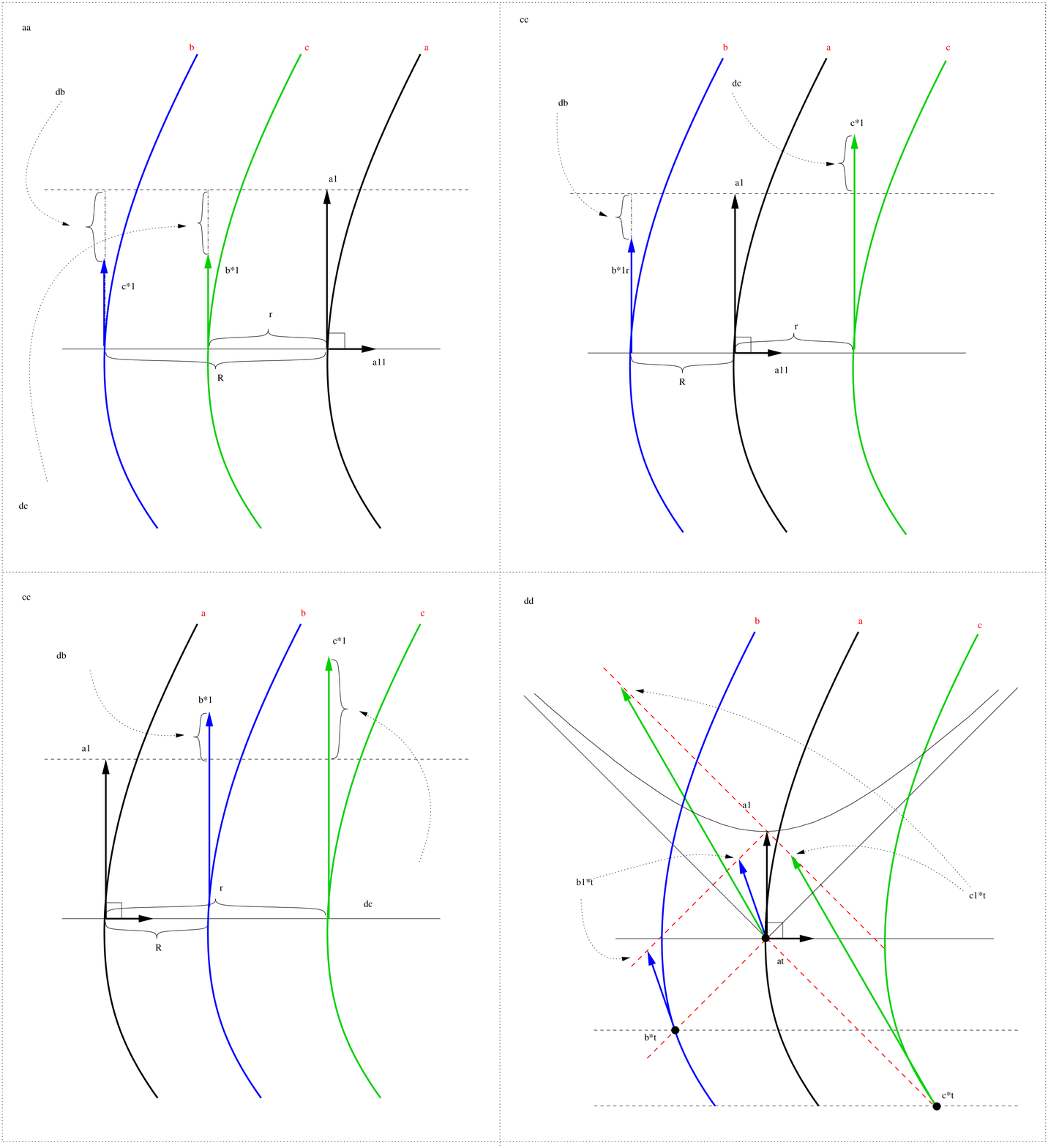}
\caption{\label{fig-minkthm} Illustration for the proof of
  Thm.~\ref{thm-mu} verifying requirement (iii) in Lem.~\ref{lem-main}}
\end{center}
\end{figure}

\begin{proof}
The proof of this theorem is based on the very same ideas and lemmas
as the proof of Thm.~\ref{thm-rad}. The only difference is that we
should use Minkowski simultaneity and Minkowski
re\-pa\-ram\-e\-tri\-za\-tion instead of radar ones, and in the proof
of Item (1) we should use Prop.~\ref{prop-mink} instead of
Prop.~\ref{prop-rad} when verifying requirement (iii) in
Lem.~\ref{lem-main}. In the proof of Item (1) of this theorem, we face
the same three cases as in the proof of Item (1) in
Thm.~\ref{thm-rad}. By (a), (b) and (c) of Fig.~\ref{fig-minkthm}, we
illustrate the proof of requirement (iii) in Lem.~\ref{lem-main} in
this three cases. Similarly, in the proof of Item (2) of this theorem,
we face the same large number of cases as in the proof of Item (2) in
Thm.~\ref{thm-rad}. By (d) of Fig.~\ref{fig-minkthm}, we illustrate the
proof of requirement (iii) in Lem.~\ref{lem-main} in one of these many
cases.  We do not go into more details here since the rest of the
proof can be put together with the help of the hints above.
\end{proof}

We have seen that gravitation (acceleration) makes ``time flow
slowly.''  However, we left the question open which feature of
gravitation (its ``magnitude'' or its ``direction'') plays a role in
this effect.  The following theorem shows that two observers, say $b$
and $c$, can feel the same gravitation while the clock of $b$ runs
slower than the clock of $c$.  Thus it is not the ``magnitude'' of the
gravitation that makes ``time flow slowly.''

\begin{thm} \label{thm-ob}
Let $d\ge 3$.
Then there is a model of \ax{AccRel}, and there
are observers $b$ and  $c$ in this model such that
$a_b(t)=a_c(t)=1$ for all $t\in\Q$, but the clock of $b$ runs slower
than the clock of $c$ as seen by $b$ by photons (or
by radar or by Minkowski simultaneity).
\end{thm}

\begin{proof}
To prove the theorem, let $\Q$ be the field of real numbers and let
\begin{equation*}
\beta(t)\leteq\big(sh(t),ch(t),0,\ldots,0\big)\quad\text{ and
}\quad\gamma(t)\leteq\big(sh(t),ch(t)+1,0\ldots,0\big)
\end{equation*}
where $sh$ and $ch$ are the hyperbolic sine and cosine functions.
Since both $\beta$ and $\gamma$ are smooth and well-parametrized
timelike curves, we can easily build a model of \ax{AccRel} such that
$\lc^b_m=\beta$ and $\lc^c_m=\alpha$ for some $m\in\IOb$.  By a
straightforward calculation, we can show that
$\mu\big(\beta''(t)\big)=\mu\big(\gamma''(t)\big)=-1$ for all
$t\in\Q$.  Hence $a_b(t)=a_c(t)=1$ for all $t\in\Q$.

\begin{figure}[h!btp]
\small
\begin{center}
\psfrag{be}[l][tl]{$\beta$}
\psfrag{ga}[l][tl]{$\gamma$}
\psfrag{b}[l][tl]{$b$}
\psfrag{c}[l][tl]{$c$}
\psfrag{o}[r][r]{$\vo$}
\psfrag{p}[tr][tr]{$\vp$}
\psfrag{q}[tr][tr]{$\vq$}
\psfrag{p1}[br][br]{$\vpp'$}
\psfrag{q1}[br][br]{$\vqq'$}
\psfrag{eq}[tl][tl]{$=$}
\psfrag{e}[r][r]{$e$}
\psfrag{ph}[l][bl]{$ph$}
\psfrag{eb}[tl][tl]{$e_b$}
\psfrag{ec}[tl][tl]{$e_c$}
\includegraphics[keepaspectratio, width=0.7\textwidth]{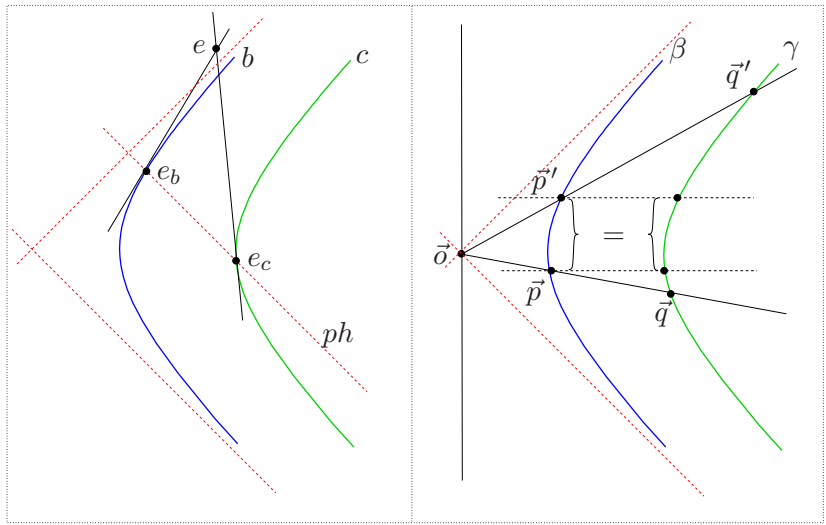}
\caption{\label{fig-unifacc} Illustration for the proof of
  Thm.~\ref{thm-ob}}
\end{center}
\end{figure}

It is easy to show that $c$ is approaching $b$ as seen by $b$ by
photons, see (a) of Fig.~\ref{fig-unifacc}.  Thus by
Lem.~\ref{lem-ph}, the clock of $b$ runs slower than the clock of $c$
as seen by $b$ by photons.  It is not difficult to show that
$ev_m(\vpp)\simrad_b ev_m(\vqq)$ iff $ev_m(\vpp)\simmu_b ev_m(\vqq)$
iff $\vo\in line(\vpp,\vqq)$.  Thus the clock of $b$ runs slower than
the clock of $c$ as seen by $b$ by both radar simultaneity and
Minkowski simultaneity, see (b) of Fig.~\ref{fig-unifacc}.
\end{proof}

Let us now prove some lemmas that were used in the proofs above. 
First let us introduce two concepts which are strongly
connected to the flow of time as seen by photons, see
Lem.~\ref{lem-ph}. We say that observer $c$ is
\df{approaching}\index{approaching} (or \df{moving away}\index{moving
  away} from) observer $b$ as seen by $b$ by photons at event $e_b$
iff the following hold
\begin{itemize}
\item $b\in e_b$,
\item for all events $e_c$ for which $c\in e_c$ and $e_b\simph_b e_c$
 hold, there is an event $e$ such that $b',c'\in e$ for every co-moving
 {\it inertial} observers $b'$ and $c'$ of $b$ at event $e_b$ and of
 $c$ at event $e_c$, respectively, and
\item $e_b$ precedes (succeeds) $e$ according to $b$,
\end{itemize}
see (b) of Fig.~\ref{figupp}. We say that $c$ is approaching (moving
away from) $b$ as seen by $b$ by photons iff it is so for every event
$e_b$ for which $b\in e_b$. The idea behind these definitions is the
following: two observers are considered approaching when they would
meet if they stopped accelerating at simultaneous events.

\begin{rem} 
Let us note that coplanar {\it inertial} observers seen by photons are
approaching each other before the event of meeting and moving away
from each other after it. This fact explains the words used for these
concepts.
\end{rem}
 
\begin{rem}
There is no direct connection between the two concepts above. For
example, it is not difficult to construct a model of \ax{AccRel} in
which there are (uniformly accelerated) observers $b$ and $c$ such
that $c$ is approaching $b$ seen by $b$ by photons while $b$ is moving
away from $c$ seen by $c$ by photons, see the proof of
Thm.~\ref{thm-ob}.
\end{rem}

\noindent
Lem.~\ref{lem-ph} can be interpreted as a refined version of the Doppler effect.

\begin{lem} \label{lem-ph}
Let $d\ge 3$.
Assume \ax{AccRel}.
Let $b$ and $c$ be coplanar observers.
Then
\begin{itemize}
\item[(1)] If $c$ is approaching $b$ as seen by $b$ by photons, the clock of $b$ runs slower
than the clock of $c$ as seen by $b$ by photons.
\item[(2)] If $c$ is moving away from $b$ as seen by $b$ by photons, the clock of $c$ runs
slower than the clock of $b$ as seen by $b$ by photons.
\end{itemize}
\end{lem}

\begin{figure}[h!btp]
\small
\begin{center}
\psfrag{g}[b][b]{$\gamma=\lc^c_m$}
\psfrag{b}[l][l]{$\beta=\beta_*=\lc^b_m$}
\psfrag{b1}[bl][bl]{$\beta'(t)$}
\psfrag{bt}[l][l]{$\beta(t)$}
\psfrag{g1}[tl][tl]{$\gamma'(\bar{t}\,)$}
\psfrag{g1*}[tl][tl]{$\gamma_*'(t)$}
\psfrag{gt}[tl][tl]{$\gamma(\bar{t}\,)$}
\includegraphics[keepaspectratio, width=0.8\textwidth]{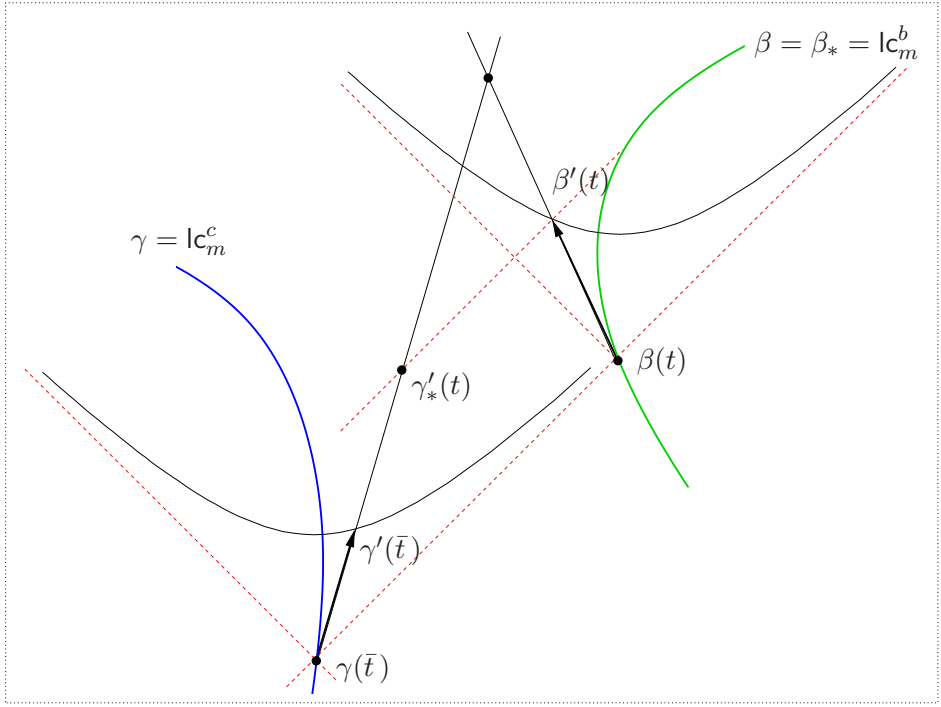}
\caption{\label{fig-thmph} Illustration for the proof of
  Lem.~\ref{lem-ph}}
\end{center}
\end{figure}
 
\begin{proof}
To prove Item (1), let $b$ and $c$ be coplanar observers, and let
$e_b$, $\bar{e}_b$, $e_c$ and $\bar{e}_c$ be such events that $b\in
e_b\cap\bar{e}_b$, $c\in e_c\cap \bar{e}_c$ and $e_b\simph_b e_c$,
$\bar{e}_b\simph_b\bar{e}_c$. Let us suppose that $c$ is approaching
$b$ as seen by $b$ by photons. We have to prove that
$\time_b(e_b,\bar{e}_b)<\time_c(e_c,\bar{e}_c)$. Since $c$ and $b$ are
coplanar, there is an {\it inertial} observer $m\in\IOb$ such that
$\wl_m(c)\cup \wl_m(b)$ is a subset of a vertical plane. Let $m$ be
such an {\it inertial} observer. We are going to apply
Lem.~\ref{lem-main}. To do so, let $\beta=\beta_*=\lc^b_m$,
$\gamma=\lc^c_m$, and let $\gamma_*$ be the photon reparametrization
of $\gamma$ according to $\beta$. By Thm.~\ref{thm-wp}, $\beta=\beta_*$
and $\gamma$ are definable and well-parametrized timelike
curves. Without losing generality, we can assume that $\beta'\upp\vet$
and $\gamma'\upp\vet$. It is easy to see that $\wl_m(b)\cap
\wl_m(c)=\emptyset$ since $c$ is approaching $b$ as seen by $b$. Thus
$\ran\beta\cap\ran\gamma=\emptyset$ since $\ran\beta=\wl_m(b)$ and
$\ran\gamma=\wl_m(c)$ by Item \eqref{item-rantr} in
Prop.~\ref{prop-lc}. Thus $\gamma_*$ is also a definable timelike curve
by Prop.~\ref{prop-ph}.  Requirement (i) in Lem.~\ref{lem-main} is clear
by the definition of the photon reparametrization. It is also clear
that there are $x_\beta,y_\beta\in\dom\beta$, $x_\gamma,y_\gamma\in
\dom\gamma$ and $x,y\in\dom\beta_*\cap\dom\gamma_*$ such that
$\beta(x_\beta)=\loc_m(e_b)=\beta_*(x)$,
$\beta(y_\beta)=\loc_m(\bar{e}_b)=\beta_*(y)$ and
$\gamma(x_\gamma)=\loc_m(e_c)=\gamma_*(x)$,
$\gamma(y_\gamma)=\loc_m(\bar{e}_c)=\gamma_*(y)$. Hence requirement
(ii) in Lem.~\ref{lem-main} also holds. Since $c$ is approaching $b$ as
seen by $b$ by photons, the tangent lines of $\beta_*$ and $\gamma_*$
at any $t\in(x,y)$ intersect in the future of $\beta_*(t)$ and
$\gamma_*(t)$. Thus
$\mu\big(\beta'_*(t)\big)=1<\mu\big(\gamma'_*(t)\big)$ for all $t\in
(x,y)$ by Prop.~\ref{prop-ph}, see Fig.~\ref{fig-thmph}; and that is
requirement (iii) in Lem.~\ref{lem-main}. Hence by Lem.~\ref{lem-main},
we have that $|x_\beta-y_\beta|<|x_\gamma-y_\gamma|$.  Consequently,
$\time_b(e_b,\bar{e}_b)<\time_c(e_c,\bar{e}_c)$ since by
Lem.~\ref{lem-time}, $\time_i(e_i,\bar{e}_i)=|x_i-y_i|$ for all
$i\in\setopen b,c\setclose$. So Item (1) is proved.

The proof of (2) is similar.
Hence we omit it.
\end{proof}

Lem.~\ref{lem-time} states that the time measured according to the
parametrization of the life-curve $\lc^k_m$ between two parameter points
and the time measured by observer $k$ between the corresponding events is the
same if \ax{AxSelf_0} and  \ax{AxPh}  are assumed and  $m\in\IOb$.

\begin{lem}
\label{lem-time}
Assume \ax{AxSelf_0}, \ax{AxPh}, and let $m\in\IOb$.
Let $k\in \Ob$.
Let $x,y\in\dom \lc^k_m$.
Then
\begin{equation}\label{eq-time}
\time_k\big(ev_m\big(\lc^k_m(x)\big),ev_m\big(\lc^k_m(y)\big)\big)=|x-y|.
\end{equation}
\end{lem}

\begin{proof}
By \eqref{item-trfunct} in Prop.~\ref{prop-lc}, $\lc^k_m$ is a function.
Thus $\lc^k_m(x)$ and $\lc^k_m(y)$ are meaningful.
We have that $k\in ev_m\big(\lc^k_m(x)\big)\bigcap ev_m\big(\lc^k_m(y)\big)$ by the definition of $\lc^k_m$.
Thus by \ax{AxSelf_0}, both events $ev_m\big(\lc^k_m(x)\big)$ and $ev_m\big(\lc^k_m(y)\big)$ have unique coordinates in $Cd_k$.
Thus the left hand side of equation \eqref{eq-time} is defined and equal to 
\begin{equation*}
\Big|\loc_k\big(ev_m\big(\lc^k_m(x)\big)\big)_\tau-\loc_k\big(ev_m\big(\lc^k_m(y)\big)\big)_\tau \Big|
\end{equation*} 
by definition.  However, by the definition of $\lc^k_m$,
\begin{equation*}\loc_k\big(ev_m\big(\lc^k_m(x)\big)\big)_\tau=x \quad\text{ and }\quad
\loc_k\big(ev_m\big(\lc^k_m(y)\big)\big)_\tau=y.
\end{equation*}
Hence equation \eqref{eq-time} holds.
\end{proof}

None of the axioms introduced so far require the existence of
accelerated observers. Our following axiom schema says that every
definable timelike curve is the world-line of an observer. Since there
are many timelike curves that are not lines, that will ensure the
existence of many accelerated observers since from \ax{AxSelf_0},
\ax{AxPh} and \ax{AxEv} it follows that the world-lines of {\it
  inertial} observers are lines, see, e.g., Thm.~\ref{thm-poi}.

We say that a \df{function $f$ is}
(parametrically) \df{definable by} $\psi(x,\vy, \vz\,)$ iff there
is an $\va\in U^n$ such that $f(b)=\vp \iff \psi(b,\vp,\va\,)$ is true in
$\mathfrak{M}$.\index{definable function}
Let $\psi$ be a FOL formula of our language.
\begin{description}
\item[\Ax{Ax\exists Ob_\psi}]\index{\ax{Ax\exists Ob_\psi}} If a
  function that is parametrically definable by $\psi$ is a timelike curve,
  then there is an observer whose world-line is the range of this
  function:
\end{description}
\begin{equation*}\index{\ax{COMPR}}\label{compr}
\Ax{COMPR}\leteq \Setopen \ax{Ax\exists Ob_\psi}\setmid \psi \text{
is a FOL formula of our language}\Setclose.
\end{equation*}
A precise formulation of \ax{COMPR} can be obtained from
that of its analogue in \cite{logst}.

The following three theorems say that the clocks can run arbitrarily
slow or fast, as seen by the three different methods.

\begin{thm} \label{thm-Rad}
Let $d\ge 3$.  Assume \ax{AccRel} and \ax{COMPR}.  Let $m$ be a
positively accelerated observer such that $\dom \lc^m_m=\Q$ and let
$e$ and $e'$ be two events such that $e\neq e'$ and $m\in e\cap e'$.
Then for all $\lambda \in \Q^+$, there are an 
observer $b$ and events $e_b$ and $e'_b$ such that $b\in e_b\cap
e'_b$, $e\simrad_m e_b$, $e'\simrad_m e'_b$ and
$\time_b(e_b,e_b')=\lambda\cdot \time_m(e,e')$.
\end{thm}

\begin{thm}\label{thm-Mu}
Let $d\ge 3$.
Assume \ax{AccRel} and \ax{COMPR}.
 Let $m$ be a
uniformly accelerated observer and let $e$ and $e'$ be two events  such that $e\neq e'$ and
$m\in e\cap e'$.
Then for all $\lambda \in \Q^+$, there are an
observer $b$ and events $e_b$ and $e'_b$ such that $b\in e_b\cap
e'_b$, $e\simmu_m e_b$, $e'\simmu_m e'_b$ and
$\time_b(e_b,e_b')=\lambda\cdot \time_m(e,e')$.
\end{thm}

\begin{thm} \label{thm-Ph}
Let $d\ge 3$.
Assume \ax{AccRel} and \ax{COMPR}.
 Let
$m$ be a positively accelerated observer and  let $e$ and  $e'$ be two events such that
$e\neq e'$ and $m\in e\cap e'$.
Then for all $\lambda \in \Q^+$,
there are an observer $b$ and events $e_b$ and $e'_b$ such that
$b\in e_b\cap e'_b$, $e\simph_m e_b$, $e'\simph_m e'_b$ and
$\time_b(e_b,e_b')=\lambda \cdot\time_m(e,e')$.
\end{thm}

\section{Concluding remarks on gravitational time dilation}

We have proved several qualitative versions of gravitational time
dilation from axiom system \ax{AccRel} by the use of Einstein's
equivalence principle. It is important to note that the axioms of
\ax{AccRel} and Einstein's equivalence principle have different
statuses. Einstein's equivalence principle is not an axiom, it
is just a guiding principle.

The theorems of this chapter can be interpreted as saying that
observers will experience time dilation in the direction of
gravitation by the corresponding measuring methods (photon, radar,
Minkowski) if all the axioms of \ax{AccRel} are ``true in our world''
and Einstein's equivalence principle is a ``good'' principle.

Since gravitation can be defined by the acceleration of dropped
\textit{inertial} bodies, Einstein's equivalence principle can be
formulated within \ax{AccRel}. It raises the possibility of checking
within \ax{AccRel} how good a principle Einstein's equivalence
principle is. That is, we can investigate for what kind of
accelerated observers the Einstein's equivalence principle can be
proved within \ax{AccRel}. For a detailed investigation on this
subject, see \cite{eep}.

\begin{rem}
By Thm.~\ref{thm-wp} and Prop.~\ref{prop-hyp}, it is not difficult to
prove that the quantity part of a model of \ax{AccRel_0} cannot be the
field of real algebraic numbers if we assume that there are uniformly
accelerated observers.
\end{rem}

\begin{rem}
By Prop.~\ref{prop-rc}, the quantity part of a model of \ax{AccRel} has
to be a real-closed field. 
\end{rem}
\noindent
These remarks generate the following three questions, each of which is unanswered yet:
\begin{que}\label{que-unif}
What can be the quantity part of a model of \ax{AccRel_0} if we also assume
that there are uniformly accelerated observers?
\end{que}

\begin{que}
What can be the quantity part of a model of \ax{AccRel_0}+\ax{COMPR}?
\end{que}

\begin{que}
What can be the quantity part of a model of \ax{AccRel}+\ax{COMPR}?
\end{que}

\chapter{A FOL axiomatization of General relativity}
\label{chp-gr}

In this chapter we extend our investigations to general relativity by
deriving a its FOL axiomatization from our theory
\ax{AccRel}, see also \cite{amnsz-wku}. The axioms of general
relativity are going to be slightly modified versions of the four
axioms of special relativity together with one more assumption which
 is a refinement of the co-moving axiom of \ax{AccRel}. We are
also going to give the connections between models of our axioms and
spacetimes that we meet in the literature on general relativity.

We slightly refine the axioms of \ax{SpecRel} and the strong
co-moving axiom of accelerated observers \ax{AxSCmv} (see
p.\pageref{axscmv}) and get an axiomatic theory of general
relativity. To do so, we ``eliminate the privileged class of inertial
reference frames'' which was Einstein's original recipe for obtaining
general relativity from special relativity, see \cite{friedman}. So
below we realize Einstein's original
program  formally and literally. We modify the axioms one by one using the following two
guidelines:
\begin{itemize}
\item let the new axioms not speak about {\it inertial} observers, and
\item let the new axioms be consequences of the old ones and our theory \ax{AccRel_0}.
\end{itemize}

To get the modified version of \ax{AxSelf}, let us note that
\ax{AxSelf_0} (see p.\pageref{axself0}) and \ax{AxSelf^+_0} (see
p.\pageref{axself+}) satisfy the requirements above. So let
\Ax{AxSelf^-}\index{\ax{AxSelf^-}} be $\ax{AxSelf_0}\land
\ax{AxSelf^+_0}$. The localized version of \ax{AxEv} contains the
following two statements: (1) every observer encounters the events in
which it is observed, and (2) if observer $k$ coordinatizes event $e$
which is also coordinatized by observer $m$, then $k$ also
coordinatizes the events which are near $e$ according to $m$. The
first statement is already formulated in \ax{AxEvTr} (see
p.\pageref{axevtr}), and the second one can be formulated by saying
that $\dom w^k_m$ is open for any observers $k$ and $m$. 

\begin{description}
\item[\Ax{AxEv^-}]\index{\ax{AxEv^-}} Every observer encounters the
  events in which it is observed; and the domains of worldview
  transformations are open, i.e.,
\begin{equation*}
 \ax{AxEvTr}\lland \forall m,k\in\Ob \quad \dom w^k_m \text{ is open}.
\end{equation*}
\end{description}
The localized version of \ax{AxPh} is  the following:
\begin{description}
\item[\Ax{AxPh^-}]\index{\ax{AxPh^-}} The instantaneous velocity of
  photons is $1$ in the moment when they are sent out according to the
  observer sending them out, and any observer can send out a photon
  in any direction with this instantaneous velocity:
\begin{multline*}
\forall k\in\Ob\enskip \forall\vp\in\Q^d\enskip  k\in ev_k(\vpp)
\then \big( \forall ph\in\Ph
 \quad ph\in ev_k(\vpp) \then
 \vv^{\,ph}_k(\vpp)=1\big)\\\lland \big(\forall \vv\in\Q^{d-1}\quad |\vv\,|=1 \then
 \exists ph\in \Ph\quad ph\in ev_k(\vpp)\lland
 \vv^{\,ph}_k(\vpp)=\vv\,\big), 
\end{multline*}
\end{description}
where $\vv^{\,b}_k(\vpp)$ is the instantaneous velocity of body $b$
according to observer $k$ at $\vp$.

Our symmetry axiom \ax{AxSymDist} has many equivalent versions with
respect to \ax{SpecRel_0}, see \cite[\S 2.8, \S 3.9, \S
  4.2]{pezsgo}. We can localize any of these versions and use it in
our FOL axiom system of general relativity. For aesthetic reasons we
use \ax{AxSymTime}, the version stating that ``{\it inertial}
observers see each others' clocks behaving in the same way,'' see
Thm.~\ref{thm-tdeq} at p.\pageref{thm-tdeq} and \cite{logst}.

\begin{description}
\item[\Ax{AxSymTime^-}]\index{\ax{AxSymTime^-}} Any two observers
  meeting see each others' clocks behaving in the same way at the
  event of meeting:
\begin{multline*}
\forall k,m\in\Ob\enskip\forall t_1,t_2 \in\Q \\\quad k,m\in ev_k(\langle t_1,\voo \rangle)\cap  ev_m\left(\langle t_2,\voo \rangle\right) \then \left| \fvv^{\,m}_k(t_1)_\tau\right|=\left| \fvv^{\,k}_m(t_2)_\tau\right|.
\end{multline*}
\end{description}

\noindent 
Now all the four axioms of theory \ax{SpecRel} are modified according to the
requirements above. 

Strictly following the guidelines above,
\ax{AxSCmv^-} would state that the worldview transformations between
observers are differentiable in their meeting-point. Instead, we
introduce a series of axioms, each of which ensures the
smoothness of worldview transformations to some degree.

\begin{description}
\item[\Ax{AxDiff_n}]\index{\ax{AxDiff_n}} The worldview transformations are $n$-times
 differentiable functions, i.e.,
\begin{equation*}
\forall k,m\in\Ob\quad w^k_m \text{ is $n$-times differentiable
 function}.
\end{equation*}
\end{description}
Let us introduce the following axiom systems of general
relativity:
\begin{equation*}\index{\ax{GenRel_n}}
\boxed{\ax{GenRel_n}\leteq \Setopen \ax{AxSelf^-}, \ax{AxPh^-},
 \ax{AxEv^-},\ax{AxSymTime^-},\ax{AxDiff_n} \Setclose\cup\ax{CONT}}
\end{equation*}
Let us note that every model of \ax{GenRel_m} is a model of \ax{GenRel_n} if $m\ge n$.
Let us also introduce a smooth version:
\begin{equation*}\index{\ax{GenRel_\omega}}
\boxed{\ax{GenRel_\omega}\leteq \Setopen \ax{AxSelf^-}, \ax{AxPh^-},
 \ax{AxEv^-},\ax{AxSymTime^-}\Setclose\cup \Setopen \ax{AxDiff_n} :
 n\ge 1 \Setclose\cup  \ax{CONT}}
\end{equation*}

\noindent
For completeness, let us mention here the localized version of
\ax{AxSymDist}, too. The reader may safely skip this axiom.
\begin{description}
\item[\Ax{AxSymDist^-}]\index{\ax{AxSymDist^-}} Observers meeting each
  other agree approximately as to the spatial distance of a
  neighbouring event if this event and the event of meeting are
  simultaneous approximately enough according to both observers:
\begin{multline*}
\forall k,m\in\Ob\;\forall \varepsilon \in\Q^+\; \forall\vp\in
wl_k(k)\cap wl_k(m)\enskip \exists \delta\in\Q^+\;\forall \vq\in
B_\delta(\vpp)\quad \\ |\vq_\tau-\vpp_\tau|<\delta \cdot\left|\vq_\sigma-\vpp_\sigma\right|
\lland \left|w^k_m(\vqq)_\tau-w^k_m(\vpp)_\tau\right|<\delta\cdot \left|
w^k_m(\vqq)_\sigma-w^k_m(\vpp)_\sigma\right|\\ \then \Big|
|\vq_\sigma-\vp_\sigma|-\left|w^k_m(\vqq)_\sigma-w^k_m(\vpp)_\sigma\right|\Big|\le
\varepsilon\cdot|\vp-\vqq|.
\end{multline*}
\end{description}

The definition of Lorentzian manifolds over arbitrary real closed
fields is a natural extension of their standard definition over $\R$.
By the following theorems, which we are going to prove in a forthcoming
paper, the models of \ax{GenRel_n} are exactly the $n$-times
differentiable Lorentzian manifolds over real closed fields; and
the models of \ax{GenRel_\omega} are exactly the smooth Lorentzian
manifolds over real closed fields.
\begin{thm}\label{thm-grn} Let $d\ge3$. Then
\ax{GenRel_n} is complete with respect to $n$-times differentiable
Lorentzian manifolds over real closed fields.
\end{thm}

\begin{thm}\label{thm-grinf} Let $d\ge3$. Then
\ax{GenRel_\omega} is complete with respect to smooth Lorentzian
manifolds over real closed fields.
\end{thm}

The proofs and formal statements of Thms.\ \ref{thm-grn} and
\ref{thm-grinf} are analogous to those of Cor.~\ref{cor-srcompl} at
p.\pageref{cor-srcompl}.  These theorems can be regarded as
completeness theorems in the following sense. Let us consider
Lorentzian manifolds as intended models of \ax{GenRel}. How to do
that? In our forthcoming paper, we will give a method for constructing
 a model of \ax{GenRel} from each Lorentzian manifold; and conversely,
we will also show that each model of \ax{GenRel} is obtained this way
from a Lorentzian manifold.  By the above, we defined what we mean by
a formula $\varphi$ in the language of \ax{GenRel} being valid in a
Lorentzian manifold, or in all Lorentzian manifolds.  Then
completeness means that for any formula $\varphi$ in the language of
\ax{GenRel}, we have $\ax{GenRel_n}\vdash \varphi$ iff $\varphi$ is
valid in all $n$-times differentiable Lorentzian manifolds over real
closed fields.  That is completely analogous to the way how Minkowskian
geometries were regarded as intended models of \ax{SpecRel} in the
completeness theorem of \ax{SpecRel}, see \cite[\S 4]{Mphd} and
\cite[Thm.11.28 p.681]{logst}.

Our theory \ax{GenRel} was obtained from \ax{AccRel} by getting rid of
the concept of inertiality in the level of axioms. However, we can
redefine this concept. We call the world-line of observer $m$
\df{timelike geodesic}\index{timelike geodesic}, if each of its points
has a neighborhood within which this observer measures the most time
between any two encountered event, i.e.,
\begin{multline*}
   \forall\vr \in wl_m(m)\;\exists \delta\in\Q^+ \enskip
  \forall \vp,\vq\in wl_m(m)\cap B_\delta(\vrr)\enskip \forall
  k\in\Ob\cap ev_m(\vpp)\cap ev_m(\vqq)\quad\\ wl_m(k)\subseteq
  B_\delta(\vrr) \then |p_\tau-q_\tau| \geq \left|w^m_k(\vpp)_\tau-w^m_k(\vqq)_\tau\right|.
\end{multline*}
In this case we also say that observer $m$ is an {\it inertial} body.  This
definition is justified by the Twin Paradox theorem of \ax{AccRel},
see Thm.~\ref{thmTwp}. This theorem says that in the models of
\ax{AccRel} the world-lines of {\it inertial} observers are timelike
geodesics in the above sense.

We can define lightlike geodesics in a similar fashion: a lightlike
geodesic $\gamma$ is a lightlike curve with the property that each point in
the curve has a neighborhood in which $\gamma$ is the unique lightlike
curve through any two points of $\gamma$.

The assumption of axiom schema \ax{COMPR} guarantees that our
definition of geodesic coincides with that of the literature on
Lorentzian manifolds. Therefore we also introduce the following theory:
\begin{equation*}
\boxed{\ax{GenRel_n^+}\leteq \ax{GenRel_n}\cup\ax{COPMR}}
\end{equation*}
So in our theory \ax{GenRel^+}, our notion of timelike geodesic
coincides with its standard notion in the literature on general
relativity.  All the other key notions of general relativity, such as
curvature or Riemannian tensor field, are definable from timelike
geodesics. Therefore we can treat all these notions (including the
notion of metric tensor field) in our theory \ax{GenRel^+} in a natural way.

{\it \underline{Connections with our results on \ax{AccRel}}}:
Theorems proved from \ax{AccRel} (our first approximation of
\ax{GenRel}) can also be reformulated and proved from \ax{GenRel},
such as the gravitational time dilation, see
Thms.\ \ref{thm-rad} and \ref{thm-mu}. For lack of space, we postpone
that to a forthcoming paper.

\chapter{The tools necessary for proving the main results}
\label{chp-a}

This chapter is about the development of the tools that were used in
the proofs of the main results of the former chapters. First we have
to build a FOL theory of real analysis. The point is to
formulate and prove theorems of real analysis staying within
FOL. We also seek for using as few assumptions as
possible.

A part of real analysis can be generalized for arbitrary ordered
fields without any real difficulty. However, a certain fragment of real
analysis can only be generalized within FOL for
\textit{definable} functions  and for proofs we
need a version of the \ax{CONT} axiom schema; and there are some
theorems of real analysis which are not provable even by the \ax{CONT}
schema. We refer to the generalizations which cannot be proved without
\ax{CONT}  by marking them ``\ax{CONT}-.'' The
FOL generalizations of some theorems, such as Chain
Rule can be proved without \ax{CONT}, so they are naturally referred
to without the ``\ax{CONT}-'' mark.

Throughout this chapter $\Df{\L}$\index{$\L$} is assumed to be a
FOL language that contains the binary relation symbol $<$
and the unary relation symbol $\Q$, such as our frame language or the
language of the ordered fields. We use notation
$\Df{\L_0}$\index{$\L_0$} for the language $\setopen \Q,<
\setclose$. Let the set of FOL formulas in language $\L$ is
denoted by $\Df{Fm(\L)}$.

In this chapter we also use the following generalized versions of
our field axiom \ax{AxEOF}:
\begin{description}
\item[\Ax{AxOF}]\index{\ax{AxOF}} $\left< \Q;+,\cdot, < \right>$ is
 an ordered field.
\end{description}
\begin{description}
\item[\Ax{AxPOS}]\index{\ax{AxPOS}} $\left< \Q;< \right>$ is a
 partially ordered set, i.e., $\le$ is a reflexive, antisymmetric
 and transitive relation on $\Q$.
\end{description}

\noindent
Naturally, we do not assume \ax{AxEOF} in the theorems of this chapter in which
 \ax{AxOF} or \ax{AxPOS} is used, see Conv.~\ref{conv-frame}.

\section{The axiom schema of continuity}

To prove some of the theorems of real analysis, we need a property of
$\R$. This property is that in $\R$ every bounded nonempty set has a
\df{supremum}, i.e., a least upper bound. It is a second-order
logic property which cannot be used in a FOL axiom system.
Instead, we use an axiom schema stating that every nonempty and
bounded subset of the quantity part that can be defined
parametrically by a FOL formula has a supremum.

This way of imitating a second-order formula by a FOL formula
schema comes from the methodology of approximating second-order
theories by FOL ones. Examples are Tarski's replacement of
Hilbert's second-order geometry axiom by a FOL axiom schema
and Peano's FOL axiom schema of induction replacing the
second-order logic induction.

Let $\{\Q\}\subseteq\L$ be a FOL language, $\mathfrak{M}$ an
\L-model with universe $M$. We say that a subset $H$ of $\Q$ is
\df{(parametrically) \L-definable by}\index{(parametrically)
 \L-definable} $\varphi\in Fm(\L)$ iff there are $a_1,\ldots,a_n\in
U$ such that
\begin{equation*}
H=\Setopen d\in\Q\setmid\mathfrak{M}\models\varphi(d,a_1,\ldots,a_n)
\Setclose. 
\end{equation*}
We say that a subset of $\Q$ is
\df{\L-definable}\index{\L-definable set} iff it is definable by an
\L-formula. More generally, an $n$-ary relation $R\subseteq \Q^n$ is
said to be \df{\L-definable}\index{\L-definable relation} in
$\mathfrak{M}$ by parameters iff there is a formula $\varphi\in
Fm(\L)$ with only free variables $x_1,\ldots,x_n, y_1,\ldots,y_{k}$
and there are $a_{1},\ldots,a_{k}\in U$ such that
\begin{equation*}
R=\Setopen\langle p_1,\ldots, p_n\rangle\in \Q^n : \mathfrak{M}\models \varphi(p_1,\ldots,p_n,a_1,\ldots,a_k)\Setclose.
\end{equation*}

\begin{description}
\item[\ax{AxSup_\varphi}] 
Every subset of $\Q$ definable by $\varphi\in Fm(\L)$ (when using $a_1,\ldots,a_n$ as fixed parameters) has a supremum if it is nonempty and bounded:
\begin{multline*}
\forall y_1,\ldots,y_n\quad[\exists x\in\Q\quad\varphi]\;\land\; [\exists b\in\Q\quad \forall x\in\Q\quad\varphi\Longrightarrow x\le b]\\ \then
\big[\exists s\in\Q\enskip \forall b\in\Q\quad (\forall x\in\Q\quad
\varphi\Longrightarrow x\le b)\ \Longleftrightarrow\ s\le b\big].
\end{multline*}
\end{description}
Our axiom schema \ax{CONT_\L} below says that every nonempty
bounded and \L-definable subset of $\Q$ has a supremum.

\begin{center}\index{\ax{CONT_\L}}
$\ax{CONT_\L}\leteq \Setopen\ax{AxSup_\varphi}:\varphi \mbox{ is a
  FOL formula of the language } \L \Setclose.$
\end{center}

When $\L$ is our frame language, we omit the subscript and write
\ax{CONT} only. When the language is $\L_0$, we write
\Df{\ax{CONT_0}}\index{\ax{CONT_0}}. The language
$\setopen\Q,+,\cdot,<\setclose$ is denoted by
$\Df{\mathcal{OF}}$\index{$\mathcal{OF}$}.

\begin{rem}
$\ax{CONT_\mathcal{L'}}$ is stronger than \ax{CONT_\L} if $\{\Q,<\}\subseteq\L\subseteq\mathcal{L'}$.
\end{rem}

An ordered field $\mathfrak{Q}$ is called
\df{real closed}\index{real closed field} if every positive element
has a square root and every polynomial of odd degree has a root.

\begin{prop}\label{prop-rc}
Let $\mathfrak{Q}$ be an ordered field. Then
\begin{equation*}
\mathfrak{Q}\models\ax{CONT_\mathcal{OF}}\enskip \text{ iff }\enskip \mathfrak{Q} \text{ is real closed.}
\end{equation*}
\end{prop}

\begin{proof} 
Let $\mathfrak{Q}$ be an ordered field such that
$\mathfrak{Q}\models\ax{CONT_\mathcal{OF}}$.  To prove that
$\mathfrak{Q}$ is real closed, let $p(y)$ be the odd degree polynomial
$a_{2n+1}y^{2n+1}+\ldots+a_1y+a_0$. It is enough to prove that $p(y)$
has a root when $a_{2n+1}>0$. Let $H\leteq \{t\in \Q: p(t)<0\}$. It is
clear that $H$ is nonempty, bounded and
$\mathcal{OF}$-definable. From $\ax{CONT_\mathcal{OF}}$, it follows
that $H$ has a supremum, let us call it $s$. Both $\{t: p(t)>0 \}$ and
$\{t : p(t)<0 \}$ are open sets, since $p(y)$ is continuous. Thus
$p(s)$ cannot be negative since $s$ is an upper bound of $H$, and
cannot be positive since $s$ is the smallest upper bound, i.e.,
$p(s)=0$ as it was required.

Let $a$ be a positive element of $\Q$ and let $H\leteq \{y\in \Q:
y^2<a\}$.  Then $H$ is nonempty, bounded and $\mathcal{OF}$-definable.
From $\ax{CONT_\mathcal{OF}}$, it follows that $H$ has a supremum and
for the same reasons as before this supremum is a square root of $a$.

\medskip\noindent If $\mathfrak{Q}$ is real closed field, it
is elementary equivalent to $\R$, see \cite[Cor.3.3.16.]{marker}. Thus
$\mathfrak{Q}\models\ax{CONT_\mathcal{OF}}$ since
$\R\models\ax{CONT_\mathcal{OF}}$.
\end{proof}

\begin{rem}
Let us note that \ax{CONT_\L} is not strong enough to prove every
theorem of real analysis, e.g., the statement that there is
a function $f$ such that $f'(x)=f(x)$ and $\ran f=\Q$ is not provable from
\ax{CONT_\L}.
\end{rem}

Let $f$ be an \L-definable function. Then we denote one of
the formulas defining $f$ by $\Df{\phi_f}$\index{$\phi_f$}, i.e.,
$\phi_f$ is a formula in the language $\L$ such that
\begin{equation*}
f=\Setopen \langle \vx,\vyy\rangle \setmid \phi_f(\vx,\vyy) \Setclose.
\end{equation*}

\begin{prop}
Let $f,g:\Q^n\parrow \Q^m$ and $h:\Q^m\parrow\Q^k$ be \L-definable
functions and let $\lambda \in \Q$. Then $\dom f$ and $\ran f$ are
\L-definable and the following functions are also \L-definable ones:
$\lambda\cdot f$, $f+g$ and $f\circ h$.
\end{prop}

\begin{proof}
Let $\phi_f(\vx,\vyy)$, $\phi_g(\vx,\vyy)$ and $\phi_h(\vx,\vyy)$ be
formulas defining $f$, $g$ and $h$ in the language $\L$,
respectively. Then we can define $\dom f$ and $\ran
f$ as
\begin{equation*}
\dom f=\Setopen \vxx \setmid \exists\vyy\quad \phi_f(\vx,\vyy) \Setclose \quad\text{and}\quad \ran f=\Setopen \vyy \setmid \exists\vxx\quad \phi_f(\vx,\vyy) \Setclose.
\end{equation*}
Furthermore,
\begin{equation*}
\begin{split}
\lambda\cdot f &=\Setopen \langle \vx,\vyy\rangle \setmid \phi_f(\vx,\vzz) \lland \vyy=\lambda\cdot\vzz \Setclose,\\
f+g &=\Setopen \langle \vx,\vyy\rangle \setmid \phi_f(\vx,\vy_1) \lland \phi_g(\vx,\vy_2)\lland \vyy=\vy_1+\vy_2 \Setclose,\\
f\circ h &=\Setopen \langle \vx,\vyy\rangle \setmid \phi_f(\vx,\vzz) \lland \phi_h(\vzz,\vyy) \Setclose.
\end{split}
\end{equation*}
From these equations, it is easy to recognize the required formulas
defining $\dom f$, $\ran f$, $\lambda\cdot f$, $f+g$ and $f\circ g$.
\end{proof}

\begin{prop}
Let $f,g:\Q^n\parrow \Q$ be \L-definable
functions. Then the $f\cdot g$ and $1/f$ functions are also \L-definable ones.
\end{prop}

\begin{proof}
Let $\phi_f(\vx,\vyy)$ and $\phi_g(\vx,\vyy)$ be formulas defining $f$ and $g$ in the
language $\L$, respectively. 
Then 
\begin{equation*}
\begin{split}
f\cdot g &=\Setopen \langle \vx,\vyy\rangle \setmid
\phi_f(\vx,\vy_1) \lland \phi_g(\vx,\vy_2)\lland \vyy=\vy_1\cdot\vy_2
\Setclose,\\
 1/f &=\Setopen \langle \vx,\vyy\rangle \setmid
\phi_g(\vx,\vzz) \lland \vzz\neq 0\lland \vyy=1/\vzz \Setclose.
\end{split}
\end{equation*}
From these equations, it is easy to recognize the required formulas defining $f\cdot g$ and $1/f$. 
\end{proof}

\section{Continuous functions over ordered fields}

In this section we define the concept of continuity within FOL and
prove some related theorems which are used in the proofs of the main
results.

\theoremstyle{definition}
\newtheorem*{cousin}{\colorbox{thmbgcolor}{\textcolor{thmcolor}{\ax{CONT}-Cousin's
   Lemma}}}
\begin{cousin}\label{lem-ind-cousin}
Let $\L\supseteq\mathcal{OF}$. Assume \ax{CONT_\L} and \ax{AxOF}.  Let
$a,b\in\Q$ such that $a<b$, and let $\mathcal{A}$ be a set of
subintervals of $[a,b]$ which has the following properties:
\begin{itemize}
\item[(i)] \df{beginable}: for each $x\in[a,b]$, $\mathcal{A}$
  contains any small enough right and left neighborhood of $x$,
  i.e.,
\begin{multline*}
\forall x \in [a,b] \; \exists c,d \in \Q \quad c<x<d
 \lland \forall y\in[c,d]\cap[a,b]\;\;\\ (y<x \then [y,x]\in \mathcal{A}) \lland (x<y \then [x,y]\in \mathcal{A}),
\end{multline*}
\item[(ii)] \df{connectable}: if $[x,y],[y,z]\in \mathcal{A}$ then $[x,z]\in \mathcal{A}$, 
\item[(iii)] \df{\L-definable}: the set $\{t\in \Q : [a,t]\in \mathcal{A}\}$ is \L-definable.
\end{itemize}
Then $[a,b]\in \mathcal{A}$.
\end{cousin}

\begin{proof}
From \ax{CONT_\L}, it follows that the set 
\begin{equation*}
 H\leteq \Setopen x\in \Q\setmid a<x \lland \forall t\in(a,x) \quad [a,t]\in \mathcal{A} \Setclose
\end{equation*}
 has a supremum since it is an \L-definable, nonempty (since $\mathcal{A}$ is beginable) and bounded set.
Let us call this supremum $s$.
We complete the proof by proving that $[a,s]\in \mathcal{A}$ and $s=b$.

Since $\mathcal{A}$ is beginable, there is a $c\in [a,s)$ such that $[c,s]\in \mathcal{A}$.
Since $s$ is the supremum of $H$, $[a,t]\in \mathcal{A}$ for all $t\in(a,s)$.
Thus $[a,c]\in \mathcal{A}$, so by the connectability of $\mathcal{A}$, we get that $[a,s]\in \mathcal{A}$.

If $s<b$, there is an $e\in (s,b]$ such that $[s,t]\in \mathcal{A}$
for all $t \in (s,e]$ since $\mathcal{A}$ is beginable.  Thus we get
   that for all $t\in (s,e] \;\; [a,t]\in \mathcal{A}$ by using the
     connectability of $\mathcal{A}$ and the fact that $[a,s]\in
     \mathcal{A}$.  Then for all $t\in(a,e]\;\;[a,t]\in\mathcal{A}$.
       This contradicts the fact that $s$ is the supremum of the set,
       $H$ therefore $s=b$.
\end{proof}

A set $G\subseteq \Q$ is called \df{open}\index{open set} if it
contains an open interval around its every element, i.e., for all
$x \in G$, there are $a,b \in G$ such that $x\in(a,b)\subseteq G$.
The open sets of $\Q$ form a topology, which is called the \df{order
 topology}\index{order topology}. A function $f: \Q \parrow \Q$ is
called \df{order-continuous}\index{order-continuous function} if the
inverse image of any open subinterval of $\Q$ is open, i.e.,
$\{x:f(x)\in (c,d)\}$ is open for all $c,d\in \Q$. It is easy to see
that while the order-topology is a second-order logic concept both the
openness of a given set or the order-continuousness of a given
function are FOL ones.

\theoremstyle{definition} \newtheorem*{ind-order-bolzano}{\colorbox{thmbgcolor}{\textcolor{thmcolor}{\ax{CONT}-order-Bolzano's Theorem}}} 
\begin{ind-order-bolzano}
\index{\ax{CONT}-order-Bolzano's Theorem}
\label{thm-indobolzano}
Assume \ax{CONT_\L} and \ax{AxPOS}. 
Let $f:\Q \parrow Q$ be an \L-definable order-continuous function such that $[a,b]\subseteq \dom f$.
If $ f(a) < c < f(b)$, then there is a $t\in [a,b]$ such that $f(t)=c$.
\end{ind-order-bolzano}

\begin{proof}
Let 
\begin{equation*}
\mathcal{A}\leteq \Setopen[x,y]\subseteq [a,b]: \big(\,\forall t\in [x,y]\enskip f(t)<c\,\big) \llor \big(\,\forall t\in [x,y]\enskip f(t)>c\,\big)\Setclose
\end{equation*}
and assume that there is no such $t\in [a,b]$ that $f(t)=c$.
$\mathcal{A}$ is \L-definable since $f$ is such.
$\mathcal{A}$ is beginable since $f$ is order-continuous.
The connectability of $\mathcal{A}$ is also clear.
Thus from \ax{CONT}-Cousin's lemma we get that $f(t)<c$ for all $t\in [a,b]$ or $f(t)>c$ for all $t\in [a,b]$.
So if $f(a)<c$ and $f(b)>c$, then there must be a $t$ where $f(t)=c$.
This completes the proof of the theorem.\end{proof}

\begin{thm}
\label{thm-indsup}
Assume \ax{CONT_\L} and \ax{AxPOS}. 
Let $f: \Q \parrow \Q$ be an \L-definable order-continuous function such that $[a,b]\subseteq \dom f$.
Then $\sup\{f(x): x\in [a,b]\}$ exists and there is a $t\in [a,b]$ where $f(t)=\sup\{f(x): x\in [a,b]\}$.
\end{thm}

\begin{proof}
Let $H\leteq \{f(x): x\in [a,b]\}$ and 
\begin{equation*}
\mathcal{A}\leteq \Setopen [x,y]\subseteq [a,b] \setmid \exists c\in \Q \enskip \forall t\in [x,y] \enskip f(t)<c \Setclose.
\end{equation*}
Since $\mathcal{A}$ is \L-definable, beginable and connectable, $H$ is bounded by \ax{CONT}-Cousin's Lemma.
Thus from \ax{CONT_\L} it follows that $\sup H$ exists since $H$ is nonempty, \L-definable and bounded.
If there is no $t\in [a,b]$ such that $f(t)=\sup H$, then 
\begin{equation*}
\mathcal{A}\leteq \Setopen [x,y]\subseteq [a,b] \setmid \exists q \in \Q \enskip \forall t\in [x,y] \enskip f(t)<q<\sup H \Setclose
\end{equation*}
 is also \L-definable, beginable and connectable.  Thus $[a,b]\in
 \mathcal{A}$ by Cousin's lemma, therefore there is a $q<\sup H$ such that
 $f(t)<q$ for all $t\in [a,b]$ and this contradicts the supremum
 property.  This completes the proof of the theorem.\end{proof}

A function $f:\Q^n \parrow\Q^m$ is called \df{continuous} at $\vq\in\dom f$ if the usual formula of continuity holds for $f$, i.e.:
\begin{equation*}
\forall \varepsilon \in \Q^+ \; \exists \delta \in \Q^+ \enskip \forall \vp \in \dom f \quad \left| \vp-\vqq \right| < \delta  \then \left| f(\vpp)-f(\vqq) \right| <\varepsilon.
\end{equation*}
The function $f$ is called continuous if it is continuous at every
$\vq\in \dom f$. Let us note that if $f:\Q^n \parrow\Q^m$ is a continuous function,
$f\big|_H$ is also continuous for all $H\subseteq\Q^n$.

\theoremstyle{definition} \newtheorem*{ind-bolzano}{\colorbox{thmbgcolor}{\textcolor{thmcolor}{\ax{CONT}-Bolzano's Theorem}}} 
\begin{ind-bolzano}
\label{thm-ind-bolzano}
\index{\ax{CONT}-Bolzano's Theorem}
Let $\L\supseteq\mathcal{OF}$.
Assume \ax{CONT_\L} and \ax{AxOF}.
Let $f: \Q \parrow \Q$ be an \L-definable and continuous function such that $[a,b]\subseteq\dom f$.
If $c$ is between $f(a)$ and $f(b)$, there is an $s\in [a,b]$ such that $f(s)=c$.
\end{ind-bolzano}

\begin{proof}
Let $c$ be between $f(a)$ and $f(b)$.  We can assume that $f(a)<f(b)$.
Let $H\leteq \{x\in [a,b] : f(x) <c\}$.  Then $H$ is \L-definable,
bounded and nonempty.  Thus by \ax{CONT_\L}, the supremum of $H$
exists.  Let us call it $s$.  Both $\{x\in (a,b) : f(x)<c \}$ and
$\{x\in (a,b): f(x)>c \}$ are nonempty open sets since $f$ is
continuous on $[a,b]$.  Thus $f(s)$ cannot be less than $c$ since $s$
is an upper bound of $H$ and cannot be greater than $c$ since $s$ is
the least upper bound.  Hence $f(s)=c$ as it was required.
\end{proof}

\begin{thm}
\label{thm-sup}
Let $\L\supseteq\mathcal{OF}$.
Assume \ax{CONT_\L} and \ax{AxOF}.
Let $f: \Q \parrow \Q$ be an \L-definable and continuous function such that $[a,b]\subset\dom f$.
Then the supremum $s$ of $\{f(x): x\in [a,b]\}$ exists and there is a $y\in [a,b]$ such that $f(y)=s$.
\end{thm}

\begin{proof}
The supremum of $H\leteq \{y\in [a,b]: \exists c\in \Q\enskip \forall x\in [a,y]\quad f(x)<c\}$ exists by \ax{CONT_\L} since $H$ is \L-definable, nonempty and bounded.
This supremum has to be $b$ and $b\in H$ since $f$ is continuous on $[a,b]$.
Thus $Ran(f)\leteq \{f(x): x\in [a,b]\}$ is bounded.
Thus by \ax{CONT_\L}, it has a supremum, say $s$, since it is \L-definable and nonempty.
We can assume that $f(a)\neq s$.
Let $A\leteq \{y\in [a,b]: \exists c \in \Q\enskip \forall x\in [a,y]\quad f(x)<c<s\}$.
By \ax{CONT_\L}, $A$ has a supremum.
At this supremum, $f$ cannot be less than $s$ since $f$ is continuous on $[a,b]$ and $s$ is the supremum of $Ran(f)$.
\end{proof}

We call function $f$ \df{monotonic}\index{monotonic} if it preserves or reverses the relation $<$, i.e., $f(x)<f(y)$ [or $f(x)>f(y)$] for all $x,y\in \dom f$ if $x<y$.

\begin{lem}
\label{lem-moncont}
If $f:\Q\parrow\Q$ is monotonic and $\ran f$ is connected, $f$ is continuous.\qed
\end{lem}

\begin{lem}
\label{lem-injcont}
Assume \ax{CONT}.
Let $f:\Q\parrow\Q$ be definable and continuous such that $\dom f$ is connected.
Then
\begin{enumerate}
\item $\ran f$ is also connected.
\item If $f$ is injective, it is also monotonic.
Moreover, $f^{-1}$ is also a definable monotonic and continuous function.
\end{enumerate}
\end{lem}

\begin{proof}
Item (1) is a consequence of \ax{CONT}-Bolzano theorem. To
prove Item (2), let us first note that if $f$ were not monotonic, it
would not be injective by \ax{CONT}-Bolzano theorem.  It is clear that
$f^{-1}$ is definable and monotonic since $f$ is such.  Thus by
Lem.~\ref{lem-moncont}, $f^{-1}$ is continuous.
\end{proof}

\noindent
The following can be easily proved without any  of the \ax{CONT} axiom schemas.
\begin{prop}
Assume \ax{AxOF}.
Let $f:\Q\parrow\Q$ be a function. Then $f$ is continuous iff it is order-continuous.\qed
\end{prop}

\begin{prop}
Assume \ax{AxOF}.
Let $f,g:\Q\parrow\Q$ be continuous functions.
Then $f+g$, $f\cdot g$ and $f\circ g$ are also continuous ones.\qed
\end{prop}

\begin{example}
Let $exp:\R \rightarrow \R$ be the exponential map.
Then $exp$ is a continuous function but it is not $\mathcal{OF}$-definable.
\end{example}

We call a set $Z\subseteq \Q^n$ \df{closed} \index{closed set} iff $\Q^n\setminus Z$ is open.
Let us note that $\{p\}$ is closed for all $\vp\in \Q^n$.
The following can be easily proved without any \ax{CONT} schema.

\begin{prop}
Let $\Q$ be an ordered field. Let $f:\Q^n\rightarrow \Q^m$. The
following three statements are equivalent:
\begin{itemize}
\item[(i)] $f$ is continuous.
\item[(ii)] The $f^{-1}$-image of an open set is open.
\item[(iii)] The $f^{-1}$-image of a closed set is closed.\hfill\qedsymbol
\end{itemize}
\end{prop}

We say that $f: \Q^n\parrow \Q^k$ \df{tends to} $\vq\in \Q^k$ while
$\vx\in \dom f $ tends to $\vp\in \Q^n$ if the usual formula for the
limit of a function holds for $f$:
\begin{equation*}
\forall \varepsilon \in \Q^+ \;\exists \delta \in \Q^+ \enskip \forall
\vx \in \dom f \;\; 0< \left| \vx - \vpp \right| < \delta \then \left|
f(\vxx) -\vqq\right| <\varepsilon.
\end{equation*} 
This  $\vq$ is unique iff $\vp$ is not isolated from the set $\dom
f\setminus\{\vpp\}$, i.e., $B_\varepsilon(\vpp)\cap \dom f
\setminus\{\vpp\}\neq \emptyset$ for all $\varepsilon\in\Q^+$. In
this case we call $\vq$ the \df{limit}\index{limit} of the function
$f$ at $\vp$ and we write that
\begin{equation*}
 {\lim_{\vxx\rightarrow \vq}f(\vxx)}=\vq.
\end{equation*}

\section{Differentiable functions over ordered fields}
\label{sec-diff}
In this section we define the concept of differentiability
within FOL and prove some theorems about it which are
used in the proofs of the main theorems.

We say that a function $f: \Q^n \parrow \Q^m$ is
\df{differentiable}\label{p-diff}\index{differentiable function} at $\vq\in \dom f $
if the usual formula
\begin{equation*}
\forall \varepsilon \in \Q^+\; \exists \delta\in \Q^+ \enskip \forall
  \vp\in \dom f\cap B_\delta(\vqq) \quad\left| f(\vpp) - f(\vqq)
  - L(\vp-\vqq)\right| \le \varepsilon\cdot\left| \vp-\vqq \right|
\end{equation*}
 holds for a linear map $L: \Q^n\rightarrow \Q^m$.  In this case $L$
 is called \df{a derivative}\index{derivative} of $f$ at $\vq$.  The
 \df{set of derivative maps} of $f$ at $\vq$ is denoted by
 \Df{\der_{\vq}f}\index{$\der_{\vqq}f$} and any
 derivative of $f$ at $\vq$ is denoted by
 \Df{d_{\vqq}f}.\index{$d_{\vqq}f$} Function $f$ is called {\bf
   uniquely differentiable} at $\vq$ if it has one and only one
 derivative at $\vq$. In this case, $d_{\vqq}f$ is called \df{the
   derivative}\index{the derivative} of $f$ at $\vqq$.

\begin{rem}\label{rem-reldiff}
We say that a binary relation is differentiable at $\vq$ if it is
equal to a differentiable function on a small enough neighbourhood of
$\vq$.
\end{rem}

\begin{rem} \label{rem-diff}
If $f$ extends $f_0$ (i.e., $f\supseteq f_0$) and $f$ is
differentiable at $\vq$, then $f_0$ is also differentiable at $\vq$
and every derivative of $f$ at $\vq$ is also a derivative of $f_0$ at
$\vq$.
\end{rem}

Several theorems can be proved about differentiable functions without
using any \ax{CONT} axiom schema. Here we prove some of them. To do
so, we will use the following easily provable and well-known fact about linear
maps.

\begin{lem}\label{lem-linbound}
Every linear map $L$ is bounded in the following sense: there is a bound $M\in\Q^+$ such that 
$|L(\vxx)-L(\vyy)|\le M\cdot|\vx-\vyy|$ for all $\vx,\vy\in \dom L$.
\end{lem}

\begin{thm}\label{thm-difflip}
Let $f$ be differentiable at $\vx\in\dom f$. Then there are
$\delta,K\in\Q^+$ such that $|f(\vxx)-f(\vyy)|\le K \cdot|\vx-\vyy|$
for all $\vy\in\dom f\cap B_\delta(\vxx)$.
\end{thm}

\begin{proof}
We have to choose $\delta$ and $K$ appropriately. 
Since $f$ is differentiable at $\vx\in\dom f$, there is a linear map $L$
and $\delta$ such that 
\begin{equation*}
\left| f(\vyy) - f(\vxx) - L(\vy-\vxx)\right| \le | \vy-\vxx |
\end{equation*}
for all $\vy\in \dom f\cap B_{\delta}(\vxx)$.
Let $M$ be a bound of $L$ which exists by Lem.~\ref{lem-linbound}, and let $K$ be $M+1$.
Then, by the triangle inequality and the linearity of $L$, 
\begin{equation*}
|f(\vyy) - f(\vxx)|\le| f(\vyy) - f(\vxx) - L(\vy-\vxx)| +| L(\vyy)-L(\vxx) |.
\end{equation*}
Thus, by Lem.~\ref{lem-linbound},
\begin{equation*}
|f(\vyy) - f(\vxx)|\le (1+M)\cdot|\vy-\vxx |=K\cdot|\vy-\vxx|
\end{equation*}
for all $\vy\in\dom f\cap B_\delta(\vxx)$.
\end{proof}

\begin{cor}
If $f$ is differentiable at $\vx\in\dom f$, then $f$ is also
continuous at $\vx$. \hfill\qed
\end{cor}

\begin{thm}\label{thm-chn}
Let $\Q$ be an ordered field. Let $g:\Q^n \parrow \Q^m$ and
$f:\Q^m\parrow \Q^k$ such that $g$ is differentiable at $\vq\in \Q^n$
and $f$ is differentiable at $g(\vqq)$. Let $L_g\in \der(\vq, g)$ and
$L_f\in \der(g(\vqq),f)$. Then $L_g\circ L_f\in \der(\vq ,g\circ f)$.
\end{thm}

\begin{proof}
Let $\varepsilon\in\Q^+$ be fixed.  Let $L_g$ be a derivative of $g$
at $\vq$ and $L_f$ be a derivative of $f$ at $g(\vqq)$.  Let $K$,
$\delta_0$ be the bounding constants of $g$ given by
Thm.~\ref{thm-difflip} and let $M$ be the bounding constant of $L_f$
given by Lem.~\ref{lem-linbound}.  Then there is a $\delta\in\Q^+$ such
that $\delta\le \delta_0$,
\begin{equation*}
|g(\vpp) - g(\vqq) - L_g(\vp-\vqq)|\le \frac{\varepsilon}{2M}\cdot|\vp-\vqq | \text{ and}
\end{equation*}
\begin{equation*}
\left|f\big(g(\vpp)\big) - f(g(\vqq)) - L_f\big(g(\vpp)-g(\vqq)\big)\right|\le \frac{\varepsilon}{2K}\cdot|g(\vpp)-g(\vqq)|
\end{equation*}
for all $\vp\in\dom g\cap B_\delta(\vqq)$ and $g(\vpp)\in\dom f \cap B_\delta(\vqq)$.
By the triangle inequality and the linearity of $L$, 
\begin{multline*}
\left|f\big(g(\vpp)\big) - f\big(g(\vqq)\big) - L_f\big(L_g(\vp-\vqq)\big)\right|\\
\le \left|f\big(g(\vpp)\big) - f\big(g(\vqq)\big) - L_f\big(g(\vpp)-g(\vqq)\big)\right|
 + \left|L_f\big(g(\vpp)-g(\vqq)- L_g(\vp-\vqq)\big)\right|\\
\le \frac{\varepsilon}{2K}\cdot|g(\vpp)-g(\vqq) |+ M\cdot |g(\vpp) - g(\vqq) - L_g(\vp-\vqq)| \\
\le \frac{\varepsilon}{2K}\cdot K\cdot|\vp-\vqq | + M\cdot \frac{\varepsilon}{2M}\cdot|\vpp-\vqq|= \varepsilon\cdot|\vpp-\vqq|
\end{multline*}
for all $\vp\in\dom f\circ g\cap B_\delta(\vqq)$; and that is what we wanted to prove.
\end{proof}

\begin{example} 
Let $g:\Q\rightarrow\Q$ and $f:[-1,0]\rightarrow \Q$ be defined as
$g(x)=x^2$ for all $x\in\Q$ and $f(x)=0$ for all $x\in[-1,0]$.  Then
$g\circ f=\setopen \langle 0,0\rangle\setclose$.  So $f$, $g$ and
$g\circ f$ are differentiable at $0$ by definition, but derivative of
$g\circ f$ is not unique though the derivatives of $f$ and $g$ are
such.
\end{example}

The following Thms.\ \ref{thm-unider} and \ref{thm-uniderv} can be
proved by the usual textbook proofs of the uniqueness of derivatives.
\begin{thm}[uniqueness of derivatives]\label{thm-unider}
Assume \ax{AxOF}.  If $f:\Q^n\parrow \Q^m$ is differentiable at $\vq$
and there is a $\delta\in\Q^+$ such that $B_\delta(\vqq)\subseteq \dom
f$ then the derivative of $f$ at $\vq$ is unique.
\end{thm}

When $f$ is a function from a subset of $\Q$ to $\Q^k$, a
derivative of $f$ can be defined as a
limit of the function $h\mapsto\frac{f(x+h)-f(x)}{h}$, i.e., as
$\lim_{h\rightarrow 0}\frac{f(x+h)-f(x)}{h}$. In this situation the
derivatives of $f$ at $x$ are not linear maps but vectors. In this
case, we use the notation $f'(x)$ for a derivative and we call it a
\df{derivative vector}\index{derivative vector} of $f$ at $x$.\label{derivative} The
connection between the two definitions is the following: $d_xf(t)=t\cdot f'(x)$ for all $t\in \Q$.
By the following theorem, the derivatives of a differentiable curve are unique.    

\begin{thm}[uniqueness of derivative vectors]
\label{thm-uniderv}
Assume \ax{AxOF}.  The derivative of $f$ at $t$ is unique, if $f:\Q\parrow \Q^k$ is a curve (i.e., $\dom f$ is
connected and has at least two elements) and $f$ is differentiable at $t$.
\end{thm}

In the case when $f:Q\parrow\Q^k$, we define the
derivative function $f'$ of $f$ as the binary relation that relates the
derivatives of $f$ at $x$ to $x\in\dom f$. Of course, $f'$ is a function only if $f$ is
uniquely differentiable.
\begin{prop} 
Let $f:\Q\parrow\Q^k$ be a $\L$-definable function. Then
$f'\subseteq\Q\times\Q^k$ is also $\L$-definable.
\end{prop}

\begin{proof}
Let $\phi_f(x,\vyy)$ be a formula defining $f$ in the language $\L$.
Then
\begin{multline*} 
f' =\Setopen \langle x_0,\vzz\rangle \setmid \phi_f(x_0,\vy_0) \lland
\forall \varepsilon\in\Q^+\enskip\exists\delta\in\Q^+\enskip \forall
x\in\Q\right.\\\left. |x-x_0|\lland \phi_f(x,\vyy) \then
|\vy-\vy_0-(x-x_0)\cdot\vzz|\le\varepsilon\cdot|x-x_0| \Setclose;
\end{multline*}
and from this equation, it is easy to recognize the formula defining
$f'$.
\end{proof}

\begin{cor} 
If $f:\Q\parrow\Q^k$ is a uniquely differentiable and $\L$-definable function, 
$f':\Q\parrow\Q^k$ is an $\L$-definable function.
\end{cor}

Let $h:\Q\parrow \Q$ and let $H\subseteq \dom h$.
Then $h$ is said to be \df{increasing}\index{increasing} on $H$ iff
$h(x)<h(y)$ for all $x,y\in H$ for which $x<y$; and
$h$ is said to be \df{decreasing}\index{decreasing} on $H$ iff 
$h(y)<h(x)$ for all $x,y\in H$ for which $x<y$.
The proof of the following theorem also uses only the ordered field
property of the real numbers, see, e.g., \cite{Rudin}, \cite{Laczkovich-Soos}.

\begin{prop}\label{propDiff}
Let $\Q$ be an ordered field. Let $f,g:\Q \parrow \Q^n$ and
$h:\Q\parrow \Q$. Then (i)--(v) below hold.
\begin{itemize}
\item[(ii)] Let $\lambda\in \Q$.
If $f$ is differentiable at $x$, then $\lambda\cdot f$ is also differentiable at $x$ and $(\lambda\cdot f)'(x)=\lambda\cdot f'(x)$.
\item[(iii)] If $f$ and $g$ are differentiable at $x$ and $x$ is an accumulation point of $\dom f\cap
\dom g$, then $f+g$ is differentiable at $x$ and $(f+g)'(x)=f'(x)+g'(x)$.
\item[(v)] If $h$ is increasing (or decreasing) on $(a,b)$, differentiable at $x\in (a,b)$ and
$h'(x)\neq0$, then $h^{-1}$ is differentiable at $h(x)$.
\end{itemize}
\end{prop}

\begin{onproof}
Since the proofs of the statements are based on the same calculations and 
ideas as in 
real analysis, we omit the proof,
see \cite[Thms.\ 28.2, 28.3, 28.4 and 29.9]{ross}.\hfill\qed
\end{onproof}

\theoremstyle{definition} \newtheorem*{chainrule}{\colorbox{thmbgcolor}{\textcolor{thmcolor}{Chain Rule}}} 
\begin{chainrule}
\label{chainrule}
\index{chain rule}
Assume \ax{AxOF}.
Let $g:\Q^n \parrow \Q^m$ and $f:\Q^m\parrow \Q^k$.
If $g$ is differentiable at $\vp\in \Q^n$ and $f$ is differentiable at $g(\vpp)$, then $g\circ f$ is differentiable at $\vp$ and $d_{\vpp}g\circ d_{g(\vpp)}f$ is one of its derivatives, i.e., 
\begin{equation*} 
d_{\vpp}(g \circ f) = d_{\vpp}g \circ d_{g(\vpp)}f.
\end{equation*}
In particular, if $g: \Q\parrow \Q^m$, and  $g$ is differentiable at $t\in\Q$ and  $f$ is differentiable at $g(t)$, then  
\begin{equation*} 
(g\circ f)'(t)=d_{g(t)}f\big(g'(t)\big).
\end{equation*} 
\end{chainrule}

\begin{prop}\label{propAff} 
Assume \ax{AxOF}.
The derivative of an affine map is its linear part; i.e., if $A:\Q^n\rightarrow \Q^m$ is an affine map, then 
$(d_{\vqq}A)(\vpp)=A(\vpp)-A(\voo)$, where $\vp,\vq\in \Q^n$ and $\vo$ is the origin of $\Q^n$.
\end{prop}
\begin{proof} 
The proof is straightforward from our definitions.\end{proof}

\begin{cor}\label{colChain}
Let $\Q$ be an ordered field. If $g: \Q\rightarrow \Q^n$ and $A:\Q^n
\rightarrow \Q^m$ is an affine map, then $(g\circ A)'(t)=
A\big(g'(t)\big)-A(\voo)$. \hfill\qedsymbol
\end{cor}

We say that function $f:\Q^n\rightarrow \Q$ is \df{locally
  maximal}\index{local maximality} at $x \in \dom f$ iff there is a
$\delta \in \Q^+$ such that $f(y)\le f(x)$ for all $y\in B_\delta(x)$.
The \df{local minimality}\index{local minimality} is analogously
defined.

\begin{prop}\label{prop-max}
Assume \ax{AxOF}.
If $f:\Q\parrow \Q$ is differentiable on $(a,b)$ and 
locally maximal or minimal at $x\in (a,b)$, its 
derivative is $0$ at $x$, i.e., $f'(x)=0$.
\end{prop}

\begin{onproof}
The proof is the same as in real analysis, see e.g.,
\cite[Thm.5.8]{Rudin}.\hfill\qed
\end{onproof}

Function $f$ is said to be \df{differentiable on set $H$} if $H\subseteq \dom f$ and $f$ is differentiable at $x$ for all $x\in H$.


\theoremstyle{definition} \newtheorem*{meanvalue}{\colorbox{thmbgcolor}{\textcolor{thmcolor}{\ax{CONT}-Mean--Value Theorem}}} 
\begin{meanvalue}
Let $\L\supseteq\mathcal{OF}$.
Assume \ax{CONT_\L} and \ax{AxOF}.
Let $f:\Q\parrow \Q$ be an \L-definable function which is differentiable on $(a,b)$ and continuous on $(a,b)$.
If $a\neq b$, there is an $s\in(a,b)$ such that $f'(s)=\frac{f(b)-f(a)}{b-a}$.\qed
\end{meanvalue}

\begin{proof}
Let $h(t)\leteq \big(f(b)-f(a)\big)\cdot t - (b-a)\cdot f(t)$.  Then
$h$ is continuous on $[a,b]$, differentiable on $(a,b)$ and
$h(a)=f(b)\cdot a-b\cdot f(a)=h(b)$.  If $h$ is constant then
$h'(t)=0$ for all $t\in(a,b)$.  Otherwise, by Thm.~\ref{thm-sup}, there
is a maximum/minimum of $h$ different from $h(a)=h(b)$ in a $t\in
(a,b)$.  Hence $h'(t)=0$ by Prop.~\ref{prop-max}.  This completes the
proof since $h'(t)=f(b)-f(a)-(b-a)\cdot f'(t)$.\end{proof}


\theoremstyle{definition}
\newtheorem*{rolle}{\colorbox{thmbgcolor}{\textcolor{thmcolor}{\ax{CONT}-Rolle's
      Theorem}}}
\begin{rolle}
Let $\L\supseteq\mathcal{OF}$.
Assume \ax{CONT_\L} and \ax{AxOF}.
Let $f:\Q\parrow\Q$ be definable and \L-differentiable function which is differentiable on $(a,b)$ and continuous on $(a,b)$.
If $a\neq b$ and $f(a)=f(b)$, there is an $s\in(a,b)$ such that $f'(s)=0$.\qed
\end{rolle}


\begin{cor}\label{colHplane}
Let $\L\supseteq\mathcal{OF}$.
Assume \ax{CONT_\L} and \ax{AxOF}.
Let $\gamma:\Q\rightarrow \Q^n$ be an \L-definable and differentiable curve.
Then for all distinct $a,b\in\Q$ and for every $(n-1)$-dimensional subspace $H$ that contains $\gamma(a)-\gamma(b)$, there is at least one $c$ between $a$ and $b$ such that $\gamma'(c)\in H$.
\end{cor}

\begin{proof}
The derivative vector of a curve $\gamma$ composed with a linear map
$A$ at $t\in\Q$ is the $A$-image of $\gamma'(t)$ by
Cor.~\ref{colChain}.  Since any $(n-1)$-dimensional subspace of $\Q^n$
can be taken to $\{0\}\times \Q^{n-1}$ by a linear transformation, we
can assume that $H=\{0\}\times \Q^{n-1}$.  Recall that the function
$\pi_t:\Q^n\rightarrow \Q$ is defined as $p\mapsto p_t$.  Then
$\gamma\circ\pi_t(a)=\gamma\circ\pi_t(b)$ since
$\gamma(a)-\gamma(b)\in H$.  By applying Rolle's Theorem to
$\gamma\circ \pi_t$, we get that there is a $c\in\Q$ such that
$(\gamma\circ \pi_t)'(c)=0$.  Thus $\gamma'(c)$ is an element of $H$
since $(\gamma\circ\pi_t)'(c)=\pi_t\big(\gamma'(c)\big)=\gamma'(c)_t$
by Cor.~\ref{colChain}.
\end{proof}

\begin{prop}
Let $\L\supseteq\mathcal{OF}$.
Assume \ax{AxOF} and \ax{CONT_\L}.
Let $f:\Q\parrow \Q$ be an \L-definable differentiable function $(a,b)\subseteq \dom f$ for some $a,b\in\Q$.
If $f'(t)=0$ for all $t\in(a,b)$ then there is a $c\in\Q$ such that $f(t)=c$ for all $t\in(a,b)$.
\end{prop}

\begin{proof}
If there are $x,y\in (a,b)$ such that $f(x)\neq f(y)$ and $x\neq y$,
then from \ax{CONT}-Mean-Value Theorem there is a $t$ between $x$ and
$y$ such that $f'(t)\cdot(y-x)=f(y)-f(x)\neq 0$ and this contradicts
that $f'(t)=0$.
\end{proof}

\theoremstyle{definition} \newtheorem*{darboux}{\colorbox{thmbgcolor}{\textcolor{thmcolor}{\ax{CONT}-Darboux's Theorem}}} 
\begin{darboux}
Let $\L\supseteq\mathcal{OF}$.
Assume \ax{AxOF} and \ax{CONT_\L}.
Let $f:\Q\parrow \Q$ be an \L-definable and differentiable function such that $(a,b)\subseteq\dom f$ for some $a,b\in \Q$.
If $c\in\big(f'(a),f'(b)\big)$, there is an $s\in(a,b)$ such that $f'(s)=c$.\qed
\end{darboux}

\begin{proof}
We can assume that $f'(a)>d>f'(b)$. Let $g(t)\leteq f(t)-t\cdot d$.
Then $g$ is differentiable and $g'(a)>0$, $g'(b)<0$. Thus $g$ cannot
be maximal at $a$ or $b$ by Prop.~\ref{prop-max}. Thus, from
Thm.~\ref{thm-sup}, we get that there is a point, say $c$, between $a$
and $b$ where $g$ is maximal. Thus from Prop.~\ref{prop-max}, we also
get that $g'(c)=f'(c)-d=0$.
\end{proof}

Let $i\leq n$.  $\pi_i:\Q^n\rightarrow \Q$ denotes the $i$-th
projection function, i.e., $\pi_i:p\mapsto p_i$.  Let $f:\Q\parrow
\Q^n$.  We denote the $i$-th coordinate function of $f$ by $f_i$,
i.e., $\Df{f_i}\leteq f\circ\pi_i$\index{$f_i$}.  Sometimes $f_\tau$
is used instead of $f_1$.  A function $A:\Q^n\rightarrow \Q^j$ is said
to be an \df{affine map}\index{affine map} if it is a linear map
composed with a translation.%
\footnote{That is, $A$ is an affine map if there are
  $L:\Q^n\rightarrow \Q^j$ and $a\in \Q^j$ such that
  $A(\vpp)=L(\vpp)+a$, $L(p+q)=L(\vpp)+L(\vqq)$ and $L(\lambda\cdot
  p)=\lambda\cdot L(\vpp)$ for all $\vp,\vq\in \Q^n$ and $\lambda \in
  \Q$.}

The following proposition says that the derivative of a function $f$
composed with an affine map $A$ at $x$ is the image of the
derivative $f'(x)$ taken by the linear part of $A$.

\begin{prop}\label{propAff2}
Assume \ax{AxOF}.
Let $f:\Q \parrow \Q^n$ be differentiable at $x$ and let $A:\Q^n\rightarrow
\Q^j$ be an affine map.
Then $f\circ A$ is differentiable at $x$ and $(f\circ A)'(x)=A\big(f'(x)\big)-A(o)$.
In particular, $f'(x)=\langle f'_1(x),\dots,f'_n(x)\rangle$, 
i.e., $f'_i(x)=f'(x)_i$.
\end{prop}
\begin{onproof}
 The statement follows straightforwardly from the respective definitions.\hfill\qed
\end{onproof}

\begin{prop} \label{propInt}
Let $\L\supseteq\mathcal{OF}$.
Assume \ax{AxOF} and \ax{CONT_\L}.
Let $f,g:\Q\parrow \Q$ be \L-definable and differentiable functions on $(a,b)$.
If $f'(x)=g'(x)$ for all $x\in(a,b)$, then there is a $c\in \Q$ such that 
$f(x)=g(x)+c$ for all $x\in(a,b)$.
\end{prop}

\begin{proof}
Assume that $f'(x)=g'(x)$ for all $x\in(a,b)$. Let $h\leteq f-g$.
Then $h'(x)=f'(x)-g'(x)=0$ for all $x\in(a,b)$ by (ii) and (iii) of
Prop.~\ref{propDiff}. If there are $y,z\in (a,b)$ such that
$h(y)\neq h(z)$ and $y\neq z$, then by the \ax{CONT}-Mean-Value
Theorem, there is an $x$ between $y$ and $z$ such that
$h'(x)=\frac{h(z)-h(y)}{z-y}\neq 0$ and this contradicts $h'(x)=0$.
Thus $h(y)=h(z)$ for all $y,z\in (a,b)$. Hence there is a $c\in \Q$
such that $h(x)=c$ for all $x\in(a,b)$.
\end{proof}

\section{Tools used for proving the Twin Paradox}

In this section we develop the tools which were used in
Chap.~\ref{chp-twp}. To do so, let us first introduce a notation. We
say that $\vp\in \Q^d$ is {\bf vertical} iff $\vp_\sigma=\vo$.

\begin{lem} \label{lemWp}
Let $f:\Q\parrow \Q^d$ be a well-parametrized timelike curve. Then the following hold:
\begin{itemize}
\item[(i)] Let $x\in \dom f$. Then $f_\tau$ is differentiable at $x$
 and $1\le|f'_\tau(x)|$. Furthermore, $|f'_\tau(x)|=1$ iff $f'(x)$
 is vertical.
\item[(ii)] Assume \ax{CONT} and let $f$ be definable. Then $f_\tau$
 is increasing or decreasing. Moreover, $1\le f'_\tau(x)$ for all $x\in
 \dom f$ if $f_\tau$ is increasing.
\end{itemize}
\end{lem}

\begin{proof}
As $f$ is a well-parametrized curve, $f'(x)$ is of Minkowski length $1$.
By Prop.~\ref{propAff}, $f_\tau$ is differentiable at $x$ and $f'_\tau(x)=f'(x)_\tau$.
Now, Item (i) follows from the fact that the absolute value of the time component of a
vector of Minkowski length 1 is always at least 1 and it is 1 iff the vector is vertical.

\smallskip
\noindent
Let us now prove Item (ii).  From Item (i), we have $f'_\tau(x)\neq 0$
for all $x\in\dom f$.  Thus by \ax{CONT}-Rolle's Theorem, $f_\tau$ is
injective.  Consequently, by \ax{CONT}-Bolzano's Theorem, $f_\tau$ is
increasing or decreasing since $f_\tau$ is continuous and injective.
Let us now assume that $f_\tau$ is increasing.  Then $0\le f'_\tau(x)$
for all $x\in \dom f$ by our definition of the derivative.  Hence by
Item (i), $1\le f'_\tau(x)$ for all $x\in\dom f$; and that is what we
wanted to prove.
\end{proof}

\begin{thm}\label{thmFtwp}
Assume \ax{CONT}. Let $f:\Q\parrow \Q^d$ be a definable
well-parametrized timelike curve, and let $a,b\in \dom f$ such that
$a<b$. Then the following hold:
\begin{itemize}
\item[(i)] $b-a\le \left|f_\tau(b)-f_\tau(a)\right|$, and
\item[(ii)]$b-a<\big|f_\tau(b)-f_\tau(a)\big|$ if $f(x)_\sigma\neq
 f(a)_\sigma$ for any $x\in[a,b]$.
\end{itemize}
\end{thm}

\begin{proof}
For every $i\leq d$, we have that $f_i$ is definable and
differentiable by Prop.~\ref{propAff}.  Hence by the
\ax{CONT}-Mean-Value Theorem, there is an $s\in (a,b)$ such that
$f'_\tau(s)=\frac{f_\tau(b)-f_\tau(a)}{b-a}$.  By Item (i) of
Lem.~\ref{lemWp}, we have $1\le |f'_\tau(s)|$.  Then $b-a\le
\big|f_\tau(b)-f_\tau(a)\big|$.  That completes the proof of Item (i).

\smallskip
\noindent
To prove Item (ii), 
let $x\in [a,b]$ such that $f(x)_\sigma\neq f(a)_\sigma$.
Then there is an $i>1$ such that $f_i(x)\neq f_i(a)$.
Hence by the \ax{CONT}-Mean-Value Theorem, there is a
$y\in (a,b)$ such that $f'_i(y)=\frac{f_i(x)-f_i(a)}{x-a}\neq 0$.
Thus $f'(y)$ is not vertical.
Therefore, by Item (i) of Lem.~\ref{lemWp}, we have $1<|f'_\tau(y)|$.
Thus by our definition of the derivative,
there is a $z\in(y,b)$ such that $1<\frac{|f_\tau(z)-f_\tau(y)|}{z-y}$.
Hence we have
\begin{equation*}
z-y<|f_\tau(z)-f_\tau(y)|.
\end{equation*}
Let us note that $a<y<z<b$.
By applying Item (i) 
to $[a,y]$ and $[z,b]$ we get
\begin{equation*}
y-a\le \big| f_\tau(y)-f_\tau(a) \big|\quad\text{and}\quad b-z\le \big|f_\tau(b)-f_\tau(z)\big|.
\end{equation*}
$f_\tau$ is increasing or decreasing by Item (ii) of Lem.~\ref{lemWp}.
Thus $f_\tau(a)<f_\tau(y)<f_\tau(z)<f_\tau(b)$ or $f_\tau(a)>f_\tau(y)>f_\tau(z)>f_\tau(b)$.
Therefore, by adding up the last three inequalities, we get $b-a<\big|f_\tau(b)-f_\tau(a)\big|$.
\end{proof}

\begin{thm}\label{thmJeq}
Assume \ax{CONT}.
Let $f,g:\Q\parrow \Q^d$ be definable well-parametrized timelike curves.
Let $a,b,a',b'\in\Q$ such that
\begin{itemize}
\item $a\le b$ and $a'\le b'$, 
\item $[a,b]\subseteq \dom f$ and $[a',b']\subseteq \dom g$, and
\item $\{f(t):t\in[a,b]\}=\{g(t'):t'\in [a',b']\}$.
\end{itemize}
Then $b-a=b'-a'$.
\end{thm}

\begin{proof}
By (ii) of Lem.~\ref{lemWp},
$f_\tau$ is increasing or decreasing on $[a,b]$ and
so is $g_\tau$ on $[a',b']$.
Without losing generality, we can assume that $\dom f=[a,b]$,
$\dom g=[a',b']$
and that $f_\tau$ and $g_\tau$ are increasing on $[a,b]$ and $[a',b']$,
respectively.%
\footnote{ It can be assumed that $f_\tau$ is increasing on $[a,b]$
  because the assumptions of the theorem remain true when $f$ and
  $[a,b]$ are replaced by $-Id\circ f$ and $[-b,-a]$, respectively,
  and $f_\tau$ is decreasing on $[a,b]$ iff $(-Id\circ f)_\tau$ is
  increasing on $[-b,-a]$.}  Then $\ran(f)=\ran(g)$ by the assumptions
of the theorem.  Furthermore, $f$ and $g$ are injective since $f_\tau$
and $g_\tau$ are such.  Since $\ran(f)=\ran(g)$ and $g_\tau$ is
injective, $f\circ g^{-1}=f_\tau\circ g_\tau^{-1}$.  Let $h\leteq
f\circ g^{-1}=f_\tau\circ g_\tau^{-1}$.  Since
$\ran(f_\tau)=\ran(g_\tau)$ and $f_\tau$ and $g_\tau$ are increasing,
$h$ is an increasing bijection between $[a,b]$ and $[a',b']$.  Hence
$h(a)=a'$ and $h(b)=b'$.  We prove that $b-a=b'-a'$ by proving that
there is a $c\in \Q$ such that $h(x)=x+c$ for all $x\in [a,b]$.  We
can assume that $a\neq b$ and $a'\neq b'$.  By Lem.~\ref{lemWp},
$f_\tau$ and $g_\tau$ are differentiable on $[a,b]$ and $[a',b']$,
respectively, and $f'_\tau(x)>0$ for all $x\in[a,b]$ and
$g'_\tau(x')>0$ for all $x'\in[a',b']$.  By Chain Rule and (v) of
Prop.~\ref{propDiff}, $h=f_\tau\circ g_\tau^{-1}$ is also
differentiable on $(a,b)$.  By $h=f\circ g^{-1}$, we have $f=h\circ
g$.  Thus $f'(x)=h'(x)g'\big(h(x)\big)$ for all $x\in(a,b)$ by Chain
Rule.  Since both $f'(x)$ and $g'\big(h(x)\big)$ are of Minkowski
length $1$ and their time-components are positive%
\footnote{That is, $f'_\tau(x)>0$ and $g'_\tau\big(h(x)\big)>0$.}  for
all $x\in(a,b)$, we conclude that $h'(x)=1$ for all $x\in(a,b)$.  By
Prop.~\ref{propInt}, we get that there is a $c\in \Q$ such that
$h(x)=x+c$ for all $x\in(a,b)$ and thus for all $x\in[a,b]$ since $h$
is an increasing bijection between $[a,b]$ and $[a',b']$.
\end{proof}

A curve is called slower than light ({\sf STL}) and faster than light
({\sf FTL}) iff any of its chords is timelike and spacelike,
respectively.

\begin{prop}\label{propstl}
Let $\L\supseteq\mathcal{OF}$.
Assume \ax{AxOF} and \ax{CONT_\L}.
Let $\gamma: \Q \parrow \Q^d$ be an \L-definable and continuous curve.
Then (i) and (ii) below hold:
\begin{itemize}
\item[(i)] $\gamma$ is timelike $\Longrightarrow$ $\gamma$ is {\sf STL}.
\item[(ii)] $\gamma$ is {\sf STL}, {\sf FTL} or it has a lightlike chord.
\end{itemize}
\end{prop} 

\begin{proof}
To prove the first statement, let us assume that $\gamma$ is not
\textsf{STL}.  Then it has a lightlike or spacelike chord, say
$\{\vp,\vqq\}$.  Let $H$ be a $(d-1)$-dimensional subspace that
contains $\vp-\vq$ and does not contain timelike vectors.  Thus by
Cor.~\ref{colHplane}, we get that there is a $t\in\Q$ such that
$\gamma'(t)$ is in $H$.  Since $H$ does not contain timelike vectors,
$\gamma'(t)$ is not timelike.  Thus $\gamma$ is not timelike.

To prove the second statement, let us assume that $\gamma$ is not {\sf
  STL} or {\sf FTL} and does not have a lightlike chord.  Then
$\gamma$ has both timelike and spacelike chords.  Then there are
distinct points $\va,\vb,\vc\in \ran(\gamma)$ such that the triangle
$\{\va,\vb,\vc\,\}$ determines two timelike and one spacelike or two
spacelike and one timelike chords of $\gamma$.  We can assume that
$\gamma(0)=\vc$ and $\vc$ is the intersection of the chords of
same type.  See Fig.~\ref{keeep}.

\begin{figure}
\begin{center}
\small
\psfrag*{t}[r][r]{$t$}
\psfrag*{y}[r][r]{$y$}
\psfrag*{0}[r][r]{$0$}
\psfrag*{1}[t][t]{$1$}
\psfrag*{a}[b][b]{$\va$}
\psfrag*{b}[l][l]{$\vb$}
\psfrag*{c}[lb][lb]{$\vc$}
\psfrag*{c'}[l][l]{$\vc$}
\psfrag*{2}[lt][lt]{$\frac{1}{\sqrt{2}}=f\big(\gamma(y)\big)$}
\psfrag*{g}[b][b]{$\gamma$}
\psfrag*{f}[r][r]{$f$}
\psfrag*{gy}[b][b]{$\gamma(y)$}
\psfrag*{gy'}[r][r]{$\gamma(y)$}
\psfrag*{or}[c][c]{or}
\psfrag*{text1}[l][l]{spacelike}
\psfrag*{text2}[l][l]{timelike}
\psfrag*{text3}[l][l]{lightlike}
\includegraphics[keepaspectratio, width=0.8\textwidth]{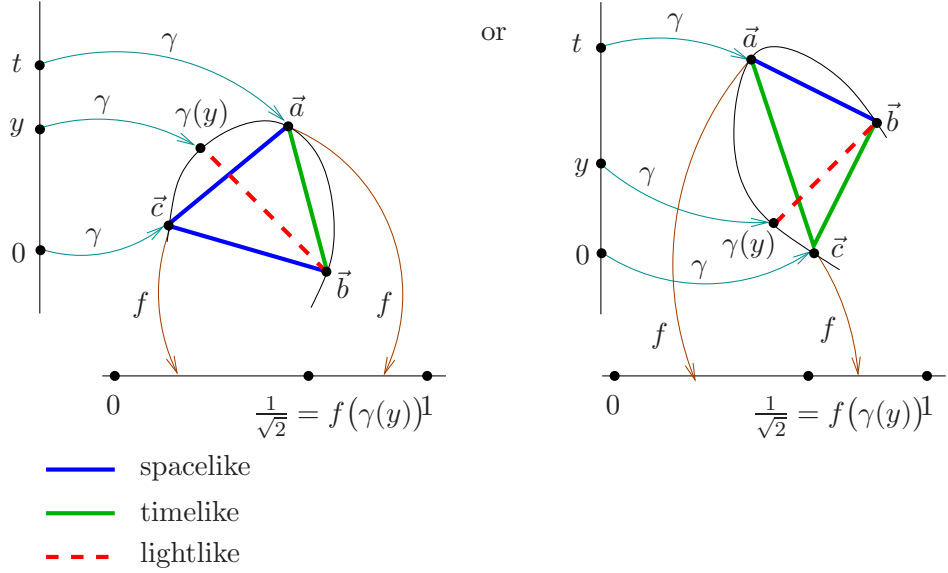}
\caption{\label{keeep} Illustration for the proof of Prop.~\ref{propstl}}
\end{center}
\end{figure}

For every $\vp\in \ran(\gamma)$ by \ax{CONT}, 
there is a  closest $t\in\Q$ to $0$ such that
$\gamma(t)=\vp$, i.e., the set $H\leteq \{\,|x|\: :\: \gamma(x)=\vp\,\}$ has
a minimal element.%
\footnote{That is so because of the following.  Let $s$ be the
  supremum of the nonempty bounded definable set $\{\, -|x|\: :\:
  \gamma(x)=\vp\,\}$.  By the continuity of $\gamma$, one of
  $\gamma(s)$ and $\gamma(-s)$ must be $\vp$.  Then $-s$ is the
  minimal element of $H$.}  Thus there is a $t\in\Q$ such that
$\gamma(t)$ is $\va$ or $\vb$ and there is no $t'$ between $0$ and $t$
such that $\gamma(t')$ is $\va$ or $\vb$.  We can assume that
$\gamma(t)=\va$ and $t> 0$.

Let $f:\Q^d\setminus\{\vb\,\}\rightarrow \Q$ be the function defined as
$\vp\mapsto \frac{|p_\tau-b_\tau|}{|\vp-\vb\,|}$.
It is easy to see that $f$ is continuous
and for all $\vp\in \Q^d\setminus\{\vb\,\}$
\begin{eqnarray}\label{csill}
f(\vpp)={1}/{\sqrt{2}} & \Longleftrightarrow & \vp-\vb \text{ is lightlike},\nonumber\\
f(\vpp)>{1}/{\sqrt{2}} & \Longleftrightarrow & \vp-\vb \text{ is timelike}, \\
 f(\vpp)<{1}/{\sqrt{2}} & \Longleftrightarrow & \vp-\vb \text{ is spacelike}.\nonumber
\end{eqnarray} 
Consider the function $g\leteq \gamma\big|_{[0,t]}\circ f$.
It is a 
continuous function.
Furthermore, $\dom g=[0,t]$ since
there is no $t'\in [0,t]$ such that
$\gamma(t')=b$.
By (\ref{csill}) above and by the fact that
$\gamma(0)=\vc$ and $\gamma(t)=\va$, we have that 
\begin{equation*}
\big(\, g(0)>1/\sqrt{2}\text{ and }g(t)<{1}/{\sqrt{2}}\,\big)\ \text{ or
}\ \big(\, g(0)<{1}/{\sqrt{2}}\text{ and }g(t)>{1}/{\sqrt{2}}\,\big)
\end{equation*} 
since one of the chords $\{\vb,\vc\,\}$ and $\{\vb,\va\,\}$ is timelike and the other is
spacelike.
However, by \ax{CONT}-Bolzano's Theorem, there is a $y\in[0,t]$
such that $g(y)=1/\sqrt{2}$. Hence by (\ref{csill}) above,
we have that $\gamma(y)-\vb$ is lightlike 
for this $y$.
Consequently, $\{\vb,\gamma(y)\}$
is a lightlike chord of $\gamma$.
This contradiction proves our proposition.
\end{proof}

\section{Tools used for simulating gravity by accelerated observers}
\label{lem-sec}
In this section we develop the tools which were used in
Chap.~\ref{chp-grav}. To do so, let us first introduce the following
convenient notation. We say that $\alpha:\Q\parrow \Q$ is a
\df{nice map}\index{nice map} if it is differentiable such that
$0\not\in\ran\alpha'$, and $\dom \alpha$ is connected.

\begin{lem}
\label{lem-tlnice}
Let $\alpha$ be a timelike curve.
Then $\alpha_\tau$ is a nice map.
\end{lem}

\begin{proof} 
Since $\alpha$ is a timelike curve, $\dom\alpha$ is connected and
$\alpha'(x)_\tau\neq 0$ for all $x\in \dom \alpha$.  But
$\dom\alpha_\tau=\dom\alpha$ and $(\alpha_\tau)'=(\alpha')_\tau$.
Thus $\alpha_\tau$ is a nice map.
\end{proof}

\begin{lem}
\label{lem-inj}
Assume \ax{CONT}.
Let $\alpha$ be a definable nice map.
Then $\alpha$ is injective.
Moreover, $\alpha$ is monotonic.
\end{lem}

\begin{proof} 
If $\alpha$ were not injective, then $\alpha'(x)$ would be $0$ for
some $x$ by \ax{CONT}-Rolle's Theorem.  But $\alpha'(x)$ cannot be $0$
since $\alpha$ is a nice map.  Thus $\alpha$ is injective.  Then
$\alpha$ is also monotonic by (2) in Lem.~\ref{lem-injcont}.
\end{proof}
 
\begin{lem}
\label{lem-nice}
Assume \ax{CONT}.
If $\alpha$ and $\delta$ are nice maps, $\delta^{-1}$ and $\alpha\circ\delta$ are also nice maps.\qed
\end{lem}

\begin{lem}
\label{lem-mink} Assume \ax{CONT}.
Let $\alpha$ and $\delta$ be definable timelike curves such that $\ran\alpha\subseteq \ran\delta$ (or $\ran\delta\subseteq \ran\alpha$), and let $h\leteq \alpha\circ\delta^{-1}$.
Then $h$ is a nice map and 
\begin{equation}\label{eq-mink}
|h'(x)|=\frac{\mu\big(\alpha'(x)\big)}{\mu\big(\delta'(h(x))\big)} \quad \text{ for all } \enskip x\in \dom h.
\end{equation}
\end{lem}

\begin{proof} 
First we show that $h=\alpha_\tau\circ\delta_\tau^{-1}$.  Since
$\alpha$ and $\delta$ are definable timelike curves, $\alpha_\tau$ and
$\delta_\tau$ are definable nice maps by Lem.~\ref{lem-tlnice}.  Thus
$\alpha_\tau$ and $\delta_\tau$ are injective by Lem.~\ref{lem-inj}.
Consequently, $\alpha$ and $\delta$ are also injective.  Therefore,
$\langle x,y\rangle \in \alpha_\tau\circ \delta_\tau^{-1}$ iff
$\alpha_\tau(x)=\delta_\tau(y)$ and $\langle x,y\rangle\in
\alpha\circ\delta^{-1}$ iff $\alpha(x)=\delta(y)$.  Since
$\alpha(x)=\delta(y)\rightarrow \alpha_\tau(x)=\delta_\tau(y)$ is clear, we
have to show the converse implication only.  By symmetry, we can
assume that $\ran\alpha\subseteq \ran\delta$.  Then there is a $z\in
\dom\delta$ such that $\delta(z)=\alpha(x)$, so
$\delta_\tau(z)=\alpha_\tau(x)=\delta_\tau(y)$.  Thus $z=y$ since
$\delta$ is injective, so $\alpha(x)=\delta(y)$.  That proves
$h=\alpha_\tau\circ\delta_\tau^{-1}$.

By Lem.~\ref{lem-nice}, $h$ is a nice map, so $\dom h$ is an interval.
We have that $\alpha\supseteq h\circ\delta$ since $h=\alpha\circ\delta^{-1}$.
Thus by Chain Rule, $\alpha'(x)=h'(x)\cdot\delta'\big(h(x)\big)$ for all $x\in \dom h$.
Since $\mu(\lambda\vpp)=|\lambda|\cdot\mu(\vpp)$ for all $\lambda\in \Q$ and $\vpp\in\Q^d$, we have that $\mu\big(\alpha'(x)\big)=|h'(x)|\cdot\mu\big(\delta'(h(x))\big)$ for all $x\in \dom h$.
We have that $\mu\big(\delta'(h(x))\big)\neq0$ since $\delta$ is timelike.
Hence equation \eqref{eq-mink} holds.
\end{proof}

\begin{lem}
\label{lem-main}
Assume \ax{CONT}.
Let $\beta$ and $\gamma$ be definable and well-parametrized timelike curves; let $\beta_*$ and $\gamma_*$ be definable timelike curves; let $x_\beta,y_\beta\in \dom\beta$, $x_\gamma,y_\gamma\in \dom\gamma$ and $x,y\in \dom\beta_*\cap \dom\gamma_*$ such that
\begin{itemize}
\item[(i)] $\ran\beta_*\subseteq \ran\beta$ and $\ran\gamma_*\subseteq \ran\gamma$.
\item[(ii)] $\beta(x_\beta)=\beta_*(x)$, $\beta(y_\beta)=\beta_*(y)$, $\gamma(x_\gamma)=\gamma_*(x)$, $\gamma(y_\gamma)=\gamma_*(y)$.
\item[(iii)] $x\neq y$ and $\mu\big(\gamma'_*(z)\big)>\mu\big(\beta'_*(z)\big)$ for all $z\in (x,y)$.
\end{itemize}
Then $\big|x_\gamma-y_\gamma\big|>\big|x_\beta-y_\beta\big|$.
\end{lem}

\begin{figure}[h!btp]
\small
\begin{center}
\psfrag{yb}[r][r]{$y_\beta$}
\psfrag{xb}[r][r]{$x_\beta$}
\psfrag{y}[r][r]{$y$}
\psfrag{x}[r][r]{$x$}
\psfrag{yc}[r][r]{$y_\gamma$}
\psfrag{xc}[r][r]{$x_\gamma$}
\psfrag{i}[t][t]{$i$}
\psfrag{j}[t][t]{$j$}
\psfrag{bb}[bl][bl]{$\beta_*,\beta$}
\psfrag{cc}[bl][bl]{$\gamma_*,\gamma$}
\psfrag{b}[b][b]{$\beta$}
\psfrag{c}[b][b]{$\gamma$}
\psfrag{bx}[b][b]{$\beta_*$}
\psfrag{cx}[b][b]{$\gamma_*$}
\psfrag{bxy}[l][l]{$$}
\psfrag{cxy}[l][l]{$$}
\psfrag{bxx}[l][l]{$$}
\psfrag{cxx}[l][l]{$$}
\includegraphics[keepaspectratio, width=0.8\textwidth]{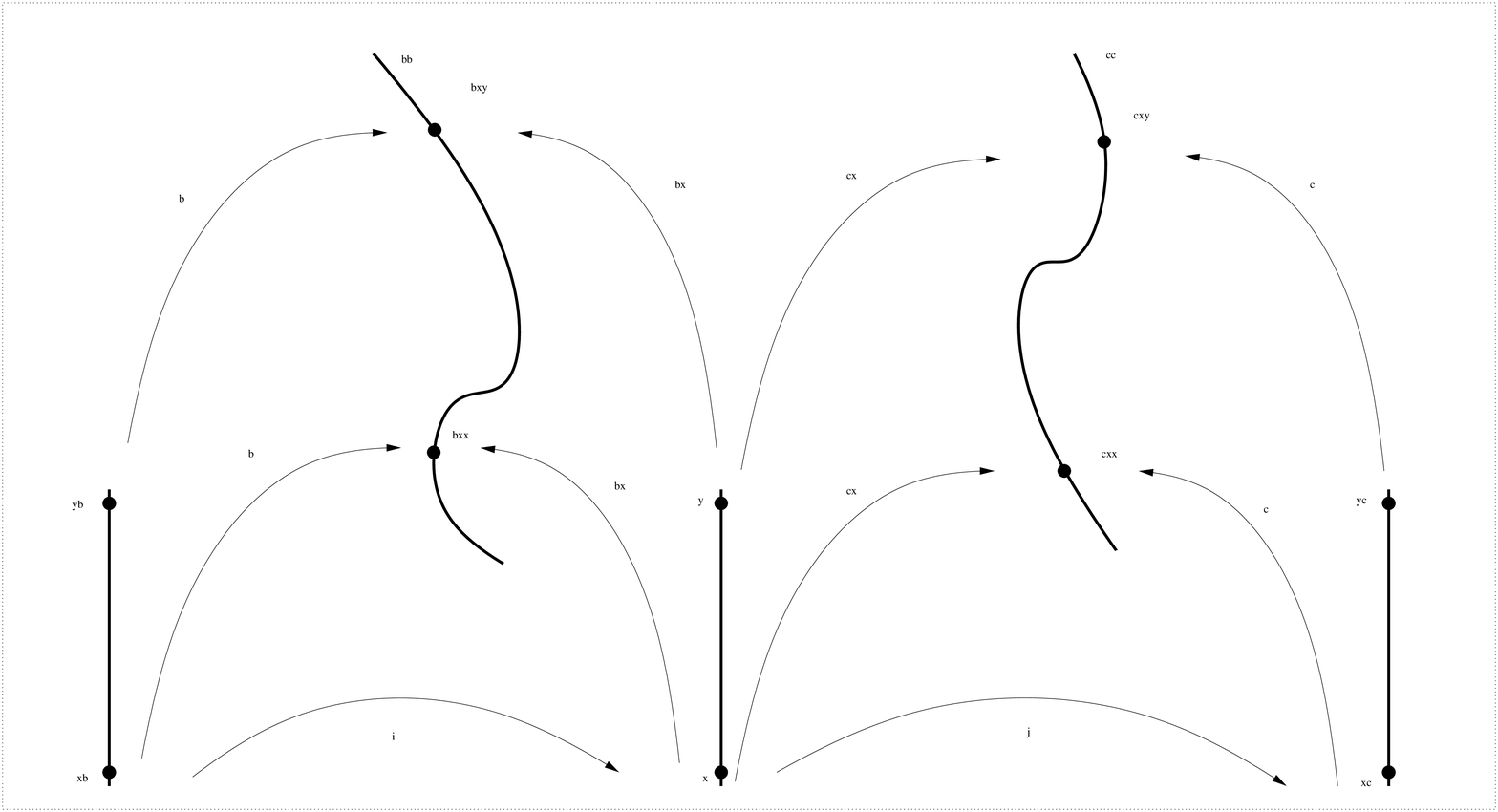}
\caption{\label{fig-lemma} Illustration for the proof of Lem.~\ref{lem-main}}
\end{center}
\end{figure}

\begin{proof}
Since $\beta$, $\beta_*$, $\gamma$ and $\gamma_*$ are definable timelike curves, they are injective by Lems.\ \ref{lem-tlnice} and \ref{lem-inj}.
Thus $x_\beta\neq y_\beta$ and $x_\gamma\neq y_\gamma$ since $x\neq y$.
Let 
\begin{equation*}
i\leteq \beta\circ\beta_*^{-1}\quad\text{and}\quad j\leteq \gamma_*\circ \gamma^{-1},
\end{equation*} 
see Fig.~\ref{fig-lemma}.
Then $i$, $j$ and $i\circ j$ are nice maps by Lem.~\ref{lem-nice} and \ref{lem-mink}.
Furthermore,
\begin{equation*}
i(x_\beta)=x,\enskip i(y_\beta)=y,\enskip j(x)=x_\gamma,\enskip j(y)=y_\gamma,\enskip
(i\circ j)(x_\beta)=x_\gamma\enskip \text{and}\enskip (i\circ j)(y_\beta)=y_\gamma.
\end{equation*}
Since $x_\beta, y_\beta\in \dom i\circ j$, and $i\circ j$ is a nice
map, we have that $(x_\beta,y_\beta)\subseteq \dom i\circ j$.

Now we will show that
\begin{equation*}
\label{megjelolt}
\forall t\in (x_\beta, y_\beta) \quad \big|(i\circ j)'(t)\big|>1.
\end{equation*}
To prove this statement, let $t\in (x_\beta, y_\beta)$.
Since $i$ is a nice map, it is monotonic by Lem.~\ref{lem-inj}, thus $i(t)\in (x,y)$.
By Lem.~\ref{lem-mink} and the fact that $\beta$ and $\gamma$ are well-parametrized, we have that 
\begin{equation}\label{i}
\big|i'(t)\big|=\frac{\mu\big(\beta'(t)\big)}{\mu\big(\beta_*'(i(t))\big)}=\frac{1}{\mu\big(\beta_*'(i(t))\big)}\\ 
\end{equation}
and
\begin{equation}\label{j}
\big|j'\big(i(t)\big)\big|=\frac{\mu\big(\gamma_*'(i(t))\big)}{\mu\big(\gamma'\big(j(i(t))\big)\big)}=\mu\big(\gamma'_*(i(t))\big).
\end{equation}
From equations \eqref{i}, \eqref{j} and Item (iii) by Chain Rule, we
have that
\begin{equation*}
\big|(i\circ j)'(t)\big|=\big|i'(t)\cdot
j'\big(i(t)\big)\big|=\frac{\mu\big(\gamma_*'(i(t))\big)}{\mu\big(\beta'_*(i(t))\big)}>1.
\end{equation*}
This completes the proof of \eqref{megjelolt}.

By \ax{CONT}-Mean--Value Theorem there is a $z\in (x_\beta,y_\beta)$ such that
\begin{equation*}
(i\circ j)'(z)=\frac {(i\circ j)(x_\beta)-(i\circ j)(y_\beta)}{x_\beta-y_\beta}=\frac{x_\gamma-y_\gamma}{x_\beta-y_\beta}\,.
\end{equation*}
By this and \eqref{megjelolt}, we conclude that $\big|\frac{x_\gamma-y_\gamma}{x_\beta-y_\beta}\big|>1$.
Hence $|x_\gamma-y_\gamma|>|x_\beta-y_\beta|$, as it was required.
\end{proof} 

\begin{rem}
Lem.~\ref{lem-main} remains true even if we substitute ``$=$'' or ``$\ge$'' for ``$>$''.
The proof can be achieved by the same substitution in the original proof.
\end{rem}

\begin{prop}$\,$
\label{prop-vmorta}
\begin{enumerate}
\item Let $\alpha$ be a well-parametrized timelike curve.
If $\alpha$ is twice differentiable at $t\in\dom\alpha$, then $\alpha'(t)\mort\alpha''(t)$.
\item Let $d\ge 3$.
Assume \ax{AccRel_0}.
Let $k\in\Ob$ and $m\in \IOb$.
Then $\fvvkm(t)\mort\fvakm(t)$ for all $t\in \dom \fvakm$.
\end{enumerate}
\end{prop}

\begin{proof}
To prove Item (1), let $t\in\dom\alpha$ such that $\alpha$ is twice differentiable at $t$.
Since $\alpha$ is a well-parametrized timelike curve, we have that
\begin{equation}\label{eq-wellp}
\big(\alpha'_1(t)\big)^2-\big(\alpha'_2(t)\big)^2-\ldots-\big(\alpha'_d(t)\big)^2=1.
\end{equation}
By derivation of both sides of equation \eqref{eq-wellp} we have that
\begin{equation*}\label{eq-mort}
2\alpha'_1(t)\cdot\alpha''_1(t)-2\alpha'_2(t)\cdot\alpha''_2(t)-\ldots-2\alpha'_d(t)\cdot\alpha''_d(t)=0.
\end{equation*}
Thus $\alpha'(t)\mort\alpha''(t)$, which is what we wanted to prove.

Item (2) is a consequence of Item (1) since $\fvvkm=(\lc^k_m)'$,
$\fvakm=(\lc^k_m)''$, $\lc^k_m$ is a well-parametrized timelike curve
by Thm.~\ref{thm-wp}, and $\lc^k_m$ is twice differentiable at $t$ iff
$t\in \dom \fvakm$.
\end{proof}

If $f:\Q\parrow\Q$, we abbreviate $f(t)>0$ for all $t\in\dom f$ to $f>0$.
We also use the analogous notation $f<0$.

\begin{lem}$\,$ 
\label{lem-vmon}
\begin{enumerate}
\item Assume \ax{CONT}.
Let $\alpha$ be a definable and twice differentiable timelike curve such that $\ran\alpha\subset \txPlane$.
If $\alpha''\circ\mu<0$, then $\alpha'_2$ is increasing or decreasing.
\item Let $d\ge 3$.
Assume \ax{AccRel}.
Let $k\in\Ob$ and $m\in \IOb$ such that $\wl_m(k)\subset \txPlane$ and $\dom\fvakm=\dom \lc^k_m$.
If $k$ is positively accelerated, $(\fvvkm)_2$ is increasing or decreasing.
\end{enumerate}
\end{lem}

\begin{proof}
To prove Item (1), let $t\in\dom \alpha$.  By
Prop.~\ref{prop-vmorta}, $\alpha''(t)$ is a spacelike vector
since it is Minkowski orthogonal to a timelike one.  Therefore,
$\mu(\alpha''(t))<0$ iff $|\alpha''_\sigma(t)|\neq0$.  Thus, since
$\ran\alpha\subset \txPlane$, we have that $\mu(\alpha''(t))<0$ iff
$\alpha''_2(t)\neq0$.  Thus by \ax{CONT}-Darboux's Theorem,
$\alpha''\circ\mu<0$ iff $\alpha''_2>0$ or $\alpha''_2<0$ since
$\alpha'_2$ is definable and $\dom\alpha'=\dom\alpha$ is connected.
Then by \ax{CONT}-Mean-Value Theorem, $\alpha'_2$ is increasing or
decreasing.

Item (2) is a consequence of Item (1) because of the following.
Let $\alpha=\lc^k_m$.  Then by Thm.~\ref{thm-wp}, $\alpha$ is definable
(well-parametrized) timelike curve.  $\alpha$ twice differentiable
since $\dom\alpha''=\dom\fvakm=\dom \lc^k_m=\dom\alpha$; and
$\ran\alpha\subset \txPlane$ since by \eqref{item-rantr} in
Prop.~\ref{prop-lc}, $\wl_m(k)=\ran \lc^k_m$.  Then $k$ is positively
accelerated iff $\alpha''\circ\mu<0$.  Hence by Item (1), if $k$ is
positively accelerated, $(\fvvkm)_2=\alpha'_2$ is increasing or
decreasing.
\end{proof}

Let us introduce the following notation: 
\begin{equation*}\index{$dw^k_m$}
\Df{dw^k_m}(\vpp)\leteq w^k_m(\vpp)-w^k_m(\vo\,).
\end{equation*}

\begin{prop}
\label{prop-inv} 
Let $d\ge3$.
Assume \ax{SpecRel}.
Let $m,k\in\IOb$ and $h\in \Ob$.
Then
\begin{enumerate}
\item \label{item-mort} $\vpp\mort\vqq$ iff $dw^k_m(\vpp)\mort dw^k_m(\vqq)$.
\item \label{item-veloctransf} 
$\fvvh_m=\fvvh_k \circ dw^k_m$ and $\dom\fvvh_m=\dom\fvvh_k$.
\item \label{item-acctransf} 
$\fvah_m=\fvah_k\circ dw^k_m$ and $\dom\fvah_m=\dom\fvah_k$.
\end{enumerate}
\end{prop}

\begin{proof}
To prove Item \eqref{item-mort}, observe that
$\mu\big(dw^k_m(\vpp)\big)=\mu(\vpp)$ by Thm.~\ref{thm-poi}. The
statement $\vpp\mort\vqq$ iff
$\mu(\vpp+\vqq)^2=\mu(\vpp)^2+\mu(\vqq)^2$ can be proved by
straightforward calculation. Thus Item \eqref{item-mort} is clear
since $dw^k_m$ is linear by Thm.~\ref{thm-poi}.

To prove Items \eqref{item-veloctransf} and \eqref{item-acctransf},
let us note that $\lc^h_m$ and $\lc^h_k$ are functions by Item
\eqref{item-trfunct} in Prop.~\ref{prop-lc}. Thus $\fvvh_m=\fvvh_k
\circ dw^k_m$ follows by Chain Rule because $\lc^h_k= \lc^h_m\circ
w^m_k$ (by \eqref{item-tr} in Prop.~\ref{prop-lc}), the derivative of
$w^k_m$ is $dw^k_m$ (since $w^k_m$ is affine transformation by
Thm.~\ref{thm-poi}), and $\fvvh_x=(\lc^h_x)'$ (by definition). Hence
$\dom\fvvh_m=\dom\fvvh_k$ also holds since $dw^k_m$ is a
bijection. $\fvah_m=\fvah_k\circ dw^k_m$ follows from
\eqref{item-veloctransf} of this proposition by Chain Rule because the
derivative of $dw^k_m$ is $dw^k_m$ (since $dw^k_m$ is a linear
transformation), and $\fvah_x=(\fvvh_x)'$ (by definition). Hence
$\dom\fvah_m=\dom\fvah_k$ also holds since $dw^k_m$ is a bijection.
\end{proof}

The \df{light cone}\index{light cone} of $\vpp\in\Q^d$ is defined as
$\Df{\Lambda}{}_{\vpp}\leteq \setopen \vqq\in\Q^d:\vpp\pheq\vqq
\setclose$.\index{$\Lambda{}_{\vpp}$} The \df{past light
 cone}\index{past light cone} of $\vpp\in\Q^d$ is defined as
$\Df{\Lambda}{}^-_{\vpp}\leteq \setopen\vqq\in\Q^d\setmid \vpp\pheq\vqq
\lland q_\tau\le p_\tau\setclose$.\index{$\Lambda{}^-_{\vpp}$} The
\df{future light cone}\index{future light cone} of $\vpp\in\Q^d$ is
defined as $\Df{\Lambda}{}^+_{\vpp}\leteq \setopen\vqq\in\Q^d\setmid
\vpp\pheq\vqq \lland q_\tau\ge
p_\tau\setclose$.\index{$\Lambda{}^+_{\vpp}$} We say that
$\vpp\in\Q^d$ \df{chronologically precedes}\index{chronologically
 precedes} $\vqq\in\Q^d$, in symbols $\vpp\Df{\ll}\vqq$, iff
$\vpp\teq\vqq$ and $p_\tau<q_\tau$.\index{$\ll$} The \df{chronological
 past}\index{chronological past} of $\vpp\in\Q^d$ is defined as
$\Df{I}{}^-_{\vpp}\leteq \Setopen\vqq\in\Q^d\setmid
\vqq\ll\vpp\Setclose$.\index{$I{}^-_{\vpp}$} The \df{chronological
 future}\index{chronological future} of $\vpp\in\Q^d$ is defined as
$\Df{I}{}^+_{\vpp}\leteq \Setopen\vqq\in\Q^d\setmid \vpp\ll\vqq
\Setclose$.\index{$I{}^+_{\vpp}$} The \df{chronological
 interval}\index{chronological interval} between $\vpp\in\Q^d$ and
$\vqq\in\Q^d$ is defined as $\Df{\llangle\vpp, \vqq\rrangle}\leteq
\setopen \vr\in\Q^d\setmid \vpp\teq\vr \lland \vqq\teq\vr \lland r_\tau\in
(p_\tau,q_\tau)\setclose$.\index{$\llangle\vpp, \vqq\rrangle$} We also
use the notation $\Df{I}{}_{\vpp}\leteq I^-_{\vpp}\cup I^+_{\vpp} \cup
\{\vpp\}$.\index{$I{}_{\vpp}$}

\begin{lem}
\label{lem-cau}
Let $\vpp,\vqq\in\Q^d$.
Then
\begin{enumerate}
\item \label{item-distjcones} If $\vpp\teq \vqq$, then $\Lambda^-_{\vpp}\cap \Lambda^-_{\vqq}=\Lambda^+_{\vpp}\cap \Lambda^+_{\vqq}=\emptyset$.
\item If $\vpp\ll\vqq$, then $\Lambda^-_{\vqq}\cap I^-_{\vpp}=\emptyset$, and $\Lambda^-_{\vpp}\cup I^-_{\vpp}\subset I^-_{\vqq}$.
\item $\vpp\ll\vqq$ iff $I^+_{\vpp}\cap I^-_{\vqq}\neq\emptyset$.\qed
\end{enumerate}
\end{lem}

\begin{lem}
\label{lem-chord}
Assume \ax{CONT}.
Let $\gamma$ be a definable timelike curve, and let $x,y\in \dom\gamma$ such that $x\neq y$.
Then
\begin{enumerate}
\item\label{item-chord} All the chords of $\gamma$ are timelike, i.e., $\gamma(x)\teq\gamma(y)$.
\item\label{item-pcone} If $\gamma(x)\in I^-_{\vpp}$ and $\gamma(y)\not\in I^-_{\vpp}$, there is a $z\in[x,y]$ such that $\gamma(z)\in\Lambda^-_{\vpp}$.
\item\label{item-fcone} If $\gamma(x)\in I^+_{\vpp}$ and $\gamma(y)\not\in I^+_{\vpp}$, there is a $z\in[x,y]$ such that $\gamma(z)\in\Lambda^+_{\vpp}$.
\item\label{item-cord} If $\gamma_\tau$ is increasing (decreasing), $\gamma(x)\ll\gamma(y)$ iff $x<y$ ($y<x$).
\item\label{item-cinv} $z\in(x,y)$ iff $\gamma(z)\in \llangle \gamma(x),\gamma(y)\rrangle$.
\end{enumerate}
\end{lem}
\begin{proof}
 Item \eqref{item-chord} follows from Prop.~\ref{propstl}.
To prove Item \eqref{item-pcone}, let 
\begin{equation*}
H\leteq \Setopen t\in[x,y]\setmid \mu\big(\gamma(t),\vpp\big)<0 \lland \gamma(t)_\tau<p_\tau\Setclose.
\end{equation*}
It is clear that $H\subseteq \dom \gamma$ is definable, bounded and
nonempty.  Let $z\leteq \sup H$ which exists by \ax{CONT}.  Thus by
continuity of $t\mapsto \mu(\gamma(t),\vpp)$ and $\gamma_\tau$, we
have that $\gamma(z)\not\in I^-_{\vpp}$ since $z$ is an upper bound of
$H$. Furthermore, $\mu(\gamma(t),\vpp)\le0$ and $\gamma(t)_\tau\le
p_\tau$ since $z$ is the least upper bound of $H$.  But
$\gamma(t)_\tau=p_\tau$ and $\mu(\gamma(t),\vpp)<0$ is impossible.
Thus $\gamma(t)_\tau\le p_\tau$ and $\mu(\gamma(t),\vpp)=0$.  Hence
$\gamma(\vpp)\in\Lambda^-_{\vpp}$.

Item \eqref{item-fcone} is clear from Item \eqref{item-pcone} since the continuous bijection $\vpp\mapsto -\vpp$ takes $I^+_{\vpp}$ to $I^-_{\vpp}$ and $\Lambda^+_{\vpp}$ to $\Lambda^-_{\vpp}$.

Item \eqref{item-cord} is clear by Item \eqref{item-chord}.

Item \eqref{item-cinv} is a consequence of Item \eqref{item-cord} since $\gamma_\tau$ is either increasing or decreasing by Lems.\ \ref{lem-inj} and \ref{lem-nice}.
\end{proof}

\noindent
We use the following notations: 
\begin{equation*}\index{$Cone_{\varepsilon}(\vpp;\vqq)$}\index{$\Lambda^-[H]$}
\Df{Cone_{\varepsilon}(\vpp;\vqq)}\leteq \bigcup_{\vr\,\in B_\varepsilon(\vqq)}line(\vpp,\vr\,) \quad \text{ and }\quad
\Df{\Lambda^-[H]}\leteq \bigcup_{\vpp\in H}\Lambda^-_{\vpp}.
\end{equation*}

Let $\alpha$ and $\beta$ be timelike curves.
We say that $\beta_*$ is the \df{photon reparametrization of $\beta$ according to $\alpha$}\index{photon reparametrization} if 
\begin{equation*}
\beta_*=\setopen \langle t,\vpp\rangle\in\dom\alpha\times\ran\beta \setmid\vpp\in\Lambda^-_{\alpha(t)}\setclose.
\end{equation*}

\begin{prop}
\label{prop-ph}
Assume \ax{CONT}.
Let $\alpha$ and $\beta$ be definable timelike curves.
Let $\beta_*$ be the photon reparametrization of $\beta$ according to $\alpha$.
\begin{enumerate} 
\item \label{item-phcont} Then $\beta_*$ is a definable, continuous and injective curve.
\item \label{item-phder} If $\ran\alpha\cap\ran\beta=\emptyset$, and $\ran \alpha\cup \ran \beta$ is in a vertical plane, $\beta_*$ is a timelike curve, and $\beta_*(t_0)+\beta'_*(t_0)\in\Lambda^-_{\alpha(t_0)+\alpha'(t_0)}$.
\end{enumerate}
\end{prop}

\begin{figure}[h!btp]
\small
\begin{center}
\psfrag{p}[l][l]{$\vpp$}
\psfrag{q}[bl][bl]{$\vqq$}
\psfrag{r}[l][l]{$\vr$}
\psfrag{ph}[l][l]{$ph$}
\psfrag{e}[tr][tr]{$\varepsilon$}
\psfrag{e1}[tr][tr]{$\varepsilon_1$}
\psfrag{e2}[bl][bl]{$\varepsilon_2$}
\psfrag{a}[l][l]{$\alpha$}
\psfrag{a'0}[bl][bl]{$\alpha'(t_0)$}
\psfrag{b}[tr][tr]{$\beta,\beta_*$}
\psfrag{b'0}[br][br]{$\beta'(\bar{t}_0)$}
\psfrag{a0}[l][l]{$\alpha(t_0)$}
\psfrag{b0}[l][l]{$\beta(\bar{t}_0)=\beta_*(t_0)$}
\psfrag{adif}[bl][bl]{$\frac{\alpha(t)-\alpha(t_0)}{t-t_0}$}
\psfrag{b*'0}[br][br]{$\beta_*'(t_0)$}
\psfrag{b*t}[tl][tl]{$\beta_*(t)$}
\psfrag{at}[l][l]{$\alpha(t)$}
\includegraphics[keepaspectratio, width=0.9\textwidth]{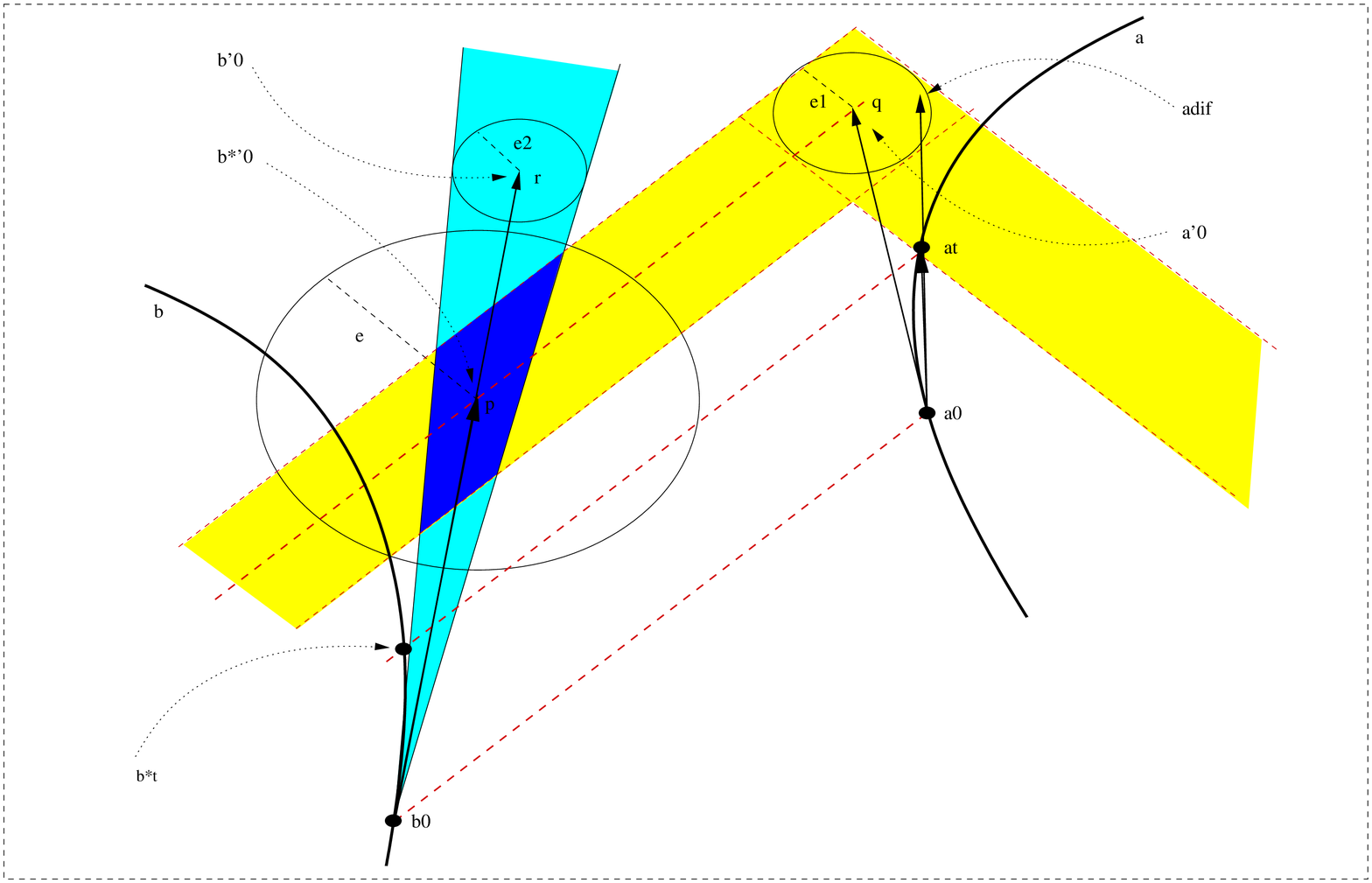}
\caption{\label{figphpar} Illustration for the proof of Prop.~\ref{prop-ph}}
\end{center}
\end{figure}

\begin{proof}
It is clear that $\beta_*$ is definable.

To show that $\beta_*$ is a function, we need to prove that
$\Lambda^-_{\alpha(t)}\cap\ran \beta$ has one element at the most for
all $t\in\dom \alpha$.  It is clear by Lem.~\ref{lem-chord} since if
it had two distinct elements, say $\vpp$ and $\vqq$, then $\setopen
\vpp,\vqq\setclose$ would not be a timelike chord of $\beta$, but a
lightlike one.

For all $t\in\dom\beta_*$, let $\bar{t}\in\dom\beta$ such that
$\beta(\bar t\,)=\beta_*(t)$, and let $f:t \mapsto \bar t$ be the
reparametrization map, i.e., $f\leteq \beta_*\circ\beta^{-1}$.  First
we show that
\begin{equation*} t\in (t_1,t_2) \iff
\alpha(t)\in\llangle \alpha(t_1),\alpha(t_2)\rrangle \stackrel{(*)}{\iff}
\beta_*(t)\in\llangle \beta_*(t_1),\beta_*(t_2)\rrangle
\iff \bar{t}\in (\bar{t}_1,\bar{t}_2)
\end{equation*}
if $t,t_1,t_2\in\dom\beta_*$.  The first and the last equivalence are
clear by \eqref{item-cinv} in Lem.~\ref{lem-chord} since $\alpha$,
$\beta$ are timelike curves and $\beta(\bar{t}\,)=\beta_*(t)$ for all
$t\in\dom\beta_*$.  To prove $(*)$, we can assume that
$\alpha(t_1)\ll\alpha(t)\ll\alpha(t_2)$.  Thus
$\beta_*(t)\ll\alpha(t_2)$ since $\Lambda^-_{\alpha(t)}\subset
I^-_{\alpha(t_2)}$ (2) of by Lem.~\ref{lem-cau}.  Therefore,
$\beta_*(t)\ll\beta_*(t_2)$ since $\beta_*(t)\in I_{\beta_*(t_2)}$ by
\eqref{item-chord} in Lem.~\ref{lem-chord}, but $I^+_{\beta_*(t_2)}\cap
I^-_{\alpha(t_2)}=\emptyset$ (3) by Lem.~\ref{lem-cau}.  A similar
argument can show that $\beta_*(t_1)\ll\beta_*(t)$, so $(*)$ is
proved.  Now we have that $f$ preserves betweenness, so it is
monotonic.

To show that $\dom \beta_*$ is connected, let $x,y\in\dom\beta_*$, and let $z\in(x,y)$.
Then $z\in\dom\alpha$ since $x,y\in\dom \alpha$ and $\dom\alpha$ is connected.
Since $\alpha$ is a timelike curve, $\alpha(z)\in\llangle \alpha(x),\alpha(y)\rrangle$.
Without losing generality, we can assume that $\alpha(x)\ll \alpha(z)\ll\alpha(y)$.
Then $\beta_*(x)\in I^-_{\alpha(z)}$ since $\beta_*(x)\in\Lambda^-_{\alpha(x)}\subset I^-_{\alpha(z)}$;
and $\beta_*(y)\not\in I^-_{\alpha(z)}$ since $\beta_*(y)\in\Lambda^-_{\alpha(y)}$ and $\Lambda^-_{\alpha(y)}\cap I^-_{\alpha(z)}=\emptyset$, see Lem.~\ref{lem-cau}.
Then by \eqref{item-pcone} in Lem.~\ref{lem-chord}, there is a $\widehat{z}\in\dom\beta$ such that $\beta(\widehat{z})\in\Lambda^-_{\alpha(z)}$ since $\beta(\bar{x})\in I^-_{\alpha(z)}$ and $\beta({\bar{y}})\not\in I^-_{\alpha(z)}$.
Thus $\langle z,\beta(\widehat{z})\rangle\in\beta_*$.
Consequently, $z\in\dom\beta_*$.
Hence $\dom\beta_*$ is connected.

Now using a similar argument, we show that $\ran f\subseteq\dom\beta$
is also connected.  To do so, let $\bar{x},\bar{y}\in\ran f$ and
$\widehat{z}\in(\bar{x},\bar{y})$.  Then $\widehat{z}\in\dom \beta$.
We can assume that
$\beta(\bar{x})\ll\beta(\widehat{z})\ll\beta((\bar{y})$.  Then
$\alpha(x)\in I^+_{\beta(\widehat{z})}$ and $\alpha(y)\not\in
I^+_{\beta(\widehat{z})}$.  Thus there is a $z\in\dom\alpha$ such that
$\alpha(z)\in\Lambda^+_{\beta(\widehat{z})}$.  Consequently,
$\beta(\widehat{z})\in\Lambda^-_{\alpha(z)}$, so $\langle
z,\beta(\widehat{z})\rangle\in\beta_*$.  Therefore,
$\widehat{z}\in\ran f$, and hence $\ran f$ is connected.

Since $\ran f$ is connected and $f$ is monotonic, $f$ must be
continuous by Lem.~\ref{lem-moncont}.  Hence $\beta_*=f\circ\beta$ is
also continuous and $\beta_*$ injective since both $\beta$ and $f$ are
such.  So Item \eqref{item-phcont} is proved.

To prove Item \eqref{item-phder}, let $\vqq=\alpha'(t_0)+\alpha(t_0)$, $\vr=\beta'(\bar{t}_0)+\beta(\bar{t}_0)$, and let $\vpp$ be the unique element of $\Lambda^-_{\vqq}\cap line(\beta(\bar{t}_0),\vr\,)$, see Fig.~\ref{figphpar}.
We will show that $\beta'_*(t_0)=\vpp-\beta_*(t_0)$.
To do so, let $\varepsilon\in\Q^+$ be fixed.
We have to show that there is a $\delta\in\Q^+$ such that $\frac{\beta_*(t)-\beta_*(t_0)}{t-t_0}\in B_\varepsilon(\vpp)$ if $t\in\dom\beta_*\cap B_\delta(t_0)$.
It is clear that we can choose $\varepsilon_1$ and $\varepsilon_2$ such that 
\begin{equation}\label{eq-inball}
\Lambda^-[B_{\varepsilon_1}(\vqq)]\cap Cone_{\varepsilon_2}(\beta_*(t_0);\vr\,)\subset B_\varepsilon(\vpp).
\end{equation}
Since $\alpha$ is differentiable at $t_0$, there is a $\delta_1\in \Q^+$ such that
\begin{equation}\label{eq-alpha}
\frac{\alpha(t)-\alpha(t_0)}{t-t_0}+\alpha(t_0)\in B_{\varepsilon_1}(\vqq)
\end{equation}
if $t\in\dom\alpha\cap B_{\delta_1}(t_0)$.
Since $\ran\beta\cap\ran\alpha=\emptyset$, and $\ran\beta\cup\ran\alpha$ is in a vertical plane, $line\big(\beta_*(t),\alpha(t)\big)$ and $line\big(\beta_*(t_0),\alpha(t_0)\big)$ are parallel.
Hence
\begin{equation*}
\frac{\beta_*(t)-\beta_*(t_0)}{t-t_0}+\beta_*(t_0)\in \Lambda^-_{\frac{\alpha(t)-\alpha(t_0)}{t-t_0}+\alpha(t_0)}.
\end{equation*}
Thus by \eqref{eq-alpha}, we have that 
\begin{equation}\label{eq-inpast}
\frac{\beta_*(t)-\beta_*(t_0)}{t-t_0}+\beta_*(t_0)\in \Lambda^-[B_{\varepsilon_1}(\vqq)]
\end{equation}
if $t\in\dom\beta_*\cap B_{\delta_1}(t_0)$.
Since $\beta$ is differentiable at $\bar{t}_0$, there is a $\bar{\delta}_2\in \Q^+$ such that 
\begin{equation} \label{eq-beta}
\frac{\beta(\bar{t}\,)-\beta(\bar{t}_0)}{\bar{t}-\bar{t}_0}+\beta(\bar{t}_0)\in B_{\varepsilon_2}(\vr\,)
\end{equation}
if $\bar{t}\in\dom\beta\cap B_{\bar{\delta}_2}(\bar{t}_0)$.
Since $f:t\mapsto\bar{t}$ is continuous, there is a $\delta_2\in\Q^+$ such that \eqref{eq-beta} holds if $t\in\dom\beta_*\cap B_{\delta_2}(t_0)$.
Since
\begin{equation*}
\frac{\beta_*(t)-\beta_*(t_0)}{t-t_0}=\frac{\beta(\bar{t}\,)-\beta(\bar{t}_0)}{\bar{t}-\bar{t}_0}\cdot\frac{\bar{t}-\bar{t}_0}{t-t_0},
\end{equation*}
we have that 
\begin{equation}\label{eq-incone}
\frac{\beta_*(t)-\beta_*(t_0)}{t-t_0}+\beta_*(t_0)\in Cone_{\varepsilon_2}(\beta_*(t_0);\vr\,)
\end{equation}
if $t\in\dom\beta_*\cap B_{\delta_2}(t_0)$.  Let
$\delta=\min(\delta_1,\delta_2)$.  Therefore, by equations
\eqref{eq-inpast} and \eqref{eq-incone}, we have that
\begin{equation*}
\frac{\beta_*(t)-\beta_*(t_0)}{t-t_0}+\beta_*(t_0)\in \Lambda^-[B_{\varepsilon_1}(\vqq)]\cap Cone_{\varepsilon_2}(\beta_*(t_0);\vr\,)
\end{equation*}
if $t\in\dom\beta_*\cap B_{\delta}(t_0)$.
But the latter is a subset of $B_\varepsilon(\vpp)$ by equation \eqref{eq-inball}.
Consequently,
\begin{equation*}
\frac{\beta_*(t)-\beta_*(t_0)}{t-t_0}+\beta_*(t_0)\in B_\varepsilon(\vpp)
\end{equation*}
if $t\in\dom\beta_*\cap B_{\delta}(t_0)$.
Hence $\beta_*$ is differentiable at $t_0$ and $\beta'_*(t_0)=\vp-\beta_*(t_0)$, as it was required.
\end{proof}

\begin{figure}[h!btp]
\small
\begin{center}
\psfrag{a}[l][l]{$\alpha$} \psfrag{b}[l][l]{$\beta$}
\psfrag{t}[r][r]{$t$} \psfrag{l+1}[r][r]{$\lambda+1$}
\psfrag{l}[r][r]{$\lambda$} \psfrag{-l}[l][l]{$-\lambda$}
\psfrag{1-l}[l][l]{$1-\lambda$} \psfrag{o}[tl][tl]{$\vo$}
\psfrag{p}[r][r]{$\langle -\lambda, -1, 0, \ldots, 0 \rangle$}
\psfrag{q}[l][l]{$\langle \lambda, 1, 0, \ldots, 0 \rangle$}
\psfrag{qt}[l][l]{$\langle \frac{\lambda}{\lambda+1}\cdot t,
  \frac{1}{\lambda+1}\cdot t , 0, \ldots, 0 \rangle$}
\includegraphics[keepaspectratio, width=0.7\textwidth]{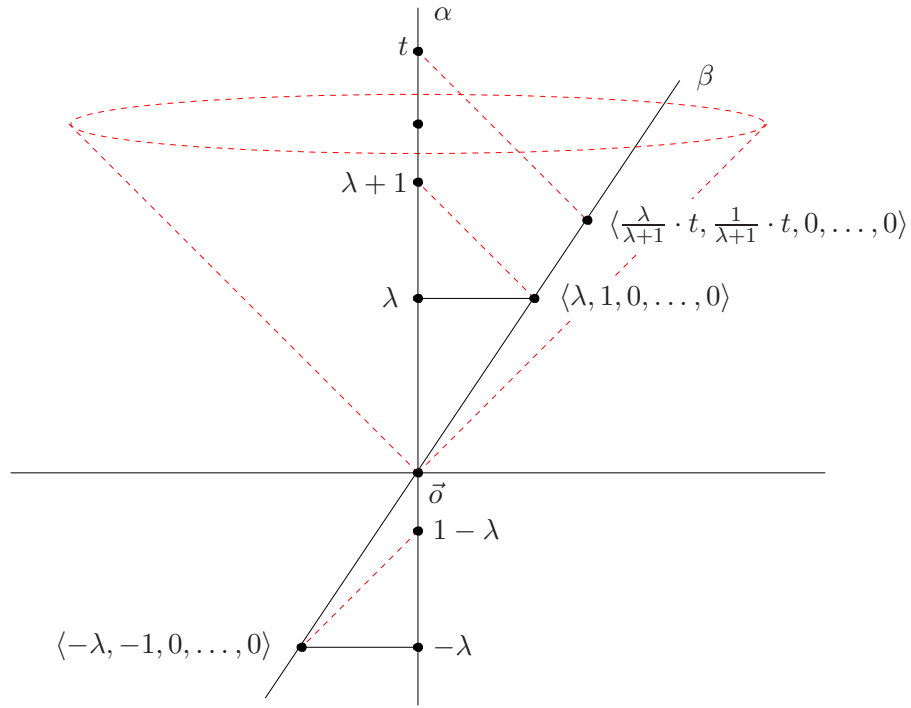}
\caption{\label{fig-nondif} Illustration of Example \ref{xmpl-nondif}}
\end{center}
\end{figure}

The following example shows that the assumption $\ran\alpha\cap\ran\beta=\emptyset$ is necessary in item 2 in Prop.~\ref{prop-ph}.
\begin{example}\label{xmpl-nondif}
Let $\lambda\in\Q$ such that $1<\lambda$ or $\lambda<-1$.  Let
timelike curves $\alpha$ and $\beta$ be defined as $\alpha(t)=\langle
t,0,\ldots, 0\rangle$ and $\beta(t)=\langle \lambda \cdot t, t, 0,
\ldots, 0\rangle$ for all $t\in\Q$.  Then the photon reparametrization
of $\beta$ according to $\alpha$ is:
\begin{equation*}
\beta_*=\left\{
\begin{array}{ll}
\langle\frac{\lambda}{\lambda+1}\cdot t, \frac{1}{\lambda+1}\cdot t, 0, \ldots, 0 \rangle &t\ge0\\
\langle\frac{\lambda}{\lambda-1}\cdot t, \frac{1}{\lambda-1}\cdot t, 0, \ldots, 0 \rangle &t\le 0,
\end{array}
\right.
\end{equation*}
see Fig.~\ref{fig-nondif}.  Therefore, $\beta_*$ is continuous, but
it is not differentiable at $t=0$.
\end{example}

Let $\vpp,\vqq\in \txPlane$.
Then the \df{photon sum}\index{photon sum} of $\vpp$ and $\vqq$, in symbols $\vpp\Df{\phsum}\vqq$\index{$\phsum$}, is the intersection of the two photon lines $\setopen \vpp+\langle A,A,0,\ldots, 0\rangle:A\in\Q \setclose$ and $\setopen \vqq+\langle B,-B,0,\ldots, 0\rangle:B \in\Q \setclose$.

\begin{figure}[h!btp]
\small
\begin{center}
\psfrag{ph1}[l][l]{$ph_1$}
\psfrag{ph2}[r][r]{$ph_2$}
\psfrag{p2}[r][r]{$\left\langle \frac{p_\tau-p_2}{2}, -\frac{p_\tau-p_2}{2}, \vo\,\right\rangle$}
\psfrag{p}[br][br]{$\langle p_\tau, p_2,\vo\,\rangle=\vpp$}
\psfrag{q1}[l][l]{$\left\langle\frac{q_\tau+q_2}{2}, \frac{q_\tau+q_2}{2}, \vo\,\right\rangle$}
\psfrag{q}[bl][bl]{$\vqq=\langle q_\tau, q_2,\vo\,\rangle$}
\psfrag{qlp}[l][l]{$\vqq\phsum\vpp$}
\psfrag{plq}[l][l]{$\vpp\phsum\vqq$}
\psfrag{t}[l][l]{$t$}
\psfrag{x}[l][l]{$x$}
\psfrag{pt}[bl][bl]{$p_\tau$}
\psfrag{ps}[t][t]{$p_\sigma$}
\psfrag{qt}[tr][tr]{$q_\tau$}
\psfrag{qs}[t][t]{$\vq_\sigma$}
\includegraphics[keepaspectratio, width=0.8\textwidth]{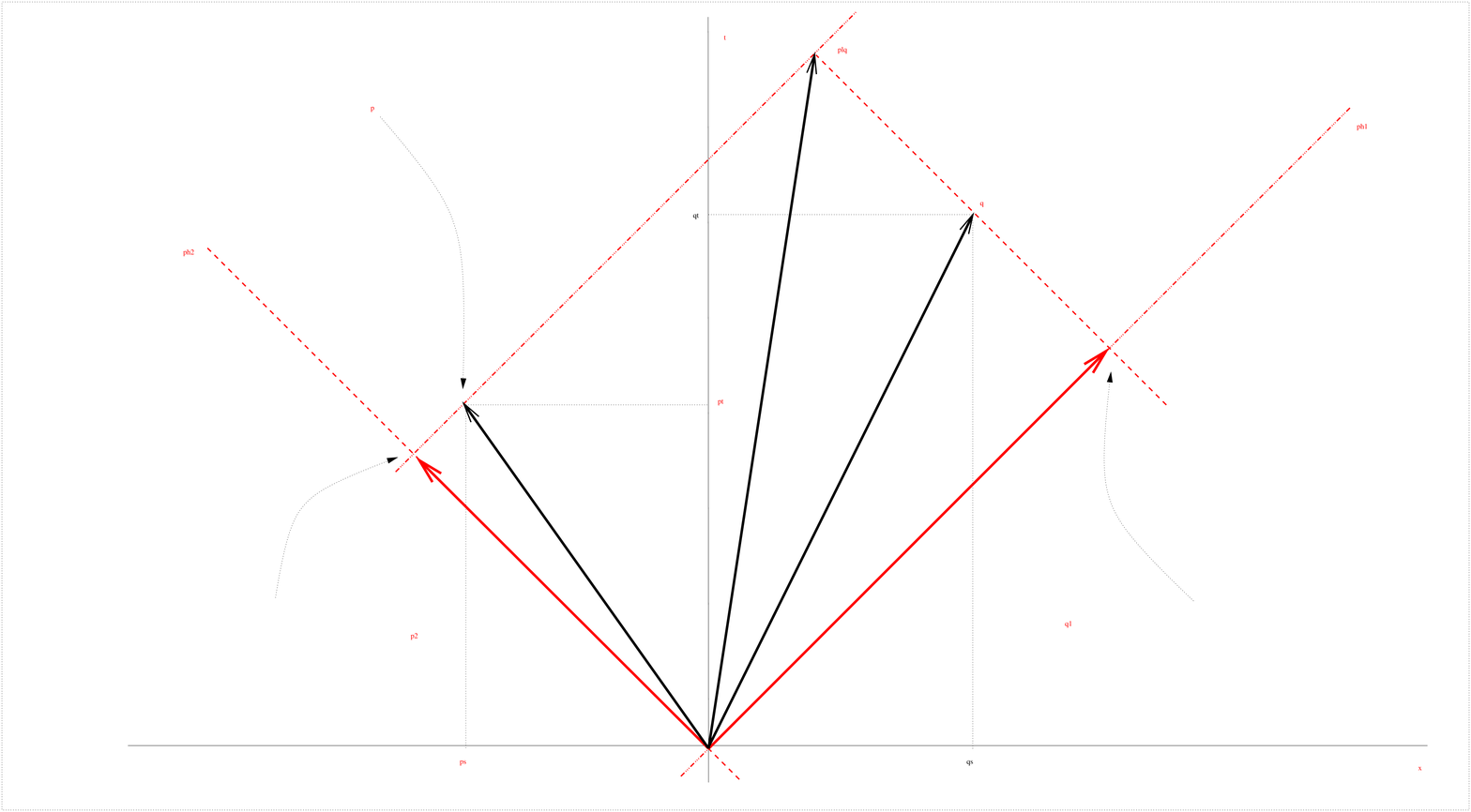}
\caption{\label{figphcrd} Illustration of the photon sum
  $\vpp\phsum\vqq$, and for the proof of Lem.~\ref{lem-phsum}}
\end{center}
\end{figure}

\begin{lem}
\label{lem-phsum} 
Let $\vpp,\vqq\in \txPlane$, and let $a=\frac{q_\tau+q_2}{2}$ and $b=\frac{p_\tau-p_2}{2}$.
Then $\vpp\phsum\vqq=\langle a+b,a-b,0,\ldots,0\rangle$.
\end{lem}
\begin{proof}
The proof is straightforward by the respective definitions, see Fig.~\ref{figphcrd}.
\end{proof}

\begin{lem}
\label{lem-tlinvcont}
Assume \ax{CONT}.
Let $\beta$ be a definable timelike curve.
Then $\beta^{-1}:\ran\beta\rightarrow \dom\beta$ is definable, injective and continuous.
\end{lem}

\begin{proof}
It is clear that $\beta^{-1}$ is definable and injective.

Since by Lems.\ \ref{lem-tlnice} and \ref{lem-inj} $\beta$ is
injective, $\beta^{-1}$ is a function from $\ran\beta$ to $\dom\beta$.
To prove that it is also continuous, let $t_0\in\dom\beta$.  We have
to show that for all $\varepsilon \in \Q^+$, there is a $\delta\in\Q^+$
such that if $t\in \dom\beta$ and $|\beta(t)-\beta(t_0)|<\delta$, then
$|t-t_0|<\varepsilon$.  By Lem.~\ref{lem-chord},
$t\in(t_0-\varepsilon,t_0+\varepsilon)$ iff $\beta(t)\in\llangle
\beta(t_0-\varepsilon),\beta(t_0+\varepsilon)\rrangle$.  Thus, since
$\llangle \beta(t_0-\varepsilon),\beta(t_0+\varepsilon)\rrangle$ is an
open set, there is a good $\delta$.
\end{proof}

\begin{lem}
\label{lem-repar}
Assume \ax{CONT}.
Let $\beta$ be a definable timelike curve and $\beta_*$ a definable continuous curve such that $\ran \beta_*\subseteq \ran\beta$, and let $f\leteq \beta_*\circ\beta^{-1}$.

\begin{enumerate}
\item Then $f$ is a definable and continuous function.
\item If $\beta_*$ is injective, $f$ is also injective.
Moreover, $\dom f$ and $\ran f$ are connected and $f^{-1}$ is also a definable, monotonic and continuous function.
\item If $\beta_*$ is differentiable such that $\beta_*'(t)\neq \vo$ for all $t\in\dom\beta$, then $f$ is injective and differentiable, and $f'(t)\neq0$.
Hence $f^{-1}$ is also a differentiable function.
\end{enumerate}
\end{lem}

\begin{proof}
Item (1) is clear by Lem.~\ref{lem-tlinvcont}.

Item (2) is clear by Item (1) and Lem.~\ref{lem-injcont} since $\dom f=\dom \beta_*$ which is connected.

To prove Item (3), let $t_0\in\dom f$.
Since $\ran\beta_*\subseteq\ran\beta$, we have that there is a $\lambda\in\Q$ such that $\lambda\cdot\beta'(t_0)=\beta'_*(f(t_0))$.
Since $\big(f(t)-f(t_0)\big)/(t-t_0)$ is the ratio of parallel vectors 
\begin{equation*}
\frac{\beta(t)-\beta(t_0)}{t-t_0} \quad\text{ and }\quad \frac{\beta_*\big(f(t)\big)-\beta_*\big(f(t_0)\big)}{f(t)-f(t_0)},
\end{equation*} 
we have that $\big(f(t)-f(t_0)\big)/(t-t_0)$ tends to $\beta'(t_0)/\beta'_*(f(t_0))=1/\lambda$ if $t$ tends to $t_0$.
Thus $f$ is differentiable, and $f'(t_0)=1/\lambda$.
\end{proof}

\begin{lem}
\label{lem-cordmap}
Assume \ax{CONT}.
Let $\alpha$ be a definable timelike curve.
Let $t\in\dom\alpha$ and $x=\alpha_\tau(t)$.
Let $f_\alpha\leteq \alpha^{-1}_\tau\circ\alpha_\sigma$.
\begin{enumerate}
\item Then $f_\alpha$ is a differentiable curve, and $f'_\alpha(x)=\alpha'_\sigma(t)/\alpha'_\tau(t)$.
\item If $\alpha$ is twice differentiable at $t$, then so is $f_\alpha$ at $x$, and 
\begin{equation*}f''_\alpha(x)=\frac{\alpha'_\tau(t)\alpha''_\sigma(t)-\alpha''_\tau(t)\alpha'_\sigma(t)}{\alpha'_\tau(t)^3}.\end{equation*}
\end{enumerate}
\end{lem}

\begin{proof}
Let us first prove Item (1). We have that $\alpha_\tau$ is injective
by Lems.\ \ref{lem-tlnice} and \ref{lem-inj}. Hence $f_\alpha$ is a
function. $\dom f_\alpha$ is connected since $\dom
f_\alpha=\ran\alpha_\tau$ and $\ran\alpha_\tau$ is connected by
Lem.~\ref{lem-injcont}. Thus $f_\alpha$ is a curve. Since $\alpha_\tau$
is an injective differentiable curve, $\alpha^{-1}_\tau$ is also such
and $(\alpha^{-1}_\tau)'(x)=1/\alpha'_\tau(t)$. Thus by Chain Rule, we
have that $f'_\alpha(x)=\alpha'_\sigma(t)/\alpha'_\tau(t)$.

Now let us prove Item (2).
If $\alpha$ is twice differentiable at $t$, then so are $\alpha_\sigma$ and $\alpha_\tau$.
By Item (1), $f'_\alpha=\alpha^{-1}_\tau\circ\alpha'_\sigma/\alpha'_\tau$.
Thus $f_\alpha$ is twice differentiable at $x$ and a straightforward calculation based on the rules of differential calculus can show that $f''_\alpha(x)$ is what was stated.
\end{proof}

\begin{lem}
\label{lem-twocurve}
Assume \ax{CONT}.
Let $\alpha$ and $\beta$ be definable timelike curves such that $\ran\alpha\cup\ran\beta$ is in a vertical plane.
Let $t_1,t_2\in\dom\alpha$ and $\bar{t}_1,\bar{t}_2\in\dom\beta$ such that $\alpha(t_1)\seq\beta(\bar{t}_1)$, $\alpha(t_2)\seq\beta(\bar{t}_2)$ and $\big(\beta(\bar{t}_1)-\alpha(t_1)\big)\upp\big(\alpha(t_2)-\beta(\bar{t}_2)\big)$.
Then there is  a $t\in(t_1,t_2)$ such that $\alpha(t)\in\ran\beta$.
Hence $Ran\alpha\cap\ran\beta\neq\emptyset$.
\end{lem}

\begin{proof}
Since $\ran\alpha\cup\ran\beta$ is in a vertical plane, we can assume,
without losing generality, that $d=2$. By \ax{CONT}-Bolzano's Theorem,
we can also assume that $\alpha(t_1)_\tau=\beta(\bar{t}_1)_\tau$ and
$\alpha(t_2)_\tau=\beta(\bar{t}_2)_\tau$. Let $x_1=\alpha(t_1)_\tau$
and $x_2=\alpha(t_2)_\tau$. Let $f_\alpha\leteq
\alpha_\tau^{-1}\circ\alpha_\sigma$ and $f_\beta\leteq
\beta_\tau^{-1}\circ\beta_\sigma$. Then $f_\alpha$ and $f_\beta$ are
continuous curves, see Lem.~\ref{lem-cordmap}. By the assumption
$\big(\beta(\bar{t}_1)-\alpha(t_1)\big)\upp\big(\alpha(t_2)-\beta(\bar{t}_2)\big)$,
we have that
$\big(f_\beta(x_1)-f_\alpha(x_1)\big)\big(f_\alpha(x_2)-f_\beta(x_2)\big)<0$.
Thus by \ax{CONT}-Bolzano's Theorem, there is an $x\in(x_1,x_2)$ such
that $f_\alpha(x)=f_\beta(x)$. Let $t\leteq \alpha_\tau^{-1}(x)$.
Then $\alpha(t)\in\ran\beta$.
\end{proof}

Let $\alpha$ and $\beta$ be timelike curves.
We say that $\beta_*$ is the \df{radar reparametrization of $\beta$ according to $\alpha$}\index{radar reparametrization} if 
\begin{equation*}
\beta_*=\setopen \langle t,\vpp\rangle\in\dom\alpha\times\ran\beta \setmid\exists r\in\Q \quad \vpp\in\Lambda^-_{\alpha(t+r)}\cap\Lambda^+_{\alpha(t-r)}\setclose.
\end{equation*}
We say that $\beta$ is at constant radar distance $r$ from $\alpha$ iff 
\begin{equation*}
\ran\beta\subseteq \bigcup_{t\pm r\in \dom\alpha}\Lambda^-_{\alpha(t+r)}\cap\Lambda^+_{\alpha(t-r)}.
\end{equation*}
Let us note that this $r$ can be negative if $\alpha_\tau$ is
decreasing since by this definition $\alpha(t-r)\ll\alpha(t+r)$.

\begin{prop}
\label{prop-rad}
Assume \ax{CONT}.
Let $\alpha$ and $\beta$ be definable timelike curves.
Let $\beta_*$ be the radar reparametrization of $\beta$ according to $\alpha$.
\begin{enumerate} 
\item\label{item-radcont} Then $\beta_*$ is a definable, injective, and continuous curve.
\item\label{item-raddiff} If $\ran \alpha\cup \ran \beta$ is in a
  vertical plane, and $\beta$ is at constant radar distance $r$ from
  $\alpha$, then $\beta_*$ is differentiable.
\item\label{item-radder} Let us further assume that this vertical plane is the $\txPlane$.
Then 
\begin{itemize}
\item[] $\beta'_*(t)=\alpha'(t-r)\phsum\alpha'(t+r)\enskip\text{ iff }\enskip \big(\beta_*(t)-\alpha(t)\big)\upp\phantom{-}\vex$, 
\item[] $\beta'_*(t)=\alpha'(t+r)\phsum\alpha'(t-r)\enskip\text{ iff }\enskip \big(\beta_*(t)-\alpha(t)\big)\upp -\vex$.
\end{itemize}
\end{enumerate}
\end{prop}

\begin{proof}
It is clear that $\beta_*$ is definable.  Without losing generality,
we can assume that $\alpha_\tau$ is increasing, see Lems.\
\ref{lem-tlnice} and \ref{lem-inj}.

\begin{figure}[h!btp]
\small
\begin{center}
\psfrag{df}[r][r]{$\dom f$}
\psfrag{t}[l][l]{$t$}
\psfrag{dg}[r][r]{$\dom g$}
\psfrag{rh}[r][r]{$\ran h$}
\psfrag{g-1tt}[l][l]{$g^{-1}(\tilde{t})$}
\psfrag{htt}[l][l]{$h(\tilde{t})$}
\psfrag{f}[b][b]{$f$}
\psfrag{g}[bl][bl]{$g$}
\psfrag{h}[br][br]{$h$}
\psfrag{a}[tl][tl]{$\alpha$}
\psfrag{b}[tl][tl]{$\beta,\beta_*$}
\psfrag{tt}[l][l]{$\tilde{t}$}
\psfrag{dh}[l][l]{$\dom h$}
\psfrag{rf}[r][r]{$\ran f$}
\psfrag{rg}[l][l]{$\ran g$}
\includegraphics[keepaspectratio, width=0.8\textwidth]{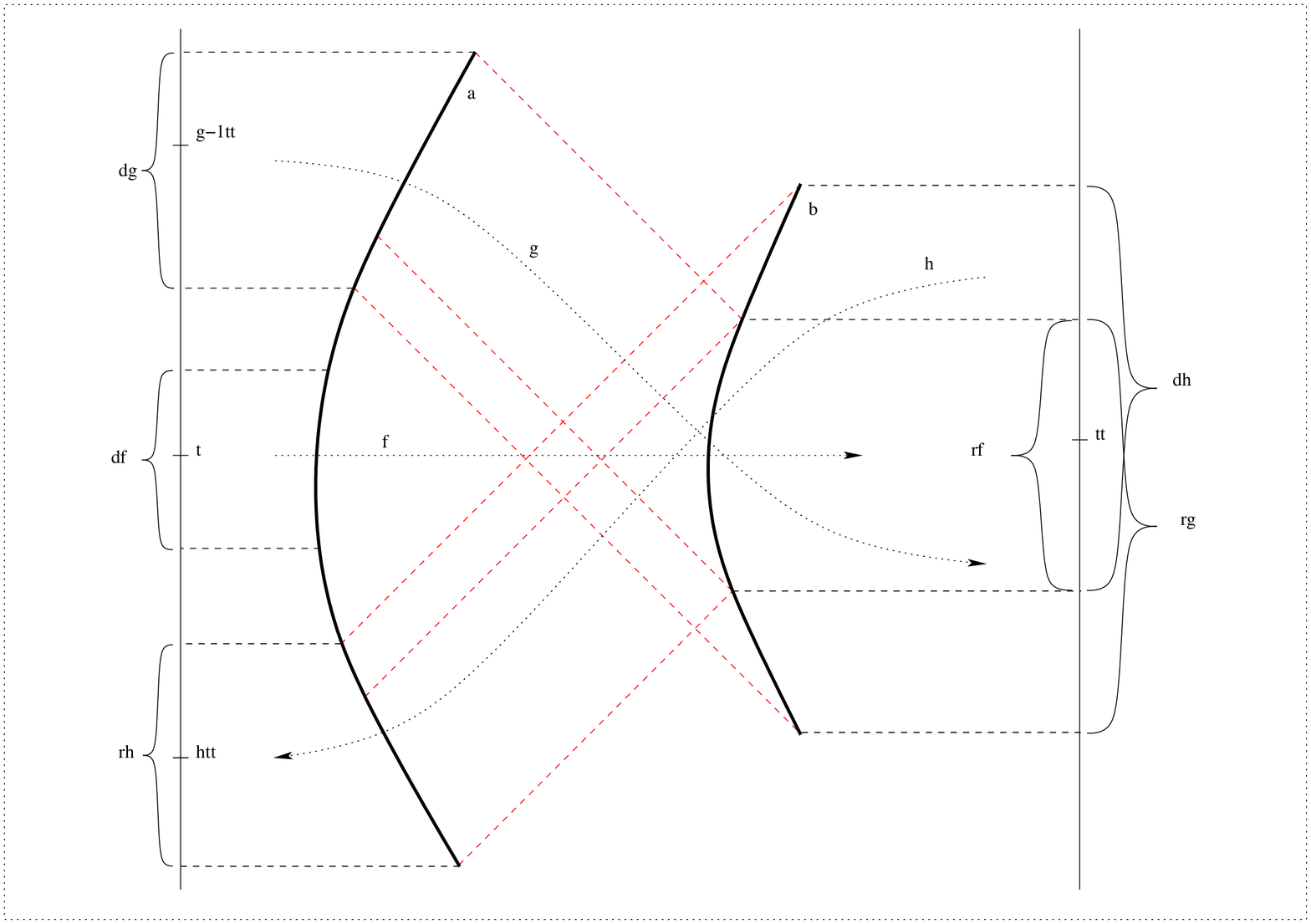}
\caption{\label{fig-radpar} Illustration for the proof of Prop.~\ref{prop-rad}}
\end{center}
\end{figure}

To show that $\beta_*$ is a function, let $\langle
t,\vpp\rangle,\langle t,\vqq\rangle \in\beta_*$.  Then there are
$r,s\in \Q$ such that $\vpp\in
\Lambda^-_{\alpha(t+r)}\cap\Lambda^+_{\alpha(t-r)}$ and $\vqq\in
\Lambda^-_{\alpha(t+s)}\cap\Lambda^+_{\alpha(t-s)}$.  We can assume
that $0\le r\le s$.  Since both $\alpha$ and $\beta$ are timelike
curves, $\vp=\vq$ iff $r=s$.  Therefore, if $\vp\neq\vq$,
$\alpha(t+r)\ll\alpha(t+s)$ and $\alpha(t-s)\ll\alpha(t-r)$.  Thus
$\vq\not\in I^-_{\vp}$ since $I^-_{\vpp}\subset I^-_{\alpha(t+r)}$ and
$I^-_{\alpha(t+r)}\cap\Lambda^-_{\alpha(t+s)}=\emptyset$; and
$\vq\not\in I^+_{\vp}$ since $I^+_{\vpp}\subset I^+_{\alpha(t-r)}$ and
$I^+_{\alpha(t-r)}\cap\Lambda^+_{\alpha(t-s)}=\emptyset$.  Thus
$\vpp=\vqq$ since $\vqq\in I_{\vpp}$ by Lem.~\ref{lem-chord}.

For all $t\in\dom\beta_*$, let $\tilde{t}\in\dom\beta$ such that
$\beta(\tilde{t}\,)=\beta_*(t)$, and let $f:t \mapsto \tilde{t}$ be
the (radar) reparametrization map, i.e., $f\leteq
\beta_*\circ\beta^{-1}$. Then $f$ is injective since if
$\Lambda^-_{\alpha(t_1+r)}\cap\Lambda^+_{\alpha(t_1-r)}\cap\Lambda^-_{\alpha(t_2+s)}\cap\Lambda^+_{\alpha(t_2-s)}\neq
\emptyset$, then $t_1=t_2$ and $r=s$, see \eqref{item-distjcones} in
Lem.~\ref{lem-cau}. Let $g$ and $h$ be the photon reparametrization
maps of $\beta$ according to $\alpha$ and of $\alpha$ according to
$\beta$, respectively. Then $g$, $g^{-1}$ and $h$, $h^{-1}$ are
monotonic and continuous bijections between connected sets, see
Prop.~\ref{prop-ph} and Lem.~\ref{lem-injcont}. It is clear by
the respective definitions, that
\begin{equation*}
f^{-1}(\tilde{t}\,)=t=\frac{g^{-1}(\tilde{t}\,)+h(\tilde{t}\,)}{2}
\end{equation*}
for all $\tilde{t}\in\ran f$, see Fig.~\ref{fig-radpar}. Thus $f^{-1}$
is continuous since both $h$ and $g^{-1}$ are such. It is clear that
$\dom f^{-1}=\ran f=\dom h \cap \ran g$. Thus $\dom f^{-1}$ is
connected since both $\dom h$ and $\ran g$ are such.  Therefore,
$\dom\beta_*=\dom f=\ran f^{-1}$ is also connected and $f$ is
definable and continuous, see Lem.~\ref{lem-injcont}. Hence
$\beta_*=f\circ\beta$ is also continuous; and $\beta_*$ is injective
since both $\beta$ and $f$ are such. So Item \eqref{item-radcont} is
proved.

Now let us prove Item \eqref{item-raddiff}. If $r=0$, then $\beta_*$
is the restriction of $\alpha$ to $\dom\beta_*$ which is connected,
thus it is obviously differentiable. If $r\neq0$, then $\ran
\alpha\cap \ran\beta=\emptyset$. Thus by \eqref{item-phder} in
Prop.~\ref{prop-ph} and Lem.~\ref{lem-repar}, we have that $h$
and $g^{-1}$ are differentiable. Thus $f$ is also differentiable.

To prove Item \eqref{item-radder}, let $\ran \alpha\cup \ran\beta\subset \txPlane$.
By Item \eqref{item-raddiff} of this proposition, $\beta_*$ is differentiable.
It is not difficult to see that 
\begin{equation}\label{eq-radpar}
\begin{split}
&\beta_*(t)=\alpha(t-r)\phsum \alpha(t+r) \text{ iff } \big(\beta_*(t)-\alpha(t)\big)\upp\vex \text{ and}\\
&\beta_*(t)=\alpha(t+r)\phsum\alpha(t-r) \text{ iff } \big(\beta_*(t)-\alpha(t)\big)\upp -\vex
\end{split}
\end{equation}
if $t\in\dom\beta_*$ since $\beta$ is at constant radar distance $r$ from $\alpha$.
By Lem.~\ref{lem-twocurve}, we have that the direction of $\beta_*(t)-\alpha(t)$ cannot change.
Thus it is always the same equation in \eqref{eq-radpar} that holds for $\beta_*$.
Hence Item \eqref{item-radder} follows from Lem.~\ref{lem-phsum} by an easy calculation.
\end{proof}

If $\alpha:\Q\parrow\Q^d$ and $\vp\in\Q^d$, we abbreviate
$\alpha(t)\upp\vp$ for all $t\in\dom\alpha$ to $\alpha\upp\vp$.  We
use analogously the notation $\alpha\upp\beta$ if
$\alpha,\beta:\Q\parrow\Q^d$.  Let $\Df{\bar{\alpha}}\leteq \langle
\alpha_2,\alpha_1,\alpha_3,\ldots,\alpha_d\rangle$ for all $\alpha:
\Q\rightarrow\Q^d$, i.e., the first two coordinates are
interchanged.\index{$\bar{\alpha}$}

\begin{lem}
 \label{lem-accdir}
Assume \ax{CONT}.
Let $\alpha$ be a definable timelike curve.
\begin{enumerate}
\item \label{item-velocdir} Then $\alpha'\upp\vet$ or $\alpha'\upp-\vet$.
\item \label{item-condir} If $\ran\alpha\subset \txPlane$, then $\bar{\alpha}'\upp\vex$ iff $\alpha'\upp\vet$ and $\bar{\alpha}'\upp-\vex$ iff $\alpha'\upp-\vet$.
\item \label{item-accdir} If $\alpha$ is twice differentiable, $\ran\alpha\subset \txPlane$ and $\vo\not\in\ran\alpha''$,
then $\alpha''\upp\vex$ ($\alpha''\upp-\vex$) iff $\alpha'_2$ is increasing (decreasing).
\item \label{item-accsame} If $\alpha$ is twice differentiable, $\ran\alpha$ is in a vertical plane and $\vo\not\in\ran\alpha''$, then $\alpha''(t_1)\upp\alpha''(t_2)$ for all $t_1,t_2\in \dom\alpha$.
\item \label{item-conacc} If $\alpha$ is twice differentiable and $\ran\alpha\subset \txPlane$, then for all $t\in\dom\alpha$, there is a $\lambda_t\in\Q$ such that $\lambda_t\alpha'(t)=\alpha''(t)$.
Furthermore, if $\vo\not\in\ran\alpha''$, the sign of $\lambda_t$ is the same for all $t\in\dom\alpha$ and
\begin{equation}\label{eq-accdir}
\begin{split}
& \lambda_t>0 \quad\text{ iff }\enskip \mbox{$\phantom{-}\bar\alpha'\upp\alpha''$}\\
&\lambda_t<0 \quad\text{ iff }\enskip \mbox{$-\bar\alpha'\upp\alpha''$}
\end{split}
\end{equation}
\end{enumerate}
\end{lem}

\begin{proof}
Item \eqref{item-velocdir} is easy to prove since by
Lem.~\ref{lem-tlnice}, $0\not\in\ran\alpha_\tau$.  Thus by
\ax{CONT}-Darboux's Theorem, we have that $\alpha'_\tau>0$ or
$\alpha'_\tau<0$.

To prove Item \eqref{item-condir}, let us first note that
$\alpha=\langle\alpha_\tau,\alpha_2,0,\ldots,0\rangle$ since
$\ran\alpha\subset \txPlane$.  Therefore,
$\bar\alpha'=\langle\alpha'_2,\alpha'_\tau,0,\ldots,0\rangle$.  Hence
$\bar\alpha'\upp\vex$ iff $\alpha'_\tau>0$, and $\bar\alpha'\upp-\vex$
iff $\alpha'_\tau<0$.

To prove Item \eqref{item-accdir}, let $t\in\dom\alpha$. It is clear
that $\alpha''(t)$ is spacelike or $\vo$ since
$\alpha''(t)\mort\alpha'(t)$ by Prop.~\ref{prop-vmorta}. Thus
$\vo\not\in\ran\alpha''$ iff $\vo\not\in\ran\alpha''_\sigma$. We have
that $\alpha_\sigma=\langle\alpha_2,0,\ldots,0\rangle\in\Q^{d-1}$
since $\ran\alpha\subset \txPlane$. Thus
$\vo\not\in\ran\alpha''_\sigma$ iff $0\not\in\ran\alpha''_2$. Hence
$0\not\in\ran\alpha''_2$. Therefore, by \ax{CONT}-Darboux's Theorem,
we have that $\alpha''_2>0$ or $\alpha''_2<0$. Consequently,
$\alpha''\upp\vex$ iff $\alpha''_2>0$, and $\alpha''\upp-\vex$ iff
$\alpha''_2<0$. Thus, since $0\not\in\ran\alpha''_2$,
$\alpha''\upp\vex$ iff $\alpha'_2$ is increasing, and
$\alpha''\upp-\vex$ iff $\alpha'_2$ is decreasing.

Let us now prove Item \eqref{item-accsame}.
Without losing generality, we can assume that the vertical plane is the $\txPlane$.
By Lem.~\ref{lem-vmon}, we have that $\alpha'_2$ is increasing or decreasing since $\alpha''\circ\mu<0$ iff $\vo\not\in\ran\alpha''$.
Thus Item \eqref{item-accsame} follows by Item \eqref{item-accdir}.

Let us finally prove Item \eqref{item-conacc}.  Since both
$\bar\alpha'(t)$ and $\alpha''(t)$ are Minkowski orthogonal to
$\alpha'(t)$ and are in the $\txPlane$, there is a $\lambda_t\in\Q$ such
that $\bar\alpha'(t)=\lambda_t\alpha''(t)$.  By Items
\eqref{item-condir} and \eqref{item-accdir}, equation
\eqref{eq-accdir} is clear.
\end{proof}

Let $\alpha$ and $\beta$ be timelike curves.
We say that $\beta_*$ is the \df{Minkowski reparametrization of $\beta$ according to $\alpha$}\index{Minkowski reparametrization} if 
\begin{equation*}
\beta_*=\setopen \langle t,\vpp\rangle\in\dom\alpha \times\ran\beta \setmid\big(\vpp-\alpha(t)\big)\mort \alpha'(t)\setclose.
\end{equation*}
We say that \df{$\beta$ is at constant Minkowski distance $r\in\Q^+$ from $\alpha$} iff for all $\vp\in\ran\beta$, there is a $t\in \dom\alpha$ such that $-\mu\big(\vp,\alpha(t)\big)=r$.\index{constant Minkowski distance}

\begin{prop}
\label{prop-mink}
Assume \ax{CONT}.
Let $\alpha$ and $\beta$ be definable timelike curves such that $\alpha$ is well-parametrized, and let $\beta_*$ be the Minkowski reparametrization of $\beta$ according to $\alpha$ such that.
\begin{itemize}
\item[(i)]$ \alpha$ is twice differentiable, and $\vo\not\in\ran\alpha''$.
\item[(ii)] $\ran \alpha\cup \ran \beta$ is in a vertical plane.
\item[(iii)] If $\langle t,\vpp\rangle\in\beta_*$ and
  $\big(\alpha(t)-\vpp\big)\upp \alpha''(t)$, then
  $-\mu\big(\vp,\alpha(t)\big)<-1/\mu(\alpha''(\tau))$ for all
  $\tau\in\dom \alpha$.\item[(iv)] $\beta$ is at constant 
  Minkowski distance $r\in\Q^+$ from $\alpha$.
\end{itemize}
Then $\beta_*$ is a definable timelike curve.
Furthermore,
\begin{equation}\label{eq-muder}
\begin{split}
&\beta'_*(t)=\alpha'(t)+ r\cdot\bar{\alpha}''(t)\text{ iff } \alpha''(t)\upp \big(\beta_*(t)-\alpha(t)\big),\\
&\beta'_*(t)=\alpha'(t)- r\cdot\bar{\alpha}''(t)\text{ iff } \alpha''(t)\upp \big(\alpha(t)-\beta_*(t)\big)
\end{split}
\end{equation}
 if $\ran \alpha\cup \ran \beta\subseteq \txPlane$, $\alpha'\upp\vet$ and $\alpha''\upp\vex$.
\end{prop}

\begin{proof}
It is clear that $\beta_*$ is definable.

To see that $\beta_*$ is a function, let $\langle
t,\vqq\rangle,\langle t,\vpp\rangle\in\beta_*$. Then
$(\vpp-\vqq)\mort\alpha'(t)$. If $\vpp\neq\vqq$, they are
timelike-separated by Lem.~\ref{lem-chord} since
$\vpp,\vqq\in\ran\beta$. Thus, since two timelike vectors cannot be
Minkowski orthogonal, we have that $\vpp=\vqq$. Hence $\beta_*$ is a
function.

Without losing generality, we can assume that the vertical plane that
contains $\ran \alpha\cup \ran \beta$ is the $\txPlane$,
$\alpha'\upp\vet$ and $\alpha''\upp\vex$, see Lems.\ \ref{lem-vmon}
and \ref{lem-accdir}.

Since $\beta$ is at constant Minkowski distance $r$ from $\alpha$, 
\begin{equation}\label{eq-mupar}
\begin{split}
&\beta_*(t)=\alpha(t)+r\cdot\bar\alpha'(t) \text{ iff } \bar\alpha'(t)\upp \big(\beta_*(t)-\alpha(t)\big),\\ 
&\beta_*(t)=\alpha(t)-r\cdot\bar\alpha'(t) \text{ iff } \bar\alpha'(t)\upp \big(\alpha(t)-\beta_*(t)\big)
\end{split}
\end{equation}
if $t\in\dom\beta_*$.

Since $\beta$ is at constant Minkowski distance $r\in\Q^+$ from $\alpha$, we have that $\ran\alpha\cap\ran\beta=\emptyset$.
Hence by Lem.~\ref{lem-twocurve}, we have that the direction of $\beta_*(t)-\alpha(t)$ cannot change.
Thus it is always the same equation in \eqref{eq-mupar} that holds for $\beta_*$.

Since $\alpha$ is twice differentiable, so is $\bar\alpha$.  Thus both
$\alpha+r\cdot\bar\alpha'$ and $\alpha-r\cdot\bar\alpha'$ are
definable differentiable curves.

Now we will show that $\alpha+r\cdot\bar\alpha'$ is a timelike curve
and if $\bar\alpha'(t)\upp \big(\alpha(t)-\beta_*(t)\big)$ for some
$t\in\dom\beta_*$, then $\alpha-r\cdot\bar{\alpha}'$ is also a
timelike curve. It is clear that $(\alpha\pm
r\cdot\bar\alpha')'=\alpha'\pm r\cdot\bar\alpha''$. Let
$t\in\dom\alpha$.
By (5) in Lem.~\ref{lem-accdir}, we have that
$\mu\big(\alpha'(t)+r\cdot\bar\alpha''(t)\big)=\mu\big(\alpha'(t)\big)+r\mu\big(\bar\alpha''(t)\big)$
and
$\mu\big(\alpha'(t)-r\cdot\bar\alpha''(t)\big)=\mu\big(\alpha'(t)\big)-r\mu\big(\bar\alpha''(t)\big)$.
By Thm.~\ref{thm-wp}, we have that $\mu\big(\alpha'(t)\big)=1$.  Thus
$\mu\big((\alpha+ r\cdot\bar\alpha')'(t)\big)>0$. Hence
$\alpha+r\cdot\bar\alpha'$ is a timelike curve. Since
$\alpha'\upp\vet$ and $\alpha''\upp\vex$, we have that
$\alpha''(t)\upp\bar{\alpha}'(t)$ by Lem.~\ref{lem-accdir}. Thus by
assumption (iii) and the fact that $\beta$ is at constant Minkowski
distance $r$ from $\alpha$, we have that $r<-1/\mu(\alpha''(\tau))$
for all $\tau\in\dom\alpha$ if
$\bar\alpha'(t)\upp\big(\alpha(t)-\beta_*(t)\big)$ for some
$t\in\dom\alpha$. Since $\ran \alpha\subseteq \txPlane$, we have that
$\mu(\alpha''(t))=-\mu(\bar\alpha''(t))$. Thus
$\mu(\bar{\alpha}''(t))<1/r$. Consequently, $\mu\big(\alpha'(t)-
r\cdot\bar\alpha''(t)\big)>0$. Hence $\alpha-r\cdot\bar\alpha'$ is
also a timelike curve.

Here we only prove that $\dom\beta_*$ is connected when $\bar\alpha'(t)\upp\big(\alpha(t)-\beta_*(t)\big)$ for some $t\in\dom\beta_*$ because the proof in the other case is almost the same.
Let $t_1,t_2\in \dom\beta_*$, and let $t\in(t_1,t_2)$.
Then $t_1,t_2\in\dom\alpha$, and thus $t\in\dom\alpha$ since $\dom\alpha$ is connected.
Since $\alpha-r\cdot\bar\alpha'$ is a timelike curve and $\alpha'-r\cdot\bar\alpha''\upp\vet$, we have that 
\begin{equation*}
\beta_*(t_1)=\alpha(t_1)-r\cdot\bar{\alpha}'(t_1)\ll\alpha(t)-r\cdot\bar\alpha'(t)
\ll\alpha(t_2)-r\cdot\bar{\alpha}'(t_2)=\beta_*(t_2).
\end{equation*}
Thus by \ax{CONT}-Bolzano's Theorem, there is a $\bar{t}\in\dom\beta$
such that $\big(\beta(\bar{t}\,)-\alpha(t)\big)\mort \alpha'(t)$.
Since $\beta$ is at constant Minkowski distance $r$ from $\alpha$, we
have that $\beta(\bar{t}\,)=\alpha(t)-r\cdot\bar\alpha'(t)$.  Hence
$t\in\dom\beta_*$, as it was required.

Since $\beta_*$ agrees with one of the two timelike curves
$\alpha+r\cdot\bar\alpha'$ and $\alpha-r\cdot\bar\alpha'$ on the
connected set $\dom\beta_*$, we have that $\beta_*$ is also a timelike
curve. Since $\alpha''\upp\vex$ and $\alpha\upp\vet$ we have that
$\alpha''\upp\bar{\alpha}'$. Therefore, by derivation of the equations
of \eqref{eq-mupar}, we have that the derivative of $\beta_*$ is what
was stated in \eqref{eq-muder}.
\end{proof}

\chapter{Why do we insist on using FOL for foundation?}
\label{chp-whyfol}

In this chapter we are going to give a detailed explanation why
FOL is the best logic to be used in foundational works, such as
this one.

\section{On the purposes of foundation}\label{sec-purpfound}

The main purpose of foundation is to get a deeper understanding of
fundamental concepts of a theory by stating axioms about them and
studying the relationship between the axioms and their consequences. There
are three main kinds of question to ask in the course of foundation:
\begin{itemize}
\item What are the consequences of the given axioms? 
\item What axioms are responsible for a certain theorem?
\item How do statements independent from the theory relate to one another?
\end{itemize}

The first one is a usual question of ordinary axiomatic
mathematics. The other two are new kinds of question in foundational
thinking and reverse mathematics.  The third one is meaningful only in
the case of incomplete theories; but there are a lot of incomplete
theories, e.g., any consistent axiom system containing arithmetic is
incomplete by G\"odel's first incompleteness theorem. Moreover, it is
usually reasonable to weaken a complete theory to make it possible to
ask this third type of question. For example, to facilitate studying
the role of the axiom of parallels, Euclid's complete axiom system of
geometry was weakened to an incomplete one. These three kinds of
question are studied in the hope that they will lead to a more refined
and deeper understanding of the fundamental concepts and assumptions
of the given theory. For more details on the role and importance of
foundational thinking, see, e.g., \cite[Introduction]{pezsgo} and
\cite{FriFOM1}.

\section{The success story of foundation in mathematics}
Experience shows that foundational thinking does lead to deeper
understanding. For example, in geometry it clarified the
status of the axiom of parallels and led to the discovery of
hyperbolic geometry. It has been shown by foundation that this axiom is
independent from the other basic assumptions of Euclidean geometry.

Foundation also eliminated Russel's antinomy from set theory and thus
from mathematics. It helped to gain a deeper understanding of many
statements of set theory by providing many other statements which are
weaker, equivalent or stronger according to some axiom system of set
theory, such as the Zermelo--Fraenkel set theory (ZF).\index{ZF} For
example, the axiom of choice is equivalent to Zermelo's well-ordering
theorem, Zorn's lemma and the existence of basis in every vector
space; the Baire Category Theorem, Stone's representation Theorem and
the Banach-Tarski paradox are some of its many consequences; and the
statement ``every subset of real numbers $\R$ is Lebesgue measurable''
is stronger than the negation of the axiom of choice. These results tell us more about
what it means to postulate the axiom of choice or its negation.  And
there are lots of other statements of set theory which are
investigated in this way.

Second-order arithmetic%
\footnote{Let us note that the name second-order arithmetic is
  misleading since it is a two-sorted FOL theory, i.e., it is
  a FOL theory which studies two kinds of individual: sets and
  numbers. }  is also a good example of the successfulness of
foundational thinking.  The main goal of second-order arithmetic is to
investigate how strong a set existence axiom is needed to prove
certain theorems of mathematics, such as the Bolzano-Weierstrass
Theorem or K\"onig's Lemma. For more details, see, e.g.,
Simpson~\cite{simpson}.
 
The examples above show that foundational thinking has been fruitful
in many fields of mathematics. Hence it seems to be a good idea to
apply it in a wider range; for example, in certain fields of physics,
such as relativity theory as suggested by Harvey Friedman~\cite{FriFOM2}, for instance.

There are lots of interesting assumptions, statements and questions of
relativity theory (both special and general), such as the
possibility/impossibility of faster than light motion, the twin
paradox, gravitational time dilation or the existence of closed
timelike curves (i.e., the possibility of time travel), to mention
only a few. There is much hope that foundation will help to clarify
and understand the statuses of these statements and questions as well
as the concepts related to them. This is one of the many reasons why
the relatively large group led by Andr\'eka Hajnal and Istv\'an
N\'emeti have devoted so much effort and enthusiasm to providing
foundation for spacetime theories. See \cite[pp.144 footnote
  137]{leszabo} for a physicist's reflection on some of that.

\section{Choosing a logic for foundation}

To provide logical foundation of any field of science, we have to
choose a formal logic. In the following sections we show that our
choosing FOL is the best possible choice in several
senses. To do so, we compare it to other logics from different
aspects.

Since we would like to treat the physical world as a possible model of
our theory in certain physical interpretation, we need a logic with
semantical consequence relation. Even after this restriction, there
are a great many different logics which we could use for axiomatic
foundation. The two most popular candidates are FOL and
(standard or full) second-order logic. The main difference between
them is that in second-order logic it is possible to quantify over
$n$-ary relations while in FOL we can quantify just over
individuals.

Because of its great expressive power, it would be convenient to use
second-order logic. However, as it will be showed in the forthcoming
paragraphs, its great expressive power is rather a disadvantage.

A main problem with second-order logic is that it contains tacit
assumptions about sets. That is so because unary relations and subsets
are essentially the same things. Hence if we use second-order logic,
we tacitly build set theory into our theory, and that generates
several problems. On the other hand, FOL does not
contain any assumptions about sets.\footnote{Of course, to prove
  nontrivial theorems about FOL, we need some basic set
  theory as a metatheory. However, that does not contradict the fact
  that FOL is free of any hidden assumptions about
  sets.}

V\"a\"an\"anen in \cite{vaananen} says: ``First-order set theory and
second-order logic are not radically different: the latter is a major
fragment of the former.'' In \cite[\S Set theory in Sheep's
 Clothing]{quine}, Quine also argues that second-order logic is none other
than a set theory in disguise. So if we do not want to be burdened (or loaded) by any
hidden assumptions about sets, we cannot use second-order logic for
foundation. The same argument applies to standard higher-order logic and type theory.

\section{Completeness}

Completeness is also a fundamental property of the logic we
choose for foundation since without it we cannot have control over the true
statements in the models of our axioms. FOL is complete
by G{\"o}del's completeness theorem, but second-order logic is
not, see, e.g., \cite[\S IX.1.]{ETF}. That means that the
semantical consequence relation of second-order logic is vague, which
by itself is enough to exclude second-order logic from the list of
possible logics for foundation.

Let us, however, dwell on the vagueness of the semantical consequence
relation of second-order logic. Not just there is no sound and
complete system of derivation rules for second-order logic, but the
set of G{\"o}del numbers of second-order logic validities is not
definable by any second-order logic formula (in the standard model of
arithmetic), see \cite[Thm.41C]{enderton}. Hence it is not just not
recursively enumerable, but it is not at any level of the arithmetical
hierarchy of FOL definable sets of numbers.

The complexity of second-order validities in a language containing one
binary relation symbol is also very high since it cannot be defined by any
higher-order logic formula in the language of Peano arithmetic, and a
 formula of complexity $\Pi_2$ is needed to define it in the language of set theory,
see V{\"a}{\"a}n{\"a}nen~\cite{vaananen}. These results show that the
validity relation of second-order logic is too blurred and vague for
our purposes.

In contrast, the set of FOL validities is recursively enumerable, see
\cite[\S X. Prop.1.6]{ETF}, and the set of consequences of any
recursive enumerable FOL theory is recursively enumerable by
\cite[Thm.35I]{enderton} and \cite[Thm.15.1]{monk}.

\section{Absoluteness}

Naturally, we would like to choose logic $\LL$ such that its semantical
consequence relation ($\models_\LL$) is as independent from set theory
as possible. This property of a logic is called {\it
  absoluteness}. Absoluteness of a logic roughly means that the truth
or falsity of $\mathfrak{M}\models_\LL\varphi$ does not depend on the
entire set theoretical universe, only on the sets required to exist by
some fixed list of axioms (e.g., ZF or a fragment of ZF) and on the
transitive closures of the sets $\mathfrak{M}$ and
$\varphi$ under discussion. For exact definition, see, e.g.,
\cite{barwise}, \cite{vaananen85}.

Let us now see some examples that show what can happen if we use a
non-absolute logic, such as second-order logic. We can formulate the
continuum hypothesis (CH) in second-order logic, see,
e.g., \cite{ETF}, \cite{mosterin},
\cite{Sain}.  Let $\varphi_{CH}$ be a second-order formula expressing
CH and let $\models_2$ be the semantical consequence relation of
second-order logic.

\begin{example}
The answer to the simple question whether $\R\models_2 \varphi_{CH}$
or $\R\not\models_2 \varphi_{CH}$ holds, depends on the model of set
theory we are working in. So it is {\it unknowable}. Moreover, this
dependence is so strong that the answer may alter by moving from the
set theoretical universe $V$ we work in to a transitive submodel of $V$.
\end{example}
\noindent
Let us now see a more general example.
\begin{example}
Let $\varphi_{\infty}$ be a formula of second-order logic expressing
that its model contains infinitely many elements; it is not difficult to write up
such a formula, see \cite[\S IX. 1.3]{ETF}.  Let $\psi$ be the
following formula of second-order logic: $\varphi_{\infty} \rightarrow
\varphi_{CH}$. Then for any infinite structure $\mathfrak{M}$, the
question whether $\mathfrak{M}\models_2 \psi$ or
$\mathfrak{M}\not\models_2 \psi$ holds is also {\it unknowable}.
\end{example}

On the basis of the many independent statements of set theory, we can generate
a great many {\it unknowable} sentences of second-order logic.

Let us note here that a statement being {\it unknowable} and being
independent from a theory does not mean the same. {\it Unknowability}
of a statement means that its validity depends on what class model of
the metatheory we are working in.  So {\it unknowability} is highly
undesirable, while independence is not problematic at all. Moreover, in
foundations, it is useful to study incomplete theories, see
Section~\ref{sec-purpfound}. Hence independence can be useful.

The examples above show that absoluteness is a desired property of a
logic used for foundation. It is important to note that the above 
situations cannot occur in FOL because it is
absolute in a strong sense, i.e., it is absolute in relation to the
Kripke--Platek set theory\footnote{KP consists only the axioms of
  extensionality, foundation, pair, union, and the separation and
  collection schemas restricted to formulas containing only bounded
  quantifiers, see, e.g., \cite{Barwise72} } (KP), which is
considerably weaker than ZF, see, e.g., \cite[Example 2.1.3
]{vaananen85}.

\section{Categoricity}
\label{sec-cat}

First of all let us note that, in logical foundation, the fewer axioms
a theory contains the better it is; so categoricity (and even
completeness) of an axiom system is not a desired property.  Moreover,
searching for strong (e.g., categorical or complete) axiom systems is a
fallacy in foundation of a physical theory since in physics we do not really
know whether an axiom is true or not, we just presume so.

Nevertheless, it is often considered as a great advantage of
second-order logic that it is possible to axiomatize something
categorically within it, i.e., it can capture structures up to
isomorphism. 
The following examples will show that categoricity is rather a
disadvantage as it can obscure things we are interested in.

\begin{example} 
Let us consider the ``nice'' (finite and categorical) second-order
axiomatization $\mathbf{RCF}_2$ of real numbers. Since
$\mathbf{RCF}_2$ is categorical, there is only one model of it ($\R$).
Now we can think that we have captured what we wanted and nothing
else.  However, if we take a closer look, we will see that we can ask
many {\it unanswerable} questions about $\R$. For example, since CH is
independent from set theory, we do not know whether there is or there is not an
uncountable subset $H$ of $\R$ such that there is no bijection between
$H$ and $\R$.  That is inconvenient because there is only one model of
$\mathbf{RCF}_2$. So either $\R\models_2 \varphi_{CH}$ or $\R\models_2
\lnot\varphi_{CH}$ must be valid but we cannot know which one. At
first glance it is not clear at all how a concrete yes-or-no question
can exist without a definite answer? The problem results from the fact
that we have captured one $\R$ in each set theory model. However,
since there are several models of set theory, we have several $\R$'s,
too.
\end{example}

Let us now see another example which shows that we can lose important
information about the model we intend to capture if we use
second-order logic.

\begin{example}
In second-order logic, thanks to its expressive power, we can
formulate an axiom that states that if CH is true, there is an isomorphism
between its model and $\N$, and if CH is false, its model is
isomorphic to the ordered ring of integers ($\Z$):
\begin{equation*}
(\varphi_{CH}\rightarrow \mathfrak{M}\cong
  \N)\land(\lnot\varphi_{CH}\rightarrow \mathfrak{M}\cong \Z).
\end{equation*}
So the axiom system containing the above formula only is categorical, yet
it is {\it unknowable} whether there is a least element of the
structure which is captured up to isomorphism.
\end{example} 

By the trick of the above example we can provide many categorical
axiom systems where some very basic properties (e.g.,
finiteness/infiniteness) about the unique model are {\it unknowable},
see Andr\'eka--Madar\'asz--N\'emeti~\cite[\S Why FOL?]{pezsgo}.

These examples show that categoricity is not at all as good a thing as
it seems to be, and sometimes a non-categorical FOL
axiomatization can provide more information about its several models
than a categorical second-order logic axiomatization can about its
unique model. Second-order logic only makes us believe that we have
one particular object in hand, but in fact, we have many.

\section{Henkin semantics of second-order logic}
\label{sec-henkin}
There are also other semantics of second-order logic in addition to
standard (or full) semantics where all the relations are
present. Henkin further generalized standard semantics and introduced
such ones in which just some of the relations (but at least all the
definable ones) are present such that these relations satisfy certain
requirements, see Henkin~\cite{henkin}.  V\"a\"an\"anen in
\cite{vaananen} argues that if second-order logic is used for
foundation, we cannot meaningfully ask which semantics is being
used. That is so because eventually everything boils down to writing
proofs in the computable inference system of the logic used. So the
standard version of second-order logic cannot be used for foundation,
only the generalized Henkin second-order logic is suitable for this
purpose. The Henkin second-order logic is actually a theory of
many-sorted FOL, so one can only pretend using standard second-order
logic for foundation. Furthermore, when mathematicians are apparently
using higher-order logic, they are actually using Henkin higher-order
one.
 
In the present approach Henkin higher-order logic is considered
absolutely acceptable. It has a completeness theorem and is absolute. So we do not
hesitate to use it when needed. Hence higher-order logic tools are
acceptable and available for us if they are treated with appropriate
caution.

\section{Our choice of logic in the light of Lindstr\" om's Theorem}

So far we have mainly argued for choosing a complete and absolute
logic for foundation, such as FOL. Thus second-order
logic and hence any higher-order logic is too strong for our
aims. However, there are many model-theoretic logics which are
stronger than FOL but weaker than second-order
logic, e.g., weak second-order logic, infinitary logics or
logics with generalized quantifiers, etc. Can any of these logics be
good for our purpose?  So our question can be restated as follows: Is there
any complete and absolute model-theoretic logic stronger than
FOL?\footnote{For the exact definition of  comparing
  the strength of two logics, see \cite{ebbinghaus}.}  To answer this
question, let us recall two properties of abstract model-theoretic
logics. A  logic is {\it compact} iff every set of sentences $\Sigma$
of the logic has a model if all finite subsets of $\Sigma$ have
models. A  logic has the {\it L\"owenheim--Skolem property} iff every
 sentences $\varphi$ of the logic has a countable model if $\varphi$ has
a model.  There is a well-known theorem of Lindstr\"om that
characterizes FOL as the strongest model-theoretic logic
with these properties, see, e.g., \cite{flum}.
\begin{thm}[Lindstr\"om]
FOL is the strongest compact model-theoretic logic with
L\"owenheim--Skolem property.
\end{thm}

However, Lindstr\"om's theorem does not answer our question by itself
since we have required two different properties (absoluteness and
completeness) of a logic to judge it {\it suitable
  for foundation}. A theorem of V\"a\"an\"anen's comes to our aid, see
Cor.2.2.3 in \cite{vaananen85}.

\begin{thm}[V\"a\"an\"anen]
Every absolute model-theoretic logic has the L\"owenheim--Skolem property.
\end{thm}

\noindent
Since completeness implies compactness, by putting Lindstr\"om's and
V\"a\"an\"anen's theorems together, we get the following:

\begin{cor}
FOL is the strongest abstract model-theoretic logic which is complete and absolute.
\end{cor}

This corollary  implies that FOL is the strongest abstract model-theoretic logic suitable for foundation. 

\addcontentsline{toc}{chapter}{Index}
\input{phd.ind}

\addcontentsline{toc}{chapter}{References}
\bibliography{myphd}
\bibliographystyle{plain}

\end{document}